\documentclass[10pt,letterpaper]{article}
\usepackage[includeheadfoot,
            marginratio={1:1,2:3},
            width=412pt,
            height=688pt,]{geometry}

\usepackage{geometry}
\usepackage{fancyhdr}
\usepackage[bookmarks]{hyperref}
\usepackage[latin1]{inputenc}
\usepackage{graphicx}
\usepackage{float}
\usepackage{amsmath}
\usepackage{amsfonts}
\usepackage{amssymb}
\usepackage{stmaryrd}
\usepackage{empheq}
\usepackage{amsthm}
\usepackage[all,cmtip]{xy}
\usepackage{longtable}
\usepackage{simplewick}

\newcommand{\nc}{\newcommand}
\nc{\gl}{\llbracket}
\nc{\gr}{\rrbracket}

\newcommand{\eq}[1]{\begin{equation}
                     \begin{split} #1 \end{split}
                     \end{equation}}
\newcommand{\ul}{\underline}
\newcommand{\ov}{\overline}

\newcommand{\fa}{\hat}

\newcommand{\Lie}{{\cal L}}
\newcommand{\tri}{\hspace{-1pt}\vartriangle\hspace{-1pt}}
\nc{\triN}{\hspace{-1pt}\vartriangle_N\hspace{-1pt}}
\nc{\triNm}{\hspace{-1pt}\vartriangle_{N-1}\hspace{-1pt}}
\nc{\trizw}{\hspace{-1pt}\vartriangle_2\hspace{-1pt}}
\nc{\tridr}{\hspace{-1pt}\vartriangle_3\hspace{-1pt}}
\nc{\trifue}{\hspace{-1pt}\vartriangle_5\hspace{-1pt}}
\nc{\klie}{\fa{\mathcal L}_{\xi}}
\nc{\kl}{\fa{\mathcal L}}
\nc{\ks}{_{KS(\beta)}}
\nc{\cq}{{\cal Q}}

\newtheorem{definition}{Definition}[subsection]
\newtheorem{prop}[definition]{Proposition}

\allowdisplaybreaks[1]
\numberwithin{equation}{section}

\begin{document}

\vspace*{-1.5cm}
\begin{flushright}
  {\small
  MPP-2013-269
  }
\end{flushright}

\vspace{1.5cm}
\begin{center}
{\LARGE
Lie algebroids, non-associative structures \\[0.3cm]
and non-geometric fluxes
}
\vspace{0.2cm}

\end{center}

\vspace{0.35cm}
\begin{center}
 Andreas Deser$^{1,2}$
\end{center}

\vspace{0.1cm}
\begin{center}
\emph{$^{1}$ Max-Planck-Institut f\"ur Physik (Werner-Heisenberg-Institut) \\
   F\"ohringer Ring 6,  80805 M\"unchen, Germany } \\[0.1cm]
\vspace{0.25cm}
\emph{$^{2}$ Institut f\"ur Theoretische Physik, Leibniz Universit\"at Hannover \\
Appelstraße 2, 30167 Hannover, Germany}  \\[0.1cm]

\end{center}

\vspace{1cm}


\begin{abstract}
In the first part of this article, the geometry of Lie algebroids as well as the Moyal-Weyl star product and some of its generalizations in open string theory are reviewed. A brief introduction to T-duality and non-geometric fluxes is given. Based on these foundations, more recent results are discussed in the second part of the article. On the world-sheet level, we will analyse closed string theory with flat background and constant $H$-flux. After an odd number of T-dualities, correlation functions allow to extract a three-product having a pattern similar to the Moyal-Weyl product. We then focus on the target space and the local appearance of the various fluxes. An algebra based on vector fields is proposed, whose structure functions are given by the fluxes. Jacobi-identities for vector fields allow for the computation of Bianchi-identities. Based on the latter, we give a proof for a special Courant algebroid structure on the generalized tangent bundle, where the fluxes are realized by the commutation relations of a basis of sections. As reviewed in the first part of this work, in the description of non-geometric $Q$- and $R$-fluxes, the $B$-field gets replaced by a bi-vector $\beta$, which is supposed to serve as the dual object to $B$ under T-duality. A natural question is about the existence of a differential geometric framework allowing the construction of actions manifestly invariant under coordinate- and gauge transformations, which couple the $\beta$-field to gravity. It turns out that Lie algebroids are the right language to answer this question positively. We conclude by giving an outlook on future directions. 

\end{abstract}

\clearpage

\tableofcontents

\section{Introduction}
\label{ch-intro}

The mathematical and physical description of our perception of space and time is most likely the oldest part of natural sciences, starting with the famous \emph{Elements} of Euclid. It took human creativity about two thousand years to condense at the foundations of our present conception of nature which is determined by two pillars: General relativity as a geometric model of large scale structures like the universe and quantum theory to understand the process of measurement especially in the microscopic world.

The formulation of general relativity 
is intimately connected to major concepts in differential geometry. Spacetime is given by a four-dimensional Lorenzian mani\-fold and gravity manifests itself as curvature. In addition, geometry cannot be separated from matter: Energy and momentum determine the shape of spacetime through Einstein's field equations. Force-free motion happens on geodesics which in turn can be measured as was magnificently confirmed by the observation of light rays in gravitational fields. But despite all the tremendous successes, the theory is not a complete description of nature. It contains black holes as singular solutions and under realistic assumptions singularities cannot be avoided \cite{Hawking:1969sw, 0531.53051}.

Historically in the same period as relativity, the microscopic description of the fundamental constituents of matter was provided by a theory of measurement processes. Non-relativistic quantum mechanics uses the language of functional analysis to model physical observables as operators acting on a Hilbert space. One of its great successes is the precise reproduction of atomic spectra, where one can actually \emph{see} the spectrum of an operator in the mathematical sense. The unification of quantum theory with special relativity revealed new phenomena like particle production and annihilation which finally lead to the formulation of quantum field theory. Treated perturbatively, the latter suffers from infinities coming from loop-corrections where virtual particles can have arbitrary high momenta. These can be cured in so-called \emph{renormalizable} theories, where the infinities can be absorbed into redefinitions of a finite number of parameters, like masses and charges. The effect of this procedure is an energy dependence of the coupling constants, called renormalization group flow. In the last century quantum field theory produced some of the most precise agreements of theoretical physics with experiment like the anomalous magnetic moment of the electron \cite{Barnett:1996hr}, but there are still open questions, the most prominent being the failure to apply quantum field theory to gravity, which is not renormalizable.  

\subsection{Gauge theories and the standard model}
Combining the powerful techniques of quantum field theory with another mathematical field, the geometry of fiber bundles opened the possibility to describe all microscopic interactions of elementary particles in a unified way. There were still severe problems to handle, like the masslessness of all gauge mediators predicted by gauge invariance, which was solved by spontaneous symmetry breaking and the Brout-Englert-Higgs mechanism. In addition, gauge theories with fermions have symmetries which are broken at the quantum level. They are called anomalies and only gauge theories with very specific particle representations are free of these problems.

The result of these efforts was the formulation of the \emph{standard model} of elementary particles. It is a quantum gauge theory with gauge group $SU(3) \times SU(2) \times U(1)$, whose $SU(3)$-factor describes the strong interactions (quantum chromodynamics) and the $SU(2) \times U(1)$ give the electroweak theory, whose $SU(2)$-gauge bosons get massive through the Higgs-mechanism. In addition, the standard model is anomaly-free and renormalizable and its matter content is given by three generations of fermions (quarks and leptons) and a bosonic Higgs field. The latter was found most probably at the LHC (investigations about its precise structure are not finished up to now).

Despite its tremendous successes, the standard model leaves lots of questions unanswered. To mention only a few, first of all there is no explanation to the number of generations and why it has this specific matter content, which ensures for example the absence of anomalies. Furthermore, there are lots of free para\-meters like the values of the gauge couplings and Yukawa couplings. On the more conceptual level, there is the \emph{hierarchy} problem, which roughly states that the quantum corrections to the Higgs mass, quadratically in the cutoff scale, are of many orders of magnitude greater than the uncorrected mass. To cancel these corrections, \emph{supersymmetry} would be one possibility, which however has not been observed up to now. Finally, the standard model does not contain gravity and is therefore incomplete.

\subsection{Non-commutative geometries}
The infinities of quantum field theories can be regularized by a cutoff momentum scale, which means that the theory is not sensible to probe distances smaller than the inverse of the cutoff. Viewing this from the opposite direction, considering a spacetime-geometry where points get meaningless and only finite volumes can be measured might cure divergence problems of quantum field theory, and possibly gravity \cite{Snyder:1946qz}.

 The idea of replacing points of a manifold by cells of finite volume appears already in the quantization of \emph{phase spaces} and the resulting Heisenberg uncertainty relations in quantum mechanics. These are a result of replacing the commutative algebra of observables in classical mechanics by a non-commutative algebra, given by operators on a Hilbert space in quantum mechanics. The generalization of this procedure to arbitrary phase space manifolds is achieved by the theory of deformation quantization \cite{0377.53024, 0377.53025}, where the observables are still classical functions, but their product gets replaced by a non-commutative star product. It can be shown that this procedure on the one hand reproduces results of quantum mechanics like atomic spectra, and on the other hand can be extended to capture arbitrary Poisson manifolds by the famous formality theorems of Kontsevich \cite{1058.53065}.

Switching from the Hamiltonian to the Lagrangian view point, i.e. to the \emph{con\-fi\-gu\-ration space} and its tangent space, the above ideas have been generalized to gauge field theories, but the hope of improving the divergences was disappointed with the discovery of a complicated ultraviolet-infrared mixing behavior and problems to apply renormalization theory (as reviewed in \cite{Szabo:2001kg}). Nevertheless, star-products can also be applied to Einstein gravity \cite{Aschieri:2005yw}, however with similar problems and in addition it is not clear how to get a diffeomorphism covariant theory.

Finally, the attempt to replace the algebra of observables by an abstract non-commutative algebra acting on a Hilbert space, together with a Dirac-operator (which is called a \emph{spectral triple}) and then create an abstract differential calculus resulted in \emph{non-commutative geometry}, initiated by Connes \cite{Connes:1994yd}.  It was possible to reconstruct standard manifolds out of spectral triples, and in addition, combining them with spectral triples corresponding to discrete extra dimensions, Connes, Lott and Chamseddine were able to get the standard model coupled to general relativity \cite{Connes:1990qp}, at least as an effective theory. 

\subsection{String theory and particle physics}
The idea that the concept of a point in the spacetime manifold is not appropriate in quantum field theory and has to be replaced by extended objects is also the basis of \emph{string theory}. The change of paradigm is performed in this case by taking one-dimensional strings as probes of spacetime instead of point particles. 

First discovered through investigations of meson resonances around forty years ago, string theory at the beginning lost much of its interest because it could not describe meson scattering appropriately and it had a number of shortcomings, like an unstable (tachyonic) ground state and could only be quantized consistently in 26 dimensions. This point of view changed dramatically by the discovery of gauge and gravitational degrees of freedom in the massless spectrum of open and closed strings. 
Furthermore, adding supersymmetry projected out the unstable vacuum state and reduced the number of dimensions in which string theories could exist to ten. It was then possible to prove that the low energy effective field theories matched the ten dimensional supergravity theories, and a mechanism to handle anomalies was discovered by Green and Schwarz.

The bridge to our four dimensional world was built by compactifying string theory on a (warped) direct product of an external flat Minkowski space and a compact internal manifold. The geometry of the latter determines the four-dimensional supersymmetry and particle spectrum. This picture received even more attention by the discovery of non-perturbative higher-dimensional dynamical objects, called \emph{D-branes}. The low-energy spectrum of multiple coincident branes contains the degrees of freedom needed for the description of gauge theories and intersecting such stacks of branes supplied another important ingredient: Strings stretching from one stack to another get massless at the intersection and contain representations of massless fermions in their spectrum. This opened the huge field of intersecting brane models as low energy particle physics descriptions \cite{Blumenhagen:2006ci}.

In spite of all the breakthroughs, the above philosophy of modeling four-dimensional physics by the geometry of string theory is far from being complete. To mention a few shortcomings, compactification introduces lots of scalar fields which describe for example the shape of the compact space. These give rise to unwanted phenomenology like long range forces and unobserved particles. These \emph{moduli} fields have to be avoided in low energy physics. As an example, com\-pac\-ti\-fi\-cation on manifolds equipped with additional fields such as the $B$-field and its flux \cite{Grana:2005jc}, introduces scalar potentials whose minimization gives specific values to the moduli. Such solutions are called \emph{vacua} and one of the big problems is the huge variety of the latter which is often referred to as the landscape problem. 

\subsection{Sigma models and T-duality}
The dynamics of a string propagating in (curved) background spaces is described by a two-dimensional quantum field theory. Its ``spacetime'' is given by a time parameter and one spatial coordinate, which describe the propagation of a 1-dimensional string. If this \emph{world sheet} of the string has boundary, we are dealing with open strings and the string is closed otherwise. 

The dynamical fields of the theory are interpreted as the coordinates of a \emph{target space}, which in the bosonic case can have the properties of an ordinary manifold or a supermanifold in the case of additional fermionic coordinates. It is intriguing to relate properties of the two-dimensional field theory of both, classical and quantum nature to the geometrical appearance of the target space. As a prominent example, the renormalization group equations of the two-dimensional theory give rise to the target space Einstein equations and Maxwell equations for additional fields like the $B$-field. Another important example is the presence of a $B$-field in the case of open strings. By considering correlation functions of the two-dimensional field theory it is possible to show that in this case, the target space coordinates do not commute any more and the target space is described more properly by non-commutative geometry \cite{Schomerus:1999ug, Seiberg:1999vs}.

Taking the opposite viewpoint, properties of the target space geometry give rise to statements about the underlying two-dimensional field theory: If the manifold has isometries, it is possible to show that it is equivalent to a \emph{T-dual} field theory which describes another target space with dual metric and $B$-fields \cite{Buscher:1987sk}, and as a consequence also dual field strengths, i.e. dual fluxes.  The simplest example of this \emph{T-duality} is given by a target space with a compact dimension given by a circle. The dual space is then given by a circle with inverse radius, i.e. a propagating string cannot distinguish between a geometry of radius $R$ and radius $l_s ^2 /R$, where $l_s$ is the length of the string.   

One of the advantages of the view-point of two-dimensional field theory is simplicity. In certain cases, one can describe a closed string by two independent sectors, called left- and right-moving. The above T-duality, leading to complicated relations for the dual metrics and $B$-fields on the target space level, can be realized very simply by changing the sign of one of the sectors on the world sheet level. Even more: by using the latter, one is able to arrive at world sheet theories, whose target space geometries cannot be interpreted as ordinary manifolds. In addition, it is not clear how one can define fields like the metric, $B$-field or fluxes on such spaces, which are usually referred to as \emph{non-geometric}. One of the first discoveries in this direction was the fact, that instead of a $B$-field and its corresponding $H$-flux, such configurations are better described by a two-vector $\beta$ and its $R$-flux. The former possibly hints at deep connections to the mathematical field of \emph{Poisson geometry}, as the basic objects of such geometries are so-called Poisson-two-vectors. The mathematical description of geometry and gravity in the presence of Poisson tensors with the help of differential geometric notions could contribute to the understanding of the description of T-dual geometries and is one of the main motivations for this thesis.      


\subsection{Outline}
The beginning of this work is devoted to the mathematical foundations of the structures and language which are used later on. The notion of a Lie algebroid generalizes naturally both Lie algebras and the tangent bundle of a manifold. Lie algebras are characterized by a set of structure constants which determine the commutation relations of the generators of the algebra. On the one hand, making the structure ``constants'' spacetime dependent can be considered as attaching to every point in spacetime a Lie algebra. This results in a bundle of Lie algebras which is a special case of a Lie algebroid. On the other hand, replacing the tangent bundle of a manifold by an arbitrary vector bundle enables us to extend notions of differential geometry like Lie- and covariant derivatives to arbitrary vector bundles. 

After this purely mathematical review we are going to set the physical stage, where the structures introduced before will become important. First it is reviewed how notions of standard geometry have to be extended in the case of open string theory in the presence of magnetic fields. It turns out, that non-commutative geo\-metry is needed in the case of constant $B$-field. The standard Moyal-Weyl star product can be rederived from correlation functions of tachyon vertex operators in open string theory. This is another example of the intimate connection between two-dimensional sigma models and spacetime-geometry. In the case of non-constant $B$-field this can be extended to an even more complicated star-product which is non-commutative and non-associative. 

The following section deals with the second physical aspect, where the previously introduced mathematical structures become important. T-duality is reviewed in detail and its consequences on target space fields like the metric and $B$-field are given. After discussing the most prominent example of an approximate solution with constant $H$-flux, the appearance of so-called non-geometric fluxes is motivated from different viewpoints. This is also the first time where possible connections to Poisson geometry can be discovered in terms of the existence of a bi-vector in order to characterize $Q$- and $R$-fluxes.

After these preparations, the next sections are dedicated to the achievements obtained in this thesis. First, we are going to generalize the open-string non-commutativity of section \ref{ch-quant} to the closed string case and investigate the consequences of T-duality implemented at the world sheet level. It turns out that an $n$-product structure on the algebra of observables on the target space is needed to capture the case of non-vanishing $R$-flux, which is the case completely T-dual to the starting model with constant $H$-flux. 

We then turn to the target space in the next two sections and analyse its geometric structures in the presence of the different kinds of fluxes. First, we use Courant algebroids to re-derive an algebra first discovered by Roytenberg in which the commutation relations of the basis sections are determined by the four different kinds of flux. Using the Jacobi-identities of this algebra, we give Bianchi-identities which are constraints on the fluxes if they are turned on together. The starting point of the following section is a simple but intriguing observation: Similar to the $H$-flux, which is the exterior differential of the $B$-field, the $R$-flux can be expressed as the differential of a bi-vector $\beta$, where the differential now acts on vector fields. This can be made precise by replacing the tangent bundle of the target space by a special Lie algebroid. Even more: A complete differential geometry setup can be constructed to write down a diffeomorphism- and gauge invariant action for the metric and bi-vector $\beta$, which is one of the main results of this thesis. Furthermore, the new action can be related to the well known low energy supergravity actions by a field redefinition which was already discovered by Seiberg and Witten \cite{Seiberg:1999vs}.

We conclude by giving an outlook of potential applications and extensions of the results obtained: After briefly explaining the relation of the previous results to T-duality and non-geometry, the application of Lie algebroids to supersymmetry is shortly discussed and  the construction of similar structures in the case of Filippov 3-algebroids is given.

\section{Mathematical background}
\label{ch-math}

In general relativity or Riemannian geometry, the basic objects are given by a manifold together with a metric. Adding a connection, it is possible to establish a dynamical theory of gravitation. In contrast to this, in symplectic or Poisson geometry, used for the description of phase spaces, a metric is not needed a priori. The basic structure is given by a closed two-form in symplectic geometry and a closed bi-vector in Poisson geometry. In string theory, both a metric and a two-form $B$-field are contained in the massless spectrum and therefore are part of the low energy effective field theory. One of the inherent symmetries of string theory is T-duality which mixes $B$-field and metric components and thus suggests a unified treatment of the two fields. Such a description is provided by generalized geometry \cite{Hitchin:2004ut, Gualtieri:2003dx, Ellwood:2006ya, Grana:2008yw, Berman:2010is}. In some situations, as encountered for example in T-fold backgrounds \cite{Hull:2004in, Hull:2006va, Hull:2007jy,Bouwknegt:2008kd, Schulz:2011ye}, it is even more elegant to use an anti-symmetric bi-vector together with a metric to describe T-dual backgrounds. An immediate question is about the geometric analogues e.g. to Lie derivatives and connections needed to formulate a dynamical theory containing the metric and bi-vector as basic fields and to identify symmetries like diffeomorphism covariance and gauge symmetry.

In this section, we will introduce in detail the basic mathematical concepts which are used in this work to answer this question. It turns out that Lie algebroids \cite{mackenzie1987lie, mackenzie2005general} are the appropriate language to construct a suitable differential geometric framework. We begin with an introduction of the main constructions and describe the most important examples. This is followed by a detailed analysis of differential geometry extended to the Lie algebroid framework. Finally, we will relate these structures to Lie bi-algebroids and Courant algebroids \cite{0885.58030}, which are important tools in generalized geometry.

\subsection{Differential geometry of Lie algebroids}
\label{sec-liealgdiffgeo}
The notion of a Lie algebroid can be imagined as a generalization of two mathe-matical structures. One the one hand, it generalizes Lie algebras in the sense that the constant structure coefficients become spacetime-dependent. In other words, one has to deal with ``bundles of Lie algebras'', as the fibers over spacetime carry the structure of a Lie algebra. On the other hand, Lie algebroids generalize the tangent bundle of a manifold as they still allow its basic operations like the Lie bracket of vector fields. But the latter get replaced by sections in a general vector bundle while one still wants to act on functions. One therefore needs to relate sections in a Lie algebroid to vector fields by a bundle homomorphism called ``anchor-map''.

\subsubsection{Lie algebroids}
\label{sec-Lie}
To distil a mathematical concept out of these different ways of thinking, let us give the definition and provide the most direct consequences. For more details, the reader is referred to \cite{1139.53001}.

\begin{definition} \label{Liealgdef}
Let $M$ be a manifold, $E \rightarrow M$ a vector bundle together with a bracket $[\cdot,\cdot]_E : E \times E \rightarrow E$ satisfying the Jacobi identity, and a homomorphism $\rho : E \rightarrow TM$ called the anchor-map. Then $(E,[\cdot,\cdot]_E,\rho)$ is called  \emph{Lie algebroid} if the following Leibniz rule is satisfied
\eq{
\label{Leibniz}
[s_1, f s_2]_E = f \hspace{1pt}[s_1,s_2]_E + \rho(s_1)(f) s_2  \,,
}
for $f\in {\cal C}^{\infty}(M)$ and sections $s_i$ of $E$. For simplicity, if the context is clear we often denote the Lie algebroid just by the total space $E$.
\end{definition}

\noindent An immediate consequence of the definition is the following homomorphism pro\-per\-ty of the anchor map, which relates the bracket on the Lie algebroid to the Lie bracket $[\cdot,\cdot]_L$ on the tangent space $TM$ of the manifold.

\begin{prop} \label{homoeigenschaft}
\eq{
\rho\left([s_1,s_2]_E\right) = [\rho(s_1),\rho(s_2)]_L
}
\end{prop}

\begin{proof}
On the one hand, using the Leibniz rule, we have for $f\in {\cal C}^{\infty}(M)$
\begin{displaymath}
\left[[s_1,s_2]_E,fs_3\right]_E = f\,\left[[s_1,s_2]_E,s_3\right]_E + \rho\left([s_1,s_2]_E\right)(f)\,s_3
\end{displaymath}
On the other hand, using the Jacobi identity and then the Leibniz rule, we can write the left hand side also as
\begin{displaymath}
 \left[[s_1,s_2]_E,fs_3\right]_E = \left[\rho(s_1),\rho(s_2)\right]_L(f)\, s_3 + f\,\left[[s_1,s_2]_E,s_3\right]_E \;.
\end{displaymath}
Comparing the two and noting that the section $s_3$ was arbitrary, we get the result.
\end{proof}
The bracket $[\cdot,\cdot]_E$ can be generalized to arbitrary alternating multisections in $\Gamma(\wedge^{\bullet}E)$ via defining its action on functions $f,g \in {\cal C}^{\infty}(M)$ and on sections $s_1,s_2 \in \Gamma(E)$ by
\eq{
\gl f, g \gr = 0 \;, \hspace{40pt}
\gl f , s \gr = -\rho(s)\,f \;, \hspace{40pt}
\gl s_1, s_2 \gr = [ s_1, s_2 ]_E \;,
}
and extending it to sections of arbitrary degree $a\in \Gamma(\wedge^k E)$, $b\in \Gamma(\wedge^l E)$ and $c\in \Gamma(\wedge^{\bullet}E)$ by the following relations
\eq{
\label{Gerstenhaber}
  \gl a, b\wedge c \gr &= \gl a, b \gr \wedge c + (-1)^{(k-1)l}\, b\wedge \gl a, c \gr  \;, \\[3pt]
 \gl a,b \gr &= - (-1)^{(k-1)(l-1)}\, \gl b, a \gr \;,
}
which, together with the graded Jacobi identity
\begin{equation} \label{gradedjacobi}
\gl a, \gl b, c \gr \gr = \gl\gl  a, b \gr , c \gr + (-1)^{(k-1)(l-1)}
\hspace{1pt} \gl b, \gl a, c \gr\gr\; ,
\end{equation}
constitute the axioms of a so-called \emph{Gerstenhaber algebra}. In fact, it can be shown \cite{1139.53001} that the Gerstenhaber algebra property of $\Gamma(\wedge^{\bullet}E)$  (where multiplication is given by the wedge-product and Gerstenhaber bracket given by $\gl\cdot,\cdot \gr$) and $(E,[\cdot,\cdot]_E,\rho)$ being a Lie algebroid are equivalent statements.
Another equi\-va\-lent characterization of a Lie algebroid can be given by considering the exterior algebra of its dual space $E^*$. Because it is important for later developments, we want to state it as a proposition.

\begin{prop} \label{d-E}
The exterior algebra $\Gamma(\wedge^{\bullet}E^*)$ is differential graded with differential $d_E : \Gamma(\wedge^k E^*) \rightarrow \Gamma(\wedge^{(k+1)}E^*)$ given by
\eq{
\label{algebroiddiff}
(d_E \,\omega)(s_0, \dots, s_k) =\hspace{10pt}& \sum_{i=0} ^k \;(-1)^i \rho(s_i)\left( \omega(s_0,\dots,\hat{s}_i,\dots,s_k) \right)  \\
 +&\sum_{i<j} \; (-1)^{i+j} \omega \left([s_i, s_j]_E,s_0,\dots,\hat{s}_i,\dots,\hat{s}_j,\dots,s_k \right)
 \;,
}
where $\omega \in \Gamma(\wedge^k E^*)$, $s_i \in \Gamma(E)$ and where the hat stands for deleting the corresponding entry.
\end{prop}

\begin{proof}
We only want to prove the nilpotency of the differential $d_E ^2 = 0$. The essential features of the proof can already be seen in the case of $\alpha \in \Gamma(E^*)$. In this case, applying two times the definition \eqref{algebroiddiff} leads to
\eq{
(d_E ^2 \alpha) (s_0,s_1,s_2) =& \phantom{+}\left(\left[\rho(s_0),\rho(s_1)\right] - \rho([s_0,s_1])\right)\alpha(s_2) \\
&-\left(\left[\rho(s_0),\rho(s_2)\right] - \rho([s_0,s_2])\right)\alpha(s_1) \\
&+\left(\left[\rho(s_1),\rho(s_2)\right] - \rho([s_1,s_2])\right)\alpha(s_0) \\
&+ \alpha([[s_0,s_1],s_2]) -\alpha([[s_0,s_2],s_1]) + \alpha([[s_1,s_2],s_0]) \\
&= 0\;,
}
because of the homomorphism property \ref{homoeigenschaft} and the Jacobi-identity. The proof for general $\alpha$ is similar.
\end{proof}
\vspace{0.5cm}
To illustrate the concepts introduced above, we now give three examples of Lie algebroids. The first will be a trivial example, serving to see that Lie algebroids are a natural generalization of the tangent bundle equipped with the Lie bracket. The second will show that Lie algebroids reduce to Lie algebras if we collapse the base manifold to a point, i.e. the structure functions determining the bracket will be structure constants. Finally the third example will play a fundamental role in the following sections because it replaces the tangent bundle by its dual bundle.

\paragraph{Examples}
\label{subsubsec-ex}

\begin{itemize}
\item As a trivial example consider $ E=(TM,[\cdot,\cdot]_{L},\rho= \textrm{id})$ where the anchor is the identity map and the bracket is given by the usual Lie bracket
$[X,Y]_{L}$ of vector fields. The extension to multi-vector fields in
$\Gamma(\wedge^{\bullet}TM)$ is given by
the relations \eqref{Gerstenhaber}, which result in the so-called Schouten--Nijenhuis  bracket $[\cdot,\cdot]_{SN}$.
The differential on the dual space $\Gamma(\wedge^{\bullet}T^* M)$ is the standard de Rham differential.

\item Consider a Lie group $G$ with corresponding Lie algebra $\mathfrak{g}$. We can define a vector bundle over a single point $\{p\}$ by just taking $\mathfrak{g}$ as single fiber:
\eq{
E : \; \mathfrak{g} \rightarrow \{p\} \,.
}
The bracket for elements $g_i \in \mathfrak{g}$ is given by the Lie bracket on $\mathfrak{g}$:
\eq{
[g_i,g_j]_{\mathfrak{g}} = f^k{}_{ij} g_k
}
where $f^k{}_{ij}$ are the structure constants of $\mathfrak{g}$. The anchor is defined to be the zero map and the corresponding differential on $\Gamma(\wedge^{\bullet}\mathfrak{g}^*) = \wedge^{\bullet}\mathfrak{g}^*$ is given by the so-called \emph{Chevalley-Eilenberg}-operator $d_{CE}$:
\begin{multline}
(d_{CE}\, \alpha)(g_0,\dots,g_k) := 
\sum_{i<j} (-1)^{i+j}\,\alpha \left([g_i,g_j]_{\mathfrak{g}},g_0,\dots,\hat{g_i},\dots,\hat{g_j},\dots,g_k \right)\;,
\end{multline}
where $\alpha \in \wedge^k \mathfrak{g}^*$. Thus, one can imagine a Lie algebroid over a general manifold as a ``bundle of Lie algebras''.
\item Finally, let $(M,\beta)$ be a Poisson manifold
with Poisson structure tensor $\beta =  \frac{1}{2}\,\beta^{ab} \partial_a \wedge \partial_b$.
Note that if $\beta$ is a proper Poisson tensor, it follows that the 3-vector given by $\Theta:=\frac{1}{2}\,[\beta,\beta]_{SN}$ vanishes.
The Lie algebroid is  given by
$E^*=(T^* M,[\cdot,\cdot ]\ks,\rho =  \beta^\sharp)$, where
the anchor $\beta^{\sharp}$ is defined as
\eq{ \label{anchor}
\beta^{\sharp} (dx^a) := \beta^{am}\partial_m \;,
}
for $\{dx^a\}$ a basis of one-forms.
The bracket on $T^*M$ is the \emph{Koszul bracket}, which for one-forms is defined as
\eq{
\label{koszul}
	[\xi,\eta]\ks := L_{\beta^\sharp(\xi)}\eta
 -\iota_{\beta^\sharp(\eta)}\,d\xi \; ,
}
where the Lie derivative on forms is given by $L_{X} = \iota_X \circ d + d \circ \iota_X$ with $d$ the de Rham differential.
The associated bracket for forms with arbitrary degree is again determined by \eqref{Gerstenhaber} and is called the \emph{Koszul--Schouten bracket}.
The corresponding differential on the dual space $\Gamma(\wedge^{\bullet}TM)$ is given in terms of the Schouten--Nijenhuis bracket as
\eq{
  \label{betad}
  d_{\beta} := [\beta,\cdot\,]_{SN} \; .
}
The proof of the Lie algebroid properties can be found e.g. in \cite{1139.53001}. An additional important property is given by the Koszul bracket of exact forms $df,dg$, which gives the relation to the Poisson bracket on $M$:
\eq{
[df,dg]\ks =\, d\{f,g\} \;.
}
\begin{proof}
Writing out the left-hand side gives
\eq{
[df,dg]\ks &= L_{\beta^\sharp(df)}\,dg - \iota_{\beta^\sharp(dg)}d(df) \nonumber \\
&= d(\iota_{\beta^\sharp(df)}\, dg) \nonumber \\
&= d \left( \beta^\sharp (df) (dg)\right) \nonumber \\
&= d\,\{f,g\}\;. \nonumber
}
\end{proof}
To conclude this example, we remark that in the case of a Poisson manifold, the nilpotency of the differential $d_{\beta}$ follows easily from the graded Jacobi identity for the Schouten-Nijenhuis bracket. For $X\in \Gamma(\wedge^k TM)$ we get:
\eq{
d_{\beta}^2 X &= \left[\beta,[\beta,X]_{SN}\right]_{SN} \\
&=\tfrac{1}{2} \left[ [\beta,\beta]_{SN},X \right]_{SN} = 0 \, .
}
It is therefore possible to define a cohomology theory for this differential: The corresponding complex is given by $(\Gamma(\wedge^{\bullet}TM),\wedge,d_{\beta} = [\beta, \cdot]_{SN})$ and similar to de Rham cohomology, the $k^{\textrm{th}}$ \emph{Poisson cohomology} for vector fields is defined by
\eq{ \label{Poissonkohomologie}
\textrm{H}^k _{\beta} (M) := \frac{\textrm{ker}\, d_{\beta} |_{\Gamma(\wedge^k TM)}}{\textrm{im} \, d_{\beta} |_{\Gamma(\wedge^{k-1}TM)}} \;.
}
A similar construction can be performed for a general Lie algebroid, see for example \cite{1139.53001}. One can show that for invertible anchor-map, the corresponding Lie algebroid cohomology is isomorphic to the standard de Rham cohomology.
\end{itemize}
The last fact about Lie algebroids which is important for the following work is the notion of a homomorphism. Let us mention the precise definition:

\begin{definition}
Let $(E_1,[\cdot,\cdot]_{E_1},\rho_1)$ and $(E_2,[\cdot,\cdot]_{E_2},\rho_2)$ be two Lie algebroids over the same base manifold $M$. A bundle homomorphism $\Phi: E_1 \rightarrow E_2$ is called \emph{Lie algebroid homomorphism} if the following compatibility with the two anchors and brackets hold:
\eq{ \label{Liealghomo}
\rho_2 \circ \Phi &= \rho_1 \\
\Phi \left([s_1,s_2]_{E_1} \right) &= \left[ \Phi(s_1),\Phi(s_2) \right]_{E_2} \;,
}
for sections $s_1,s_2 \in \Gamma(E_1)$.
\end{definition}

\noindent The most important consequence of the compatibility relations \eqref{Liealghomo} which we will need is the relation between the two differentials corresponding to $E_1$ and $E_2$. Let us define the \emph{transposed} homomorphism $\Phi^*$ by
\eq{ \label{transposed}
\Phi^* : \Gamma(\wedge^{k}E_2) &\rightarrow \Gamma(\wedge^{k}E_1)  \\
\Phi^* \alpha (s_1,\dots,s_k) &:=\, \alpha \left(\Phi(s_1),\dots,\Phi(s_k)\right) \;,
}
where $s_i \in \Gamma(E_1)$. Then we have the following fact:

\begin{prop} \label{Algebroidhomodiff}
Let $\Phi : E_1 \rightarrow E_2 $ be a Lie algebroid homomorphism and $d_{E_i} : \Gamma(\wedge^{\bullet}E_i ^*) \rightarrow \Gamma(\wedge^{\bullet}E_i ^*)$ be the differential on $E_i$ ($i=1,2$). Then we have the following relation:
\eq{
\Phi^* \circ d_{E_2} = d_{E_1} \circ \Phi^* \;.
}
\end{prop}

\noindent The proof is done by writing the definition \ref{transposed} of the transposed homomorphism into the expression \eqref{algebroiddiff} and then using the properties \eqref{Liealghomo} of the Lie algebroid homomorphism $\Phi$.

\vspace{12pt}

This completes our list of technical facts about Lie algebroids. As we have seen, these objects closely resemble the tangent bundle of a manifold together with the Lie bracket of vector fields. The whole machinery of (Riemannian) differential geometry is based heavily on the latter structures and therefore the immediate question of generalizing differential geometry to Lie algebroids arises. The positive answer to this question will be the topic of the next section.

\subsubsection{Generalizing differential geometry}
\label{subsec-gendiffgeo}

The tangent bundle $TM$ of a manifold $M$ arises as the set of velocity vectors tangent to curves in a manifold. One part of differential geometry is concerned with defining proper derivatives of tensorial quantities along the directions of such velocity vectors. Constructions like Lie/covariant derivatives, curvature and torsion are operators which use vector fields to act on other tensor fields. In most of these constructions, properties of the Lie bracket like the Jacobi identity and the Leibniz rule play an important role. As these are imitated by Lie algebroids, one is tempted to construct Lie derivatives and covariant derivatives with respect to sections in a general Lie algebroid. The possibility of this program and its extension also to torsion and curvature is described for example in \cite{pre06081788} (see also \cite{Gualtieri:2007bq}). In the following we are going to review the most important constructions in this setting in order to prepare the formalism to be used in later sections and to set our conventions.

\paragraph{Lie derivative}

Let $(E,[\cdot,\cdot]_E,\rho)$ be a Lie algebroid. We start with the generalization of the Lie derivative. In the standard case, the Lie derivative of a function $f \in {\cal C}^{\infty}(M)$ with respect to a vector field $X \in \Gamma(TM)$ is given by acting with the vector field on the function: $L_X(f) = X(f) = X^m \partial_m (f)$. For a section $s$ of $E$ it is a priori not clear how to act on functions on the manifold. Only the anchor map relates sections in $E$ to ordinary vector fields. We therefore define the action on functions $f$ by

\begin{definition}
\label{Liedef}
The Lie derivative of a function $f$ on $M$ with respect to a section $s \in \Gamma(E)$ is given by
\eq{
{\cal L}_s (f) := s(f) := \rho(s)(f) \;.
}
\end{definition}

\noindent In the trivial example of $TM$  this definition coincides with the original Lie derivative, by using the identity map as an anchor. For the example of $T^*M$ mentioned in the last section, formula \eqref{Liedef} allows us to define derivatives in the direction of a \emph{one-form}. In particular, for $dx^a$ we have:
\eq{\label{D}
\Lie_{dx^i} (f) = \beta^{\sharp}(dx^i)(f) = \beta^{ij}\partial_j f =: D^i f \;,
}
where we introduced the differential operator $D^a=\beta^{ab}\partial_b$, which can be considered to be a generalization of the standard partial derivative.
Note that \eqref{Liedef} is compatible with the Lie bracket on $E$ due to the following relation for a function $f$:
\eq{\label{consistency}
\bigl [\Lie_{s_1},\Lie_{s_2}\bigr] \hspace{1pt}f = \Lie_{[s_1, s_2]_E} f  \;,
}
which is a simple consequence of the homomorphism property \eqref{homoeigenschaft} of the anchor. Note the similarity to the standard case!

The Lie derivative acting on sections of $E$ is defined using the bracket on the total space $E$, while for sections of the dual $E^*$ the Cartan formula and the associated differential $d_E$ on $E^*$ are employed. Again, the constructions are done in complete analogy to the standard case of the tangent bundle.  Let us formulate the precise statement in the following definition:

\begin{definition}
Let $s,s_i$ be sections of $E$ and $\alpha$ a section of $E^*$. Then the Lie derivative of $s_2$ with respect to $s_1$ and of the dual section $\alpha$ with respect to $s$ are given by
\eq{ \label{Lie2}
\Lie_{s_1} s_2 &= [s_1, s_2]_E \; , \hspace{40pt}
\Lie_s \alpha = \iota_s \circ d_E \, \alpha + d_E \circ \iota_s \, \alpha \;,
}
where the insertion map $\iota$ is defined in the standard way, that is for a local basis $\{s_i\}$ of $\Gamma(E)$ and dual basis $\{s^j\}$ of $E^*$ we have $\iota_{s_i} s^j = \delta^j_i$.
\end{definition}

\noindent The extension of \eqref{Lie2} to multi-sections is given by using the product rule as it is done for the standard Lie derivative. With the definitions \eqref{Liedef} and \eqref{Lie2} it is now easy to prove the following properties of the Lie derivative for a Lie algebroid:

\begin{prop}
The Lie derivative $\Lie$ has the following properties (the first two hold for sections in $\Gamma(\wedge^{\bullet} E^*)$, whereas the last one is valid in both, $\Gamma(\wedge^{\bullet} E)$ and $\Gamma(\wedge^{\bullet} E^*)$)
\begin{gather}
\label{Lierel}
\Lie_s \circ d_E = d_E \circ \Lie_s \;, \\
\iota_{[s_1,s_2]_E} = \Lie_{s_1} \circ \iota_{s_2} - \iota_{s_2} \circ \Lie_{s_1}  \;, \\
\bigl[\Lie_{s_1},\Lie_{s_2} \bigr] = \Lie_{[s_1,s_2]_E} \;.
\end{gather}
\end{prop}

\begin{proof}
The proof can be done in the same way as for the standard Lie derivative.
\end{proof}

\paragraph{Covariant derivative}

The next step is to generalize the notion of connections and covariant differentiation to a Lie algebroid $E$, which was done for example in  \cite{pre06081788}. It turns out that this can be performed in analogy to the standard case. Linearity can be directly generalized, whereas for the Leibniz rule one has to know how to act with sections in $E$ on functions. But this was given in definition \ref{Liedef} by using the anchor map. Thus we have the following:

\begin{definition} \label{defconnection}
Let $(E,[\cdot,\cdot]_E,\rho)$ be a Lie algebroid. A \emph{covariant derivative} on $E$ is a bilinear map $\nabla: \Gamma(E) \times \Gamma(E) \rightarrow \Gamma(E)$ which has the properties:
\eq{\label{covder}
\nabla_{fs_1} s_2 &=  f\hspace{1pt}\nabla_{s_1} s_2\;,  \\
\nabla_{s_1} fs_2 &=  \rho(s_1)(f) s_2 + f\hspace{1pt}\nabla_{s_1} s_2 \,.
}
for $s_1,s_2 \in \Gamma(E),\, f \in {\cal C}^{\infty}(M)$.
\end{definition}

\noindent This definition is extended in the standard way to direct sums and tensor products of Lie algebroids (because these operations can be performed in general for vector bundles), e.g. for sections $s_i \in \Gamma(E)$ we have
\eq{
\nabla_{s_1} (s_2 + s_3) &= \nabla_{s_1} s_2 + \nabla_{s_1} s_3 \;, \\
\nabla_{s_1} (s_2 \otimes s_3) &= \nabla_{s_1} s_2 \otimes s_3 + s_2 \otimes \nabla_{s_1} s_3
\;.}

Following these definitions, as a next step it is possible to obtain curvature and torsion operators. They are given by formulas in analogy to the standard case on the tangent bundle except for the use of the appropriate Lie algebroid bracket.

\begin{definition}\label{defcurv}
Let $(E,[\cdot,\cdot]_E,\rho)$ be a Lie algebroid, $s_i \in \Gamma(E)$ and $\nabla$ a covariant derivative on $E$. Then \emph{curvature} and \emph{torsion} are defined by
\eq{ \label{curv}
R(s_1,s_2)s_3 &=\; \nabla_{s_1}\nabla_{s_2} s_3 - \nabla_{s_2}\nabla_{s_1} s_3
- \nabla_{[s_1,s_2]_{E}}\, s_3  \;, \\[3pt]
T(s_1,s_2) &= \; \nabla_{s_1}s_2 - \nabla_{s_2}s_1 - [s_1,s_2]_{E}  \,.
}
\end{definition}

\noindent To see that these expressions are tensors with respect to standard diffeomorphisms it suffices to check that they are ${\cal C}^{\infty}(M)$-linear in every argument. The reason is  that for a general ${\cal C}^{\infty}(M)$ multi-linear map $ A : \Gamma\bigl( (\otimes^r TM) \otimes (\otimes^s T^*M)\bigr) \rightarrow {\cal C}^{\infty}(M)$ and coordinates $x^i, y^{i '}$ we have
\eq{\label{tensor}
A^{i_1 \dots i_r}{}_{j_1 \dots j_s} &= A(dx^{i_1},\dots,dx^{i_r},\partial_{j_1},\dots, \partial_{j_s}) \\
&= A\left(\tfrac{\partial x^{i_1}}{\partial y^{i_1 '}} dy^{i_1 '},\dots,\tfrac{\partial x^{i_r}}{\partial y^{i_r '}}dy^{i_r '},\tfrac{\partial y^{j_1 '}}{\partial x^{j_1}} \partial_{j_1 '},\dots,\tfrac{\partial y^{j_s '}}{\partial x^{j_s}} \partial_{j_s '}\right) \\
&=\tfrac{\partial x^{i_1}}{\partial y^{i_1 '}} \cdots \tfrac{\partial x^{i_r}}{\partial y^{i_r '}}\tfrac{\partial y^{j_1 '}}{\partial x^{j_1}} \cdots \tfrac{\partial y^{j_s '}}{\partial x^{j_s}} A^{i_1' \dots i_r'}{}_{j_1' \dots j_s'} \;.
}
The proof of  ${\cal C}^{\infty}(M)$-linearity for both expressions in \eqref{curv} is now a straightforward calculation using the definition \eqref{covder} and the Leibniz rule \eqref{Leibniz}. As it illustrates nicely the importance of the properties of a Lie algebroid and its use for generalizing differential geometry, we present as an example the ${\cal C}^{\infty}$-linearity of the curvature operator in definition \ref{curv} in the third argument. This ensures that the operator $R(s_1,s_2)$ is a ${\cal C}^{\infty}(M)$-endomorphism:
\eq{
R(s_1,s_2)(f s_3) &= \nabla_{s_1}\nabla_{s_2} (fs_3) - \nabla_{s_2}\nabla_{s_1} (f s_3)
- \nabla_{[s_1,s_2]_{E}}\, (f s_3) \\
&= f \left(\nabla_{s_1}\nabla_{s_2} s_3 - \nabla_{s_2}\nabla_{s_1} s_3
- \nabla_{[s_1,s_2]_{E}}\, s_3 \right) \\
&\phantom{=} + \rho(s_1)\left(\rho(s_2)(f)\right) - \rho(s_2)\left(\rho(s_1)(f)\right) - \rho\left([s_1,s_2]_E\right)(f) \\
&= f\,R(s_1,s_2)s_3 \;,\nonumber
}
where in the last line we used the homomorphism property \ref{homoeigenschaft} of the anchor in a Lie algebroid. The proof of linearity in the other arguments uses in addition the Leibniz-property in definition \ref{Leibniz}. Similar arguments hold for the torsion operator.

\paragraph{Metric}

Finally, to generalize Riemannian geometry to the case of Lie algebroids, we have to give the definition of a metric. Together with the tensor properties of the curvature and torsion operators, it is then possible to write down actions consisting of scalar quantities, which are composed for example of curvature operators contracted in the right way with the metric, as it is done in ordinary gravity theory. This will be one of our results in later sections.

A \emph{metric} on a Lie algebroid $E$ is an element of $\Gamma(E^*\otimes_{\textrm{sym}}E^*)$ which gives rise to a scalar product for sections in $E$. The latter will be denoted by
\eq{
\label{metric}
\langle s_i, s_j \rangle = g_{ij} \;.
}
Therefore, if we denote by $s^i,s^j$ sections in the dual Lie algebroid $E^*$, we can write the metric $g$ in the form of a symmetric tensor field $g = g_{ij}\, s^i \otimes_{\textrm{sym}} s^j$.

The last concept which we want to introduce in this brief list of differential geometric notions is an analogue of the Levi-Civita connection for Lie algebroids. It turns out that a similar statement about existence and uniqueness as known from Riemannian geometry is possible. For the purpose of reminding the reader about the precise conditions and to see the generalization, let us formulate the following definition:

\begin{definition}
Let $(E,[\cdot,\cdot]_E,\rho)$ be a Lie algebroid with metric $g$, giving rise to the scalar product $\langle \cdot,\cdot\rangle $. Then there exists a unique connection $\mathring \nabla$ having the following properties:
\begin{itemize}
\item vanishing torsion: \hspace{30pt}$\mathring\nabla_{s_1} s_2 - \mathring\nabla_{s_2} s _1 = [s_1,s_2]_E$,
\item metricity: \hspace{70.5pt}$\rho(s_1)\langle s_2,s_3 \rangle = \langle \mathring\nabla_{s_1}s_2, s_3 \rangle + \langle s_2, \mathring\nabla_{s_1} s_3 \rangle $.
\end{itemize}
In analogy to the Riemannian case, it is called \emph{Levi-Civita} connection.
\end{definition}

\noindent The connection $\mathring\nabla$ is characterized by the Koszul formula, whose proof uses the same techniques as in standard Riemannian geometry. Later on, we are going to use it to calculate the connection coefficients (\emph{Christoffel symbols}) for specific Lie algebroids. The Koszul formula allows to express the connection in terms of the anchor and the metric components and is given by
\eq{\label{Koszulformula}
2\hspace{1pt}\bigl\langle \mathring\nabla_{s_1} s_2,s_3\bigr\rangle = & \;s_1\bigl(\langle s_2,s_3 \rangle \bigr)
+ s_2\bigl(\langle s_3,s_1 \rangle \bigr) - s_3\bigl(\langle s_1 ,s_2 \rangle \bigr)  \\[3pt]
&- \langle s_1 ,[s_2,s_3]_E\rangle + \langle s_2,[s_3,s_1]_E\rangle + \langle s_3,[s_1,s_2]_E\rangle  \;.
}
where the action of sections $s_i$ in $E$ is again given by applying the anchor map, as was defined in \ref{Liedef}.

This completes our survey in generalizing notions of Riemannian geometry to Lie algebroids. As we have seen, similar constructions like Lie/covariant derivative, curvature, torsion and Levi-Civita connections are possible. The anchor map is used to define how sections in a general Lie algebroid act on functions and therefore establishes the connection to the tangent bundle of the base manifold. It is important to note that this simple statement has far reaching consequences for the type of differential geometry constructed on a general Lie algebroid: In the most important expressions for physics, like the curvature tensor (expressed in terms of the Levi-Civita connection), the anchor is built in non-trivially. One could go even further by saying that the generalization of differential geometry to the Lie algebroid setting introduces the anchor as a new basic tensor field into the formalism which is of equal importance as the metric (which in Riemannian geometry was the only basic field variable).

\subsection{Lie bi-algebroids and Courant algebroids}
\label{sec-bialg}
One of the most important mathematical structures used in the generalized geometry description of supergravity is that of a \emph{Courant algebroid}. The unified description of one-forms and vector fields by sections in the generalized tangent bundle $TM \oplus T^*M$ needs an extension of the standard Lie bracket to include vector fields and one-forms on an equal footing. Mathematically, it turns out that the combination of a Lie algebroid with its dual into a so-called \emph{Lie bi-algebroid} results in the structure of a Courant algebroid \cite{0885.58030, Roytenberg:01}.

\subsubsection{Lie bi-algebroids}
\label{subsec-bialg}
Consider a Lie algebroid $(E,[\cdot,\cdot]_E,\rho)$ (which we simply call $E$ in the following) over a manifold $M$ and assume the existence of a bracket $[\cdot,\cdot]_{E^*}$ on the dual vector bundle $E^*$ and a bundle homomorphism $\rho^* : E^* \rightarrow TM$ such that $(E^*,[\cdot,\cdot]_{E^*},\rho^*)$ is again a Lie algebroid. According to the last section (especially \eqref{algebroiddiff}), from $E$ we can construct a differential on sections of the dual bundle and similar for $E^*$:
\eq{
d_{E} :& \;\Gamma(\wedge^k E^*) \rightarrow \Gamma(\wedge^{k+1} E^*) \;, \\
d_{E^*}:& \; \Gamma(\wedge^k E) \rightarrow \Gamma(\wedge^{k+1} E)\;.
}
In addition we know that the brackets on the two Lie algebroids can be extended to the algebra of alternating multisections $\Gamma(\wedge^{\bullet}E)$ and $\Gamma(\wedge^{\bullet} E^*)$. Therefore, also differentiating sections in the bracket is a well defined operation. With this information, a Lie bi-algebroid is given by the following definition \cite{Roytenberg:01}:

\begin{definition}
Let $E$ and $E^*$ be two Lie algebroids that are dual as vector bundles. Then the pair $(E,E^*)$ is called a \emph{Lie bi-algebroid} if the differential $d_E$ is a graded derivation of the bracket $[\cdot,\cdot]_{E^*}$ on $E^*$, i.e. the following compatibility condition for sections $s_1 \in \Gamma(\wedge^k E^*), s_2 \in \Gamma(\wedge^{\bullet} E^*)$ holds:
\eq{\label{liebialg}
d_E \left( [s_1,s_2]_{E^*} \right) =\, [d_E\, s_1, s_2]_{E^*} + (-1)^k[s_1,d_E \,s_2]_{E^*} \;.
}
\end{definition}

\noindent To illustrate the concept, let us mention a simple example which is also relevant for later discussions. Let $M$ be a Poisson manifold with Poisson tensor $\beta = \tfrac{1}{2} \beta^{ij}\,\partial_i \wedge \partial_j$. Consider $E=T^*M$ to be the Lie algebroid of the third example in section \ref{sec-Lie}, whose bracket is given by the Koszul-Schouten bracket \eqref{koszul} and $E^* = (T^*M)^* = TM$ the trivial Lie algebroid. As pointed out there, the differential $d_E$ is given by the Schouten bracket  $d_{\beta} = [\beta,\cdot]_{SN}$ whereas $d_{E^*}$ is the standard de Rham differential. From the graded Jacobi identity \eqref{gradedjacobi} of the Schouten--Nijenhuis bracket one easily infers the compatibility relation for sections $s_1 \in \Gamma(\wedge^k TM), s_2 \in \Gamma(\wedge^{\bullet} TM)$
\eq{
d_{\beta}[s_1,s_2] =&\, \left[\beta,[s_1,s_2]_{SN}\right]_{SN} \\
=&\,\left[\,[\beta,s_1]_{SN},s_2 \right]_{SN} + (-1)^k\left[s_1,[\beta,s_2]_{SN}\right]_{SN} \\
=&\,[d_{\beta}\,s_1,s_2]_{SN} + (-1)^k[s_1, d_{\beta}\,s_2]_{SN} \;.
}
Thus we proved that the pair $(T^*M, TM)$ together with the corresponding brackets is a Lie bi-algebroid. In the same way but starting with the trivial Lie algebroid $TM$, one can see that also the pair $(TM,T^*M)$ is a Lie bi-algebroid. This is a special case of the fact that given a Lie bi-algebroid $(E,E^*)$, also the dual pair $(E^*,E)$ is a Lie bi-algebroid.

\subsubsection{Courant algebroids}
\label{subsec-courantalg}
It turns out that to every Lie bi-algebroid, one can associate a Courant algebroid structure. The most important example used in physics is the generalized tangent bundle $TM \oplus T^*M$, which we describe below. But first of all let us give the precise definitions and properties of a Courant algebroid. We closely follow the work \cite{Roytenberg:01}. Consider a vector bundle $E$ together with a bracket $[\cdot,\cdot]_E$. For sections $s_i \in \Gamma(E)$ we define the \emph{Jacobiator} to be the following operator:
\eq{
\mathfrak{J}(s_1,s_2,s_3) = \left[\,[s_1,s_2]_E,s_3\right]_E + \left[\,[s_2,s_3]_E,s_1\right]_E + \left[\,[s_3,s_1]_E, s_2 \right]_E \;.
}
Thus the Jacobi identity of a Lie algebroid is given by the condition $\mathfrak{J} = 0 $. Now we are ready to give the following definition:

\begin{definition}
\label{def-Courantalg}
Let $M$ be a manifold and $E \rightarrow M$ be a vector bundle together with a non-degenerate symmetric bilinear form $\langle \cdot,\cdot \rangle$, a skew-symmetric bracket $[\cdot,\cdot]_E$ on its sections $\Gamma(E)$ and a bundle map $\alpha : E \rightarrow TM$. Then $(E,[\cdot,\cdot]_E,\langle \cdot,\cdot \rangle, \alpha)$ is called a \emph{Courant algebroid} if the following properties hold:
\begin{itemize}
\item For $s_1,s_2 \in \Gamma(E)$: $\alpha([s_1,s_2]_E) = [\alpha(s_1),\alpha(s_2)]_E$.
\item For $s_1,s_2 \in \Gamma(E), \, f\in {\cal C}^{\infty}(M)$:
\eq{
[s_1, f\,s_2]_E = f\,[s_1,s_2]_E + \alpha(s_1)(f)\, s_2 - \tfrac{1}{2}\langle s_1, s_2 \rangle {\cal D} f \;.}
\item $\alpha \circ {\cal D} = 0$, i.e. for $f,g \in {\cal C}^{\infty}(M)$: $\langle {\cal D}f , {\cal D} g \rangle = 0 $.
\item For $e, s_1, s_2 \in \Gamma(E)$:
\eq{
\alpha(e)\langle s_1, s_2 \rangle =\, \langle [e,s_1]_E + \tfrac{1}{2}{\cal D}\langle e, s_1 \rangle, s_2 \rangle + \langle s_1, [e,s_2]_E + \tfrac{1}{2} {\cal D}\langle e, s_2 \rangle \rangle \;.
}
\item For $s_i \in \Gamma(E)$: $\mathfrak{J}(s_1,s_2,s_3) = {\cal D}T(s_1,s_2,s_3)$,
\end{itemize}
where we defined the map ${\cal D}: {\cal C}^{\infty}(M) \rightarrow \Gamma(E)$ by
\eq{
\langle {\cal D} f, s\rangle =\, \alpha(s)(f)\;,
}
and the map $T(s_1,s_2,s_3)$ is a function on the base space $M$ defined by
\eq{
T(s_1,s_2,s_3) =\, \tfrac{1}{6} \langle [s_1,s_2]_E, s_3 \rangle + \textrm{cyclic} \;.
}
\end{definition}

\noindent Before giving the most important example, let us clarify the relation to Lie bi-algebroids. Suppose that we have a Lie bi-algebroid $(E,E^*)$. We want to write down a Courant algebroid structure on the vector bundle direct sum $S = E \oplus E^*$. In order to describe the corresponding bracket and anchor we use the following notation for the objects in $E$ and $E^*$:
\begin{itemize}
\item Lie algebroid $E$: Anchor map $\rho$, sections $X_1, X_2$, Lie derivatives ${\cal L}^E _{X_i}$, exterior derivative $d_E$;
\item Lie algebroid $E^*$: Anchor map $\rho^*$, sections $\xi_1, \xi_2 $, Lie derivatives ${\cal L}^{E^*}_{\xi_i}$, exterior derivative $d_{E^*}$.
\end{itemize}
In addition, we introduce the following two bilinear forms on $S$, of which the first, indexed by ``$+$'' is symmetric and the second indexed by ``$-$'' is antisymmetric:
\eq{
\langle X_1 + \xi_1, X_2 + \xi_2 \rangle_{\pm} = \, \iota_{X_2} \xi_1 \pm \iota_{X_1} \xi_2 \;.
}
To get the structure of a Courant algebroid on $S$, we introduce the map $\alpha$ and the derivative map ${\cal D}$ of the above definition as:
\eq{
\alpha(X + \xi) :=\, \rho(X) + \rho^* (\xi); \quad {\cal D} := \,d_E + d_{E^*} \;.
}
The bracket on $S$ is given by a combination of brackets on the algebroids $E$ and $E^*$ together with Lie derivatives and exterior derivatives of the corresponding duals. More precisely we define for sections $s_i =\, X_i + \xi_i  \in \Gamma(S)$:
\eq{
[s_1, s_2]_S =&\,[X_1, X_2]_E + {\cal L}_{\xi_1} ^{E^*} X_2 - {\cal L}_{\xi_2} ^{E^*} X_1 -\tfrac{1}{2} d_{E^*}\,\langle s_1, s_2 \rangle_- \\
&+ [\xi_1, \xi_2]_{E^*} + {\cal L}_{X_1}^{E} \xi_2 - {\cal L}_{X_2}^E \xi_1 + \tfrac{1}{2} d_E \,\langle s_1, s_2 \rangle_- \;.
}
With these definitions, it is possible to show the following result which gives the connection between Lie bi-algebroids and Courant algebroids \cite{0885.58030, Roytenberg:01}:

\begin{prop}
If $(E,E^*)$ is a Lie bi-algebroid, then $(S,[\cdot,\cdot]_S, \langle \cdot, \cdot \rangle_+, \alpha)$ is a Courant algebroid.
\end{prop}

\noindent It is a natural question if we can decompose a Courant algebroid into the direct sum of two Lie algebroids. This is possible if we can decompose it into so-called \emph{Dirac structures}. Let us give the definition of the latter:
\begin{definition}
\label{def-diracstructure}
Let $(E,[\cdot,\cdot]_E,\langle \cdot, \cdot \rangle, \alpha)$ be a Courant algebroid. A subbundle $D \subset E$ is called a \emph{Dirac structure} if it is maximally isotropic  under $\langle \cdot, \cdot \rangle$ and its sections are closed under $[\cdot,\cdot]_E$.
\end{definition}
Here, a subbundle $D$ is isotropic under $\langle \cdot , \cdot \rangle$, if $\langle s_1, s_2 \rangle = 0 $ for sections $s_i \in \Gamma(D)$. It is called maximally isotropic if it has the maximal possible dimension of a subbundle having the latter property.  If there exists a decomposition of a Courant algebroid into the direct sum of Dirac structures, we have the following result \cite{0885.58030, Roytenberg:01}:

\begin{prop}
Let $(E,[\cdot,\cdot]_E,\langle\cdot,\cdot \rangle, \alpha)$ be a Courant algebroid. If it can be decomposed into transversal Dirac structures, i.e. $E = D_1 \oplus D_2$, then $(L_1, L_2)$ is a Lie bi-algebroid, where $L_2$ can be considered as the dual to $L_1$ with respect to the pairing $\langle \cdot, \cdot \rangle$.
\end{prop}

\noindent To conclude this section, we describe a simple example of a Courant algebroid. Let $TM$ be the trivial Lie algebroid and $T^*M$ be the cotangent bundle together with zero anchor and zero bracket. Then it is easy to show that we get a Courant algebroid $TM \oplus T^*M$ by taking the following bracket which was originally introduced by Courant \cite{0850.70212}:
\eq{
[X_1 + \xi_1, X_2 + \xi_2] = [X_1, X_2] + L_{X_1} \xi_2 - L_{X_2} \xi_1 + \tfrac{1}{2} d\,\left(\xi_1(X_2) - \xi_2 (X_1) \right) \;.
}
The two subbundles $TM$ and $T^*M$ are maximally isotropic subbundles and therefore we have the Lie bi-algebroid $(TM, T^*M)$, which one can check also directly.

\section{Open strings and deformation quantization}
\label{ch-quant}
In the standard formulation of quantum mechanics, classical observables, which are determined by the Poisson $^*$-algebra\footnote{A $^*$-algebra over $\mathbb{C}$ is an algebra ${\cal A}$ with a $\mathbb{C}$-antilinear, involutive antiautomorphism $^*$, i.e. for $\mathfrak{a},\mathfrak{b} \in {\cal A}, \; z,w\in \mathbb{C}$: $\;
(z \mathfrak{a} + w\mathfrak{b})^* = \bar z  \mathfrak{a}^* + \bar w \mathfrak{b}^*\,, \quad (\mathfrak{a}^*)^* = \mathfrak{a} \,, \quad (\mathfrak{a}\mathfrak{b})^* = \mathfrak{b}^* \mathfrak{a}^* \;.$
} of smooth functions on phase space, get replaced by a $^*$-algebra of operators on a Hilbert space. An important consequence of this procedure is the replacement of classical commutativity by quantum non-commutativity of observables. For the simplest phase space $\mathbb{R}^{2n}$ quantum observables are given by (unbounded) operators on the space of square integrable functions. Examples are position and momentum operators and polynomials thereof. The relation of the classical Poisson structure and the commutator of operators is given by the correspondence principle:
\eq{
\{f,g\} \rightarrow \frac{1}{i\hbar}[\hat f, \hat g] \;,
}
where $\hat f, \hat g$ are the operators corresponding to the classical observables $f,g$ (e.g. for $\mathbb{R}^{2n}$ with coordinates $x^i, p_i$ given by polynomials of multiplication operators by $x^i$ and momentum operators $\tfrac{\hbar}{i}\tfrac{\partial}{\partial x^i}$).

The generalization of a quantization procedure to more complicated phase spaces like general Poisson manifolds is done in the most transparent way by using another approach to mathematically describe the algebra of quantum observables. In \emph{deformation quantization} \cite{0377.53024,0377.53025}, the observables get replaced by formal power series of ${\cal C}^\infty(M)[[\hbar]]$ in a deformation parameter (which is physically interpreted as Planck's constant $\hbar$). Non-commutativity of the quantum algebra of observables is encoded by using a star product instead of pointwise multiplication, which is defined by:
\begin{gather}
\star : \, {\cal C}^\infty(M)[[\hbar]] \times {\cal C}^\infty(M)[[\hbar]] \rightarrow {\cal C}^\infty(M)[[\hbar]] \\
f \star g =\, \sum_{k=0}^\infty \, \hbar^k\,C_k(f,g) \;,
\end{gather}
where the bilinear maps $C_k(\cdot,\cdot)$ are determined by the following properties
\begin{itemize}
\item $\star\;$ should be associative,
\item $C_0(f,g) =\, fg\;, $
\item $C_1(f,g) - C_1(g,f) =\, i \{f,g\}\;,$
\item $f \star 1 =\, 1 \star f =\,f\;.$
\end{itemize}
The second axiom states that the product is a deformation of the classical commutative product whereas the third axiom gives again the correspondence principle.

It turns out that the formalism of star products reproduces many of the important results of standard quantum mechanics like the spectrum of the hydrogen atom \cite{0377.53025} and allows for a generalization to phase spaces which are arbitrary symplectic or Poisson manifolds \cite{0812.53034, 1058.53065}.

So far we described the approach of deformation quantization as a \mbox{mathematical} concept introduced independent of physical motivations, which turned out to be suitable to describe quantum properties of observables on general phase spaces. In the last two decades, string theory was able to give a derivation of the form of important star products by considering spacetime itself (or better the world-volume of D-branes) instead of phase spaces. By considering open string theory in the presence of a constant Neveu-Schwarz $B$-field, it was possible to derive the structure of the Moyal-Weyl star product by considering correlation functions of open string vertex operators \cite{Schomerus:1999ug,Seiberg:1999vs}. Later it was possible to generalize this to the case of non-constant but closed $B$-field, resulting in a derivation of Kontsevich's star product (e.g. \cite{1038.53088}). Further generalizations for non-vanishing $H$-flux were studied for example in \cite{Herbst:2001ai, Cornalba:2001sm}.

In this section we first describe the simplest case of constant $B$-field and then sketch the generalization to non-constant and non-closed $B$-fields. It will turn out that this case results in non-associative star products. The mathematical description of the latter is still not understood completely.

\subsection{Constant $B$-field: Non-commutativity}
In \cite{Schomerus:1999ug, Seiberg:1999vs}, it was realized that spacetime seen by the endpoints of open strings in the presence of a NS-NS $B$-field is non-commutative in the sense, that the product of functions depending on the coordinates of the D-brane where the open string is located is given by the Moyal-Weyl star product
\eq{ \label{moyalweyl}
(f \star g) \,(x) =&\, \exp\left(i\theta^{ij}\tfrac{\partial}{\partial x^i}\otimes \tfrac{\partial}{\partial y^j}\right)\,f(x)g(y)|_{x=y}\\
=&\, f(x)g(x) + i\theta^{ij}\partial_if(x)\,\partial_j g(x) + \dots
}
We are going to review the main arguments for this observation and establish an identification of the non-commutativity parameter $\theta^{ij}$ in terms of the (inverse) $B$-field. Consider the following open string sigma model\footnote{We are interested in the classical approximation to open strings, therefore the world sheet will be the upper half-plane $\mathbb{H}$.}
\eq{\label{sigmamodel}
S=\frac{1}{4\pi\alpha'} \int_{\mathbb{H}} \,d^2 z\;\bigl(g_{ij}\;\partial X^i \bar{\partial} X^j - 2\pi  \alpha'\; B_{ij}\;\partial X^i \bar{\partial} X^j \bigr)\;.
}
Varying with respect to $X^i$ leads to the following conditions at the boundary $\partial \mathbb{H}$ due to integration by parts:
\eq{\label{boundarycond}
g_{ij}\left(\bar{\partial} - \partial\right)X^j|_{z=\bar{z}} - 2\pi\alpha'B_{ij}\left(\bar{\partial} + \partial \right)X^j |_{z=\bar{z}} = 0 \;.
}
The exact propagator in two dimensions with these boundary conditions is a standard result in mathematics. Before stating it, let us try to understand special cases of \eqref{boundarycond}. If $B$ is invertible, in the limit $g \rightarrow 0$ (meaning that $B$ is very strong), the boundary conditions become Dirichlet. Therefore the real line behaves like a conducting line in electrostatics. The solution is thus given by the method of image charges, i.e. the sum of two opposite point charges at positions $w$ and $\bar{w}$:
\eq{\label{neumannprop}
\langle X^i(z)X^j(w) \rangle =\,g^{ij}\ln |z-w| - g^{ij}\ln |z-\bar{w}| \;.
}
Conversely, for $ B \rightarrow 0$, we get Neumann boundary conditions, which give a plus sign in the above formula. Later we will consider another special limit to exhibit non-commutativity. But now, let us state the precise propagator:
\eq{\label{exprop}
\langle X^i(z)X^j(w) \rangle =&\,-\alpha'\bigl[g^{ij}\ln|z-w| - g^{ij}\ln|z-\bar{w}| \\
&\,+ G^{ij}\ln|z-\bar{w}|^2 + \frac{1}{2\pi \alpha'}\theta^{ij} \ln \frac{z - \bar{w}}{\bar{z}- w} + C \bigr]\;,
}
where we defined the following fields \cite{Seiberg:1999vs}:
\eq{ \label{fullmetric}
G^{ij} =&\, \Bigl(\frac{1}{g + 2\pi \alpha' B} g \frac{1}{g - 2\pi \alpha' B} \Bigr)^{ij}\;,  \\
G_{ij} =&\; g_{ij} - (2\pi\alpha')^2 (B g^{-1}B)_{ij} \;,\\
\theta^{ij} =& -(2\pi \alpha')^2 \Bigl(\frac{1}{g + 2\pi \alpha' B}B\frac{1}{g - 2\pi\alpha' B}\Bigr)^{ij}\;.
}
Performing the above mentioned limits, we again recover \eqref{neumannprop} if we fix the integration constant $C$ to vanish. We are interested in open strings, so the insertion of the corresponding vertex operators is at the boundary, i.e. the real line. Restricting the propagator to it, we get
\eq{\label{expropR}
\langle X^i(\tau_1) X^j(\tau_2)\rangle =\, -\alpha' G^{ij} \ln (\tau_1 - \tau_2)^2 + \frac{i}{2}\theta^{ij} \epsilon(\tau_1 - \tau_2) \;,
}
where $\epsilon(\tau)$ gives $1$ for positive $\tau$ and $-1$ for negative $\tau$. From this result, we can already see the non-commutativity of spacetime probed by open strings if we calculate the equal-time commutator of two fields $X^i$, $X^j$ at the boundary of the world-sheet:
\eq{\label{ncspacetime}
\bigl[ X^i(\tau), X^j(\tau)\bigr] :=&\, \lim_{\delta_\tau \rightarrow 0} \langle T\left(X^i(\tau)X^j(\tau - \delta_\tau) - X^i(\tau)X^j(\tau + \delta_\tau) \right)\rangle \\
=&\,i\,\theta^{ij} \;,
}
where $T(\dots)$ denotes time ordering on the real line and $\delta_\tau$ is a shift in the world-sheet time. We observe that the result is independent of the world sheet coordinates, i.e. we can interpret it as a real spacetime property.

To finally see directly the Moyal-Weyl product \eqref{moyalweyl} of functions, we look at open string vertex operators in which $\alpha' \rightarrow 0$ with $G$ and $\theta$ kept fixed (which is also called \emph{Seiberg-Witten}-limit, see \cite{Seiberg:1999vs}). Let us assume that the matrix $B$ has full rank equal to the dimension of spacetime. The limit is done by setting the following scaling
\eq{\label{scaling}
\begin{array}{cccc}
\alpha'& \propto & \sqrt{\epsilon} &\rightarrow 0 \;,\\
g_{ij} &\propto &  \epsilon & \rightarrow 0 \;, \\
B_{ij} & \propto & \hspace{10pt}\epsilon^0 \;.& \\
\end{array}
}
Thus, \eqref{fullmetric} can be simplified to
\eq{\label{swlimit}
G^{ij} &=\,-\frac{1}{(2\pi\alpha')^2}\left(B^{-1}\,g\,B^{-1}\right)^{ij} \;,\\
G_{ij} &=\,\;(2\pi \alpha')^2 (B g^{-1}B)_{ij} \;, \\
\theta^{ij} &=\,\; \left(B^{-1} \right)^{ij} \;.
}
In this limit, the propagator on the real line \eqref{expropR} simplifies to
\eq{\label{limprop}
\langle X^i(\tau_1)X^j(\tau_2) \rangle = \, \frac{i}{2}\theta^{ij}\epsilon(\tau_1 - \tau_2)\;.
}
 Using this propagator, we are now able to calculate normal ordered products of field operators. The simplest but non-trivial functional dependence on the coordinates is given by the exponential function, so we consider tachyon vertex operators $V_p \stackrel{\textrm{def}}{=} :e^{ip_i X^i(\tau)}:$. Using the formula for the product of two such operators derived in the appendix, we get:
\eq{
:e^{ip_i X^i(\tau)}::e^{iq_jX^j(0)}: =\, e^{-\frac{i}{2}\theta^{ij}p_i q_j \epsilon(\tau)}\,:e^{ip_i X^i(\tau) + iq_j X^j(0)}: \;.
}
We can iteratively apply this result to a product on $N$ vertex operators of tachyons. As a consequence, the $N$-point tachyon correlation function is given by
\eq{ \label{phasecorrelator}
\langle V_{p_1} \cdots V_{p_N} \rangle =\, \exp \left( -\frac{i}{2}\sum_{1 \leq n < m \leq N} p_{n,i} \theta^{ij} p_{m,j} \,\epsilon(\tau_n - \tau_m) \right) \,\langle V_{p_1} \cdots V_{p_N} \rangle |_{\theta = 0 } \;.
}
In the case at hand, the last factor gives a momentum conservation delta function. Off shell, i.e. without implementing momentum conservation, doing a permutation of vertex operators gives a non-trivial phase factor due to the $\epsilon$-function. It is intriguing that precisely the same phase is reproduced by exchanging two factors of the Moyal-Weyl star product of $N$ exponential functions:
\eq{
e^{i\,p_1 \cdot X} \star \cdots \star e^{i\, p_N \cdot X} =\,\exp \left(-\frac{i}{2}\sum_{1 \leq n < m \leq N  } p_{n,i}\theta^{ij}p_{m,j} \right)\, e^{i\, \left(\sum_{n = 1 } ^N p_n\right) \cdot X} \;.
}
More generally for arbitrary functions (e.g. approximated arbitrary precise by polynomials) depending on the coordinate fields, denoted by $f(X(\tau)),g(X(\tau))$, the following relation can be shown (e.g. \cite{Schomerus:1999ug}) for normal ordered products:
\eq{
:f(X(\tau))::g(X(0)): = :e^{\frac{i}{2}\epsilon(\tau)\theta^{ij}\frac{\partial}{\partial X^i(\tau)}\frac{\partial}{\partial X^j(0)}} f(X(\tau)) g(X(0)):
}
and for the limit $\tau \stackrel{>}{\rightarrow} 0 $ we recover the Moyal-Weyl star-product \eqref{moyalweyl} of the functions depending on $X(0)$.

\subsection{General $B$-field: Non-associativity}

In the previous section we sketched the derivation of the Moyal-Weyl star product by open string perturbation theory. The assumption of constant and therefore closed $B$-field simplified the calculation of correlation functions and resulted in the most basic example of an associative star-product.

Assuming now general fields $g_{ij}(X)$ and $B_{ij}(X)$ in the open string sigma model \eqref{sigmamodel}, correlators of vertex operators can be computed perturbatively by using a \emph{background field expansion} (\cite{Herbst:2001ai, Cornalba:2001sm, Herbst:2003we}) in the following way:
\eq{ \label{backgroundfieldexp}
X^i(z,\bar z) &=\, x_0^i + \zeta^i(z,\bar z) \;, \\
g_{ij}(x_0 + \zeta) &=\, \eta_{ij} - \tfrac{1}{3} R_{ikjl}\zeta^k \zeta^l + \dots \;, \\
B_{ij}(x_0 + \zeta) &=\, B_{ij}(x_0) + \partial_k B_{ij}(x_0)\zeta^k + \tfrac{1}{2} \partial_k \partial_l B_{ij}(x_0) \zeta^k\zeta^l + \dots \;,
}
where $\zeta^i(z,\bar z)$ is a fluctuation around the constant point $x_0$ on the target space, $\eta_{ij}$ is the flat background metric and  $R_{ikjl}$ is the Riemann tensor corresponding to the metric $g_{ij}(x_0)$. The calculation of correlation functions involving the fluctuation $\zeta^i$ was done e.g. in \cite{Herbst:2001ai} and results in the following product ``$\bullet$'' of functions $f,g$:
\eq{ \label{ncastar}
f \bullet g =\, &f\star g - \tfrac{1}{12} \theta^{i l}\partial_l \theta^{j k} \,\left(\partial_i \partial_j f \star \partial_k g + \partial_k f \star \partial_i \partial_j g \right) \\
&+ {\cal O}\left((\partial \theta)^2, \partial^2 \theta \right) \;,
}
where $\star$ denotes the standard Moyal-Weyl product \eqref{moyalweyl} and we expanded up to combinations where only one derivative of the non-commutativity parameter $\theta$ is involved. The latter is determined by the background fields $\eta, B$ introduced in \eqref{backgroundfieldexp} and takes a similar form as in the last section:
\eq{\label{ncatheta}
\theta^{ij} =\, -\bigl(\frac{1}{\eta - B} B \frac{1}{\eta + B} \bigr)^{ij} \;.
}
It was observed in addition that the $\bullet$-product whose terms can be computed order by order perturbatively is not associative :
\eq{
\left(f \bullet g \right) \bullet h &- f \bullet \left(g \bullet h \right) \\
&=\, \tfrac{1}{6}\,\theta^{i m}\theta^{j n}\theta^{k l} \, H_{m n l }\, \partial_i f \star \partial_j g \star \partial_k h + {\cal O}(\partial^2) \;,
}
where we introduced the $H$-flux as the exterior differential of the two-form which we get by inverting the non-commutativity parameter $\theta$: $H:= d\,(\theta^{-1})$. The \mbox{associativity} of the $\bullet$-product is restored in the limit $H \rightarrow 0$.

Finally, in analogy to the non-commutativity of spacetime coordinates \eqref{ncspacetime} we can see the non-associativity of spacetime by introducing a bracket with three arguments which is determined by the Jacobi-identity of the commutator $[\cdot,\cdot]_\bullet$ of the $\bullet$-product:
\eq{
[f,g,h] := \left[f,[g,h]_\bullet\right]_\bullet + \left[g,[h,f]_\bullet\right]_\bullet + \left[h,[f,g]_\bullet \right]_\bullet \;.
}
Evaluating this bracket on the spacetime coordinates $x^i, x^j, x^k$ leads to a very simple expression \cite{Herbst:2001ai}, determined by the measure of non-associativity $H$:
\eq{ \label{thetaflux}
\left[x^i, x^j, x^k \right] =\, \theta^{i m}\theta^{j n}\theta^{k l}\,H_{m n l} \;.
}
To sum up, in the case of general $B$-field backgrounds, it is still possible to calculate star products perturbatively by using a background field expansion. However, the resulting deformation of the classical product is non-associative. The strength of this non-associativity is measured by $d (\theta)^{-1}$.

If we look again at the expression \eqref{ncatheta} for the non-commutativity parameter, we observe that for strong $B$-fields $B \gg \eta$, the expression reduces to $\theta = B^{-1}$, i.e. the non-associativity in \eqref{thetaflux} is determined by the original NS-NS $H$-flux.

It is well known from closed string theory that considering T-duals to configurations with $H$-flux leads to so-called geometric and non-geometric fluxes\footnote{We will introduce some important facts about geometric and non-geometric fluxes in the next section.}. Thus a natural question would be if one can see similar effects of non-associativity also in the closed string case with non-vanishing $H$-flux and its T-dual configurations. We will come back to this question in later sections.

\section{Non-geometric flux backgrounds}
Compactifications of string theory on geometric manifolds serve as a very rich source of interesting mathematics and physics. Calculations of correlation functions in the topological string lead to far reaching insights into topics on the mathe\-ma\-ti\-cal frontier like mirror symmetry or Gromov-Witten theory. On the physical side, they were even more inspiring: Compactification on complex three-dimensional Calabi-Yau manifolds and their orientifolds together with intersecting branes are the ingredients of constructing realistic models of particle physics and cosmology.

However, from a conformal field theory point of view, geometric compactifications are only a subset of possible string theory models, maybe they are not even the generic case. For example, in theories like asymmetric orbifolds, there is no geometric interpretation of the target space. 

In the case of geometric compactifications involving NS-NS three-form flux $H$, in the last decade it became clear that T-duality may connect such compactifications to backgrounds which go beyond the framework of differential geometry \cite{Lust:2010iy, Dabholkar:2002sy, Hellerman:2002ax, Shelton:2005cf, Shelton:2006fd, Hull:2004in, Hull:2006va, Dabholkar:2005ve, Grana:2008yw, Andriot:2011uh, Berman:2012vc, Andriot:2012an, Andriot:2012wx, Condeescu:2012sp}. Starting with an approximate solution to the string equations of motion given by a torus with $H$-flux, T-duality in an isometric direction leads to a twisted torus. Whereas this can still be des\-cribed by geometric $f$-flux, a second T-duality results in an object whose set of transition functions between coordinate charts has to be extended to include also the T-duality group. It is called T-fold and is characterized by the non-geometric $Q$-flux. Even though no isometry direction being left, one can perform a formal third T-duality and there are hints that the resulting space which carries $R$-flux is non-commutative or even non-associative. This is often summarized in the following chain of dualities:
\eq{
H_{abc} \xleftrightarrow{\;\;T_a\;\;}\; f^a{}_{bc}\; \xleftrightarrow{\;\;T_b\;\;} \;Q_c{}^{ab}\; \xleftrightarrow{\;\; T_c\;\;}\; R^{abc} \;.  
}
In this section, we start by describing the action of T-duality when there are isometric directions on the target space manifold. The Buscher rules give a des\-crip\-tion how to get the dual geometry in this case and we study in more detail the duality action for constant background fields in the case of compactification on a torus. We continue by giving a brief motivation for the mathematical structure of non-geometric $Q$- and $R$-fluxes, which is followed by the most prominent example of compactification on a torus with constant $H$-flux and its T-dual versions. Finally, we sketch the four-dimensional effective viewpoint. The aim for having an effective superpotential whose coefficients map bijectively to each other when going from type IIA to type IIB compactifications was one of the first motivations to introduce nongeometric fluxes.

Due to the enormous amount of activity in the field, this will be far from being complete. Important topics like the formalism of double field theory (e.g. \cite{Hull:2009mi, Hohm:2010jy, Aldazabal:2011nj, Hull:2009zb})  will not be touched. The goal of this section is to concentrate on the structural aspects which are important for later sections.

\subsection{Introduction: The Buscher rules}

T-duality is an example of a symmetry which is inherent to string theory due to its fundamental objects, which are one-dimensional instead of the pointlike character of particles. Intuitively, the statement of T-duality is that string theories on T-dual backgrounds are equivalent, i.e. a string cannot distinguish between dual geometries. This has the advantage that one can extract information about exotic geometries by studying string theory on dual backgrounds. The easiest examples are circle compactifications of radii $R$ and $R^{-1}$, which we review briefly in the following. 

Take for simplicity closed bosonic string theory with target space topology $\mathbb{R}^{1,24}\times S^1_R$, i.e. compactification on a circle with radius $R$. Denoting by $n,w$ the momentum and winding numbers, respectively and by $N,\bar{N}$ the oscillator numbers, the mass formula for the closed string spectrum (e.g. \cite{Green:1987sp, Blumenhagen:2013fgp}) is given by: 
\eq{\label{massformula}
M^2 =\, \frac{n^2}{R^2} + \frac{w^2 R^2}{(\alpha') ^2} + \frac{2}{\alpha'}\left(N + \bar{N} - 2\right) \;.
}

\noindent Thus the spectrum is invariant under the T-duality transformation
\eq{\label{tdual1}
R \rightarrow \frac{\alpha'}{R}, \quad n \leftrightarrow w \;,
}

\noindent meaning that the two theories on circles with radii $R$ and $\alpha'/R$ are equivalent. On the level of worldsheet sigma-models, this phenomenon can be seen by Buscher's procedure \cite{Buscher:1987sk, Buscher:1987qj}. To simplify notation, we set $\alpha' = 1$. Consider the sigma model on the sphere $\mathbb{P}^1$ with target space metric $G_{ab}$ and NS-NS $B$-field $B_{ab}$:
\eq{\label{sigmamodel}
S =\, \frac{1}{2\pi} \int_{\mathbb{P}^1} d^2 z \, \left(G_{ab} + B_{ab}\right)\partial X^a \bar{\partial}X^b \;.
}
We observe that the theory is invariant under the change of target space coordinate fields $\delta X^i = \epsilon v^i$, if $v^i$ are the components of an isometry direction ${\bf v}$ of the metric and the $H$-field and $\epsilon$ is a small parameter:
\eq{\label{isometry}
L_{\bf v} G  =&\, 0\;, \\
L_{\bf v} H =&\, 0 \quad \rightarrow L_{\bf v} B =\, d\omega \;,
}
where $H =\, dB$ and $\omega$ is an arbitrary one-form. In addition, if we have a dilaton term
\eq{\label{dilaton}
S_d =\, \frac{1}{2\pi} \int d^2 z \, \phi R^{(2)} \;,
}
the dilaton condition $v^i \partial_i \phi$ has to be satisfied. Having an isometry of this type, we can introduce coordinates $\{\theta := X^0\,,X^a\}$ in such a way that the fields do not depend on the isometry direction $\theta$. Consequently, we are able to introduce an action with trivial dependence on $\theta$ and an additional variable $\tilde{\theta}$ which first appears as a Lagrange multiplier but later turns out to be the isometry coordinate of the dual theory. Together with auxiliary fields $A,\bar{A}$ this action is given by:
\eq{\label{1order}
S_1 =\, \int d^2z \, \bigl[& G_{00} A\bar{A} + \left(G_{0a} + B_{0a}\right) A \bar{\partial}X^a + \left(G_{a0} + B_{a0}\right)\partial X^a \bar{A} \\
&+ \left(G_{ab} + B_{ab}\right)\partial X^a \bar{\partial}X^b + \tilde{\theta}\left(\partial \bar{A} - \bar{\partial}A\right) \bigr] \;.
}
Integrating out the Lagrange multiplier field $\tilde{\theta}$ results in
\eq{
\partial \bar{A} - \bar{\partial} A =\,0\;,
}
i.e. in our case we can take $\bar{A} =\,\bar{\partial}\theta,\, A =\,\partial \theta$ as a solution and we recover the original action \eqref{sigmamodel}. Integrating out instead the auxiliary fields $A$ and $\bar{A}$ results in
\eq{
A =&\, -\frac{1}{G_{00}}\left(G_{a0} + B_{a0}\right) \partial X^a + \frac{1}{G_{00}} \partial \tilde{\theta} \;, \\
\bar{A} =&\,-\frac{1}{G_{00}}\left(G_{0a} + B_{0a}\right) \bar{\partial} X^a - \frac{1}{G_{00}} \bar{\partial} \tilde{\theta} \;,
}
and we arrive at a \emph{dual} sigma model with similar structure if we identify the field $\tilde{\theta}$ with the dual isometry direction $\tilde{X}^0$:
\eq{\label{dualsigmamodel}
\tilde{S} =\, \frac{1}{2\pi} \int_{\mathbb{P}^1} d^2 z \, \left(\tilde{G}_{ab} + \tilde{B}_{ab}\right)\partial \tilde{X}^a \bar{\partial}\tilde{X}^b \;,
}
where the new fields are given by the Buscher rules \cite{Buscher:1987sk, Buscher:1987qj, Rocek:1991ps}:
\begin{gather} \label{Busher}
\tilde{G}_{00} =\, \frac{1}{G_{00}}\;, \quad \tilde{G}_{0a} =\,\frac{B_{0a}}{G_{00}}\,,\quad \tilde{G}_{ab} =\, G_{ab} - \frac{G_{a0}G_{0b} + B_{a0}B_{0b}}{G_{00}} \;, \\
\tilde{B}_{0a} =\,\frac{G_{0a}}{G_{00}}\,, \quad \tilde{B}_{ab} =\, B_{ab} - \frac{G_{a0}B_{0b} + B_{a0}G_{0b}}{G_{00}}\;.
\end{gather}
Noting that the metric on a circle of radius $R$ can be given by $G_{00} = R^2$, we can recover the statement concluded from the mass formula \eqref{massformula} as a special case.

\noindent The above arguments are completely on the classical level. Quantum mechanically, the proof of the equivalence of T-dual theories is subtle (see for example \cite{Rocek:1991ps}). We only want to mention that the transformation laws for the fields receive corrections coming from the Jacobian of integrating out the auxiliary fields. At one loop this leads to a shift in the dilaton:
\eq{
\phi \rightarrow \phi + \ln G_{00} \;.
}
In the following sections we only concentrate on classical aspects and the last result was only listed for completeness.

\subsection{Torus compactification and $O(d,d;\mathbb{Z})$-duality}
\label{sec-torusodd}
We now want to present a different viewpoint on the above insights about the \mbox{action} of T-duality on the background $G$- and $B$-fields (see for example \cite{Zwiebach:2011rg, Blumenhagen:2013fgp}). Using the Hamiltonian formalism for the action \eqref{sigmamodel} of the last section, we will analyse symmetries of the spectrum and rederive the Buscher rules \eqref{Busher}. To apply world-sheet Hamiltonian methods, let us rewrite the sigma model action in local coordinates $\tau \in \mathbb{R}$ and $\sigma \in [0,2\pi]$ :
\eq{\label{hamiltonsigmamodel}
 S =\, -\frac{1}{4\pi} \int_0 ^{2\pi} d\sigma \,\int_{\mathbb{R}} d\tau \,\left(\eta^{\alpha \beta}\partial_\alpha X^\mu \partial_\beta X^\nu G_{\mu\nu} + \epsilon^{\alpha \beta} \partial_\alpha X^\mu \partial_\beta X^\nu B_{\mu \nu} \right)\;,
}
where we again take $\alpha' = 1$, and the conventions for the world-sheet are:
\eq{\label{wsconventions}
\eta^{\alpha \beta} =\,\textrm{diag}(-1,1), \quad \epsilon^{10} =\, -\epsilon^{01} =\,1, \quad \partial_\alpha =\,(\partial_\tau, \partial_\sigma) \;.
}
We consider compactification on a $d$-dimensional torus $T^d$ with flat external space, thus the metric and $B$-field can be split into internal and external parts:
\eq{\label{metricsplit}
G_{\mu\nu} =\,
\begin{pmatrix}
G_{ij} & 0 \\
0 & \eta_{ab}
\end{pmatrix}\,,\quad
B_{\mu\nu} =\,
\begin{pmatrix}
B_{ij} & 0 \\
0 & 0
\end{pmatrix} \;.
}
From now on, we only concentrate on the internal part, i.e. our dynamical fields are $G_{ij}$ and $B_{ij}$. Denoting derivatives with respect to $\tau$ and $\sigma$ with dot and prime, respectively, the canonical momentum of the theory on the internal $d$-torus is given by:
\eq{\label{canmomentum}
2\pi \Pi_i = G_{ij} \dot{X}^j + B_{ij}X^{'j}\;.
}
To write down a simple expression for the Hamiltonian of the theory, we assume $G_{ij}$ and $B_{ij}$ to be constant. In order to compare this situation with the pre\-ce\-ding section, we note that equally simple calculations could be performed if the fields had isometries in some coordinate directions of the torus. In this sense the two sections are equivalent, or in other words, if there is no isometric coordinate direction, the discussion in this section is not applicable.

The world-sheet Hamiltonian density $\mathfrak{h}$ is given by a Legendre transform of the Lagrangian in \eqref{hamiltonsigmamodel} and can be written in the following convenient form:
\eq{\label{Hamilton}
\mathfrak{h} =\, \frac{1}{4\pi} (X', 2\pi \Pi)\, \mathcal{H}(G,B)
\begin{pmatrix}
X' \\
2\pi \Pi
\end{pmatrix}\;,
}
where we have introduced the $2d$-dimensional \emph{generalized metric} $\mathcal{H}$, given in terms of the metric and $B$-field:
\eq{ \label{genmetric}
\mathcal{H} =\,
\begin{pmatrix}
G - BG^{-1}B & BG^{-1} \\
-G^{-1}B & G^{-1}
\end{pmatrix} =\,
\begin{pmatrix}
1 & B \\
0 & 1
\end{pmatrix}
\begin{pmatrix}
G & 0 \\
0 & G^{-1}
\end{pmatrix}
\begin{pmatrix}
1 & 0 \\
-B & 1
\end{pmatrix} \;.
}
Now, to write down the Hamiltonian, we recall the mode expansions of the internal closed string coordinates if we compactify on a $d$-torus (e.g. \cite{Blumenhagen:2013fgp}). The identification of the torus coordinates are:
\eq{\label{modes}
X^i \sim  X^i + 2\pi w^i =\, X^i + 2\pi \sum_{k = 1}^d n^k e^i_k\;,
}
where the $\{e_k\}_{k=1,\dots, d}$ generate the basis of the $d$-dimensional lattice ${\bf \Lambda}_d$ defining the torus and $n^k \in \mathbb{Z}$. The constant $B$-field does not contribute to the equations of motion and thus the mode expansion of the compact coordinates is:
\eq{\label{mode1}
X^i = x_{0}^i + p^i \tau + w^i\sigma + i\sum_{n\neq 0} \frac{1}{n} \left(\alpha_n e^{-in(\tau - \sigma)} + \bar{\alpha}_n e^{-in(\tau + \sigma)}\right) \;.
}
Comparing with the canonical momentum \eqref{canmomentum}, we can read off the center of mass momentum $\pi_i$:
\eq{\label{centerofmassmom}
\pi_{i} =\,\int_0 ^{2\pi} d\sigma \, \Pi_i =\, G_{ij} p^j + B_{ij}w^j \;.
}
It is this quantity which generates translations and as a consequence of the compactness of the internal directions, it has to be quantized. More precisely it takes values in the dual lattice ${\bf \Lambda}_d ^*$ spanned by the dual basis vectors $\{e^{*i}\}$, defined by $e^{*i}(e_j) =\,\delta^i_j$. Thus the center of mass momentum can be expanded in terms of the dual basis by:
\eq{\label{quantmom}
\pi_i =\,\sum_{k=1} ^d m_k e^{*k}_i , \qquad m_k \in \mathbb{Z}\;.
}
The expansion \eqref{mode1} can be split into the sum of left- and right moving sectors  by introducing zero modes $x^i_{0L}, x^i_{0R}$ and $p^i_L, p^i_R$: $X^i = \tfrac{1}{\sqrt{2}}\left(X^i_L + X^i _R\right)$, where
\eq{\label{mode2}
X_L ^i =&\, x_{0L}^i + p_L^i(\tau + \sigma) + i\sum_{n\neq 0}\,\frac{1}{n}\alpha_n ^i e^{-in(\tau + \sigma)}\;, \\
X_R ^i =&\, x_{0R}^i + p_R ^i(\tau - \sigma) + i\sum_{n\neq 0}\,\frac{1}{n}\bar{\alpha}_n ^i e^{-in(\tau - \sigma)}\;.
}
The momentum modes $p_{L/R}^i$ (which also determine the mass formula) can be expressed in terms of the center of mass momentum $\pi_i$ and winding $w^i$ by using the relation \eqref{centerofmassmom}. They are given by:
\eq{\label{plpr}
(p_i)_{L/R} =\,\frac{1}{\sqrt{2}}\left[\pi_i \pm \left(G_{ij} \mp B_{ij}\right)w^j\right] \;.
}
Putting together the results \eqref{Hamilton}, \eqref{modes} and \eqref{quantmom}, we are able to express the Hamiltonian as follows:
\eq{\label{Hamiltonf}
H =\, \int_{0}^{2\pi} d\sigma \, \mathfrak{h} = \frac{1}{2} \mathbf{k}^t \mathcal{H}(G,B) \mathbf{k} + ... \;,
}
where the dots consist of terms coming from the oscillator modes and the integer-valued vector $\mathbf{k}^t =\, (n^i,m_i)$ is given by the center of mass momentum- and winding quantum numbers introduced in \eqref{modes} and \eqref{quantmom}, respectively.

For completeness, we want to express the level-matching condition in terms of the quantum numbers $n^i$ and $m_i$. The mass formulas for the left and right moving sectors are given by:
\eq{ \label{massformlr}
m_{L/R}^2 =\,p_{L/R}^2 + 2(N_{L/R} -1)\;,
}
where $N_{L/R}$ are the eigenvalues of the corresponding number operators. Therefore level-matching can be written as $m_L^2 =\, m_R ² $, and using the expressions \eqref{plpr} to write down the masses in terms of the previously introduced quantum numbers leads to the constraint:
\eq{ \label{levelmatch}
N_R - N_L =\,  m_i n^i =\, \frac{1}{2}\,\mathbf{k}^t \eta \mathbf{k} \;,
}
where in the last term we have defined the metric $\eta$, which is a $2d \times 2d$-matrix and explicitely given by:
\eq{\label{oddmetric}
\eta = \,
\begin{pmatrix}
0 & 1 \\
1 & 0
\end{pmatrix}\;.
}
For every compactification on a $d$-torus, the vectors having components $p_{L/R} ^i$ lie on an even, self dual lattice\footnote{We now consider $2d$-dimensional vectors $(p_L, p_R) \in {\bf \Gamma}_{d,d}$ in contrast to $n^k e_k ^i \in {\bf \Lambda}_d$.} ${\bf \Gamma}_{d,d}$ \cite{Blumenhagen:2013fgp}. It is a mathematical fact that every such lattice can be generated by an $O(d,d;\mathbb{R})$-transformation out of the lattice coming from the background with flat metric and zero $B$-field. Equivalent theories are created by rotations of the momenta by the maximal compact subgroup $O(d,\mathbb{R})\times O(d,\mathbb{R})$ and therefore the moduli space of inequivalent torus compactifications is \emph{locally} isomorphic to the coset:
\eq{
\mathcal{M_{\textrm{loc}}} =\, \frac{O(d,d;\mathbb{R})}{O(d;\mathbb{R})\times O(d;\mathbb{R})} \;.
}
We will see in the following that this can be further reduced by \emph{discrete} transformations. To show this, we analyse the transformations under which the Hamiltonian \eqref{Hamiltonf} and the level matching condition \eqref{levelmatch} are invariant. For the latter, mapping the vector $\mathbf{k}$ by an integer valued matrix $M$ (taking the conventions $\tilde{\mathbf{k}} =\, M^t \mathbf{k}$), it is easy to see that we get:
\eq{ \label{oddbed}
M \eta M^t =\, \eta \;.
}
This is the definition of the group $O(d,d; \mathbb{Z})$. To get useful relations and introduce notation which became standard, let us write a general $O(d,d;\mathbb{Z})$-matrix in the form:
\eq{\label{oddmatrix}
M =\,
\begin{pmatrix}
a & b \\
c & d
\end{pmatrix}\;,
}
with $d$-dimensional integer valued matrices $a,b,c,d$. From the condition \eqref{oddbed}, we infer the relations
\eq{\label{oddbed2}
a^tc + c^ta =&\, b^td + d^t b =\, 0\,, \quad a^td + c^tb =\,1\;, \\
ab^t + ba^t =&\, cd^t + dc^t =\, 0\,, \quad ad^t +bc^t =\, 1\;,
}
where the second line is not an independent condition but is equivalent to the first one. We listed both for later reference.

As a next step let us analyse the invariance of the Hamiltonian \eqref{Hamiltonf}, i.e. the invariance of the spectrum of the theory. To shorten notation let us introduce a combination of the metric and $B$-field which turns out to be more natural than the two fields separately:
\eq{ \label{E}
E_{ij} =\, G_{ij} + B_{ij} \;.
}
Recalling \eqref{Hamiltonf}, a necessary condition for the invariance of the spectrum is given by:
\eq{\label{specinv}
\tilde{\mathbf{k}}^t \mathcal{H}(\tilde{E}) \tilde{\mathbf{k}} =\,\mathbf{k}^t \mathcal{H}(E)\mathbf{k} \;,
}
where we also assumed that the generalized metric $\mathcal{H}$ depends on a transformed $\tilde{E}_{ij}$. Thus we get the condition for the transformed generalized metric
\eq{
\mathcal{H}(E) =\, M \mathcal{H}(\tilde{E}) M^t \;.
}
A short calculation shows that this can be achieved by the following action of the $O(d,d;\mathbb{Z})$-matrix $M$ on the matrix $\tilde{E}_{ij}$:
\eq{ \label{oddaction}
E = (a\tilde{E} + b)(c\tilde{E} + d)^{-1} =:
\begin{pmatrix}
a & b \\
c & d
\end{pmatrix}\,\tilde{E} \;.
}

\noindent As a result, this action of the group $O(d,d;\mathbb{Z})$ on the background matrix $\tilde{E}$ leaves the spectrum invariant and hence we also get equivalent theories. The moduli space of inequivalent torus compactifications is therefore:
\eq{
\mathcal{M} =\,\frac{O(d,d;\mathbb{R})}{O(d;\mathbb{R})\times O(d;\mathbb{R})} / O(d,d;\mathbb{Z})\;.
}
Let us now decompose $O(d,d;\mathbb{Z})$ in order to see how T-duality transformations in different torus directions are contained. Note that we take constant $G$ and $B$, and thus every direction is an isometry in the sense of the previous section about the Buscher rules.

The generators of $O(d,d;\mathbb{Z})$ consist of three groups which we describe separately.
\begin{itemize}
\item Discrete $B$-shifts:
\eq{
g_b =\,
\begin{pmatrix}
1 & b \\
0 & 1
\end{pmatrix}
\;, \quad b^t = \, -b\;. }
These matrices correspond to integer-valued shifts of the $B$-field, as one can see using the general action of $O(d,d;\mathbb{Z})$ given in \eqref{oddaction}. Thus they are the discrete analogue of $B$-transformations in generalized geometry.
\item Discrete change of bases:
\eq{
g_A =\,
\begin{pmatrix}
A & 0 \\
0 & (A^t)^{-1}
\end{pmatrix} \,, \quad A \in \textrm{GL}(d;\mathbb{Z})\;.
}
As one infers again from \eqref{oddaction}, these transformations are the discrete analogues of a change of bases. For example a separate reshuffling of momentum and winding quantum numbers belongs to this class of transformations.
\item T-dualities
\eq{ \label{tduality}
g_T =\,
\begin{pmatrix}
1 - e_i & e_i \\
e_i & 1 - e_i
\end{pmatrix}\,, \quad (e_i)_{mn} =\, \delta_{im}\delta_{in} \,,}
where $1$ denotes the identity matrix in $d$ dimensions and $e_i$ is the matrix with only non-trivial entry at the diagonal element $(ii)$. These transformations are the T-duality transformations, which we are going to analyse in more detail in the following.
\end{itemize}
Let us now concentrate on the last group of transformations. We want to show that they reproduce the Buscher rules \eqref{Busher}. For simplicity we again take constant background fields. The T-dualities \eqref{tduality} act on the  generalized metric ${\cal H}$ given in \eqref{genmetric} by conjugation:
\eq{ \label{dualityaction}
\tilde{\cal H} =\, g_T^t {\cal H} g_T =\, g_T {\cal H} g_T \;.
}
From the new generalized metric we can now read off the transformed metric and $B$-field. Instead of giving the proof of the equivalence of this operation to the Buscher rules in full generality, we take a simple example\footnote{We closely follow the review of this well known example, as given in \cite{Andriot:2012vb}.} to illustrate the method and to take the opportunity to introduce useful notation. We analyse the action of T-duality in two different isometry directions of the internal manifold. This can be given for example by a flat two-torus together with a constant $B$-field, parametrized as follows:
\eq{\label{configuration}
G =\,
\begin{pmatrix}
R^2_1 & 0 \\
0 & R^2_2
\end{pmatrix}\;, \quad B=\, b\,
\begin{pmatrix}
0 & 1 \\
-1 & 0
\end{pmatrix}\;,
}
where $R_1, R_2$ are the radii of the two basis cycles of the torus. A common parametrization of this torus is given by the complex structure $\tau$ and the com\-plexi\-fied K\"ahler class $\rho$ defined in terms of the metric and $B$-field:
\eq{ \label{modularparam}
\tau =&\;\frac{G_{01}}{G_{00}} + \frac{i}{G_{00}} \sqrt{|\det G|} \;, \\
\rho =&\; B_{01} + i\sqrt{|\det G|}\;.
}
In our case we have $\tau =\, iR_2/R_1$ and $\rho =\, b + i R_1 R_2$. T-dualizing in the $0$-direction leads to the interchange of these two parameters by the Buscher rules \eqref{Busher}. Let us now do the calculation with the generalized metric . For this it is convenient to parametrize the metric and $B$-field by the parameters \eqref{modularparam}:
\eq{\label{gbmodular}
G =\, \frac{\textrm{Im}\, \rho}{\textrm{Im}\,\tau}
\begin{pmatrix}
1 & \textrm{Re}\, \tau \\
\textrm{Re}\, \tau & |\tau|^2
\end{pmatrix} \;, \quad B=\,
\begin{pmatrix}
0 & \textrm{Re}\,\rho \\
-\textrm{Re} \, \rho & 0
\end{pmatrix} \;.
}
Thus the generalized metric \eqref{genmetric} is given by:
\eq{ \label{genmetric3}
{\cal H} =\,\frac{1}{\textrm{Im}\,\rho \,\textrm{Im}\,\tau}
\begin{pmatrix}
|\rho|^2 & |\rho|^2 \textrm{Re} \,\tau & -\textrm{Re}\,\rho \,\textrm{Re}\,\tau & \textrm{Re}\,\rho \\
|\rho|^2 \textrm{Re}\,\tau & |\rho \tau|^2 & -|\tau|^2 \textrm{Re}\,\rho & \textrm{Re}\,\rho \,\textrm{Re}\,\tau \\
-\textrm{Re}\,\rho \,\textrm{Re}\,\tau & -|\tau|^2\textrm{Re}\,\rho & |\tau|^2 & -\textrm{Re}\, \tau \\
\textrm{Re} \,\rho & \textrm{Re}\,\rho \,\textrm{Re}\, \tau & -\textrm{Re}\,\tau & 1
\end{pmatrix} \;.
}
For applying the action of T-duality given in \eqref{dualityaction} in the $0$-direction we use the matrix $e_0$ in definition \eqref{tduality},
\eq{
e_0 =\,
\begin{pmatrix}
1 & 0 \\
0 & 0
\end{pmatrix} \;.
}
The new generalized metric after conjugation is of a similar structure as \eqref{genmetric3}, but with exchanged parameters $\tau \leftrightarrow \rho$, which means the following form of the dual metric and $B$-field:
\eq{\label{0-direction}
G =\, \frac{\textrm{Im}\,\tau }{\textrm{Im}\,\rho}
\begin{pmatrix}
1 & \textrm{Re}\,\rho \\
\textrm{Re}\, \rho & |\rho|^2
\end{pmatrix}\;, \quad
B =\,
\begin{pmatrix}
0 & \textrm{Re}\,\tau \\
-\textrm{Re}\,\tau & 0
\end{pmatrix} \;.
}
Thus, the action \eqref{dualityaction} reproduces the Buscher rules. In the same way it is possible to calculate T-dualities in the other directions (note that we assumed constant background fields). For later reference we list two more dualities. First, dualizing in the $1$-direction leads to:
\eq{
G =\, \frac{\textrm{Im}\,\tau}{|\tau|^2 \textrm{Im}\, \rho}
\begin{pmatrix}
|\rho|^2 & -\textrm{Re}\,\rho \\
-\textrm{Re} \,\rho & 1
\end{pmatrix}\;, \quad
B =\, \frac{\textrm{Re}\,\tau}{|\tau|^2}
\begin{pmatrix}
0 & -1 \\
1 & 0
\end{pmatrix} \;.
}
Thus in this case, we get:
\eq{
\tau \rightarrow -\frac{1}{\rho} , \quad \rho \rightarrow -\frac{1}{\tau} \;.
}
Finally, dualizing along the directions $0$ and $1$ results in:
\eq{\label{01-direction}
G=\, \frac{\textrm{Im}\,\rho}{|\rho|^2 \textrm{Im}\,\tau}
\begin{pmatrix}
|\tau|^2 & -\textrm{Re}\,\tau \\
-\textrm{Re}\,\tau & 1
\end{pmatrix} \;, \quad
B=\,\frac{\textrm{Re}\,\rho}{|\rho|^2}
\begin{pmatrix}
0 & -1 \\
1 & 0
\end{pmatrix} \;,
}
which means the following action on the torus parameters:
\eq{ \label{modulartrafo}
\tau \rightarrow -\frac{1}{\tau} \;, \quad \rho \rightarrow -\frac{1}{\rho} \;.
}
We will come back to these results in later examples, where we still have isometric directions but for example non-constant $B$-field. The above rules will then lead to configurations which are well described if one introduces so-called non-geometric fluxes. The motivation for the introduction of these new objects is the topic of the following section.

\subsection{Non-geometric fluxes}
\label{sec-ngeofluxes}

The existence of flux compactifications on manifolds T-dual to tori with NS-NS $H$-flux was realized in the work \cite{Shelton:2005cf, Shelton:2006fd}, by comparing four dimensional effective superpotentials of type IIA and type IIB orientifold compactifications. As a gui\-ding principle, dual superpotentials should map into each other ``bijectively''. From the absence of terms dual to the integrated $H$-flux, the existence of non-geometric $Q$- and $R$-flux degrees of freedom was concluded. 

In the following sections we first motivate the possibility of such fluxes and their local structure by purely geometrical considerations. After relating this discussion to world sheet Poisson sigma models, we give the most prominent example of an approximate flux compactification given by a three-torus with $H$-flux, which illustrates nicely the emergence of $f$- and $Q$-fluxes. To finally motivate also the existence of the $R$-flux, we sketch the original path of comparing four dimensional effective superpotentials. 

\subsubsection{Target space aspects: $O(d,d)$-transformations}
To get more information on a possible local form of the different geometric and non-geometric fluxes, we recall the generalized metric \eqref{genmetric}:
\eq{ \label{genmetric2}
\mathcal{H} =\,
\begin{pmatrix}
G - BG^{-1}B & BG^{-1} \\
-G^{-1}B & G^{-1}
\end{pmatrix} =\,
\begin{pmatrix}
1 & B \\
0 & 1
\end{pmatrix}
\begin{pmatrix}
G & 0 \\
0 & G^{-1}
\end{pmatrix}
\begin{pmatrix}
1 & 0 \\
-B & 1
\end{pmatrix} \;.
}
This is also valid for general spacetime dependent fields, i.e. one can interpret the generalized metric as a \emph{local} $B$-transform of the diagonal metric on the generalized tangent bundle $TM \oplus T^*M$. In the same way we can define the local $O(d,d)$ as the group of transformations leaving the metric $\eta$ defined in \eqref{oddmetric} invariant. This group can be decomposed locally into diffeomorphisms, local $B$-transformations and local $\beta$-transformations.

To get information about the structure of the various types of fluxes, we note the following. Taking a general $B$-transform,
\eq{ A_B =\,
\begin{pmatrix}
1 &b \\
0 & 1
\end{pmatrix}\;, \quad b_{ij} =\,-b_{ji} \;,
}
we see from \eqref{genmetric2} that a generalized metric with arbitrary $B$-field can be generated by such a transformation of the configuration with trivial $B$-field. Thus, to every $B$-transform, we associate the $H$-flux characterizing the background:
\eq{\label{Hflux}
H_{abc} =\, \partial_{[\underline{a}} b_{\underline{bc}]} \;.
}
Similarly, having a basis of the tangent bundle given by $\{\partial_i\}_{i=1,\dots, d}$ , a change of basis is given by a local $GL(d)$-transformation $a$. In the language of generalized tangent bundles such a transformation is given by:
\eq{ A_a =\,
\begin{pmatrix}
a & 0 \\
0 & (a^{-1})^t
\end{pmatrix}\;.
}
A basis change of this type introduces a non-holonomic basis $e_a =\, a_{a}{}^k\, \partial_k$ of the tangent bundle defining the \emph{geometric flux} as the structure constants of such a basis:
\eq{\label{geofluss12}
[e_a,e_b] =&\, \left( a_{[\underline{a}}{}^k\partial_k\,a_{\underline{b}]}{}^n \right) (a^{-1})_n{}^c \, e_c \\
&=\, f^c{}_{ab} \, e_c \;.
}
These changes of bases are parametrized by $d^2$ degrees of freedom and $B$-trans\-for\-ma\-tions are given by $\frac{1}{2}d(d-1)$ parameters. To get the full $O(d,d)$-group, we need a $\frac{1}{2}d(d-1)$-dimensional subgroup which is given by $\beta$-transformations:
\eq{
A_{\beta} =\,
\begin{pmatrix}
1 & 0 \\
\beta & 1
\end{pmatrix}\;, \quad \beta^{ij} =\, -\beta^{ji} \;.
}
The simplest combination of the potential $\beta$ with partial derivatives analogous to \eqref{geofluss12} is given by the following $Q$-flux:
\eq{\label{Qflux}
Q_a{}^{bc} =\, \partial_a \beta^{bc}\;.
}
Using the generalized tangent bundle we are able to give an interpretation of this object by seeing it in analogy to the geometric flux. As one can infer from the index structure in \eqref{Qflux}, one has to take a bracket of two \emph{one-form}-basis fields instead of vector fields. The corresponding bracket is given by the Koszul-bracket, introduced in section \ref{ch-math}, especially in \eqref{koszul}:
\eq{
\left[e^a, e^b\right]\ks =\,Q_c{}^{ab}\,e^c \;, \qquad e^a, e^b, e^c \in \Gamma(T^*M)\;.
}
Indeed, taking the standard basis $e^a =\, dx^a$, we recover the right local structure of the $Q$-flux as expected already in \eqref{Qflux}:
\eq{\label{Qfluxcoord}
\left[dx^a,dx^b\right]\ks =\, \partial_c \beta^{ab} \,dx^c \;.
}
Following the argument given in \cite{Aldazabal:2010ef}, to get information about the $R$-flux we look at the representation theory of the group $O(d,d)$. In the case $d=6$ which is relevant for torus compactifications, the $H$-flux has 20 components whereas the $f$- and $Q$-fluxes have 90 components. As physical fields, they should transform in a representation of $O(d,d)$. The smallest representation containing the three fluxes is the $\mathbf{220}$. Thus, there are 20 degrees of freedom missing, which is the dimension of an antisymmetric 3-vector (this is the only index structure which remains since a 3-form is already given by the $H$-flux).

It is intriguing to take once more the viewpoint of the cotangent bundle. As shown in \eqref{Qfluxcoord}, it is possible to express the $Q$-flux as structure constants of the Koszul-bracket similarly to the geometric flux in the case of the Lie-bracket. As was detailed in section \ref{ch-math}, one can associate the de Rham differential to the Lie bracket and looking at \eqref{Hflux}, we see that the $H$-flux is given by this differential acting on the potential $B$. Similarly, we can associate a differential $d_\beta$ to the Koszul bracket (as shown in \eqref{betad}) and acting on the bi-vector \mbox{$\beta =\,\tfrac{1}{2}\beta^{ij}\,\partial_i\wedge \partial_j$} results in a totally antisymmetric 3-vector which has the local form of the $R$-flux encountered in the literature, e.g. \cite{Grana:2008yw}:
\eq{\label{Rflux}
R =\, \tfrac{1}{2} d_\beta \beta\;, \quad \leftrightarrow \quad R^{ijk} =\, \beta^{[\underline{i} n}\partial_n \beta^{\underline{jk}]} \;. }

To sum up, we identified the local structure of the $H$-flux, the geometric flux and the $Q$-flux and identified their potentials. Whereas the $H$-flux is the differential of the $B$-field and the geometric flux is given by the structure constants of the Lie-bracket, the $Q$-flux can be interpreted as the structure constants of the Koszul-bracket. Finally by considering the representation theory of the symmetry group $O(d,d)$, we could identify the index structure of the non-geometric $R$-flux and by using the differential associated to the Koszul bracket, we were able to express it as the Poisson-differential $d_\beta$ of the bi-vector potential $\beta$.

\subsubsection{Worldsheet aspects: Poisson sigma models}
\label{sec-Poissonsigma}
In the last section, we encountered the bi-vector $\beta = \tfrac{1}{2}\beta^{ij}\,\partial_i \wedge \partial_j$ as an important object to describe the local structure of the non-geometric $Q$- and $R$-fluxes. On the level of the world sheet, this kind of geometric object plays an important role in the formulation of \emph{Poisson sigma models} (e.g. \cite{Schaller:1994es, Schaller:1995xk, Alekseev:2004np} and references therein). In this section we are going to sketch some aspects of this huge field of mathematical physics which are important for us. Giving an exhaustive treatment of this field would go far beyond the scope of this work and therefore we refer the reader to the literature for more information and results.
To write down a simple Poisson sigma model \cite{Baulieu:2001fi}, let us introduce the dynamical fields of the theory. We consider a two-dimensional world sheet $\Sigma$ with coordinates $(\sigma^0,\sigma^1)=(\tau, \sigma)$. The target space coordinate fields $X^i$ are given by the components of the map $X : \Sigma \rightarrow {\cal M}$, where ${\cal M}$ is the target space manifold. To illustrate the situation, consider the following commutative diagram:
\eq{
\begin{array}{ccc}
T^*\Sigma & \stackrel{X^*}{\longleftarrow} & T^*{\cal M} \\
\downarrow & & \downarrow \\
\Sigma & \stackrel{X}{\longrightarrow} & {\cal M}
\end{array}
}
where $X^*$ denotes the pullback operation. Thus the derivatives $dX^i$ are elements of $X^* (T^* {\cal M})$, i.e. differential forms on $\Sigma$ and can be contracted with $T^*{\cal M}$-valued differential forms $p \in \Gamma(T^*\Sigma) \otimes T^*{\cal M}$ having components $p_i =\, p_{i \alpha} \, d\sigma^\alpha$. Finally, introducing the background fields $\beta^{ij}(X)$ and $\theta^{ij}(X)$ (which can be interpreted as the metric and bi-vector on the target space), the action for the two-dimensional Poisson sigma model\footnote{Note that we added a term for the metric $g^{ij}$ and therefore the model is not topological in contrast to the original constructions.} is given by \cite{Baulieu:2001fi}:
\eq{\label{Poissonsigmamodel}
S =\, \int_{\Sigma} i\,p_i \wedge dX^i + \frac{1}{2} g^{ij}(X) p_i \wedge \star _2 \, p_j + \frac{1}{2} \beta^{ij}(X)\,p_i \wedge p_j \;,
}
where $\star _2$ denotes the two-dimensional Hodge-star operator and we dropped terms that are related to the dilaton. To simplify notation, we also drop factors of $2\pi \alpha'$ because they are not important in this section. 

The action \eqref{Poissonsigmamodel} is the starting point to get two theories by eliminating either the momenta $p_i$ or the currents $dX^i$. On the one hand the standard Polyakov sigma model \eqref{hamiltonsigmamodel} is recovered by integrating out the $p_i$-fields. In the notation introduced above, the resulting action is:
\eq{
  S_P =\, \int_\Sigma \, G_{ij}\, dX^i \wedge \star _2\, dX^j + B_{ij}\, dX^i \wedge dX^j \;.
}
On the other hand, eliminating the forms $dX^i$ results in an action of a very similar structure, but with $g^{ij}$ and $\beta^{ij}$ as background fields and $p_i$ as dynamical variables:
\eq{
S' =\, \int_\Sigma \, g^{ij}\,p_i \wedge \star_2 \, p_j + \beta^{ij}\,p_i \wedge p_j \;.
}
It is interesting to note the relation between the two different sets of background fields as was shown e.g. in \cite{Baulieu:2001fi,Halmagyi:2008dr}:
\eq{\label{basicfieldredef1}
G_{ij} =\,& (g^{-1} - \beta g \beta )^{-1}_{ij} \;, \\
B_{ij} =\,& G_{ik} g_{jm} \beta^{km} = (g^{-1} - \beta g  \beta )^{-1}_{ik} g_{jm} \beta^{km} \;.
}
For later reference, we also note the converse relations:
\eq{\label{basicfieldredef2}
g_{ij} =\,& G_{ij} - B_{im}G^{mn} B_{nj} \;, \\
\beta^{ij} =\,& -\left[(G + B)^{-1}\right]^{im}B_{mn}\left[(G - B)^{-1}\right]^{nj} \;.
}
Due to the equivalence of the two models and their similarity in structure, we take the previous field redefinition as distinguished in the context of Poisson sigma models. Indeed, as the results of the last sections show, the set $(g,\beta)$ will be more appropriate to describe situations with $Q$- and $R$-fluxes, whereas $(G,B)$ will be the choice for describing geometric configurations with $H$- and $f$-fluxes.

Remembering the construction of the Hamiltonian for the Polyakov sigma model with the generalized metric ${\cal H}(G,B)$ in section \ref{sec-torusodd}, especially \eqref{Hamilton} and \eqref{genmetric}, we can use the previous field redefinition to get the generalized metric in terms of the new variables:
\eq{
{\cal H}(g,\beta)=\,
\begin{pmatrix}
g & -g\beta \\
\beta g & g^{-1} - \beta g \beta
\end{pmatrix} =\,
\begin{pmatrix}
1 & 0 \\
\beta & 1
\end{pmatrix}
\begin{pmatrix}
g & 0 \\
0 & g^{-1}
\end{pmatrix}
\begin{pmatrix}
1 & -\beta \\
0 & 1
\end{pmatrix}\;,}
i.e. as a $\beta$-transform of the diagonal form of the generalized metric written in terms of $g$. In the following section we will see at a concrete example, how the two sets of variables $(G,B)$ and $(g^{-1},\beta)$, which are often referred to as \emph{frames}, can be used to describe configurations T-dual to a three-torus with $H$-flux. 

\subsubsection{Example: 3-torus with $H$-flux}
\label{subsec-t3h}
Let us now discuss one of the most important examples where it is possible to perform T-dualities in two different directions, starting from a configuration which has non-vanishing $H$-flux. The dual backgrounds can be characterized by geometric flux $f$ after one T-duality and by non-geometric $Q$-flux after the second T-duality. It will turn out that starting with background fields $(G,B)$, the $Q$-flux background is described most conveniently in the frame with variables $(g,\beta)$.

\paragraph{Torus with $H$-flux}
The model we are going to analyse is given by a three-dimensional torus with coordinates $x,y,z$  (which can be considered to be part of the internal space) with constant NS-NS $H$-flux and constant dilaton. It is given by the following flat metric and gauge choice for the $B$-field:
\eq{\label{torusH}
G =\, \begin{pmatrix}
R_1^2 & 0 & 0 \\
0 & R_2 ^2 & 0 \\
0 & 0 & R_3 ^2
\end{pmatrix}\;, \quad B =\,
\begin{pmatrix}
0 & Nz & 0 \\
-Nz & 0 & 0 \\
0 & 0 & 0
\end{pmatrix}\;,
}
where $N$ is a constant integer. Thus we get the constant flux $H=\,N\,dx \wedge dy \wedge dz$. Note that this configuration is not a solution to the string equations of motion: Whereas the Bianchi identity for the $H$-flux is fulfilled trivially, the string equations of motion are only approximately true in the so-called \emph{dilute flux limit}:
\eq{
\left(\frac{N}{R_1 R_2 R_3}\right) \rightarrow 0 \;.
}
For the rest of the example, we only consider this limit. It turns out that the T-dual configurations are also solutions to the supergravity equations of motion up to linear order in the flux. In addition, let us simplify our notation for the rest of the calculations to point out only the important features of the different T-duality frames. We choose units where all the radii are equal to one and we only take one unit of $H$-flux, i.e.
\eq{\label{konvention}
G =\, \begin{pmatrix}
1& 0 & 0 \\
0 & 1 & 0 \\
0 & 0 & 1
\end{pmatrix}\;, \quad B =\,
\begin{pmatrix}
0 & z & 0 \\
-z & 0 & 0 \\
0 & 0 & 0
\end{pmatrix} \;, \quad H=\,dx\wedge dy \wedge dz \;.
}
A very convenient way of looking at the above geometry is to describe it as a fibration of a two-torus (coordinates $x,y$) over a base-$S_z^1$. The torus is then described by the following parameters, introduced in section \ref{sec-torusodd}:
\eq{
\tau =\, i \;, \quad \rho =\, z + i \;,
}
and we can use the T-duality relations shown in section \ref{sec-torusodd} if we dualize in the direction of isometries. Before starting with the first T-duality, let us note that the configuration \eqref{konvention} has the following behavior if we go around the base:
\eq{
z \rightarrow z+1 :\quad B \rightarrow B + dA\;, \; A =\, xdy-ydx \;.
}
Whereas the metric is trivial,  to get a well defined configuration we should use coordinate charts where the $B$-field changes by a gauge transformation on the overlap. We want to compare this \emph{geometric} situation with the configuration after two T-dualities where it will turn out that the transition functions between two charts have to be extended to include also T-dualities.

\paragraph{Geometric $f$-flux}
The above geometry allows for a T-duality in the $x$-direction. Using equation \eqref{0-direction} of section \ref{sec-torusodd} for the transformation of the fibered torus and noting that the base circle remains unchanged, the resulting configuration is given by:
\eq{\label{conf0}
G =\,
\begin{pmatrix}
1 & z & 0 \\
z & z^2 + 1 & 0 \\
0 & 0 & 1
\end{pmatrix}\;, \quad B = 0 \;.
}
To diagonalize the above metric, we introduce the following \emph{non-holonomic} basis:
\eq{
e^x =\, dx + z dy\;, \quad e^y =\, dy\;, \quad e^z =\, dz \;.
}
To find out the structure constants of the new basis, we can either determine the corresponding basis of the tangent bundle or directly calculate the exterior differential of the basis one-forms:
\eq{
de^x =\, -\frac{1}{2}f^x{}_{mn}\, e^m \wedge e^n \; \rightarrow f^x{}_{yz} = -f^x{}_{zy} =\,1 \;,
}
where the rest of the structure constants vanish. Thus this configuration has vanishing $H$-flux, but non-vanishing geometric $f$-flux.

\paragraph{Non-geometric $Q$-flux}
The preceding configuration \eqref{conf0} does not depend on the second coordinate and therefore the Buscher rules allow for a second T-duality in this direction. The base of the fibration still remains unchanged and therefore we use equation \eqref{01-direction} of section \ref{sec-torusodd} to get the result for the metric and $B$-field in this case:
\eq{ \label{conf01}
G =\,
\begin{pmatrix}
 \frac{1}{z^2 + 1} & 0 & 0 \\
0 & \frac{1}{z^2 + 1} & 0 \\
0 & 0 & 1
\end{pmatrix} \;, \quad B=\,
\begin{pmatrix}
0 & \frac{-z}{z^2 + 1} & 0 \\
\frac{z}{z^2 + 1} & 0 & 0 \\
0 & 0 & 0
\end{pmatrix} \;.
}
The new torus parameters depend on the third coordinate and are determined by the relation \eqref{modulartrafo}. Explicitely we have:
\eq{
\tau' =\, -\frac{1}{\tau} =\, i \;, \qquad \rho' =\, -\frac{1}{\rho} =\, -\frac{1}{z+i} \;.
}
It is interesting to note the transformation behavior of the two parameters if we go around the base circle $z\rightarrow z+1$:
\eq{\label{paramono}
\tau'(z+1) =\, \tau'\;, \quad \rho'(z+1) =\, \frac{\rho'}{-\rho' + 1} =:
\begin{pmatrix}
1 & 0 \\
-1 & 1
\end{pmatrix} (\rho') \;,
}
i.e. the complexified K\"ahler class is changed by the specific modular transformation \eqref{paramono} which is also called a \emph{parabolic monodromy}. Therefore the transition functions for our fields are no longer only gauge transformations and diffeomorphisms but involve more complicated transformations. To characterize the latter in a clear form, we choose the new variables $g, \beta$ introduced in \eqref{basicfieldredef2}. In our case they are given by
\eq{\label{newframe}
g =\,
\begin{pmatrix}
1 & 0 & 0 \\
0 & 1 & 0 \\
0 & 0 & 1
\end{pmatrix} \;, \quad \beta =\,
\begin{pmatrix}
0 & z & 0 \\
-z & 0 & 0 \\
0 & 0 & 0
\end{pmatrix} \;.
}
Thus, a transformation $z\rightarrow z+1$ results in a constant shift of the bi-vector $\beta$. This is an example of a $\beta$-transform which is not a standard gauge transformation or a diffeomorphism and we thus have to extend the set of transition functions of the resulting torus fibration by $\beta$-transformations. Manifolds with this extended group of transition functions are called \emph{T-folds}.

Finally, let us note that the above configuration has non-vanishing $Q$-flux which is given by:
\eq{
Q_z{}^{xy} =\, \partial_z \beta^{xy} =\, 1 \;.
}
We observe that it is invariant under the constant shift of the bi-vector and thus is a well defined quantity on the T-fold. To sum up, we started with a geometric configuration with constant $H$-flux which was dualized into a configuration having vanishing $H$-flux but non-vanishing geometric $f$-flux. The latter was determined by the structure constants of a non-holonomic basis. Finally after T-dualizing a second time, the resulting configuration was described properly with dual variables $(g,\beta)$ and the resulting non-vanishing flux was given by the non-geometric $Q$-flux. As a result we confirmed the main part of the T-duality chain of the introduction:
\eq{ \label{3/4chain}
H_{xyz} \rightarrow f^x{}_{yz} \rightarrow Q_z{}^{xy} \;.
}
The final T-duality which would result in a frame having non-vanishing $R$-flux cannot be described by the previous construction because there is no isometric direction left to apply the Buscher rules. Nevertheless one considers this background to have one unit of $R$-flux $R^{xyz} = 1 $. The physical motivation for the existence of this flux, at least from the four-dimensional perspective, will be the topic of the next section.

\subsubsection{The four-dimensional perspective}
In the last three subsections, we analysed the structure of the geometric $H$- and $f$-fluxes and added objects which we called non-geometric $Q$- and $R$-fluxes. They were first motivated by considering the generalized tangent bundle and the representation theory of the $O(d,d)$-group. By considering world-sheet Poisson sigma models, we identified two distinct set of variables which turned out to be useful to describe configurations T-dual to a three-torus with $H$-flux, including $f$-flux and $Q$-flux geometries. However, up to now there was no \emph{physical} reason to include also the non-geometric $R$-flux. This will be the main point of this section, where we sketch the original motivation of non-geometric fluxes given in \cite{Shelton:2005cf, Shelton:2006fd}. We closely follow the detailed exposition of \cite{Wecht:2007wu}, where also additional information can be found.

To begin with, let us consider type IIB string theory. In order to set our notation, recall its field content:
\begin{itemize}
\item R-R $p\,$-forms $C_p$, $\;p \in \{ 0,2,4\}$, field strength: $F_{p+1} = dC_p$.
\item Axio-dilaton $\; S = C_0 + ie^{-\phi}$, with dilaton $\phi$.
\item $G_3$-flux $\; G_3 = F_3 - S H$, with NS-NS $H$-flux $H = dB$.
\item $\tilde{F}_5 = dC_4 - \tfrac{1}{2}\, C_2 \wedge H + \tfrac{1}{2}\, F_3 \wedge B$.
\end{itemize}
These have to be supplemented by the constraints of self-duality of $F_5$ and the non-standard Bianchi-identity $dF_5 = H\wedge F_3$. Now let us compactify this theory on the orientifold ${\cal X} = (T^2)^3/ \Omega \mathbb{Z}_2 (-1)^{F_l}$, where $\Omega$ is the world-sheet parity, which together with the left-moving fermion number operator $(-1)^{F_l}$ reduces supersymmetry from $N = 8$ to $N = 4 $ in four dimensions. The $\mathbb{Z}_2$ denotes the spatial reflection on the internal manifold: $x^i \mapsto -x^i, i \in \{1,\dots,6\}$.

It is possible to show (e.g. \cite{Kachru:2002he}) that the presence of fluxes further reduces supersymmetry to $N = 1 $ and in addition one restricts to three moduli fields given by the axio-dilaton $S$, the complex structure $\tau$ and the K\"ahler class $U$ of the torus $T^2$ (note that for simplicity we assume three identical copies of $T^2$ in the internal manifold). The remaining data of the four-dimensional effective theory are the K\"ahler potential whose form is not needed later and the superpotential. The latter can be computed by the Gukov-Vafa-Witten formula \cite{Gukov:1999ya} :
\eq{\label{GVWformula}
W =\, \int_{\cal X} \; \left( F_3 - S\,H \right) \wedge \omega \;,
}
where $\omega$ is the holomorphic three-form on ${\cal X}$. To write it in components, let us introduce coordinates on each torus, labeled by $(\alpha,i),(\beta,j),(\gamma,k)$, meaning that we have the complex coordinate $z^1 := x^\alpha + \tau x^i$ on the first $T^2$ and similar for the other copies. Thus we get for the holomorphic three-form:
\eq{
\omega = \,  dz^1 \wedge dz^2 \wedge dz^3 =\, &dx^\alpha \wedge dx^\beta \wedge dx^\gamma + \tau(dx^i \wedge dx^\beta \wedge dx^\gamma +\dots ) \\
&  + \dots  + \tau^3 \, dx^\alpha \wedge dx^\beta \wedge dx^\gamma \;.
}
With the help of this expression, we can now calculate the superpotential by doing the integral \eqref{GVWformula}. The result is a cubic polynomial in $\tau$ with coefficients given by the $F_3$ flux components, integrated over a three-dimensional internal cycle and a cubic polynomial in $\tau$ with coefficients given by the integrated $H$-flux, more precisely we have \cite{Wecht:2007wu}:
\eq{\label{IIBsuper}
W_{IIB} =\, a_0 -3a_1\tau &+ 3a_2 \tau^2 - a_3 \tau^3 + S\left( -b_0 + 3b_1 \tau - 3b_2 \tau^2 + b_3 \tau^3 \right) \;.
}
Let us also specify the fluxes which determine the corresponding coefficients in \eqref{IIBsuper}. They are given in table \ref{IIBtable}, where we denote by $\bar F_{ij\alpha}$ the $F_3$-flux integrated over the cycle given by the coordinates $(i,j,\alpha)$, and similar for the other combinations.

\begin{table}[h]
\begin{center}
\begin{tabular}{|c|c|c|}
\hline
term & integral IIB flux & coefficient \\
\hline 
$\tau^0$ & $\bar F_{ijk}$ & $a_0$ \\
$\tau$ & $\bar F_{ij\gamma}$ & $a_1$ \\
$\tau^2$ & $\bar F_{i\beta \gamma}$ & $a_2$ \\
$\tau^3$ & $\bar F_{\alpha \beta \gamma}$ & $a_3$ \\
$S\tau^0$ & $\bar H_{ijk}$ & $b_0$ \\
$S\tau$ & $\bar H_{\alpha j k}$ & $b_1$ \\
$S\tau^2 $ & $\bar H_{i \beta \gamma}$ & $b_2$ \\
$S\tau^3$ & $\bar H_{\alpha \beta \gamma}$ & $b_3$ \\
\hline
\end{tabular}
\caption{\small Integrated fluxes corresponding to the coefficients of the four-dimensional effective type IIB superpotential \cite{Wecht:2007wu}. \label{IIBtable}}
\end{center}
\end{table}

The coefficients $a_i$ and $b_i$ are not independent: As an example, the integrated Bianchi identity of $\tilde F_5$ contains the term $\int \, H \wedge F_3$ and is determined by the sum of the orientifold-plane charges of the theory. This gives a constraint of the type:
\eq{
a_0 b_3 - 3a_1 b_2 + 3a_2 b_1 - a_3b_0 = \textrm{const.}
}

As a next step, let us compare \eqref{IIBsuper} with the superpotential of a type IIA orientifold compactification on a twisted torus. The latter was defined in section \ref{subsec-t3h} and is characterized by geometric $f$-flux. To compare type IIB and type IIA, we assume we can T-dualize the directions labeled by Greek indices in the internal torus. The four-dimensional effective type IIA superpotential can be computed by directly reducing the ten dimensional theory to four dimensions, as was shown in \cite{Hull:2005hk}, from which we only take the structure of the result. Denoting the moduli in the IIA case also by $\tau, S$ and $U$, we have:
\eq{
W_{IIA} =\, &a_0 -3a_1\tau + 3a_2 \tau^2 - a_3\tau^3 \\
&+ S(-b_0 + 3b_1 \tau) + 3U\left(c_0 +(c'_1 + c''_1 -c'''_1)\tau \right) \;,
}
where the coefficients are determined again by the integrated R-R fluxes $\bar F_p$ (p even), $\bar H$ and the structure constants $f$ of the twisted torus as shown in table \ref{IIAtable}.
\begin{table}
\begin{center}
 \begin{tabular}{|c|c|c|}
\hline
term & integral IIA flux & coefficient \\
\hline
$\tau^0$ & $\bar F_{\alpha i  \beta j \gamma k}$ & $a_0$ \\
$\tau$ & $\bar F_{\alpha i \beta j }$ & $a_1$ \\
$\tau^2$ & $\bar F_{\alpha i }$ & $a_2$ \\
$\tau^3$ & $\bar F_0 $ & $a_3$ \\
$S\tau^0$ & $\bar H_{ijk}$ & $b_0$ \\
$S\tau$ & $f^\alpha{}_{jk}$ & $b_1$ \\
$U $ & $\bar H_{\alpha \beta k}$ & $c_0$ \\
$U\tau$ & $f^j{}_{k \alpha}, f^i{}_{\beta k}, f^\alpha{}_{\beta \gamma}$ & $c''_1, c' _1, c'''_1$ \\
\hline
\end{tabular}
\caption{\small Integrated fluxes and structure constants corresponding to the coefficients of the four-dimensional effective type IIA superpotential \cite{Wecht:2007wu}. \label{IIAtable}}
\end{center}
\end{table}
Comparing the two tables \ref{IIBtable} and \ref{IIAtable}, we observe that the R-R fluxes map to each other in the standard way (e.g. \cite{Green:1996bh}). In addition, the $b_0$-coefficients remain unchanged as they are determined by the $H$-flux in the directions which are not dualized. An interesting phenomenon can be seen in the coefficient of $S\tau$: The $H$-flux gets mapped to the geometric flux in a way we already encountered in the last section: $\bar H_{\alpha j k} \mapsto f^{\alpha}{}_{jk}$. There is a mismatch in the other coefficients, which directs us to the following idea: By introducing more flux-determined coefficients in the two superpotentials such that the fluxes are related in a way suggested by \eqref{3/4chain}, it should be possible to write down a superpotential whose coefficients map to one another by performing T-dualities in the directions $(\alpha,\beta, \gamma)$. To do this, we have to extend \eqref{3/4chain} to include also the non-geometric $R$-flux:
\eq{\label{fullchain}
H_{\alpha \beta \gamma} \; \xleftrightarrow{\;\;T_\alpha\;\;}\; f^\alpha{}_{\beta \gamma}\, \xleftrightarrow{\;\; T_\beta \;\;} \; Q_\gamma{}^{\alpha \beta} \; \xleftrightarrow{\;\;T_\gamma\;\;}\; R^{\alpha \beta \gamma} \;.
}
To arrive at a duality invariant superpotential, the authors of \cite{Shelton:2005cf, Shelton:2006fd} started from the above described IIA case and performed the following steps:
\begin{itemize}
\item  T-dualize in the directions $(\alpha, \beta , \gamma)$ to get to the IIB theory above. Use \eqref{fullchain} to get the NS-NS-fluxes and use the standard T-duality rule for the R-R fluxes. If there is no corresponding flux on the IIB side (\ref{IIBtable}), add a new coefficient.
\item On every $T^2$, exchange the two torus directions, e.g. $x^\alpha + \tau x^i \leftrightarrow x^i + \tau x^\alpha$, which means $1 \leftrightarrow \tau^3$ and $\tau \leftrightarrow \tau^2$. This gives the same theory, since in the above IIB case the orientifold planes are unaffected by this rotation.
\item T-dualize back in the directions $(\alpha, \beta, \gamma)$ to get the corresponding additional coefficients in the IIA case.
\end{itemize}
As an example, the integral flux $\bar F_{\alpha i \beta j \gamma k}$ in IIA will be dualized to $\bar F_{ijk}$, which gets exchanged to $\bar F_{\alpha \beta \gamma}$ which then is dualized back to $\bar F_0$. Thus no new flux is needed, which is also true for the other coefficients determined by the R-R fluxes. But there are new terms needed for the NS-NS part, as one can see for example by the $S\tau $ coefficients. Performing the analysis for all the fluxes, we get the following invariant superpotential \cite{Shelton:2005cf, Shelton:2006fd}:
\eq{\label{invariantsuperpot}
W = \, &a_0 -3a_1 \tau + 3a_2 \tau^2 - a_3 \tau^3 + S(-b_0 + 3b_1\tau -3b_2 \tau^2 + b_3 \tau^3) \\
& + 3U(c_0 + (c'_1 + c''_1 - c'''_1)\tau - (c'_2 + c''_2 + c'''_2)\tau^2 -c_3 \tau^3 )\;,
}
where the coefficients are determined by the fluxes (including the non-geometric ones) as given in table \ref{IIABtable}.

\begin{table}
\begin{center}
 \begin{tabular}{|c|c|c|c|}
\hline
term & integral IIA flux & integral IIB flux & coefficient \\
\hline
$\tau^0$ & $\bar F_{\alpha i  \beta j \gamma k}$ & $\bar F_{ijk}$ &  $a_0$ \\
$\tau$ & $\bar F_{\alpha i \beta j }$ & $\bar F_{ij\gamma}$ &  $a_1$ \\
$\tau^2$ & $\bar F_{\alpha i  }$ & $\bar F_{i\beta \gamma}$ & $a_2$ \\
$\tau^3$ & $\bar F_0 $ & $\bar F_{\alpha \beta \gamma} $ & $a_3$ \\
$S\tau^0$ & $\bar H_{ijk}$  & $\bar H_{ijk}$ & $b_0$ \\
$U $ & $\bar H_{\alpha \beta k}$ & $Q_k{}^{\alpha \beta}$ & $c_0$ \\
$S\tau$ & $f^\alpha{}_{jk}$ & $\bar H_{\alpha j k}$ & $b_1$ \\
$U\tau$ & $f^j{}_{k \alpha}, f^i{}_{\beta k}, f^\alpha{}_{\beta \gamma}$ & $ Q_k{}^{\alpha j}, Q_k{}^{i \beta}, Q_\alpha {}^{\beta \gamma} $ & $c''_1, c' _1, c'''_1$ \\
$S\tau^2$ & $ Q_k{}^{\alpha \beta}$ & $\bar H_{i\beta \gamma}$ & $b_2$ \\
$U\tau^2$ & $Q_\beta{}^{\gamma i}, Q_\gamma{}^{i\beta}, Q_k{}^{ij}$ & $Q_\gamma{}^{i \beta}, Q_\beta{}^{\gamma i}, Q_k{}^{ij}$ & $c''_2, c'_2, c'''_2$ \\
$S\tau^3$ & $R^{\alpha \beta \gamma}$ & $ \bar H_{\alpha \beta \gamma}$ & $b_3$ \\
$U\tau^3$ & $R^{ij\gamma}$ & $Q_\gamma{}^{ij}$ & $c_3$ \\
\hline
\end{tabular}
\caption{\small Integrated fluxes corresponding to the coefficients of the duality invariant four-dimensional effective superpotential \cite{Wecht:2007wu}.\label{IIABtable}}
\end{center}
\end{table}

Finally, there are constraints which the fluxes have to satisfy. We are only interested in the equations for the NS-NS fluxes, as they will be important for later parts of this work. Writing down the Bianchi identity for the $H$-flux in a non-holonomic basis with structure constants $f$, we get:
\eq{
\partial_{[\underline{a}} H_{\underline{bcd}]} + f^n{}_{[\underline{ab}} H_{n\underline{cd}]} = 0\:.
}
Noting that we only have constant (integrated) fluxes for the coefficients in the superpotential \eqref{invariantsuperpot}, the first term drops out and we get the relation $f^n{}_{[\underline{ab}} \bar H_{n\underline{cd}]} = 0$. Now taking successively the respective fluxes in table \ref{IIABtable} with one of the indices not summed over and in upper position, we get four additional constraints as is described in detail in \cite{Wecht:2007wu}:
\eq{\label{constantbianchis}
& \bar H_{n[\underline{ab}} f^n{}_{\underline{cd}]} = 0 \;, \\
& f^a{}_{n[\underline{b}} f^n{}_{\underline{cd}]} + \bar H_{n[\underline{bc}} Q_{\underline{d}]}{}^{an} = 0 \;, \\
& Q_n{}^{[\underline{ab}]}f^n{}_{[\underline{cd}]} - 4 f^{[\underline{a}}{}_{n [\underline{c}} Q_{\underline{d}]}{}^{\underline{b}]n} + \bar H_{n[\underline{cd}]} R^{[\underline{ab}] n} = 0 \;,\\
& Q_n{}^{[\underline{ab}}Q_d{}^{\underline{c}]n} + f^{[\underline{a}}{}_{nd} R^{\underline{bc}]n} = 0 \;,\\
& Q_n{}^{[\underline{ab}}R^{\underline{cd}]n} = 0 \;.
}
This means that we cannot turn on every flux in a specific direction at the same time. We will come back to a generalization of these equations in a completely different context by mathematical considerations in the next section on Bianchi identities and Courant algebroids.

To sum up, by physical arguments we have shown in a somewhat heuristic way that it is reasonable to include the non-geometric $R$-flux into the T-duality chain which starts from a torus background with $H$-flux. The mathematical appearance of such an object was motivated in the previous sections and from a four-dimensional effective perspective, the above considerations suggest that we should also include this flux into string physics. However, as was argued in the previous section, there is no clear target space interpretation of a background carrying $R$-flux, because we would have to perform a T-duality in a non-isometric direction. There are indirect arguments \cite{Bouwknegt:2004ap, Shelton:2006fd} that this target space may be rather exotic: Assume we can localize points in such a space, which would cor\-res\-pond to $D0$-branes. T-dualizing back to the $H$-flux case would lead to $D3$-branes wrapping a three-dimensional orientable submanifold of the internal space. It is a well-known mathematical fact that every orientable manifold of dimension smaller or equal to four is $\textrm{Spin}^c$ and therefore has vanishing third Stiefel-Whitney class. But this violates the Freed-Witten anomaly cancellation condition if $H$ is topologically non-trivial. From this problem, one expects that on a space with $R$-flux, one cannot localize points and therefore a description by differentiable manifolds seems not to be appropriate. Indeed, as we see in the next section by conformal field theory arguments, such target spaces may be described by considering non-associative structures or even n-ary products on the algebra of functions on the space.

\section{Results}
\label{ch-results}

\subsection{Conformal field theory with $H$-flux and T-duality}
\label{sec-cft}

As was pointed out in the last section, the nature of non-geometric fluxes is not completely clear yet. Especially in the case of the $R$-flux, which would be the third T-duality in the chain \eqref{3/4chain}, the Buscher rules are not applicable due to the missing of an isometry in the $Q$-flux configuration. As a consequence, the geometry of $R$-flux configurations is not accessible at the target space level. Topological arguments \cite{Bouwknegt:2004ap, Mathai:2004qq} show that the resulting geometries might be very exotic, like non-associative spaces. On the contrary, at the level of the world sheet, T-duality can be implemented as a reflection of right-moving sectors of the closed string. Such an operation is not possible for point particles and thus the target space interpretation of such left-right \emph{asymmetric} theories is not clear and may also go beyond standard differential geometry. But as we reviewed in section \ref{ch-quant}, by computing correlation functions one can in some cases extract target space information like non-commutativity of coordinates \eqref{ncspacetime}.

\subsubsection{Motivation: Non-associative structures}
In section \ref{ch-quant} we reviewed the derivation of non-commutativity of spacetime by computing open string correlation functions. A crucial point was the dependence of the propagator of the theory on the ordering of two points on the real line \eqref{limprop}. This was the source of non-commutativity as a spacetime-property seen in the commutator \eqref{ncspacetime} of two target space coordinate fields.

In the case of closed string theory, the first hint that a non-vanishing commutator should be replaced by a structure involving three spacetime coordinate fields, was given in \cite{Blumenhagen:2010hj, Lust:2010iy} by computing the equal-time cyclic double-commutator of three local coordinate fields $X^a(z_1,\bar z_1)$, $X^b(z_2, \bar z_2)$, $X^c(z_3, \bar z_3)$ in the case of an $SU(2)_k$ Wess-Zumino-Witten model with $H$-flux. The result for the double commutator was determined by the background-flux and was given by
\eq{
\left[X^a,X^b,X^c\right] &=\, \lim_{z_i \rightarrow z} \, \left[X^a(z_1, \bar z_1), [X^b(z_2,\bar z_2), X^c(z_3, \bar z_3)]\,\right] + \mathrm{cycl.} \\
&= \, \epsilon\,\theta^{abc} \;,
}
where $\epsilon = 0$ for the original case of the $H$-flux, but $\epsilon = 1$ after performing T-dualities in an odd number of directions. In addition, it was argued that $\theta^{abc} \propto H^{abc}$.

This remarkable result obtained from the world-sheet perspective is a power\-ful hint that the target space geometries obtained by an odd number of T-dualities from a standard $H$-flux background may carry 3-structures and may be non-associative\footnote{For a phase-space realization of this type of algebra, see \cite{Lust:2010iy}.} . We take this as a motivation to study closed string conformal field theory in the presence of $H$-flux and T-dual configurations thereof. It will turn out in the following sections that this is possible up to linear order in the flux. By computing correlation functions of coordinates and tachyon vertex operators we will get more information about the structure of the algebra of functions on the target space.

\subsubsection{Conformal field theory with $H$-flux}
Instead of the exactly solvable WZW-model used in \cite{Blumenhagen:2010hj} we consider a simpler example of flat space with constant $H$-flux (see also section \ref{subsec-t3h}). Let us again recall the metric and flux to set our conventions
\eq{
  \label{setup_01}
   ds^2=\sum_{a=1}^N \bigl(dX^a\bigr)^2, \hspace{50pt}
   H=  \frac{2}{{\alpha'}^2}\,  \epsilon_{abc}\,  dX^a\wedge dX^b\wedge dX^c\; .
}
Note that this ansatz does not solve the string equations of motion. Looking for example at the renormalization group equation for the graviton,
\eq{
  \label{targeteom}
  0 =  \alpha' R_{ab}-\frac{\alpha' }{4}\: H_{a}{}^{cd}\, H_{bcd}
  +2\alpha' \nabla_a\nabla_b\Phi
    +O({\alpha'}^2) \;,
}
we see that the ansatz \eqref{setup_01} is only a solution up to first order in the $H$-flux. It turns out that this is true also for the other equations of motion and therefore we expect to get a proper conformal field theory description for this background \emph{up to first order in the flux}. In the following we are going to formulate this theory in order to implement T-duality by a simple reflection of the right moving part of the closed string. Starting out by the geometric model \eqref{setup_01}, we can then compute correlation functions of coordinate fields and vertex operators in T-dual situations and try to extract information on their target space interpretation.

\paragraph{Sigma model and classical solution}

Let us start by the closed string sigma model
\eq{
  \label{closedstringsigma}
  \mathcal S=\frac{1}{2\pi\alpha'}\int_\Sigma d^2 z\, \bigl(\,
  g_{ab} + B_{ab} \bigr)\, \partial X^a \,\ov\partial X^b
  \;,
}
and take an the following background fields to implement the ansatz \eqref{setup_01}:
\eq{\label{ansatzlinH}
g_{ab} =\, \delta_{ab}\;, \qquad B_{ab} =\, \frac{1}{3} H_{abc}X^c \;.
}
The classical equations of motion for the fields $X^a$ are given by
\eq{\label{classeom}
\partial \bar \partial X^a = \, \frac{1}{2} H^a{}_{bc} \partial X^b \bar \partial X^c \;,
}
where we raise and lower indices with the metric $g_{ab} = \delta_{ab}$. Thus, at zeroth order in the $H$-flux, the $B$-field vanishes and we get the free theory of a closed string on the sphere. Its solution can be split into left- and right-moving parts
\eq{
   \mathsf X^a_0(z,\ov z) =  \mathsf X^a_L(z) + \mathsf{X}{}^a_R(\ov z) \;.
}
This is not possible at linear order in the flux, for which a solution to the equations of motion is readily calculated to be
\eq{
  \label{firstordersolution}
  \mathsf{X}_1^a(z,\ov z) =  \mathsf X_0^a(z,\ov z) +
  {\textstyle \frac{1}{2}}\hspace{0.5pt}H^a{}_{bc} \,  \mathsf X^b_L(z) \,
  \mathsf {X}{}^c_R(\ov z)  \;.
}
This is already a hint that for a conformal field theory also at linear order in the flux, we have to redefine the coordinate fields in an appropriate way. We will see this in the definition of currents in the following subsections.

\paragraph{Perturbation theory}
Let us shortly recall our way to compute correlation functions. The sigma model action can be split into a free and an interaction part if we use the ansatz \eqref{ansatzlinH}:
\eq{\label{perturbationansatz}
\mathcal S =\, \frac{1}{2\pi \alpha'} \int_\Sigma d^2z \, \bigl(\delta_{ab}\partial X^a \bar \partial X^b + \frac{1}{3}H_{abc} X^a \partial X^b \bar \partial X^c \bigr) =:\, \mathcal S_0 + \mathcal S_1 \;,
}
where the first part $\mathcal S_0$ determines the (free) propagator, which we recall to be
\eq{
  \label{freepropagator}
  \bigl\langle X^a(z_1,\ov z_1 )\, X^b(z_2,\ov z_2 ) \bigr\rangle_{0}
  = - \frac{\alpha'}{2} \log |z_1-z_2|^2 \:\delta^{ab}
  \;,
}
and the second part $\mathcal S_1$ is treated as a perturbation for small flux, as it is for example the case in the large volume or dilute flux limit. Correlation functions can then be computed in the standard way by path integration:
\eq{
  \label{pathintegration}
  \langle \mathcal O_1 \dots \mathcal O_N \rangle
  = \frac{1}{\mathcal Z} \int [dX]
  \,
    \mathcal O_1 \dots \mathcal O_N \,e^{-\mathcal S[X]}
  \,,
}
where we denote the vacuum functional by ${\cal Z} = \int[dX]\,e^{-{\cal S}[X]}$. Expanding the expression \eqref{pathintegration} into a power series in the perturbation and denoting the expectation value with respect to the free action by $\langle \dots \rangle_0$ gives, up to first order in the flux,
\eq{\label{pertexp}
\langle {\cal O}_1 \dots {\cal O}_n \rangle =&\, \langle  {\cal O}_1 \dots {\cal O}_n \rangle_0 -\Bigl(\langle {\cal O}_1 \dots {\cal O}_n {\cal S}_1 \rangle_0 - \langle  {\cal O}_1 \dots {\cal O}_n\rangle_0 \langle {\cal S}_1 \rangle_0 \Bigr) \\
&\, +\frac{1}{2}\Bigl(\langle  {\cal O}_1 \dots {\cal O}_n{\cal S}_1 ^2 \rangle_0 - \langle  {\cal O}_1 \dots {\cal O}_n\rangle_0\langle {\cal S}_1 ^2 \rangle_0\Bigr) \\
&\,-\Bigl(\langle {\cal O}_1 \dots {\cal O}_n {\cal S}_1 \rangle_0 - \langle  {\cal O}_1 \dots {\cal O}_n\rangle_0 \langle {\cal S}_1 \rangle_0\Bigr)\langle {\cal S}_1\rangle_0 + \dots \\
=&\,\langle  {\cal O}_1 \dots {\cal O}_n \rangle_0 -\langle {\cal O}_1 \dots {\cal O}_n {\cal S}_1 \rangle_0 + {\cal O}(H^2)\;,
}
where in the second step we used the fact that $\langle {\cal S}_1 \rangle_0 = 0 $ because ${\cal S}_1$ is a product of an odd number of fields. We will use this formula in the following to compute two- and three-point correlation functions of various fields.

\paragraph{Currents and correlation functions}
The standard currents of the free theory $J^a = i\partial X^a$ and $\bar{J}^a = i\bar \partial X^a$ are not holomorphic and anti-holomorphic objects any more. Even at linear order in the flux, the classical equations of motion \eqref{classeom} show, that the $H$-flux mixes holomorphic and anti-holomorphic parts.

To get proper holomorphic and anti-holomorphic currents we have to add additional terms to $J^a, \bar J^a $ that compensate the mixed terms in \eqref{firstordersolution} at least up to linear order in the flux. We therefore propose the following currents:
\eq{
  \label{holcurrents}
  {\cal J}^a (z,\bar z) & = i\partial X^a(z,\bar z)-
  {\textstyle \frac{i}{2}}\hspace{0.5pt}H^a{}_{bc} \, \partial X^b(z,\bar z) \,
  {X}^c(z, \bar z)  \;, \\[1mm]
  \bar {\cal J}^a (z,\bar z) & = i\bar \partial X^a(z,\bar z)-
  {\textstyle \frac{i}{2}}\hspace{0.5pt}H^a{}_{bc} \, X^b(z,\bar z) \,
  \bar \partial X^c(z,\bar z)  \;.
}
By using \eqref{firstordersolution}, it is easy to see that they are holomorphic and anti-holomorphic fields \emph{up to linear order in the flux}. We therefore introduce the notation ${\cal J}^a(z)$ and $\bar{\cal J}^a(\bar z)$ meaning (anti-)holomorphicity up to linear order, i.e.
\eq{
\bar \partial {\cal J}^a (z) =\,0 + {\cal O}(H^2)\;, \qquad \partial \bar{\cal J}(\bar z) =\, 0 + {\cal O}(H^2) \;.
}
Let us now compute two-point functions of the currents ${\cal J}$ and $\bar {\cal J}$. The interaction Lagrangian ${\cal S}_1$ in \eqref{perturbationansatz} has an odd number of fields and therefore, up to linear order in the flux, only the first term in \eqref{pertexp} is relevant:
\eq{
  \label{twocurrentcor}
   \bigl\langle {\cal J}^a (z_1) {\cal J}^b (z_2) \bigr\rangle
  &=\, \bigl\langle i\partial X^a(z_1,\ov z_1)\, i\partial X^b(z_2,\ov z_2) \bigr\rangle_0
  =\,  \frac{\alpha'}{2} \:\frac{1}{(z_1-z_2)^2}\:\delta^{ab} \;, \\
  \bigl\langle \ov{\cal J}{}^a (\ov z_1)\, \ov{\cal J}{}^b (\ov z_2) \bigr\rangle
  &=\, \bigl\langle i\ov \partial X^a(z_1,\ov z_1)\, i\ov \partial X^b(z_2,\ov z_2) \bigr\rangle_0
  =\, \frac{\alpha'}{2} \:\frac{1}{(\ov z_1-\ov z_2)^2}\:\delta^{ab} \;, \\
  \bigl\langle {\cal J}^a (z_1) \ov{\cal J}^b (\ov z_2) \bigr\rangle
  &=\, \bigl\langle i\partial X^a(z_1,\ov z_1)\, i\ov \partial X^b(z_2,\ov z_2) \bigr\rangle_0
  =\, 0\;.
}
We observe that up to linear order in the flux, the two-point functions of the new currents are the same as for the free theory. Let us therefore move on to the three-point functions. Now, the interaction Lagrangian in \eqref{pertexp} contributes to the result. As an example, for three holomorphic currents, we have to compute:
\eq{
\bigl\langle {\cal J}^a(z_1)\,&{\cal J}^b(z_2)\,{\cal J}^c(z_3) \bigr\rangle \\
&=\,i\bigl\langle \partial X^a(z_1,\bar z_1)\,\partial X^b(z_2,\bar z_2)\, \partial X^c(z_3,\bar z_3) {\cal S}_1 \bigr\rangle_0 \\
&=\,  \frac{i\,H_{pqr}}{6\pi\alpha'} \int_{\Sigma} d^2 z \,
  \bigl\langle  \partial X^a(z_1,\bar z_1)\, \partial X^b(z_2,\bar z_2)\, \partial X^c(z_3, \bar z_3) \times  \\
   &\hspace{80pt}\times \,X^p(z,\ov z ) \,  \partial X^q(z, \ov z) \,\ov\partial X^r (z, \ov z)
    \bigr\rangle_0 \;.
}
The last expression can now be evaluated by using Wick's theorem and the free propagator \eqref{freepropagator}. The computation is straightforward and we only want to mention the mixed holomorphic and anti-holomorphic derivatives of the propagator, as it involves the two-dimensional delta function:
\eq{
  \partial_{z_1} \ov \partial_{z_2} \log |z_1-z_2|^2 = -2\pi \, \delta^{(2)}(z_1-z_2) \;.
}
In contrast to the free theory, there are now non-vanishing three-point functions of purely holomorphic and purely anti-holomorphic currents. Taking the anti-symmetry of the $H$-flux into account, the result of applying Wick's theorem is given by:
\eq{
  \label{3currentcor}
 \bigl\langle {\cal J}^a(z_1)\, {\cal J}^b(z_2)\, {\cal J}^c(z_3) \bigr\rangle
 &=  -i\:\frac{{\alpha'}^2}{8}\,H^{abc}\:
   \frac{1}{ z_{12}\, z_{23}\, z_{13}} \;, \\
  \bigl\langle \ov{\cal J}{}^a(\ov z_1)\, \ov{\cal J}{}^b(\ov z_2)\,
  \ov {\cal J}{}^c(\ov z_3) \bigr\rangle
  &= +i\:\frac{{\alpha'}^2}{8}\,H^{abc}\:
   \frac{1}{ \ov z_{12}\, \ov z_{23}\, \ov z_{13}} \;,
}
and all the mixed holomorphic and anti-holomorphic currents vanish. We use the standard notation $z_{ij} = z_i - z_j$. Note again that for the non-vanishing of the above three-point correlators already at first order in the flux, the interaction term ${\cal S}_1$ was crucial.

\paragraph{Basic three-coordinate correlator}
We take the holomorphicity and anti-holomorphicity of the currents ${\cal J}^a$ and $\bar {\cal J}^a$ respectively as a motivation to introduce new coordinates ${\cal X}^a$, which are the integrals of the currents, i.e. we define them by the following relation:
\eq{\label{rightcoords}
  \mathcal J^a(z) =\, i\hspace{0.5pt}\partial \mathcal X^a(z,\ov z) \;,\quad
  \ov{\mathcal J}{}^a(\ov z) =\, i\hspace{0.5pt}\ov \partial \mathcal X^a(z,\ov z) \;.
}
As we will see later, the currents ${\cal J}^a$ are the proper conformal fields in our theory and therefore we propose the coordinates defined by \eqref{rightcoords} to be the right objects to get information about the target space geometry. Whereas the two-point function does not change in comparison to the uncorrected fields $X^a$ at linear order in the flux, the tree-point function can be determined by integrating e.g. the first correlator in \eqref{3currentcor}. To state the result, we use the \emph{Rogers dilogarithm} function of a complex variable $L(z)$, defined by:
\eq{
  L(z)={\rm Li}_2(z) + \frac{1}{ 2} \log (z) \log(1-z)\; ,
}
where ${\rm Li}_2(z)$ is the standard dilogarithm function\footnote{We refer the reader to the appendix for a short introduction and properties of the Rogers- and standard dilogarithm functions.}. The three-point correlator of the fields ${\cal X}^a$ is now given by: \eq{
  \label{3calXcor}
  &\bigl\langle {\cal X}^a(z_1,\ov z_1)\, {\cal X}^b(z_2,\ov z_2)\,   {\cal X}^{c}(z_3,\ov z_3) \bigr\rangle
    \\
  &\hspace{80pt}= \frac{{\alpha'}^2}{12}\: H^{abc}
   \biggl[ L \Bigl( { \frac{z_{12}}{z_{13}} } \Bigr)
  + L \Bigl( { \frac{z_{23}}{z_{21}} } \Bigr)
  + L \Bigl( { \frac{z_{13}}{z_{23}} } \Bigr) - {\rm c.c.} \biggr]
  + F \;.
}
Here, we included the integration constants in a single function $F$. It is determined by the condition $\partial_i \partial_j \partial_k F = 0$, where $i \in \{z_1,\bar z_1\}$, $j\in \{z_2, \bar z_2\}$, $k \in \{z_3,\bar z_3\}$. This is similar to the propagator for standard coordinate fields $X^a(z,\bar z)$. It is only determined up to integration constants:
\eq{
  \label{wick01b}
  \bigl\langle X^a(z_1,\ov z_1 )\, X^b(z_2,\ov z_2 ) \bigr\rangle_{0}  = -\frac{\alpha'}{2}
  \Bigl(\log |z_1-z_2|^2 + f(z_1,\ov z_1) + f(z_2,\ov z_2) \Bigr)\:\delta^{ab}.
}
But in the two-point case, one can show that physical amplitudes do not depend on the integration constants and therefore they can be set to zero. At the moment a similar statement for the three-point function is not possible but we still set $F = 0 $ in the following.

To further simplify the result \eqref{3calXcor}, we introduce the flux parameter $\theta^{abc} = \tfrac{(\alpha')^2}{12}\,H^{abc}$ and use the following function which is composed of Rogers dilogarithms with characteristic arguments:
\begin{equation}
\label{Rogersum}
   {\cal L}(z)=L(z)+L\left( 1-\frac{1}{z}\right)
    + L\left(\frac{1}{1-z}\right)\; .
\end{equation}
Taking this into account, we can rewrite the three-point correlator in the following compact form:
\begin{equation}
  \label{3coordinatecor}
  \bigl\langle {\cal X}^a(z_1,\ov z_1)\, {\cal X}^b(z_2,\ov z_2)\,    {\cal X}^{c}(z_3,\ov z_3) \bigr\rangle
  = {\theta^{abc}} \Bigl[{\cal L}
   \bigl( {\textstyle \frac{z_{12}}{ z_{13}} }\bigr) - {\cal L}
   \bigl({\textstyle \frac{\ov z_{12}}{ \ov z_{13}}}\bigr) \Bigr]\; .
\end{equation}
This remarkable result will be the main source of a deformed product on the target space as we will see later by computing scattering amplitudes of tachyon vertex operators. But before moving on in this direction, we will first show how a conformal field theory can be constructed up to linear order in the $H$-flux and how T-duality is realized.

\paragraph{Conformal field theory linear in $H$}

To first order in the flux, we were able to define holomorphic and anti-holomorphic currents ${\cal J}^a(z)$ and $\bar{\cal J}^a(\bar z)$. The goal of this section is to give the main arguments that, up to linear order in $H$, it is possible to construct a conformal field theory. We will analyse the operator product expansions of the currents and then define an energy momentum operator for which the currents are primary fields of dimension one, which is known as the Sugawara construction. Finally, in the next subsection, we will introduce tachyon vertex operators.

The operator product expansion\footnote{For standard techniques in conformal field theory we refer the reader to the literature, e.g. \cite{Ginsparg:1988ui, DiFrancesco:1997nk, Blumenhagen:2009zz}.} (OPE) of two currents ${\cal J}^a(z_1)$, ${\cal J}^b(z_2)$ and their anti-holomorphic counterparts can be derived by computing correlation functions with other fields. The singular part of the OPE can be fixed by looking at the two- and three-point functions \eqref{twocurrentcor} and \eqref{3currentcor} of the currents:
\eq{
  \label{currentope}
{\cal J}^a(z_1)\; {\cal J}^b(z_2)
  &=\,
   \frac{\alpha'}{2} \,\frac{\delta^{ab}}{(z_1-z_2)^2}
  -
  \frac{\alpha'}{4}\, \frac{i\, H^{ab}{}_c}{z_1-z_2}\: {\cal J}^c(z_2)
  +
   {\rm reg.} \;, \\
   \ov{\cal J}{}^a(\ov z_1)\; \ov{\cal J}{}^b(\ov z_2)
  &=\,
   \frac{\alpha'}{2} \,\frac{\delta^{ab}} {(\ov z_1-\ov z_2)^2}
  +
   \frac{\alpha'}{4}\, \frac{i\, H^{ab}{}_c}{\ov z_1-\ov z_2}\: \ov{\cal
    J}{}^c(\ov z_2)
  +
   {\rm reg.}\;.
}
We denoted the regular part by ``reg.'' and the OPE of a holomorphic and anti-holomorphic current is purely regular. As a next step, let us construct the energy momentum tensor in a way that the ${\cal J}^a$ are the right currents of the theory:
\eq{
  \label{emtensor}
  {\cal T}(z) =\, \frac{1}{\alpha'}\: \delta_{ab}:\! {\cal J}^a {\cal J}^b\!:\! (z) \;, \quad
  \ov{\cal T}(\ov z) =\, \frac{1}{\alpha'}\: \delta_{ab} :\! \ov{\cal J}{}^a\ov{\cal J}{}^b\!:\! (\ov z) \;.
}
As we are working only up to first order in the flux, we have to check whether all the axioms of an energy momentum tensor are really obeyed up to this order. At first, the OPEs for two energy momentum tensors have the right form, as we can check using the anti-symmetry of $H$:
\eq{
  \label{OPE_02}
     {\cal T}(z_1)\; {\cal T}(z_2)
& =\,
    \frac{c/2}{(z_1-z_2)^4} + \frac{2\, {\cal T}(z_2)}{(z_1-z_2)^2}
  + \frac{\partial \, {\cal T}(z_2)}{z_1-z_2}
  + {\rm reg.} \;,\\
  \ov{\cal T}(\ov z_1)\; \ov {\cal T}(\ov z_2)
  &=\,
    \frac{c/2}{(\ov z_1-\ov
   z_2)^4} + \frac{2\, \ov{\cal T}(\ov z_2)}{(\ov z_1-\ov z_2)^2}
  + \frac{\ov \partial \,\ov{\cal T}(\ov z_2)}{\ov z_1-\ov z_2}
  + {\rm reg.} \;,
}
whereas the OPE of a holomorphic and an anti-holomorphic energy-momentum tensor is purely regular. This result has the canonical form and thus, as in the standard case, we get two copies of the Virasoro algebra with the same central charge as for the free theory (for the case of the three-dimensional background we have $c=3$). For the second step, we have to check whether the currents are conformal primary of dimension one. Calculating the OPE with the energy-momentum tensor \eqref{emtensor}, we get
\eq{
  \label{emcurrentope}
  {\cal T}(z_1)\, {\cal J}^a(z_2)
 & =\,  \frac{{\cal J}^a(z_2)}{(z_1-z_2)^2} +
     \frac{\partial {\cal J}^a(z_2)}{z_1-z_2}+{\rm reg.} \;, \\
  \ov{\cal T}(\ov z_1)\, {\cal J}^a(z_2)
 &=\, {\rm reg.} \;.
}
Similar OPEs can be computed for the anti-holomorphic parts. To put it in a nutshell, up to linear order in the flux, we are able to define an energy-momentum operator \eqref{emtensor} with respect to which the currents ${\cal J}^a(z)$ and $\bar{\cal J}^a(\bar z)$ are conformal primary fields of dimension $(1,0)$ and $(0,1)$, respectively. Furthermore, their OPEs \eqref{currentope} have the form of a non-Abelian current algebra with structure constants $H^{ab}{}_c$. Up to linear order in the flux, we therefore made the first steps to construct a conformal field theory framework. In the following we want to denote this theory by $\textrm{CFT}_H$.

\subsubsection{Tachyon vertex operators}
\label{subsec-tachyonvertex}
In the last subsection we defined a conformal field theory framework $\textrm{CFT}_H$. We are now going to introduce a new set of primary fields with respect to the energy momentum tensor \eqref{emtensor}, which will allow us to extract information about the product of functions on the target space. The tachyon vertex operator is the simplest operator with non-trivial functional dependence on the coordinates. In analogy to free closed string theory with momentum $p_a$ and winding $w^b$, we define the right- and left moving momenta in the standard way
\eq{
  \label{klr}
  k_{L}^a = p^a + \frac{w^a}{\alpha'} \;, \qquad
  k_{R}^a = p^a - \frac{w^a}{\alpha'} \;.
}
With this definition, we propose the following tachyon vertex operator by using the redefined spacetime coordinates ${\cal X}^i$:
\eq{
  \label{tachyonvertex}
  {\cal V}(z,\ov z) = \, :\!\exp \bigl( i\hspace{0.5pt} k_L\cdot {\cal X}_L +
  i\hspace{0.5pt} k_R \cdot {\cal X}_R \bigr) \!: \;,
}
where the left- and right-moving coordinates ${\cal X}^a_{L/R}$ are given by integration of the currents ${\cal J}$ and $\bar{\cal J}$, respectively. In addition, we denoted the contraction by $k \cdot {\cal X} = k_a {\cal X}^a$. Furthermore, by integration of \eqref{currentope}, we can write down the OPE of a current ${\cal J}^a$ and a coordinate ${\cal X}_L^b$ and its anti-holomorphic counterpart:
\eq{
  \label{currentcoordope}
  {\cal J}^a(z_1)\, {\cal X}^b_L(z_2)
  =&\,
    -i\:\frac{\alpha'}{2} \:\frac{\delta^{ab}}{z_1-z_2}
  +
   \frac{\alpha'}{4}\,  H^{ab}{}_c\; {\cal J}^c(z_2)\, \log(z_1-z_2)
  +{\rm reg.} \;, \\
   \ov{\cal J}{}^a(\ov z_1)\, {\cal X}^b_R(\ov z_2)
  =&\,
   -i\:\frac{\alpha'}{2} \:\frac{\delta^{ab}}{\ov z_1-\ov z_2}
  -
    \frac{\alpha'}{4}\, H^{ab}{}_c\; \ov{\cal  J}{}^c(\ov z_2)\, \log(\ov z_1-\ov z_2)
  + {\rm reg.} \;.
}
Now let us come back to the tachyon vertex operator \eqref{tachyonvertex}. We have to show that it is a primary field in $\textrm{CFT}_H$ and we want to compute its conformal dimension. Note, that this is not clear a priori since we have to use the singular parts in \eqref{currentcoordope} which contain a logarithmic part. But using the anti-symmetry of the flux, the unwanted terms cancel and we get:
\eq{
  {\cal T}(z_1)\, {\cal V}(z_2,\ov z_2)
   =&\,
    \frac{1}{(z_1-z_2)^2} \, \frac{\alpha' k_L\cdot k_L}{4} \,   {\cal V}(z_2,\ov z_2)
   +
    \frac{1}{z_1-z_2} \, \partial {\cal V}(z_2,\ov z_2)
   +{\rm reg.} \;, \\
    \ov{\cal T}(\ov z_1)\, {\cal V}(z_2,\ov z_2)
   =&\,
    \frac{1}{(\ov z_1-\ov z_2)^2} \, \frac{\alpha' k_R\cdot k_R}{4} \,  {\cal V}(z_2,\ov z_2)
   +
   \frac{1}{\ov z_1-\ov z_2} \, \ov\partial {\cal V}(z_2,\ov z_2)
   +{\rm reg.} \;.
}
This shows that the vertex operator defined in \eqref{tachyonvertex} indeed is a conformal primary field of dimension $(h,\bar h)=(\tfrac{\alpha'}{4}k_L ^2, \tfrac{\alpha'}{4} k_R ^2)$ and therefore corresponds to a physical quantum state of the conformal field theory $\textrm{CFT}_H$.

Up to now, the discussion was quite similar to the free theory, so there is the question about the difference of the physics caused by the linear perturbation in the flux. To see one of the differences, let us compute the operator product expansion of a current ${\cal J}^a $ and the vertex operator ${\cal V}$:
\eq{
  \label{currentvertexope}
   {\cal J}^a(z_1)\, {\cal V}(z_2,\ov z_2) =&\,
   \frac{1}{z_1-z_2}  \frac{\alpha' k^a_L}{2}\,   {\cal V}(z_2,\ov z_2)\\
   &\,+i\,\frac{\alpha'}{4} \log (z_1-z_2) \,  H^{a}{}_{bc}\; k^b_L :\! {\cal  J}{}^c\, {\cal
     V}\!:\!(z_2,\ov z_2) + {\rm reg.} \;.
}
One way to see this is to expand the exponential in the vertex operator ${\cal V}$ and then use the OPE of a current and a coordinate \eqref{currentcoordope} in every term.
A similar statement holds for the anti-holomorphic part. Again, as in \eqref{currentcoordope}, the flux-dependent term contains a logarithmic part in contrast to the free theory. It would be possible to eliminate the term by simply demanding $H^a{}_{bc} k_L ^b = 0 $, which means that the momenta are transversal to the flux. But this would eliminate a lot of interesting phenomena so we keep these terms. We interpret them as a hint that the $\textrm{CFT}_H$ is a logarithmic conformal field theory.

Let us continue and use the result \eqref{currentvertexope} to determine the center of mass momentum and winding in analogy to the free case. Denote the zero mode of the current ${\cal J}^a$ by ${\cal P}_L ^a$. It is given via Cauchy's theorem by a contour integral over the current. We want to determine the ${\cal P}_L ^a$-eigenvalue of the vertex operator ${\cal V}$. This can be done by the definition of ${\cal P}_L^a$ and the result \eqref{currentvertexope} by acting on the vacuum state $|0\rangle$ as we now show:
\eq{
  \label{momentumcom}
     \lim_{z_2,\ov z_2 \to 0} {\cal P}_L^a\, {\cal V}(z_2,\ov z_2)\bigr|0\bigr\rangle
     &=
     \lim_{z_2,\ov z_2\to 0} \oint \frac{dz_1}{2\pi i} \, {\cal J}^a(z_1)\, {\cal V}(z_2,\ov z_2)
     \bigr|0\bigr\rangle \\
     &= \frac{\alpha' \hspace{0.5pt}k_L^a}{2} \lim_{z_2,\ov z_2\to 0} {\cal V}(z_2,\ov z_2)
     \bigr|0\bigr\rangle \;.
}
The logarithmic part drops out because we could add and subtract the following term:
\eq{
i\,\frac{\alpha'}{4} \log (\bar z_1- \bar z_2) \,  H^{a}{}_{bc}\; k^b_L :\! {\cal  J}{}^c\, {\cal
     V}\!:\!(z_2,\ov z_2)\;,
}
which is regular in the \emph{holomorphic} variables. But this changes the logarithmic part in the OPE \eqref{currentvertexope} into a single valued real function and the contour integral of the second term in \eqref{momentumcom} vanishes.

By looking at the action \eqref{perturbationansatz}, we see that in this case the canonical momentum does not coincide with the physical momentum. Whereas the former is related to the currents ${\cal J}^a$, the latter is simply the unperturbed current $J^a = \partial X^a (z,\bar z)$. Thus the zero mode ${\cal P}_L ^a$ corresponds to the canonical momentum. To determine the physical momentum, we have to do the same calculation but with the current $i\partial X^a (z,\bar z)$ whose zero mode is denoted by $P_L ^a$:
\eq{
  \label{momencom2}
   &\lim_{z_2,\ov z_2\to 0} {P}_L^a\, {\cal V}(z_2,\ov z_2) \bigr|0\bigr\rangle \\
   =&\lim_{z_2,\ov z_2\to 0} \oint \frac{dz_1}{2\pi i}\,
   J^a(z_1)\, {\cal V}(z_2,\ov z_2)\bigr|0\bigr\rangle   \\
   =& \lim_{z_2,\ov z_2\to 0} \oint \frac{dz_1}{2\pi i}
   \left[ {\cal J}^a(z_1)\, {\cal V}(z_2,\ov z_2)
    +\frac{1}{2}\, H^a{}_{bc} \, J^b(z_1) \, X^c_R(\ov z_1)\, {V}(z_2,\ov z_2)\right]
       \bigr|0\bigr\rangle \;.
}
In the last expression, the first term was computed in \eqref{momentumcom}. The second term is linear in the flux and therefore we can use the standard OPEs of the free theory to evaluate the integral. But both, the OPE of a coordinate $X^a(z,\bar z)$ and a current $i\partial^a X(z)$ with the standard tachyon vertex operator (which we denote by $V$ instead of ${\cal V}$ in the perturbed case) gives the momentum $k_{L/R}$ times a singular part. Thus the second term in \eqref{momencom2} is proportional to $H^a{}_{bc} k_L^b k_R^c$. If the vertex operator ${\cal V}(z,\bar z)$ should carry \emph{physical} momentum $(k_L,k_R)$, we consequently have to propose the following condition:
\eq{
  \label{momentumcondition}
  0= H^a{}_{bc}\, k_L^b\, k_R^c \simeq
  H^a{}_{bc}\, p^b\, w^c  \simeq  \bigl[\, \vec p \times \vec w \,\bigr]^a\;,
}
where we used that the $H$-flux is a function times the epsilon-tensor in three dimensions. The last condition can also be seen in a different way: We take the classical part of the free solution to be:
\eq{
\mathsf X_0 ^a(\sigma, \tau) =\, x_0^a + \alpha' p^a\tau + w^a \sigma\;.
}
Using this in the definition \eqref{holcurrents} of the currents (restricting every field to the classical part) and requiring the following center of mass momentum and winding:
\eq{
   \frac{1}{ 2\pi} \int_0^{2\pi} d\sigma\, \partial_\tau {\cal X}^a = \alpha'\, p^a \;,
   \qquad
   \frac{1}{2\pi} \int_0^{2\pi} d\sigma\, \partial_\sigma {\cal X}^a =  w^a \;,
}
we again arrive at the condition \eqref{momentumcondition} for the left- and right-moving momenta\footnote{We neglected terms of the form $H^a{}_{bc}x_0^b k_{L/R}^c$ because they correspond to a constant $B$-field and can be locally gauged away.} or equivalently for the center of mass momentum $p^a$ and winding $w^a$. This derivation can be seen as the classical analogue of the quantum derivation.

In a summary, if the condition \eqref{momentumcondition} is obeyed, the tachyon vertex operator \eqref{tachyonvertex} has physical momentum $(k_L,k_R)$ and is a well-defined quantum state in $\textrm{CFT}_H$. The consequences of the condition \eqref{momentumcondition} will play an important role in the analysis of T-duality in the next section.

\subsubsection{T-duality in $\textrm{CFT}_H$}
Having established some of the key features of a conformal field theory with currents ${\cal J}^a, \bar{\cal J}^b$, we are now able to give the action of T-duality on the coordinate fields ${\cal X}_L^a, {\cal X}_R ^b$, which were defined by \eqref{rightcoords}. It is given in the standard way by reflecting the right-moving sector of the theory:
\eq{ \label{tdualityaction}
  \begin{array}{c}
  {\cal X}_L^a(z) \\[1mm]
  {\cal  X}_R^a(\ov z)
  \end{array}
  \qquad
  \xrightarrow{\;\mbox{\scriptsize T-duality}\;}
  \qquad
  \begin{array}{c}
  +{\cal X}_L^a(z)\;, \\[1mm]
  -{\cal  X}_R^a(\ov z) \;.
  \end{array}
}
Because of the definition \eqref{rightcoords} of the target space coordinates by their relation to the (anti-)holomorphic currents, for the latter T-duality is realized by
\eq{
  \begin{array}{c}
  {\cal J}^a(z) \\[1mm]
  \ov{\cal  J}{}^a(\ov z)
  \end{array}
  \qquad
  \xrightarrow{\;\mbox{\scriptsize T-duality}\;}
  \qquad
  \begin{array}{c}
  +{\cal J}^a(z)\;, \\[1mm]
  -\ov{\cal  J}{}^a(\ov z) \;.
  \end{array}
}
The goal of the following sections is to compute correlation functions of tachyon vertex operators in different T-duality frames. We recall that T-duality exchanges momentum and winding quantum numbers and therefore one expects that scattering three pure momentum states in the $H$-flux background results in the same amplitude as scattering two momentum states and one winding state in the $f$-flux background, where the T-duality is done in the same direction where we exchange momentum and winding quantum numbers. The same statement holds for T-dualities in more than one direction\footnote{Note that on the level of the world sheet, we can do an arbitrary number of T-dualities because of \eqref{tdualityaction}. We refer to the background with every direction being T-dualized as the $R$-flux background, even if its target space interpretation is not clear.}. Conversely, scattering two momentum states and one winding state in the $H$-flux background should correspond to scattering of three momentum states in the $f$-flux background. This is reflected in table \ref{tablemomwind}, where we also indicated the sign between the holomorphic- and anti-holomorphic parts of the three-${\cal X}$-correlator and whether or not the momentum-winding condition \eqref{momentumcondition} is satisfied in the $H$-flux background. Only in this case it is ensured that momentum and winding quantum numbers of the tachyon state are not corrected by $H$-dependent terms as was shown in \eqref{momencom2}. Having this condition, we expect that any effect linear in $H$ is not caused by the linear redefinition of the classical solution of the tachyon but is determined by the properties of the uncorrected solution.
\begin{table}[h]
\centering
\renewcommand{\arraystretch}{1.2}
\begin{tabular}{|cc|cc|cc|cc|}
\hline
\multicolumn{2}{|c}{$H$-flux} &
\multicolumn{2}{|c}{$f$-flux} &
\multicolumn{2}{|c}{$Q$-flux} &
\multicolumn{2}{|c|}{$R$-flux}
\\ \hline\hline
$\langle p_1,p_2,p_3\rangle^-$ & \checkmark &
$\langle p_1,p_2,w_3\rangle^-$ & \checkmark &
$\langle p_1,w_2,w_3\rangle^-$ & \checkmark &
$\langle w_1,w_2,w_3\rangle^-$& \checkmark
\\ \hline
$\langle p_1,p_2,w_3\rangle^+$ & $\times$ &
$\langle p_1,p_2,p_3\rangle^+$ & $\times$ &
$\langle p_1,w_2,p_3\rangle^+$ & $\times$ &
$\langle w_1,w_2,p_3\rangle^+$ & $\times$
\\ \hline
$\langle p_1,w_2,w_3\rangle^-$ &  $\times$ &
$\langle p_1,w_2,p_3\rangle^-$ & $\times$ &
$\langle p_1,p_2,p_3\rangle^-$ & $\times$ &
$\langle w_1,p_2,p_3\rangle^-$ & $\times$
\\ \hline
$\langle w_1,w_2,w_3\rangle^+$ &  \checkmark &
$\langle w_1,w_2,p_3\rangle^+$ & \checkmark &
$\langle w_1,p_2,p_3\rangle^+$ & \checkmark &
$\langle p_1,p_2,p_3\rangle^+$ & \checkmark
\\ \hline
\end{tabular}
\caption{\small Scattering of momentum and winding states in different T-dual backgrounds. Momentum states are indicated by $p_i$ and winding states by $w_i$. Correlation functions in the same row should give the same information as in the direction of the T-duality action momentum and winding states were exchanged. The superscript indicates the sign between the holomorphic- and anti-holomorphic part in the correlator \eqref{3calXcor} and the symbols $\checkmark$/$\times$ denote the validity of  condition \eqref{momentumcondition}.\label{tablemomwind}}
\end{table}

\emph{Remarks}. To extract information about the geometry which determines the effective field theory, we are interested in scattering amplitudes of pure momentum states in the various backgrounds. The relevant tachyon correlation functions in table \ref{tablemomwind} would be the first one in the column for the $H$-flux, the second in the column for the $f$-flux, the third one in the column for the $Q$-flux and finally the last one in the column for the $R$-flux. From the table, we infer that the momentum-winding condition \eqref{momentumcondition} is only true in the $H$-flux and in the $R$-flux background. The latter corresponds to scattering pure winding states in the $H$-flux background. The corresponding basic three-point correlation functions of the coordinates are given in these two cases by:
\eq{
 \bigl\langle {\cal X}^a(z_1,\ov z_1)\, {\cal X}^b(z_2,\ov z_2)\, {\cal X}^c(z_3,\ov z_3)\bigr\rangle^{-}& = \,
    {\theta^{abc}} \Bigl[{\cal L}
   \bigl( {\textstyle \frac{z_{12}}{ z_{13}} }\bigr) - {\cal L}
   \bigl({\textstyle \frac{\ov z_{12}}{ \ov z_{13}}}\bigr) \Bigr]\;, \\
\bigl\langle {\cal X}^a(z_1,\ov z_1)\, {\cal X}^b(z_2,\ov z_2)\, {\cal X}^c(z_3,\ov z_3)\bigr\rangle^{+}&=\,
{\theta^{abc}} \Bigl[{\cal L}
   \bigl( {\textstyle \frac{z_{12}}{ z_{13}} }\bigr) + {\cal L}
   \bigl({\textstyle \frac{\ov z_{12}}{ \ov z_{13}}}\bigr) \Bigr]\; .
}
The sign change between the holomorphic and anti-holomorphic part in the last expression is caused by the reflection of the right-moving part of the coordinates.

We will use these results in the following section to calculate $n$-point tachyon correlation functions. Due to the validity of the momentum-winding condition \eqref{momentumcondition} we will be able to get information about the product of functions on the target space in the presence of $H$-flux and in its complete T-dual case, the $R$-flux.

\subsubsection{Tachyon correlation functions}

We now have all the ingredients to compute $n$-point tachyon correlation functions in the $\textrm{CFT}_H$-framework and analyse the effect of T-duality. In particular we want to infer the product structure of functions on the target space and compare the case of the original $H$-flux with its complete T-dual version, where the $R$-flux is present.

\paragraph{Three tachyon correlator}
Let us start by scattering three tachyons. To have winding, let us assume the tree-dimensional target space to be compact, e.g. a three-torus. We are going to compare the following two cases: First we want to analyse the standard case of pure momentum states in the presence of $H$-flux, where we label the momenta by 3-vectors $p_i$. The corresponding tachyon vertex operator was discussed in the previous two subsections and is given by
\eq{
{\cal V}_i^- \stackrel{\textrm{def}}{=}\, {\cal V}_{p_i}(z_i, \bar z_i) = \;:\exp\left(i\,p_i \cdot {\cal X}(z_i,\bar{z}_i)\right):\;.
}
Secondly, we are interested in pure momentum state scattering in the presence of $R$-flux, which corresponds to pure winding state scattering in the $H$-flux background, as can be seen in table \ref{tablemomwind}. We denote the winding by three-vectors $w_i$ and the corresponding tachyon vertex operator is given by
\eq{ \label{windingtachyon}
{\cal V}_i ^+ \stackrel{\textrm{def}}{=}\,{\cal V}_{w_i}(z_i,\bar z_i) =\; :\exp\left(i\,w_i \cdot \tilde{\cal X}(z_i,\bar{z}_i) \right):\;,
}
where the T-dual coordinate is denoted by $\tilde{\cal X} = {\cal X}_L - {\cal X}_R $. To proceed, let us note the following: As we are interested in scattering momentum states in an $R$-flux background, we set $w_i \rightarrow p_i$ in the vertex operator \eqref{windingtachyon}. The superscript $\langle \dots\rangle^\mp$ on correlators indicates that we are in the $H$-flux case for the minus sign and in the $R$-flux case for the plus sign. The result for the three-tachyon correlation function (for details of the calculation we refer the reader to the appendix) is then given by
 \eq{
  \label{threetachyoncorrelator}
   \bigl\langle \,{\cal V}_1 \,{\cal V}_2 \,{\cal V}_3 \,\bigr\rangle^\mp
   =\frac{\delta(p_1+p_2+p_3)}{\vert z_{12}\,z_{13}\,z_{23}\vert^2}
       \Bigl[1 -i\hspace{0.5pt}\theta^{abc}\, p_{1,a} p_{2,b} p_{3,c}  \bigl[
  {\cal L}   \bigl( {\textstyle \frac{z_{12}}{z_{13}} }\bigr) \mp {\cal L}
   \bigl({\textstyle \frac{\ov z_{12}}{ \ov z_{13}}}\bigr)\bigr] \Bigr] \,.
}
As we only calculated up to first order in the flux (i.e. the parameter $\theta^{abc}$), we cannot make a definite statement about the full result of the correlator. But the form of \eqref{threetachyoncorrelator} suggests, that it is the beginning of a power series expansion of the exponential function. To indicate that this may be possible we introduce the notation $[\dots]_\theta$, meaning that the result is only valid up to linear order in $\theta^{abc}$:
 \eq{
  \label{threetachyoncorrelator1}
   \bigl\langle \,{\cal V}_1 \,{\cal V}_2 \,{\cal V}_3 \,\bigr\rangle^\mp
   =\frac{\delta(p_1+p_2+p_3)}{\vert z_{12}\,z_{13}\,z_{23}\vert^2}
       \exp\Bigl[-i\hspace{0.5pt}\theta^{abc}\, p_{1,a} p_{2,b} p_{3,c}  \bigl[
  {\cal L}   \bigl( {\textstyle \frac{z_{12}}{z_{13}} }\bigr) \mp {\cal L}
   \bigl({\textstyle \frac{\ov z_{12}}{ \ov z_{13}}}\bigr)\bigr] \Bigr]_\theta \,.
}
Let us now investigate these results. Clearly the delta-function indicates usual momentum conservation. We first analyse the amplitude without this constraint, i.e. the off-shell correlator. Let $\sigma \in S_3$ be a permutation of three elements. Then the tachyon correlator with permuted vertex operators is given by:
\begin{equation}
  \label{phasethreeperm}
  \bigl\langle \, {\cal V}_{\sigma(1)}   {\cal V}_{\sigma(2)}  {\cal V}_{\sigma(3)}  \bigr\rangle^\mp=
  \exp\Bigl[ \,i\left({\textstyle \frac{1+\epsilon}{ 2}}\right)  \eta_\sigma\,  \pi^2\,  \theta^{abc}\, p_{1,a}
  \,p_{2,b} \,p_{3,c} \Bigr]
  \bigl\langle {\cal V}_1\,  {\cal V}_2\,  {\cal V}_3  \bigr\rangle^\mp \;,
\end{equation}
where $\eta_\sigma = 1 $ for an odd permutation and vanishes for an even one. In addition, the parameter $\epsilon$ indicates the background: We have $\epsilon = -1$ for the original $H$-flux, i.e. the phase vanishes in this case, as it is expected. In the case of the 3-times T-dual background, corresponding the $R$-flux, $\epsilon = 1$ and there is a non-trivial phase factor.  To derive this result, we used the following properties of the function ${\cal L}(z)$, introduced in \eqref{Rogersum} as sum of Rogers dilogarithms:
\begin{gather}
\label{Rogersfundamentalr}
{\cal L}(z) =\,{\cal L}\bigl(1-\frac{1}{z}\bigr) =\, {\cal L}\bigl(\frac{1}{1-z} \bigr), \\
{\cal L}(z) + {\cal L}(1-z) =\, 3{\cal L}(1) =\, \frac{\pi^2}{2}\;.
\end{gather}
These formulas follow easily from the corresponding properties of the Rogers dilogarithm introduced in the appendix. As an example, consider exchanging $z_2$ and $z_3$ in the $R$-flux case:
\eq{
 \bigl\langle \, {\cal V}_{1}   {\cal V}_{3}  {\cal V}_{2}  \bigr\rangle^+=&\,\exp\Bigl[ -i\hspace{0.5pt}\theta^{abc}\, p_{1,a} p_{3,b} p_{2,c}  \bigl[
  {\cal L}   \bigl( {\textstyle \frac{z_{13}}{z_{12}} }\bigr) + {\cal L}
   \bigl({\textstyle \frac{\ov z_{13}}{ \ov z_{12}}}\bigr)\bigr] \Bigr]_{\theta} \\
    =&\,\exp\Bigl[ i\hspace{0.5pt}\theta^{abc}\, p_{1,a} p_{2,b} p_{3,c}  \bigl[
   \pi^2 - {\cal L}   \bigl( {\textstyle \frac{z_{12}}{z_{13}} }\bigr) - {\cal L}
   \bigl({\textstyle \frac{\ov z_{12}}{ \ov z_{13}}}\bigr)\bigr] \Bigr]_{\theta} \\
=&\,\exp\left[ i\hspace{0.5pt}\pi^2\theta^{abc}\, p_{1,a} p_{2,b} p_{3,c}\right]\bigl\langle \, {\cal V}_{1}   {\cal V}_{3}  {\cal V}_{2}  \bigr\rangle^+ \;.
}
Similar calculations hold for the other permutations. Thus, off-shell (without momentum conservation), we get a non-trivial phase by performing an odd permutation of the vertex operators. Note that this is similar to the open string case, where off-shell permutations of tachyon vertex operators resulted in a phase which was then seen as a sign for the Moyal-Weyl star product, as was shown in section \ref{ch-quant}, especially in the expressions \eqref{phasecorrelator}. In the closed string case, the corresponding phase is governed by the three-index object $\theta^{abc}$ \eqref{phasethreeperm}, and we will interpret this phenomenon as a hint for a three-product.

Note, that on-shell, the phase vanishes due to momentum conservation:
\eq{
  p_{1,a} \,p_{2,b}\, p_{3,c}\,\theta^{abc} = 0
  \hspace{40pt}{\rm for}\hspace{40pt} p_3=-p_1-p_2 \;.
}
This was expected, since in scattering amplitudes a product of field operators is radially ordered and therefore changing the order of the operators will not affect the amplitude. Even in the $R$-flux background this should hold, as it is one of the defining properties of conformal field theory amplitudes.

\paragraph{N-tachyon correlator}
We are now going to generalize the results of the last subsection to the case of scattering $N$ tachyons. It turns out that this can be done inductively by applying the formulas presented in the appendix. Before stating the general result, we want to demonstrate the logic at the four point result. The correlation function of four tachyon vertex operators is given by
\eq{
\label{fourtachyonscattering}
 \hspace{-40pt} &\bigl\langle {\cal V}_1\,{\cal V}_2\,{\cal V}_3\,{\cal V}_4 \bigr\rangle^\mp =
  \bigl\langle V_1\,V_2\,V_3\,V_4 \bigr\rangle^\mp_0 \;\times \\
  &\hspace{80pt} \exp \biggl[  -i \hspace{0.5pt}
   \theta^{abc}\!\! \sum_{1\leq i<j<k\leq 4}  p_{i ,a} \,p_{j,b}\, p_{k,c} 
   \Bigl[ {\cal L}\bigl({\textstyle \frac{z_{ij}}{ z_{ik}}}\bigr)
  \mp  {\cal L}\bigl({\textstyle\frac{\ov z_{ij}}{ \ov z_{ik}}}\bigr)
  \Bigr] \biggr]_{\theta} \;.
}
Although the phase becomes more complicated, it turns out that permuting the vertex operators by $\sigma \in S_4$ in the $R$-flux case results in similar phases as the ones encountered in the previous subsection, whereas they vanish in the case of the $H$-flux. We again have to use the fundamental relations of the Rogers dilogarithm \eqref{Rogersfundamentalr} to extract a phase independent of the world-sheet coordinates. Only such properties are of interest for us as they reflect true target space facts.

To illustrate this statement, let us write out explicitly the holomorphic part of the phase appearing in \eqref{fourtachyonscattering} in terms of Rogers dilogarithm functions:
\eq{
-i\theta^{abc} \Bigl[p_{1,a}p_{2,b}p_{3,c}\;{\cal L}&\left(\frac{z_{12}}{z_{13}} \right) + p_{1,a}p_{2,b}p_{4,c}\;{\cal L}\left(\frac{z_{12}}{z_{14}} \right) \\
&+ p_{1,a}p_{3,b}p_{4,c}\;{\cal L}\left(\frac{z_{13}}{z_{14}} \right) + p_{2,a}p_{3,b}p_{4,c}\;{\cal L}\left(\frac{z_{23}}{z_{24}} \right) \Bigr]\;.
}
Exchanging for example ${\cal V}_3$ and ${\cal V}_4$, in the $R$-flux background, we get a relative phase of
\eq{
i\pi^2 \theta^{abc}\left(p_{1,a}p_{3,b}p_{4,c} + p_{2,a}p_{3,b}p_{4,c} \right)\;,
}
and similar for other permutations. Note that the phase vanishes on-shell due to $p_4 = -p_1 - p_2 - p_3$ and the anti-symmetry of $\theta^{abc}$ so that the whole amplitude is invariant under permutations of the vertex operators.

The four tachyon amplitude played an important role in the history of string theory. Many important properties of the theory were found by analyzing for example the pole structure of the \emph{Virasoro-Shapiro} amplitude. In our case it is also important to do the same steps in $\textrm{CFT}_H$. It turns out that this can be done \cite{Blumenhagen:2011ph} and it is intriguingly connected to the properties of the so-called extended- or \emph{Neumann-Rogers} dilogarithm, which we briefly mention in the appendix for completeness. This analysis goes beyond the scope of this work and therefore we want to refer the reader to the original paper \cite{Blumenhagen:2011ph} for more details in this direction.

Finally we can give the general $N$-point tachyon amplitude. It is a straightforward generalization of the previous cases:
\eq{\label{Ntachyonamplitude}
  &\bigl\langle {\cal V}_1\,{\cal V}_2\,\dots \,{\cal V}_N \bigr\rangle^\mp =
  \bigl\langle V_1\,V_2\,\ldots\, V_N \bigr\rangle^\mp_0 \;\times \\
  &\hspace{93pt} \exp \biggl[  -i  \hspace{0.5pt}
   \theta^{abc}\!\! \sum_{1\leq i<j<k\leq N}  p_{i ,a} \,p_{j,b}\, p_{k,c} 
   \Bigl[ {\cal L}\bigl({\textstyle \frac{z_{ij}}{ z_{ik}}}\bigr)
  \mp  {\cal L}\bigl({\textstyle\frac{\ov z_{ij}}{ \ov z_{ik}}}\bigr)
  \Bigr] \biggr]_{\theta} \;.
}
In addition, one can extract relative phases in the case of the $R$-flux by doing a permutation $\sigma \in S_N$ of the vertex operators which are similar to the cases before and will not be calculated explicitly. They do not depend on the coordinates on the world sheet and vanish on-shell. We will use these properties of the $N$-point correlator to speculate about the existence of an $N$-product on the algebra of functions on the target space in case of the $R$-flux background in the next section.

\subsubsection{$N$-product structures}
Comparing the previous analysis to the case of the detection of open string non-commutativity as described in section \ref{ch-quant}, we want to emphasize the following similarities:
\begin{itemize}
\item The permutation of tachyon vertex operators in correlation functions leads to characteristic momentum- and background-dependent phases. In the open string case this can be seen in the expression \eqref{phasecorrelator}, whereas we calculated the phase for the case of closed string theory in the $R$-flux background in \eqref{phasethreeperm} for three tachyons and illustrated the general case in \eqref{Ntachyonamplitude}.
\item In both cases, the phases are independent of the world sheet coordinates, i.e. they can be interpreted as pure target space effects. In the open string this was a consequence of the step function $\epsilon$ and in the closed string case it was derived from the properties of the Rogers dilogarithm.
\end{itemize}
We note also the following difference: Whereas in the open string case, momentum conservation leads to vanishing phases for cyclic permutations of the vertex operators, in the closed string case it ensures vanishing of all the phases. However, off-shell the phases in the open string case already show the structure of the Moyal-Weyl product. We take this as a motivation to follow the analogy between the open and closed case one step further and try to guess the structure of the product of functions on the target space in the case of closed strings in the $R$-flux background. Looking at the phases detected in the open string case, one infers that they can be reproduced if we multiply exponential functions $e^{ip_n \cdot X}$ by using the Moyal-Weyl star product \eqref{moyalweyl}. In the same way it is easy to see that in the case of three insertion points, the phase of \eqref{phasethreeperm} can be reproduced by introducing a new product for three factors:
\eq{
\label{threebracket}
   f_1(x)\,\tri\, f_2(x)\, \tri\, f_3(x) \stackrel{\rm def}{=} \exp\Bigl(
   {\textstyle {\pi^2\over 2}}\, \theta^{abc}\,
      \partial^{x_1}_{a}\,\partial^{x_2}_{b}\,\partial^{x_3}_{c} \Bigr)\, f_1(x_1)\, f_2(x_2)\,
   f_3(x_3)\Bigr|_{x} \;,
}
where we use the notation $(\dots)|_x = (\dots)_{x_1 = x_2 = x_3 = x}$. Indeed, if we choose the exponential function $f_n(x) = e^{i\,p_n\cdot x}$, we get
\eq{
e^{i\,p_1\cdot x} \tri e^{i\,p_2\cdot x}\tri e^{i\,p_3\cdot x} =\, \exp\bigl(-i\frac{\pi^2}{2}\,\theta^{abc}\,p_{1,a}p_{2,b}p_{3,c} \bigr) \, e^{i\,(p_1 + p_2 + p_3)\cdot x }\;.
}
Performing an odd permutation $\sigma \in S_3$ of the factors in the $\,\tri\,$-product, we readily get the following phase for $f_k = e^{i\,p_k \cdot x}$:
\eq{
f_{\sigma(1)} \tri f_{\sigma(2)} \tri f_{\sigma(3)} =\, e^{i\,\pi^2 \,\theta^{abc}\,p_{1,a}p_{2,b}p_{3,c}} \,f_1 \tri f_2 \tri f_3 \;,
}
but this is exactly the phase observed in the 3-point correlator \eqref{phasethreeperm}. We therefore propose that the algebra of functions on the target space in this case has a 3-product structure given by \eqref{threebracket}.

Let us once more go back to the open string case. The Moyal-Weyl product applied to coordinate functions $x^a, x^b$ resulted in the non-commutativity of spacetime as we deduced in section \ref{ch-quant}, see especially \eqref{ncspacetime}. What is the analogue for the closed string case? The appropriate object is not given by the commutator but by the completely anti-symmetrized sum over 3-products as was already pointed out in \cite{Blumenhagen:2010hj}. The result for coordinate functions $x^a,x^b,x^c$ is given by:
\eq{
\label{antisymtripcon}
   \bigl [x^a,x^b,x^c \bigr]\stackrel{\textrm{def}}{=}\,\sum_{\sigma\in S_3} {\rm sign}(\sigma) \;
     x^{\sigma(a)}\, \tri\,  x^{\sigma(b)}\, \tri\,  x^{\sigma(c)} =
     3\pi^2\, \theta^{abc}\; ,
}
where we introduced the three-bracket $[\cdot,\cdot,\cdot]$. This result coincides with \cite{Blumenhagen:2010hj}, where the three-bracket was defined as the Jacobi-identity of the coordinates. The latter being non-zero is only possible if the underlying spacetime is non-commutative and non-associative. In our case, the result is not a standard two-product which is not associative (note that we did not define a two-product), but a three-product structure from the beginning. Such three-algebras were described in the literature for example in \cite{DeBellis:2010sy, deAzcarraga:2010mr}.

Let us generalize this to the case of the $N$-point correlator \eqref{Ntachyonamplitude}. Similar to the three-point case, the phases which appear if the vertex operators are permuted can be reproduced if we introduce an \emph{$N$-product structure} on the algebra of functions, defined as follows:
\eq{
   f_1(x)\, \triN \,  &f_2(x)\, \triN \ldots \triN \,  f_N(x) \stackrel{\rm def}{=} \\
   &\exp\left[ {\textstyle {\pi^2\over 2}} \theta^{abc}\!\!\!\!\! \sum_{1\le i< j < k\le N}
     \!\!\!\!  \,
      \partial^{x_i}_{a}\,\partial^{x_j}_{b} \partial^{x_k}_{c} \right]\,
   f_1(x_1)\, f_2(x_2)\ldots
   f_N(x_N)\Bigr|_{x} \;.
}
Therefore as a main result of this section, considering correlation functions of tachyon vertex operators in the closed string $\textrm{CFT}_H$-framework suggests a hierarchy of $N$-products $\triN$ on the target space. We conclude this section by giving the most immediate properties of this collection of products. Firstly, we can relate an $N$-product to an $(N-1)$-product by taking the last factor to be the identity:
\eq{
\label{NtoN-1}
    f_1\,\triN \,  f_2  \,\triN \,  \ldots\,  \triN \, f_{N-1} \, \triN \, 1\,
   =\,  f_1\,\triNm \,  \ldots\,  \triNm \,  f_{N-1} \;.
}
This also enables us to conclude how to multiply two functions by considering the case $N=3$: Because there are always three derivatives in the three-product, inserting a constant collapses the exponential to the identity:
\begin{equation}
    f_1\,  \trizw \,  f_2 = f_1\,\tridr \, f_2\,\tridr \, 1= f_1 \cdot f_2  \; ,
\end{equation}
i.e. we arrive at the standard commutative and associative product\footnote{There are results which show also a non-commutative 2-product, e.g. \cite{Lust:2010iy, Andriot:2012vb, Mylonas:2012pg}.}. Secondly we want to mention one essential difference to the open string case and the Moyal-Weyl product. Whereas in the latter case, the product of $N$ functions can be calculated by successive application of the product of two factors, this is not possible for the $N$-product. As an example, the product of five functions cannot be computed by first multiplying three functions with a three-product and then multiplying the result again by a three-product with the remaining two functions:
\begin{equation}
   f_1\,\trifue \, f_2\,\trifue \, f_3 \,\trifue \,  f_4\, \trifue \,  f_5\;\neq \;
   \left( f_1\, \tridr \, f_2\, \tridr \, f_3\right)\, \tridr f_4\, \tridr \,   f_5\; .
\end{equation}
It would be intriguing to get more information about the structure of the collection of $N$-products. The complete mathematical description and also the right abstract definition of algebras carrying a hierarchy of $N$-product structures goes beyond the scope of the present work. In the next section we only want to give some possible connections to other structures which might play a role in this respect.

\subsubsection{Concluding remarks and outlook}
Starting with a sigma model with flat, three-dimensional target space together with a $B$-field linear in the spacetime coordinates (i.e. constant $H$-flux), we were able to define a conformal field theory framework up to linear order in the $H$-flux. An important ingredient was to define new (anti-)holomorphic currents ${\cal J}^a,\bar{\cal J}^b$ and spacetime coordinates ${\cal X}^a$, both fields having non-vanishing three-point functions. T-duality was implemented as a reflection of the right-moving part of the coordinates ${\cal X}^a$. Furthermore it was possible to define tachyon vertex operators as conformal primary fields of $\textrm{CFT}_H$ and to investigate their $N$-point correlation functions in the case of the $H$-flux background and its three-fold T-dual version. In the latter case, we were able to extract momentum-dependent phase factors by permuting the vertex operators in the correlation functions. These phases were independent of the world-sheet coordinates and vanished after imposing momentum conservation. By comparison to open string non-commutativity, we were able to propose the structure of $N$-products on the algebra of functions on the target space.

There are many important questions for future work. On the physics side, it would be interesting to extend the $\textrm{CFT}_H$-framework to the supersymmetric case and investigate T-duality and vertex operators. But even in the bosonic theory it would be important to investigate more complicated vertex operators like the graviton. This was partially done in \cite{Blumenhagen:2011ph}, which could be extended to calculate graviton scattering amplitudes in the case of the $R$-flux background. On the mathematics side, clearly it is important to understand $N$-algebras and quantum spaces underlying such algebras. We only were able to give a guess how such an algebra could be realized in the case of 3-fold T-duals to constant $H$-flux backgrounds (see also \cite{Lust:2010iy}). Furthermore the phenomenon of a complete hierarchy of products for every number of factors could have connections to the notion of $L_\infty$-structures (see for example \cite{Roytenberg:01} as an introduction).

\subsection{Courant algebroids and flux Bianchi identities}
As detailed in section \ref{sec-ngeofluxes}, the application of T-duality to configurations with non-vanishing $H$-flux on the one hand leads to the well-understood geometric flux but on the other hand also to non-geometric $Q$- and $R$-fluxes, whose mathematical properties still have to be described precisely. One of the most immediate properties of the standard $H$-flux and geometric $f$-flux are their Bianchi-identities. These are identities for the fluxes which are trivial if one considers their \emph{local} structure: On a sufficiently small neighborhood of a point, the $H$-flux always can be written as $H = dB$ and therefore it is closed. As a consequence we get the Bianchi identity\footnote{In the following we always set the dilaton to be constant.}:
\eq{
dH =\, 0 \;, \qquad \partial_{[\underline{a}} H_{\underline{bcd}]} =\,0 \;,
}
where the coordinate expression holds for a commuting set of basis sections. Si\-mi\-larly, the geometric $f$-flux can be realized by the structure constants of a non-holonomic frame of the tangent bundle for the internal manifold and therefore the Jacobi-identity of the Lie bracket gives a Bianchi-identity for this flux. For a basis $\{e_a\}$ of the tangent bundle we get
\eq{
0 =&\, \left[e_a, [e_b,e_c]_L\right]_L + \textrm{cycl.} \\
=&\, \left(e_{[\underline{a}}(f^n{}_{\underline{bc}]}) + f^n{}_{[\underline{a}k} f^k{}_{\underline{bc}]} \right) e_n \;.
}
To get similar relations for the non-geometric fluxes and possibly additional contributions to the previous identities, one clearly needs an extension of the tangent bundle to include vector- and form-indices on equal footing. In section \ref{sec-bialg} we presented the mathematical framework which we are going to use in order to realize all fluxes as structure constants of a Courant algebroid \cite{Blumenhagen:2012pc}. We proceed in two steps: First, to get an idea we present the corresponding algebra purely on the tangent bundle. This enables us to derive Bianchi identities from the Jacobi identity of the Lie bracket\footnote{Similar Bianchi-identities were already motivated in \cite{Blumenhagen:2012ma}.}. Secondly we will introduce the concept of quasi-Lie algebroids (which is well known in mathematics, e.g. \cite{Roytenberg:01, Halmagyi:2008dr, Halmagyi:2009te, 1079.53126} ) and their corresponding Courant algebroid structure to construct a commutator-algebra on the generalized tangent bundle which contains all of the different fluxes as structure functions.

\subsubsection{Realization on the tangent bundle}
As we have already seen in section \ref{sec-ngeofluxes}, one of the basic structures of the non-geometric regime is a bi-vector $\beta$, which is interpreted in the following as the structure tensor of a (quasi-)Poisson manifold $M$. It is used to define the Poisson bracket of functions by:
\eq{
\{f,g\} =\, \beta^{ij}\,\partial_i f\, \partial_j g \;.
}
This bracket satisfies the Jacobi identity in the standard case, but in the following we want to relax this condition to treat the more general case of a \emph{quasi-Poisson} manifold (for more mathematical details we refer the reader to \cite{1029.53090}). The Jacobi identity is now altered to a non-trivial expression:
\eq{
\{ f, \{g, h\}\} + \textrm{cycl.} =\, \left(\beta^{[\underline{i}n}\partial_n \beta^{\underline{jk}]}\right)\,\partial_i f \,\partial_j g \,\partial_k h \;.
}
Note that the expression in brackets has the same (local) structure as the non-geometric $R$-flux \eqref{Rflux}.

Even though having a quasi-Poisson manifold, it is still possible to use the bi-vector as a map from the cotangent to the tangent bundle (similar to the anchor of a Lie algebroid), which we denote by $\beta^{\sharp}$:
\eq{
\beta^{\sharp}: T^*M \rightarrow TM\;, \quad \beta^{\sharp}(\xi)(\eta) =\, \beta(\xi, \eta)\;.
}

\paragraph{Example for vanishing $H$- and $f$-flux}

Denoting the basis and dual basis of the tangent bundle $TM$ by $e_i$ and $e^i$, respectively (i.e. $e^j(e_i) = \delta^j _i $), by using the map $\beta^{\sharp}$ we obtain in addition to the standard differential operators $e_i$ (which are given in a holonomic basis by $e_i = \partial_i$) the following operators:
\eq{\label{defesharp}
e^i_{\sharp} :=\,\beta^{\sharp}(e^i) = \, \beta^{ij} \partial_j \;.
}
It is now easy to check that both types of differential operators form the following commutator algebra:
\eq{
\label{einfachealg}
       [e_i,e_j]_L&=0 \;, \\
       [e_i,e^j_\sharp]_L&=Q_i{}^{jk}\, e_k \;, \\
       [e_\sharp^i,e_\sharp^j]_L&=R^{ijk}\, e_k + Q_k{}^{ij}\, e^k_\sharp \;,
}
where we again use the definitions \eqref{Qflux} for the $Q$-flux and \eqref{Rflux} for the $R$-flux. Note that a similar type of algebraic structure was already given in \cite{Grana:2008yw}. It is now easy to get Bianchi-type relations from the above commutator algebra if we interpret it as the Lie brackets of vector fields. The latter satisfy the Jacobi identity and evaluating the Jacobiators gives
\eq{
\label{einfachebianchi}
   0&=  3\hspace{1pt} \beta^{[\ul a m} \partial_m Q_d{}^{\underline{bc}]}
     -\partial_d R^{abc}
 + 3\hspace{1pt}  Q_d{}^{[\ul a m} Q_m{}^{\ul b \ul c]} \;, \\
   0&=  2 \hspace{1pt} \beta^{[\underline{a} m}\, \partial_{m}  {R}^{\underline{bcd}]}- 3\hspace{1pt}
          {R}^{[\underline{ab}m}\, {Q}_m{}^{\underline{cd}]}  \;,
}
where for the first one we evaluated the Jacobiator of one basis vector and two vectors defined in \eqref{defesharp} and for the second one we computed the Jacobiator of three vectors of type \eqref{defesharp}. It is interesting to note that the second identity already appeared in the context of double field theory \cite{Andriot:2012wx} and both identities were also derived in \cite{Blumenhagen:2012ma} in the context of Schouten-Nijenhuis brackets. We take this as a motivation to generalize the strategy to include also geometric flux and $H$-flux.

\paragraph{The general case on the tangent bundle}

The implementation of geometric $f$-flux can be achieved by introducing a general (non-holonomic) basis $\{e_a\}$ of the tangent bundle $TM$. There are two equivalent characterizations thereof, the first given by the commutators of the basis vectors and the second given by the exterior differential of the dual basis $e^a$:
\eq{\label{non-holonomic frame}
[e_a, e_b]_L =\, f^c{}_{ab} e_c\;, \quad de^a =\, -\tfrac{1}{2}\,f^a{}_{bc}\, e^b \wedge e^c \;. }
Exterior derivatives like in the case of local $H$-flux $H \in \Gamma(\wedge^3 T^*M)$ then get additional contributions, e.g.
\eq{
H =\, \left( \partial_{[\underline{a}} B_{\underline{bc}]} - f^d{}_{[\underline{ab}} B_{d\underline{c}]}\right)\,e^a\wedge e^b \wedge e^c \;.
}
In terms of this basis, the differential operators $e_{\sharp}^a $ are defined with the bi-vector $\beta$ expressed in terms of the non-holonomic basis: $e_\sharp ^a = \beta^{ab}e_b$. Note that in case of a Riemannian manifold, a non-holonomic basis can be expressed by choosing a set of vielbeins that diagonalize the metric: $e_a = e_a{}^i \partial_i =: \partial_a$. The structure constants can then be expressed in terms of derivatives of the vielbeins:
\eq{
f^c{}_{ab} =\, e^c{}_j\left(e_a{}^i \partial_i e_b{}^j - e_b{}^i \partial_i e_a{}^j \right)\;,
}
where the $e^c{}_j$ are defined by the dual basis $e^c = e^c{}_j dx^j$. The bi-vector $\beta$ can be expressed in terms of the latter by $\beta^{ab} = e^a{}_i e^b{}_j \beta^{ij}$. But the following discussion will be independent of the choice of a Riemannian structure on the manifold $M$ and can therefore be applied to general (quasi-) Poisson manifolds.

The last two commutation relations of \eqref{einfachealg} can now be generalized to
\eq{
  \label{comm_02}
  [ e_a, e_\sharp^b ] = Q_a{}^{bc} e_c - f^b{}_{ac} e_\sharp^c \;, \hspace{40pt}
  [e_\sharp^a, e_\sharp^b] = R^{abc} e_c + Q_c{}^{ab} e_\sharp^c \;,
}
where we enhanced the fluxes to the case of the non-holonomic frame. Explicitly they are given by
\eq{
\label{QRdef}
  Q_a{}^{bc} &=\partial_a \beta^{bc}+ f^b{}_{am}\, \beta^{mc} -f^c{}_{am}\,
  \beta^{mb} \;, \\[1mm]
 R^{abc}&= 3\, \bigl( \beta^{[\ul am}\,\partial_m \beta^{\ul b\ul c]}+f^{[\ul a}{}_{mn}\,
 \beta^{\ul bm}\beta^{\ul c]n} \bigr) \;.
}
Finally, to include the three-form flux $H$ into the algebra of commutators, we have to perform a redefinition of the previous fluxes. The reason for this way of including the $H$-flux is the fact that it cannot be reproduced directly as commutator of differential operators. A similar way of including the $H$-flux into the commutator algebra can be found in \cite{Halmagyi:2009te}. The field redefinition is given by
\eq{
  \label{redef_01}
           {\cal H}_{abc}&=H_{abc} \;,\\
           {\cal F}^c{}_{ab}&=f^c{}_{ab} - H_{abm}\, \beta^{mc} \;, \\
           {\cal Q}_{a}{}^{bc}&=Q_{a}{}^{bc} +  H_{amn}\, \beta^{mb}\,
           \beta^{nc} \;,\\
           {\cal R}^{abc}&=R^{abc} - H_{mnp}\, \beta^{ma}\,
           \beta^{nb}\, \beta^{pc} \;,
}
i.e. the bi-vector $\beta$ is used to transform the index-structure of the $H$-flux in the right way to include it in the corresponding flux. With the help of the field redefinition we can finally rewrite equations \eqref{non-holonomic frame} and \eqref{comm_02} to get the complete commutator algebra:
\eq{
\label{preroytenberg}
       [e_a,e_b]_L&= {\cal F}^c{}_{ab}\, e_c +  {\cal H}_{abc}\, e_\sharp^c \;, \\
       [e_a,e_\sharp^b]_L&= {\cal Q}_{a}{}^{bc}\, e_c -  {\cal F}^b{}_{ac}\, e_\sharp^c\;, \\
       [e_\sharp^a,e_\sharp^b]_L&= {\cal R}^{abc}\, e_c +  {\cal Q}_{c}{}^{ab}\, e_\sharp^c \;.
}
An algebra of a similar form was first derived by Roytenberg in \cite{Roytenberg:01, 1027.53104}. However note that the algebra \eqref{preroytenberg} is realized completely on the tangent bundle. We therefore refer to it as \emph{pre-Roytenberg algebra}. We will come back to a realization of the algebra on the generalized tangent bundle by using Courant algebroids in the next section. But before moving to this technically more sophisticated case, we want to use one advantage of \eqref{preroytenberg}: Due to the Jacobi identity of the Lie bracket we can derive Bianchi identities for the fluxes introduced in \eqref{redef_01}.

\paragraph{Bianchi identities}
Before using the Jacobi identity for the Lie bracket, let us state the Bianchi identity for the ${\cal H}$-flux in terms of the redefined geometric flux:
\eq{
  \label{bianchi0}
   {\rm I}\ :\ 0= \partial_{[\underline{a}} \, {\cal H}_{\underline{bcd}]} - \tfrac{3}{ 2}\,
    {\cal F}^m{}_{[\ul{ab}}\, {\cal H}_{m\ul{cd}]} \;.
}
The Jacobi identities including both, $e^a _\sharp$ and $e_a$ result in four different Bianchi identities. The different independent combinations are
\eq{
  {\rm II}\,:\;  0 = \bigl[ [e_a, e_b ]_L, e_c \bigr]_L + {\rm cycl.} \;, \quad
  {\rm III}\,:\;  0 = \bigl[ [e_a, e_b ]_L, e_\sharp^c \bigr]_L + {\rm cycl.} \;, \\[2mm]
  {\rm IV}\,:\;    0 = \bigl[ [e_a, e_\sharp^b ]_L, e_\sharp^c \bigr]_L + {\rm cycl.} \;, \quad
  {\rm V}\,:\;    0 = \bigl[ [e_\sharp^a, e_\sharp^b ]_L, e_\sharp^c \bigr]_L + {\rm cycl.} \;,    }

\noindent and using the algebra \eqref{preroytenberg} they result in the Bianchi identities
\eq{\label{bianchiII}
 {\rm II}:\hspace{5pt} 0 =\,&   \Bigl( \partial_{[\underline{c}} \, {\cal F}^d{}_{\underline{ab}]} +
    {\cal F}^m{}_{[\underline{ab}}\, {\cal F}^d{}_{\underline{c}]m} +
    {\cal H}_{[\underline{ab}\,m}\, {\cal Q}_{\underline{c}]}{}^{md}\Bigr)  \\
   & +\Bigl( \partial_{[\underline{c}} \, {\cal H}_{\underline{ab}]n}
   -2 {\cal F}^m{}_{[\underline{ab}}\, {\cal H}_{\underline{cn}]m}
  \Bigr) \beta^{nd} \;,
}
\eq{
 {\rm III}:\hspace{5pt} 0=\,& \Bigl( \beta^{cm} \partial_m \mathcal F^d{}_{ab} + 2 \partial_{[ \ul a}\, \mathcal Q_{\ul b]}{}^{cd}
 - \mathcal H_{mab} \mathcal R^{mcd} - \mathcal F^m{}_{ab} \mathcal Q_{m}{}^{cd} + 4 \mathcal Q_{[\ul a}{}^{[\ul cm}
 \mathcal F^{\ul d]}{}_{m \ul b]} \Bigr) \\
 & +\Bigl( \beta^{cm}\partial_m \mathcal H_{abn} - 2 \partial_{[\ul a} \mathcal F^c{}_{\ul b]n}
  - 3 \mathcal H_{m[\ul a\ul b} \mathcal Q_{\ul n]}{}^{mc} \\
& + 3 \mathcal F^m{}_{[\ul a \ul b} \mathcal F^c{}_{m\ul n]} \Bigr) \beta^{nd} \;,
}
\eq{\label{bianchiIV}
 {\rm IV}:\hspace{5pt} 0 =\, & \Bigl( -\partial_a \mathcal R^{bcd} - 2 \beta^{[\ul c m} \partial_m \mathcal Q_a{}^{\ul b] d}
 + 3 \mathcal Q_a{}^{[\ul b m}\mathcal Q_m{}^{\ul c \ul d]}
 - 3 \mathcal F^{[\ul b}{}_{am} \mathcal R^{\ul c \ul d]m} \Bigr)  \\
 & +\Bigl( 2 \beta^{[\ul c m} \partial_m \mathcal F^{\ul b]}{}_{an} - \partial_a
 \mathcal Q_{n}{}^{bc}
 + \mathcal Q_m {}^{bc} \mathcal F^m{}_{an}  \\
 & + \mathcal R^{bcm} \mathcal H_{man}
 -4 \mathcal Q_{[\ul a}{}^{[\ul b m}\mathcal F^{\ul c]}{}_{m \ul n]} \Bigr) \beta^{nd} \;,
}
\eq{\label{bianchiV}
 {\rm V}:\hspace{5pt}   0 =\, & \Bigl( \beta^{[\underline{c} m}\, \partial_{m}  {\cal
        R}^{\underline{ab}]d} -
     2 {\cal R}^{[\underline{ab}m}\, {\cal Q}_m{}^{\underline{cd}]}
\Bigr)  \\
&+  \Bigl( \beta^{[\underline{c} m}\, \partial_{m}  {\cal
        Q}_n{}^{\underline{ab}]} + {\cal R}^{[\underline{ab}m}\,
        {\cal F}^{\underline{c}]}{}_{mn} +
     {\cal Q}_m{}^{[\underline{ab}} \, {\cal Q}_n{}^{\underline{c}]m}
\Bigr) \beta^{nd} \;.
}

\noindent To conclude this section let us make two remarks. First, the identities \eqref{bianchi0} and \eqref{bianchiII}-\eqref{bianchiV} are constructed out of Jacobi identities but they can also be checked directly by expressing the fluxes in terms of their potentials. Second, let us look at two special cases of the identities. On the one hand, choosing a holonomic basis and having vanishing $H$-flux,  the last two identities reduce again to \eqref{einfachebianchi}. On the other hand, keeping the basis non-holonomic but choosing all the fluxes to be constant results in the following set of equations
\eq{
\label{constant fluxes}
0 &= {\cal H}_{k[\underline{ab}}\, {\cal F}^k{}_{\underline{cd}]} \;,   \\
0 &= {\cal H}_{k[\underline{ab}}\, {\cal Q}_{\underline{c}]}{}^{kj} -{\cal F}^j{}_{k[\underline{a}}\, {\cal F}^{k}{}_{\underline{bc}]} \;,   \\
0 &= {\cal H}_{kab}\, {\cal R}^{kcd} + {\cal F}^k{}_{ab}\, {\cal Q}_k{}^{cd} -
4{\cal F}^{[\underline{c}}{}_{k[\underline{a}}
\, {\cal Q}_{\underline{b}]}{}^{\underline{d}]k} \;,    \\
0 &= {\cal F}^{[\underline{a}}{}_{ki}\, {\cal R}^{\underline{bc}]k} - {\cal Q}_i{}^{k[\underline{a}}\,{\cal Q}_k{}^{\underline{bc}]} \;,  \\
0 &= {\cal Q}_k{}^{[\underline{ab}}\, {\cal R}^{\underline{cd}]k} \; .
}
We observe that they have the same structure as the Bianchi identities derived earlier in the literature \cite{Shelton:2005cf, Ihl:2007ah}, which we also presented in \eqref{constantbianchis}.

\subsubsection{Realization on $TM \oplus T^*M$}
In the previous section, we were able to realize the algebra \eqref{preroytenberg} containing the fluxes \eqref{redef_01} completely on the tangent bundle by using the interpretation of the bi-vector as a map from the cotangent bundle to the tangent bundle. However, looking at the literature on flux compactifications, e.g. \cite{Dabholkar:2005ve, Grana:2008yw}, especially in the context of Hitchin's generalized geometry \cite{1076.32019, Gualtieri:2003dx, Grana:2008yw} it is more appropriate to consider the full generalized tangent bundle $TM \oplus T^*M$ in order to describe all fluxes in terms of structure functions of a commutator algebra. We thus seek for a bracket dealing with vector fields and one forms on an equal footing. In section \ref{subsec-courantalg} we presented the basics of a well-known \mbox{mathematical} structure which is appropriate to handle this case. It turns out that including all fluxes needs for a certain \emph{twisting} of this structure which we describe below. Physically, by counting degrees of freedom, it is not possible to have all fluxes turned on without further restrictions. But in the last section, we derived relations among the various types of fluxes which restrict their number of directions being turned on at the same time. Mathematically this is reflected in the fact that the resulting structure which is suitable to include all fluxes is only a  Courant algebroid if the Bianchi identities \eqref{bianchi0} - \eqref{bianchiV} are obeyed, as we will see later in the proof of the main proposition of this section.

\paragraph{Quasi-Lie algebroids}
As detailed in section \ref{sec-bialg}, one way to construct a Courant algebroid structure is to first consider a Lie bi-algebroid. It turns out that this remains valid also for the weaker structure of a quasi-Lie bialgebroid. We use it in the following to deal with the most general case where all fluxes are included. A detailed treatment of twisted- and quasi- structures in this context is given in \cite{1079.53126}.

Let us start with the trivial Lie algebroid $(TM,[\cdot,\cdot]_L,\textrm{id})$. As stated in the last subsection, the Lie bracket is characterized by the geometric flux $f^c{}_{ab}$. The inclusion of the $H$-flux was done by a field redefinition. We now realize the latter mathematically by a \emph{twist}, meaning that we extend the Lie-bracket by an $H$-dependent term. The resulting bracket for vector fields $X,Y$ is called $H$-twisted Lie bracket:
\eq{
  \label{H-Lie}
	[X,Y]_L^H = [X,Y]_L - \beta^{\sharp}\left(\iota_Y\iota_X H\right)\, .
}
Taking the tangent bundle $TM$ together with this bracket and the identity map, we get an example of a \emph{quasi}-Lie algebroid. This means that the Leibniz rule is still valid but the anchor (which is the identity in this case) is not an algebra homomorphism any more (obvious in the above case). The defect is given by the $H$-dependent term. Moreover, the bracket \eqref{H-Lie} does not satisfy the Jacobi identity. The failure is again given by an $H$-dependent term. Thus, in the limit $H\rightarrow 0$ the structure reduces to a standard Lie algebroid. Let us now evaluate \eqref{H-Lie} on basis vector fields $e_a$:
\eq{
   [e_a,e_b]_L^H = f^p{}_{ab} \, e_p - H_{abm}\, \beta^{mp}\, e_p = \mathcal{F}^c{}_{ab} \, e_c \, .
}
We see that the right structure functions ${\cal F}^c{}_{ab}$ are produced by this kind of twisted bracket.

Next, let us consider the Lie algebroid $(T^*M,[\cdot,\cdot]\ks,\beta^\sharp)$ introduced in section \ref{ch-math}. The twisting by the $H$-flux in this case leads to the so-called \emph{$H$-twisted Koszul-Schouten bracket} defined by:
\eq{
\label{H-Koszul-Schouten}
	[\xi,\eta]\ks ^H = [\xi,\eta]\ks
			+\iota_{\beta^\sharp(\eta)}\iota_{\beta^\sharp(\xi)} H \, .
}
Again this bracket has the Leibniz rule but does not obey the Jacobi identity and the anchor is not an algebra homomorphism. The corresponding flux which measures these defects is given by the ${\cal R}$-flux introduced in \eqref{redef_01}. Its vanishing completely restores the Lie algebroid properties (see also \cite{1029.53090}). Again, evaluating on basis one-forms gives:
\eq{
	[e^a,e^b]\ks ^H = \partial_p\beta^{ab} \, e^p + 2f^{[\ul a}{}_{pm}\beta^{m \ul b]} \, e^p
		+\beta^{am}\beta^{bn} H_{mnp} e^p = \mathcal{Q}_c{}^{ab} \, e^c \, .
}
To sum up, the fluxes defined in \eqref{redef_01} have a definite geometric meaning: Whereas the ${\cal F}$- and ${\cal Q}$-fluxes can be interpreted as the structure functions of the $H$-twisted Lie bracket and $H$-twisted Koszul-Schouten bracket, respectively, the ${\cal H}$- and ${\cal R}$-fluxes measure the failures of the corresponding quasi-Lie algebroids to be proper Lie algebroids.

\paragraph{The Courant algebroid}
Let us now see how we can construct a Courant algebroid out of the quasi-Lie algebroids described previously and which conditions on the fluxes \eqref{redef_01} are needed. The notion of Courant algebroids and its connection to standard Lie bi-algebroids was pointed out in section \ref{sec-bialg}. According to the definitions given there, to define a Courant algebroid, we have to specify a total space, bilinear form, bracket and an anchor. We present them in turn:
\begin{itemize}
\item For the total space we choose the generalized tangent bundle $TM \oplus T^*M$.
\item For sections $X+\xi, Y + \eta \in \Gamma(TM \oplus T^*M)$ we define a symmetric and antisymmetric\footnote{Only the symmetric bilinear form is needed for the definition of a Courant algebroid. We introduce also the antisymmetric bilinear form to write later formulas in a convenient way.}  bilinear form by
\eq{
  \langle X+\xi,Y+\eta\rangle_\pm = \xi(Y)\pm \eta(X) \;.
}
\item The bracket $\gl \cdot,\cdot \gr$ on the generalized tangent bundle is determined by the following definitions for $X,Y \in \Gamma(TM)$ and $\xi, \eta \in \Gamma(T^*M)$:
\eq{\label{Roytenbergbracket}
	\gl X,Y\gr &= [X,Y]_L^H + \iota_{Y}\iota_X H  \;, \\
	\gl X,\xi\gr &= [\iota_{X},d^H]_+\,\xi - [\iota_{\xi},d_\beta^H]_+\, X
			+ \tfrac{1}{2}\bigl(d^H-d_\beta^H\bigr)\, \langle X,\xi\rangle_- \;, \\
	\gl \xi,X\gr &= [\iota_{\xi},d_\beta^H]_+\, X - [\iota_{X},d^H]_+\,\xi
			+ \tfrac{1}{2}\bigl(d^H-d_\beta^H\bigr)\,\langle \xi,X\rangle_- \;, \\
	\gl \xi,\eta\gr &= [\xi,\eta]\ks ^H + \iota_\eta\iota_\xi \mathcal{R} \;,
}
where we introduced the differentials $d^H$ and $d_\beta ^H$ corresponding to the $H$-twisted Lie bracket and $H$-twisted Koszul-Schouten bracket. Their action can be computed using the definitions \eqref{H-Lie} and \eqref{H-Koszul-Schouten} and the general definition of the corresponding differential, given in the mathematical introduction, equation \eqref{algebroiddiff}. The symbol $[\cdot,\cdot]_+$ denotes the anti-commutator of operators.\item The anchor $\alpha$ is given in terms of the identity map and the bi-vector by
\eq{
\alpha : \, TM \oplus T^*M \rightarrow TM \;, \quad \alpha(X + \xi) =\, X + \beta^\sharp(\xi) \;.
}
\end{itemize}
With these definitions we are now able to state the following proposition:

\begin{prop}
\label{Courant theorem}
$(TM \oplus T^*M, \gl \cdot,\cdot \gr, \langle \cdot, \cdot \rangle_+, \alpha)$ is a Courant algebroid. Moreover, its bracket $\gl \cdot, \cdot \gr$ has the following algebra on basis sections:
\eq{\label{Roytenbergalgebra}
	\gl e_a,e_b\gr &= \mathcal{F}^c{}_{ab} \, e_c + \mathcal{H}_{abc}\, e^c \;, \\
	\gl e_a,e^b\gr &= \mathcal{Q}_a{}^{bc} \, e_c - \mathcal{F}^b{}_{ac} \, e^c\;, \\
	\gl e^a,e^b\gr &= \mathcal{Q}_c{}^{ab} \, e^c + \mathcal{R}^{abc} \, e_c \; .
}
\end{prop}

\begin{proof}
We have to check the defining properties of a Courant algebroid as they were presented in definition \ref{def-Courantalg}. The homomorphism property of the anchor can be calculated directly. The Leibniz rule can be checked separately for every case in \eqref{Roytenbergbracket}: The first and the last ones are trivial because of the corresponding properties of the Lie- and Koszul-Schouten bracket. Let us be more explicit for the second case (the third case is done in the same way). For $f \in {\cal C}^{\infty}(M)$ we have

\eq{
\gl X, f\xi \gr =&\, \iota_X d^H (f\xi) + d^H \iota_X (f\xi) - \iota_{f\xi} d_\beta ^H X \\
&- d_\beta ^H \iota_{f\xi} X + \tfrac{1}{2} \left(d^H - d_\beta ^H \right) \left(f \langle X, \xi \rangle_- \right) \\
=&\, f\gl X, \xi \gr + \left(\iota_X df\right) \xi - d_\beta f \, \iota_\xi X - \tfrac{1}{2} df \, \iota_X \xi + \tfrac{1}{2} d_\beta f \, \iota_X \xi \\
=&\, f\gl X, \xi \gr + X(f) \xi - \tfrac{1}{2} {\cal D}f\, \iota_X \xi \;,
}
where ${\cal D}$ was defined in \ref{def-Courantalg}. The third property of \ref{def-Courantalg} is easy in our case and the fourth property is a straightforward computation. The property for the Jacobiator is slightly more subtle so we explain it in detail. There are four independent types of Jacobiators which we calculate explicitly\footnote{Note, that applying the map $\beta^\sharp$ to these Jacobiators again results in the expressions for the Bianchi identities \eqref{bianchiII}-\eqref{bianchiV}.}:
\eq{\label{Jac1}
 \mathfrak{J}(e_a,e_b,e_c) =&   -3\Bigl( \partial_{[\underline{c}} \, {\cal F}^d{}_{\underline{ab}]} +
    {\cal F}^m{}_{[\underline{ab}}\, {\cal F}^d{}_{\underline{c}]m} +
    {\cal H}_{[\underline{ab}\,m}\, {\cal Q}_{\underline{c}]}{}^{md}\Bigr) e_d\\
   &-3\Bigl( \partial_{[\underline{c}} \, {\cal H}_{\underline{ab}]d}
   -2 {\cal F}^m{}_{[\underline{ab}}\, {\cal H}_{\underline{cd}]m}
  \Bigr)e^d +\tfrac{3}{2}\,\mathcal{D} {\cal H}_{abc}  \;,
}
\eq{\label{Jac2}
 \mathfrak{J}(e_a,e_b,e^c) =& -\Bigl( \beta^{cm} \partial_m \mathcal F^d{}_{ab}
	+ 2 \partial_{[ \ul a}\, \mathcal Q_{\ul b]}{}^{cd}
	 - \mathcal H_{mab} \mathcal R^{mcd}
	- \mathcal F^m{}_{ab} \mathcal Q_{m}{}^{cd}
	\\&+ 4 \mathcal Q_{[\ul a}{}^{[\ul cm}
	 \mathcal F^{\ul d]}{}_{m \ul b]} \Bigr) e_d
 	-\Bigl( \beta^{cm}\partial_m \mathcal H_{abd}
	- 2 \partial_{[\ul a} \mathcal F^c{}_{\ul b]d}
	\\&- 3 \mathcal H_{m[\ul a\ul b} \mathcal Q_{\ul d]}{}^{mc}
	+ 3 \mathcal F^m{}_{[\ul a \ul b} \mathcal F^c{}_{m\ul d]} \Bigr)e^d
	+\tfrac{3}{2}\,\mathcal{D}\mathcal{F}^c{}_{ab}  \;,
}
\eq{
\mathfrak{J}(e_a,e^b,e^c) =&+ \Bigl( -\partial_a \mathcal R^{bcd}
	- 2 \beta^{[\ul c m} \partial_m \mathcal Q_a{}^{\ul b] d}
	 + 3 \mathcal Q_a{}^{[\ul b m}\mathcal Q_m{}^{\ul c \ul d]}
	\\&- 3 \mathcal F^{[\ul b}{}_{am} \mathcal R^{\ul c \ul d]m} \Bigr)e_d
	 +\Bigl( 2 \beta^{[\ul c m} \partial_m \mathcal F^{\ul b]}{}_{ad}
	- \partial_a \mathcal Q_{d}{}^{bc}\\
	&+ \mathcal Q_m {}^{bc} \mathcal F^m{}_{ad}
	+ \mathcal R^{bcm} \mathcal H_{mad}
	-4 \mathcal Q_{[\ul a}{}^{[\ul b m}
	\mathcal F^{\ul c]}{}_{m \ul d]} \Bigr) e^d
	+\tfrac{3}{2}\,\mathcal{D}\mathcal{Q}_a{}^{bc}\;,
}
\eq{ \label{Jac4}
 \mathfrak{J}(e^a,e^b,e^c) =&
	-3\Bigl( \beta^{[\underline{c} m}\, \partial_{m}  {\cal R}^{\underline{ab}]d} -
	2 {\cal R}^{[\underline{ab}m}\, {\cal Q}_m{}^{\underline{cd}]} \Bigr)e_d
	 -3  \Bigl( \beta^{[\underline{c} m}\,
	\partial_{m}  {\cal Q}_d{}^{\underline{ab}]} \\ &
	+ {\cal R}^{[\underline{ab}m}\,
        {\cal F}^{\underline{c}]}{}_{md} +
     {\cal Q}_m{}^{[\underline{ab}} \, {\cal Q}_d{}^{\underline{c}]m}\Bigr) e^d
     +\tfrac{3}{2}\,\mathcal{D}\mathcal{R}^{a bc}\;.
}
To show that the Jacobiator \eqref{Jac1} takes the desired form, as a first step we apply the second Bianchi identity \eqref{bianchiII} to rewrite it in the form:
\eq{\label{firststep}
\mathfrak{J}(e_a,e_b,e_c) =\, 3\bigl( \partial_{[\underline{c}}{\cal H}_{\underline{ab}]n} - 2{\cal F}^m{}_{[\underline{ab}} {\cal H}_{\underline{cn}]m}\bigr)\bigl(\beta^{nd} e_d - e^n \bigr) + \tfrac{3}{2} {\cal D}{\cal H}_{abc}\;.
}
As a second step we now use the first Bianchi identity \eqref{bianchi0}, which can be expanded in the follwing form:
\eq{
\tfrac{3}{4} \partial_{[\underline{a}}{\cal H}_{\underline{bc}]d} - \tfrac{1}{4}\partial_d {\cal H}_{abc} - \tfrac{3}{2}{\cal F}^m{}_{[\underline{ab}}{\cal H}_{m\underline{cd}]} = 0 \;.
}
Using this in equation \eqref{firststep}, we reduce the Jacobiator to the final form
\eq{
\mathfrak{J}(e_a,e_b,e_c) =\, \tfrac{1}{2}{\mathcal D}{\mathcal H}_{abc} \;.
}
In the same way we are able to reduce the remaining Jacobiators \eqref{Jac2}-\eqref{Jac4} to the desired form by applying successively the corresponding Bianchi identities. We summarize the result:
\eq{
\label{jacobiatorsresult}
\text{Jac}(e_a,e_b,e_c)={\cal D}\,T(e_a,e_b,e_c)&=\tfrac{1}{2}\,{\cal D}{\cal H}_{abc} \;,\\
\text{Jac}(e_a,e_b,e^c)={\cal D}\,T(e_a,e_b,e^c)&=\tfrac{1}{2}\,{\cal D}\mathcal{F}^c{}_{ab} \; , \\
\text{Jac}(e_a,e^b,e^c)={\cal D}\,T(e_a,e^b,e^c)&=\tfrac{1}{2}\,{\cal D}\mathcal{Q}_{a}{}^{bc}\;, \\
\text{Jac}(e^a,e^b,e^c)={\cal D}\,T(e^a,e^b,e^c)&=\tfrac{1}{2}\,{\cal D}\mathcal{R}^{abc} \; .
}
Thus, all the axioms of a Courant algebroid are verified and we proved the first claim of the proposition. The second claim, i.e. the commutation relations \eqref{Roytenbergalgebra} can be computed directly.
\end{proof}

\paragraph{Summarizing remarks}
In the last two sections we were able to show a mathematical structure, which allows to describe the nature of the fluxes $({\cal H},{\cal F},{\cal Q},{\cal R})$. It is given by the Courant algebroid of proposition \ref{Courant theorem} and the fluxes turn out to be its structure functions if we evaluate it on a general (non-holonomic) basis of the generalized tangent bundle. A commutator algebra of this structure is called Roytenberg algebra in the mathe\-ma\-ti\-cal literature. To prove these statements it was indispensable to use the Bianchi identities \eqref{bianchi0}-\eqref{bianchiV}. The derivation of the latter used the realization of the Roytenberg algebra on the tangent bundle. Physically, this restriction on the fluxes is expected due to their nature as T-duals of each other. Let us summarize the logic in the following diagram:

\vspace{30pt}

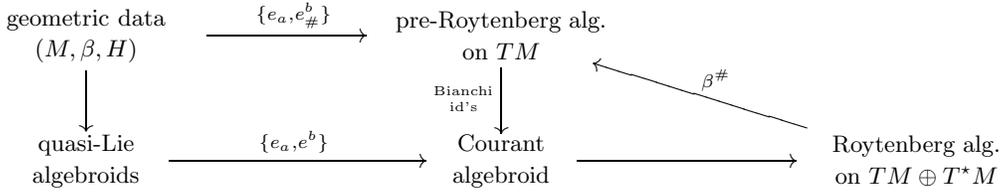
\xymatrixcolsep{5pc}\xymatrix{
  \parbox{3cm}{\centering {\small \textrm{geometric data} \\  $(M,\beta,H)$ }}\ar[r]^{\{e_a, e^b _{\#}\}} \ar[d] &\parbox{3.4cm}{\centering {\small pre-Roytenberg alg. \\ on  $TM$}} \ar[d]_{\parbox{0.8cm}{\centering \tiny{Bianchi \\id's}}} \\
\parbox{2cm}{\centering {\small quasi-Lie \\ algebroids}} \ar[r]^{\{e_a,e^b \}} & \parbox{1.8cm}{\centering {\small Courant \\ algebroid}} \ar[r] & \parbox{3.0cm}{\centering {\small Roytenberg alg.\\ on $TM \oplus T^{\star}M$}} \ar[lu]_{\beta^{\#}}
}

\vspace{40pt}

\noindent The use of Lie algebroids and Courant algebroids as a structure to describe the mathematical properties of geometric and non-geometric fluxes and their potentials turns out to be very deep and still not completely uncovered. As a further application, in the next section we will investigate a special Lie algebroid structure on the cotangent bundle appropriate to describe the dynamics of a quasi-Poisson structure. This Lie algebroid turns out to be a Dirac structure (definition \ref{def-diracstructure}) of the Courant algebroid used above.  We will then use the differential geometry of Lie algebroids presented in section \ref{ch-math} to formulate a covariant derivative, torsion- and Riemann tensor in this framework. As a consequence, we are able to write down an action similar to the standard bosonic low energy string effective action but with a metric on the cotangent bundle, the bi-vector $\beta$ and the dilaton as basic dynamical variables.

\subsection{Bi-invariant symplectic gravity}
\label{sec-biinvsymplgrav}

In the section on non-geometric fluxes we described the emergence of a bi-vector $\beta$ in the description of $Q$- and $R$-flux backgrounds. It is an important problem to formulate a dynamical theory which has a metric\footnote{In the following, we will interprete $\fa g^{ij}$ as a metric on the cotangent bundle in contrast to the standard metric $G_{ij}$ on the tangent bundle.} $\fa g^{ij}$ and this antisymmetric bi-vector $\beta^{ij}$ as basic field variables. The low energy effective action of the bosonic string,
\eq{
\label{stringaction_intro}
S= \frac{1}{2\hspace{0.5pt}\kappa^2}
\int d^nx \hspace{1pt}\sqrt{-|G|}\hspace{1pt} e^{-2\phi}
\Bigl(R-{\textstyle{1\over    12}} \hspace{0.5pt} H_{abc}\, H^{abc}
+4 \hspace{0.5pt} \partial_a \phi\hspace{2pt}\partial^a\phi
 \Bigr) \;,
}
contains the metric in form of standard Einstein-Hilbert gravity and in addition the dilaton and the kinetic term for the NS-NS $B$-field. The latter is determined by gauge invariance of the field strength $H$, but its influence on the geometry of the underlying space is only given by the energy-momentum part of the equations of motion. It is not included in the geometry a priori. T-duality changes this separation of geometry and gauge transformations dramatically by mixing metric and $B$-field components, as we have seen in \eqref{Busher}. We therefore expect a non-trivial mix between geometry and gauge theory in the description of T-dual situations where the bi-vector $\beta$ is involved: On the one hand, the gauge transformations of the original NS-NS $B$-field will be translated into a new transformation behavior of the bi-vector and also to quantities containing it. On the other hand, geometry itself will be modified to include the metric and bi-vector on a similar footing, as it is suggested already by the T-duality Buscher rules. In our section \ref{ch-math} about the mathematical background on Lie algebroids we encountered a possible way to achieve such a modified geometry. The detailed development of a democratic theory involving $\fa g^{ij}$ and $\beta^{ij}$, its implications to the theory of non-geometric fluxes and its relations to standard bosonic string theory and superstring theory is the topic of the following sections \cite{Blumenhagen:2012nk, Blumenhagen:2012nt}.

\subsubsection{Gauge transformations and quasi-Lie derivatives}
As we have seen in the examples of section \ref{subsubsec-ex}, we can interpret the bi-vector \mbox{$\beta = \tfrac{1}{2}\beta^{ij}\,\partial_i \wedge \partial_j$} in terms of Poisson-geometry. As was detailed there, we can formulate a differential operator $d_{\beta} = [\beta,\cdot]_{SN}$ on the space of multi-vectorfields on a manifold $M$. Taking the derivative of $\beta$ results in an important antisymmetric 3-vector:
\eq{\label{R-flux}
\Theta := \tfrac{1}{2}\,d_{\beta} \beta = \left(\beta^{[\underline{i}m}\partial_m \beta^{\underline{jk}]}\right)\, \partial_i \wedge \partial_j \wedge \partial_k \;.
}
The last expression was also considered in \cite{Grana:2008yw,Andriot:2011uh} as a coordinate expression for the non-geometric R-flux (here, we denote it by $\Theta$ to not confuse it with various types of Ricci scalars to be introduced later in this section). Therefore it is natural to construct a differential calculus with $d_\beta$ replacing the standard exterior differential, i.e. we have to replace the tangent bundle by the cotangent bundle equipped with an appropriate Lie algebroid structure. 

Taking the description \eqref{R-flux} of the R-flux, we see that its vanishing is equi\-va\-lent to $M$ being a Poisson manifold. To also treat the case of non-vanishing flux, we have to deal with so-called \emph{quasi}-Poisson structures. Independent of these properties, we can still take $\beta$ as an anchor-type map by defining:
\eq{\label{quasiAnker}
\beta^{\sharp}:\, T^*M \rightarrow TM, \quad \beta^{\sharp}(dx^i) = \beta^{ij}\partial_j \;,
}
and map tensor fields from $T^*M$ to $TM$. Transporting in this way the metric to the tangent bundle leads to:
\eq{
\otimes^2\beta^{\sharp}(G_{ij}\,dx^i\otimes dx^j) = \beta^{in}\beta^{jm} G_{ij}\, \partial_n \otimes \partial_m \;,
}
where $\otimes^2 \beta^\sharp$ means acting on every tensor factor. In addition, recalling the interpretation of the NS-NS $B$-field as a quasi-symplectic structure, we \emph{identify} its inverse with the bi-vector $\beta$ (this is also known in non-commutative geometry, see for example \cite{Seiberg:1999vs}). To sum up, we get the following relation between the two sets of variables $(G_{ij}, B_{ij})$ and $(\fa g^{ij},\beta^{ij})$:
\eq{ \label{translation}
B_{ij} &\;\to\; \beta^{ij} = (B^{-1})^{ij} \;, \\
G_{ij} &\;\to\; \fa g^{ij} = \beta^{im}\,\beta^{jn}\,G_{mn} \;.
}
It is now possible to convert gauge transformations of the $B$-field into cor\-res\-pon\-ding transformations of $\beta$ and also $\fa g$ (because $\beta$ is included in the second line of \eqref{translation}). Denoting in general variations under gauge transformations by $\delta^{\mathrm{gauge}}$, we get:
\eq{
\label{gaugetrafo}
\delta_{\xi}^{\mathrm{gauge}}\hspace{1pt}B_{ij} &=\partial_i\hspace{1pt}\xi_j-\partial_j\hspace{1pt}\xi_i \;, \\
\delta_{\xi}^{\mathrm{gauge}}\hspace{1pt}\beta^{ab} &= \beta^{am}\beta^{bn}
		\bigl(\partial_m\xi_n-\partial_n\xi_m\bigr) \;, \\
\delta_{\xi}^{\mathrm{gauge}}\hspace{1pt}\fa g^{ab} &= 2\hspace{1pt}\fa g^{(\underline{a}m}\beta^{\underline{b})n}
		\bigl(\partial_m\xi_n-\partial_n\xi_m\bigr) \;.
}
To illustrate the concept, let us prove the last equality:
\eq {
\delta_{\xi}^{\mathrm{gauge}} \fa g^{ij} &= \left(\delta_{\xi}^{\mathrm{gauge}}\beta^{in}\right)\beta^{jm} G_{nm} + \beta^{in}\left(\delta_{\xi}^{\mathrm{gauge}}\beta^{jm}\right)G_{nm} \\
&= \beta^{ir}\beta^{ns}\left(\partial_r \xi_s - \partial_s \xi_r \right)\beta^{jm} G_{nm} \\
&\hspace{12pt} + \beta^{in}\beta^{jr}\beta^{ms}\left(\partial_r \xi_s -\partial_s \xi_r \right) \\
&=2\hspace{1pt}\fa g^{(\underline{a}m}\beta^{\underline{b})n}
		\bigl(\partial_m\xi_n-\partial_n\xi_m\bigr) \;,
}
where in the first line we used the fact that the original metric $G_{ij}$ does not transform under gauge transformations of the $B$-field and in the last line we employed \eqref{translation} to write everything in terms of the transformed fields.

As a result of the redefinition \eqref{translation}, we see that the metric also transforms non-trivially under a redefined gauge transformation. From general relativity it is well known how the metric tensor transforms under diffeomorphisms of spacetime: Given a vector field $v = v^m \partial_m \in \Gamma(TM)$, the infinitesimal variation of the (inverse) metric under the flow parametrized by $v$ is given by the Lie derivative:
\eq{
\delta_v \fa g &= L_v \, \left(\fa g^{ij}\,\partial_i \otimes \partial_j \right) \\
&=  \left(v^m\partial_m \fa g^{ij} - \fa g^{ik}\partial_k v^j - \fa g^{jk}\partial_k v^i \right)\,\partial_i \otimes \partial_j \;.
}
In the following it is our goal to find a geometric interpretation of the transformation \eqref{gaugetrafo} of the metric. As a first step, we can construct a vector field parametrized by the gauge parameter $\xi$ by taking the anchor of $\xi$. Assuming that the metric transforms under the flow parametrized by this vector in the standard way described above, we can separate the diffeomorphism part from \eqref{gaugetrafo}:
\eq{ \label{separation}
 (L_{\beta^\sharp\xi}\hspace{1pt}\fa g)^{ij} - \delta_{\xi}^{\mathrm{gauge}}\hspace{1pt}\fa g^{ij}
  &=: (\fa{\mathcal L}_{\xi}\hspace{1pt}\fa g)^{ij}  \\
 &=  \xi_n \beta^{nm}\partial_m \fa g^{ij} - \fa g^{ik}\partial_k(\xi_n \beta^{nj}) - \fa g^{ik}\partial_k(\xi_n \beta^{ni}) \\
&\hspace{12pt} -\left(\fa g^{im}\beta^{jn} + \fa g^{jm}\beta^{in} \right)\left(\partial_m \xi_n - \partial_n\xi_m\right) \\
&=\,  \xi_n \beta^{nm}\partial_m \fa g^{ij} + g^{jm}\left(\beta^{in}\partial_n \xi_m - \xi_n \partial_m \beta^{ni} \right) \\
&\hspace{71pt} +\fa g^{im}\left(\beta^{jn}\partial_n \xi_m - \xi_n \partial_m \beta^{nj} \right) \;.
}
We now show, that the last expression can be interpreted as a \emph{quasi-Lie derivative} with respect to the one-form gauge parameter $\xi$. We have to take the cotangent bundle $T^*M$ together with the \emph{Koszul-Schouten} bracket, which was already des\-cribed in the examples of \ref{subsubsec-ex}. Let us define the following differential operator $\klie $ by its action on functions $f$, one-forms $\eta$ and vector fields $X$ by
\eq{
\label{quasi-Lie-derivative}
\klie(f) &= \beta^{\sharp}(\xi)(f)\;,\hspace{40pt} \fa{\mathcal L}_{\xi} \eta = \bigl[\xi, \eta\bigr]\ks \; , \\
&\fa{\mathcal L}_{\xi} X = \iota_{\xi} \circ d_{\beta} X + d_{\beta} \circ \iota_{\xi} X \;,
}
which in local coordinates gives
\eq{\label{quasi-Liecoordinates}
\klie f &= \xi_m \beta^{mn}\partial_n f =:\, \xi_m D^m f \,,\\
\klie \eta &= \left(\xi_m D^m\eta_a - \eta_mD^m\xi_a+\xi_m\eta_n\,Q_a{}^{mn}\right)\,dx^a \;, \\
\klie X &= \left(\xi_mD^mX^a+X^mD^a\xi_m-X^m\xi_n\,Q_m{}^{na}\right)\,\partial_a\,,
}
where we have introduced the $Q$-flux, given by $Q_k{}^{ij} := \partial_k \beta^{ij}$ and we use the derivative $D^a := \beta^{an}\partial_n$.

From the definition \eqref{quasi-Lie-derivative}, we see that $\kl$ is a derivation on the spaces of vector fields and one-forms since it is linear and satisfies the Leibniz rule, e.g. for one-forms:
\eq{
\klie f\eta &=\, \left[\xi, f \eta \right]_{KS(\beta)} =\, \beta^{\sharp}(\xi)(f)\,\eta + f\left[\xi,\eta\right]_{KS(\beta)} \\
&=\, \klie(f)\, \eta + f\,\klie \eta \;.
}
However, for non-vanishing $\Theta$-flux \eqref{R-flux} the Koszul-Schouten bracket $[\cdot,\cdot]_{KS(\beta)}$ does not satisfy the Jacobi-identity any more and therefore, the last property in proposition \ref{Lierel} acquires additional terms dependent on $\Theta$. More precisely, we have
\eq{\label{jacK}
	\hspace{-16pt}\mathfrak{J}(\eta,\chi,\zeta)
	&= \bigl[\eta,[\chi,\zeta]_{KS(\beta)}\bigr]_{KS(\beta)} +\bigl[\zeta,[\eta,\chi]_{KS(\beta)}\bigr]\ks +\bigl[\chi,[\zeta,\eta]\ks\bigr]\ks \\
	&= \bigl[\Lie_\eta,\Lie_\chi\bigr]\hspace{1pt}\zeta-\Lie_{[\eta,\chi]\ks}\zeta\\
	&= d\bigl(\Theta(\eta,\chi,\zeta)\bigr) + \iota_{(\iota_\zeta\iota_\chi\Theta)}d\eta
		+ \iota_{(\iota_\eta\iota_\zeta\Theta)}d\chi
		+ \iota_{(\iota_\chi\iota_\eta\Theta)}d\zeta \;.
}
In addition to that, the operator $\klie$ does only commute with the corresponding differential $d_{\beta}$ up to $\Theta$-flux terms (for the case of a Poisson manifold, see \eqref{Lierel}). Here we get:
\eq{
\bigl[\klie,d_{\beta}\bigr]f = \, -\Theta^{ijk}\xi_j \partial_k f \partial_i \;.
}
To sum up, $\klie$ has the same properties as an ordinary Lie derivative only \emph{up to terms depending on $\Theta$}. Sending the $\Theta$-flux to zero (which is equivalent to having a Poisson manifold, as was described at the beginning), we arrive at the example of a Lie algebroid $(T^*M,[\cdot,\cdot]\ks,\beta^{\sharp})$ (described in section \ref{subsubsec-ex}). We therefore use the term \emph{quasi-Lie derivative} for such a differential operator.

We are now able to give the geometric interpretation of the result in equation \eqref{separation}. With the mathematical structure introduced so far, the proof is now very easy.
\begin{prop}
The transformation $\delta_{\xi}^{\mathrm{gauge}}$ of the metric $\fa g$ can be separated into a diffeomorphism and an additional part described by the quasi-Lie derivative $\klie$:
\eq{\label{gaugemetric}
\delta_{\xi}^{\mathrm{gauge}} \fa g =\, L_{\beta^{\sharp}(\xi)}\,\fa g - \klie \,\fa g \;.
}
\end{prop}
\begin{proof}
We only have to compute $\klie\, \fa g$. Using the derivation property and the action on vector fields given in \eqref{quasi-Liecoordinates}, we get
\eq{
(\klie \fa g)^{ij} &=\, \xi_n D^n \fa g^{ij} + g^{kj}\left(D^i \xi_k - \xi_n Q_k{}^{ni} \right) \\
&\hspace{57pt} +\fa g^{ik}\left(D^j \xi_k - \xi_n Q_k{}^{nj} \right) \;,
}
which coincides with the last expression in \eqref{separation}.
\end{proof}
According to this separation, in the next section we are going to introduce the concept of $\beta$-diffeomorphisms and $\beta$-tensors in addition to ordinary diffeomorphisms and tensors. We could also take the name quasi-diffeomorphism but the appearance of $\beta$ which in our case is a quasi-Poisson structure already indicates the special quality of these kind of transformations.

\subsubsection{$\beta$-diffeomorphisms and -tensors}
In the last section we used the anchor map induced by the bi-vector $\beta$ to transform objects like one-forms $\xi$ or the metric $G$ from the cotangent to the tangent bundle. Together with the identification $\beta = B^{-1}$ suggested by \cite{Seiberg:1999vs} we were able to write down the transformation behavior of the new quantities under gauge transformations of the $B$-field. It turned out that the transformation rule for the new metric consists of two parts: An infinitesimal diffeomorphism and an additional transformation which we will call $\beta$-diffeomorphism, because it is possible to write them in terms of a quasi-Lie derivative $\klie$. We now generalize this observation to define $\beta$-tensors as sections in the tangent/cotangent bundles transforming with $\klie$ under an infinitesimal variation in the direction of a one-form $\xi$. To give a precise statement, we have to define the latter. We first have to specify the infinitesimal transformation of scalars with respect to a one-form $\xi$. There is only one natural way, namely to use the anchor map:
\eq{ \label{scalartrafo}
\fa \delta_\xi f :=\, \xi_m \beta^{mk}\partial_k f =\, \xi_m D^m f \;.
}
In addition we assume that the infinitesimal variation commutes with partial differentiation, which is the same as in ordinary differential geometry. Having these two properties, it is now possible to calculate the variation $\fa \delta_\xi$ of arbitrary combinations of fields and their derivatives. In standard differential geometry, there is a distinguished set of fields, called \emph{tensor} fields , whose infinitesimal transformation under standard diffeomeomorphisms is given by the Lie derivative: \mbox{$(\delta_X T)^{m_1\dots m_r}{}_{n_1\dots n_s} = (L_X T)^{m_1\dots m_r}{}_{n_1\dots n_s}$}. In complete analogy we now want to distinguish $\beta$-tensors by demanding them to transform with the quasi-Lie derivative $\klie$. As an example, $\beta$-one-forms and $\beta$-vectors are defined to transform as
\eq{ \label{betaoneform}
\left(\fa \delta_\xi \eta\right)_a &=\, \xi_m D^m\eta_a - \eta_mD^m\xi_a+\xi_m\eta_n\,Q_a{}^{mn}\,, \\
\left(\fa \delta_\xi X \right)^a &=\, \xi_m D^m X^a+X^m D^a \xi_m-X^m \xi_n\,Q_m{}^{na}\;.
}
This can be further generalized to fields of arbitrary index structure. As we are only interested in proper sections of the tangent- and cotangent bundles, we use these objects to distinguish the subclass of $\beta$-tensors:
\begin{definition}
A tensorfield $T\in\Gamma\bigl( \hspace{1pt}(\otimes^rTM)\otimes(\otimes^s T^*M) \bigr)$ is called a \emph{$\beta$-tensor} if it transforms infinitesimally in the direction of $\xi \in \Gamma(T^*M)$ as
\eq{\label{btensor}
	\fa\delta_{\xi}\hspace{1pt}T^{a_1\dots a_r}{}_{b_1\dots b_s}
	= \bigl(\fa{\mathcal{L}}_{\xi} \hspace{1pt}T \bigr)^{a_1\dots a_r}{}_{b_1\dots b_s}  \;,
}
where $\fa{\mathcal{L}}_{\xi}$ acts on general tensor fields by
\eq{
 \bigl(\fa{\mathcal{L}}_{\xi} \hspace{1pt}T \bigr)^{a_1\dots a_r}{}_{b_1\dots b_s}
  =  & \hspace{15pt} \xi_m\,D^m \hspace{1pt}T^{a_1\dots a_r}{}_{b_1\dots b_s} \\
&- \sum_{i=1}^s \bigl(D^m\xi_{b_i}+\xi_n\,Q_{b_i}{}^{mn}\bigr)\hspace{1pt}
T^{a_1\dots a_r}{}_{b_1\dots b_{i-1}\, m\, b_{i+1}\dots b_s} \\
&+ \sum_{i=1}^r \bigl(D^{a_i}\xi_m+\xi_n\,Q_m{}^{a_i n} \bigr) \hspace{1pt}
T^{a_1\dots a_{i-1}\,m\,a_{i+1}\dots a_r}{}_{b_1\dots b_s} \,.
}
 \end{definition}
To conclude this section, we want to use once more the analogy with standard differential geometry to determine the transformation behavior of the anchor $\beta$ under infinitesimal variations $\fa \delta_\xi$. As is well known, the total differential of a function should belong to the set of one-forms: $df = \partial_m f \, dx^m$. In the same way, we demand that the differential $d_{\beta}$ of a scalar function $f$ should be a $\beta$-tensor:
\eq{ \label{varbeta1}
\fa \delta_\xi \, D^a f  &= \left(\fa \delta_\xi \beta^{ak}\right)\partial_k f + \beta^{ak}\partial_k (\fa \delta_\xi f)\\
&\stackrel{!}{=} \klie (D^a f) \\
&= \xi_m \Theta^{map}\partial_p f + \xi_m D^aD^m f + D^m f D^a \xi_m \;.
}
And therefore we infer the infinitesimal transformation behavior of the anchor $\beta$:
\eq{\label{varbeta2}
\fa \delta_\xi \beta^{ij} &= \,\xi_m \Theta^{mij} \\
&=\, \klie \beta^{ij} + \beta^{im}\beta^{jn}\left(\partial_m \xi_n - \partial_n \xi_m \right) \;.
}
Thus, the anchor does not transform as a proper $\beta$-tensor, but has an additional term showing that one should consider it as a gauge field. As a non-trivial consistency check, we want to show that the field strength of $\beta$, which is given by the $\Theta$-flux, is again a proper $\beta$-tensor:
\eq{\label{Rflusstensor}
\fa \delta_\xi \Theta &= \,\fa \delta_\xi \frac{1}{2}\bigl[\beta,\beta\bigr]_{SN} = \frac{1}{2}\left(\bigl[\fa \delta_\xi \beta, \beta \bigr]_{SN} + \bigl[\beta,\fa \delta_\xi \beta \bigr]_{SN}\right) \\
&=\, \frac{1}{2}\left(\bigl[\iota_\xi \Theta, \beta \bigr]_{SN} + \bigl[\beta,\iota_\xi \Theta \bigr]_{SN} \right)=  d_\beta \iota_\xi \Theta \\
&=\, \klie \Theta \;,
}
where in the last line, we have used the Bianchi identity $d_\beta \, \Theta = 0$. Thus we can use the $\Theta$-flux to construct actions invariant under both, diffeomorphisms and $\beta$-diffeomorphisms.

\subsubsection{Algebra of $\beta$-diffeomorphisms}
Already the violation \eqref{jacK} of the Jacobi-identity for the quasi-Lie derivative $\kl$ by $\Theta$-flux terms suggests to look more closely at successive application of two $\beta$-transformations. It will turn out that the infinitesimal $\beta$-transformation of a $\beta$-tensor is not a $\beta$-tensor any more (the defect given, as expected, by terms involving $\Theta$), but the \emph{commutator} of two $\beta$-diffeomorphisms is again a $\beta$-diffeomorphism plus an ordinary diffeomorphism. In other words, only the commutator is again a meaningful object and in addition, we observe the relation to standard diffeomorphisms: Commuting a $\beta$-diffeomorphism with a standard diffeomorphism results in a standard diffeomorphism. As a consequence, to have a closed algebra of infinitesimal transformations, we have to include both, standard- and $\beta$-diffeomorphims. In\-tui\-tively, this means that we have the direct sum of these two transformations on the ``Lie algebra''-level and as a consequence, in the large  we expect to have a symmetry group generated by both, diffeomorphisms and the large analogue of $\beta$-diffeomorphisms.

At the infinitesimal level, let us now state the precise result on the algebra of both kinds of transformations:
\begin{prop}
Let $\delta_{X_1}, \delta_{X_2}$ be two infinitesimal diffeomorphisms pa\-ra\-me\-trized by vector fields $X_1, X_2$ and $\fa \delta_{\xi_1}, \fa \delta_{\xi_2}$ be two infinitesimal $\beta$-diffeomorphims pa\-ra\-me\-trized by one-forms $\xi_1, \xi_2$ respectively, then we have the following algebra of com\-mu\-ta\-tors:
\eq{\label{trafo_alg}
	\bigl[\delta_{X_1},\delta_{X_2}\bigr] &= \delta_{[X_1,X_2]_L} \;, \\
	\bigl[\fa\delta_{\xi_1},\delta_{X_1}\bigr] &= \delta_{\kl_{\xi_1} X_1} \;,\\
	\bigl[\fa\delta_{\xi_1},\fa\delta_{\xi_2}\bigr] &= \fa\delta_{[\xi_1,\xi_2]\ks}
		+ \delta_{\iota_{\xi_2}\iota_{\xi_1}\Theta} \;.
}
\end{prop}

\begin{proof}
The first one is trivial and the second and third one are proved along the same lines. We prove the third commutator. First, one has to evaluate both sides of the relation on functions, vector fields and forms. We only check it for scalars and one-forms because the calculation is similar for vector fields and the generalization to forms and vector fields of higher degree is straight forward. Let us begin with scalars $\phi$. Using the condition \eqref{varbeta1}, the left hand side of the third relation gives
\eq{
\bigl[\fa\delta_{\xi_1},\fa\delta_{\xi_2}\bigr]\,f &=\, \kl_{\xi_1} \left(\kl_{\xi_2} \phi \right) - \kl_{\xi_2} \left(\kl_{\xi_1} \phi \right) \\
&=\, \kl_{[\xi_1,\xi_2]\ks} \phi + (\xi_{1})_m (\xi_{2})_n \Theta^{mnk}\partial_k \phi \;.
}
The evaluation on forms is more complicated, because there are no simple conditions as for scalars and one has to calculate directly. Let $\gamma, \xi, \eta$ be one forms, then we have:

\vspace{8pt}

\eq{
\Bigl[\fa \delta_\gamma &\left(\fa \delta_\xi \eta \right) - \kl_\gamma \left( \klie \eta \right)\Bigr] _a  \\
=&\,\left(\gamma_r D^r \xi_m - \xi_r D^r \gamma_m + \gamma_r \xi_w Q_m{}^{rw} \right)D^m \eta_a \\
&+ \xi_m \gamma_w \Theta^{wmk} \partial_k \eta_a + \xi_m D^m\left(\gamma_r D^r \eta_a - \eta_r D^r \gamma_a + \gamma_r \eta_w Q_a{}^{rw} \right) \\
&-\left(\gamma_r D^r \eta_m -\eta_r D^r \gamma_m + \gamma_r \eta_w Q_m{}^{rw}\right)D^m \xi_a - \eta_m \gamma_w \Theta^{wmk} \partial_k \xi_a \\
&- \eta_m D^m \left(\gamma_r D^r \xi_a - \xi_r D^r \gamma_a + \gamma_r \xi_w Q_a{}^{rw} \right) \\
&+ \left(\gamma_r D^r \xi_m - \xi_r D^r \gamma_m + \gamma_r \xi_w Q_m{}^{rw}\right)\eta_n Q_a{}^{mn} \\
& + \xi_m Q_a{}^{mn} \left(\gamma_r D^r \eta_n - \eta_r D^r \gamma_n + \gamma_r \eta_w Q_n{}^{rw}\right) + \xi_m \eta_n \partial_a\left(\gamma_w \Theta^{wmn}\right) \\
& - \gamma_m D^m\left(\xi_k D^k\eta_a -  \eta_k D^k \xi_a + \xi_k \eta_n Q_a{}^{kn}\right) \\
& + \left(\xi_k D^k \eta_m - \eta_k D^k \xi_m + \xi_k \eta_n Q_m{}^{kn}\right) D^m \gamma_a \\
& - \gamma_m \left(\xi_k D^k \eta_n - \eta_k D^k \xi_n + \xi_k \eta_w Q_n{}^{kw} \right) Q_a{}^{mn} \\
=&\,\xi_m \eta_n \Theta^{mnk} \left(\partial_a \gamma_k - \partial_k \gamma_a \right)\;,
}

\noindent where in the last line we used the  Bianchi-identity \eqref{bianchiIV}, and the different anti-symmetrizations were written out to see that the following part in the first step of the calculation above vanishes:
\eq{
0 =\;  & \partial_a \Theta^{kmn} + D^m Q_a{}^{kn} - D^n Q_a{}^{km} \\
&- D^k Q_a{}^{mn} - Q_a{}^{np}Q_p{}^{km} + Q_a{}^{mp} Q_p{}^{kn} - Q_a{}^{kp} Q_p{}^{mn} \;.
}
In addition, we used the commutator algebra $\left[D^i,D^j\right] = \Theta^{ijk}\partial_k + Q_k{}^{ij}D^k $ to exchange derivatives. Thus, we get the following result on the commutator of two infinitesimal $\beta$-diffeomorphisms evaluated at a one-form $\eta$:
\eq{
\left(\left[\fa \delta_\gamma, \fa \delta_\xi \right]\eta\right)_a =\, \left(\left[\kl_{\gamma},\klie \right]\eta\right)_a &+ \xi_m \eta_n \Theta^{mnk}\left(\partial_a \gamma_k - \partial_k \gamma_a \right) \\
&- \gamma_m \eta_n \Theta^{mnk} \left(\partial_a \xi_k -\partial_k \xi_a \right) \;.
}
Finally, taking also relation \eqref{jacK} for the Jacobiator of the Koszul-Schouten bracket into account, the above can be rewritten to
\eq{
\left[\fa \delta_\gamma, \fa \delta_\xi \right]\eta =&\, \kl_{[\gamma,\xi]\ks}\eta + d\left(\Theta (\gamma, \xi, \eta)\right) + \iota_{\iota_\eta \iota_\xi \Theta} d\gamma \\
&\,+\iota_{\iota_\gamma \iota_\eta \Theta} d\xi + \iota_{\iota_\xi \iota_\gamma \Theta}d\eta -\iota_{\iota_\eta \iota_\xi \Theta}d\gamma + \iota_{\iota_\eta \iota_\gamma \Theta} d\xi \\
=&\, \kl_{[\gamma,\xi]\ks} \eta + \left( d\circ \iota_{\iota_\xi \iota_\gamma \Theta} + \iota_{\iota_\xi\iota_\gamma \Theta}\circ d \right) \eta \;,
}
and therefore the desired result.
\end{proof}

\subsubsection{Geometry with $H$- and $\Theta$-flux}
\label{subsec-geoHR}

Having established the notion of $\beta$-diffeomorphisms / tensors and their relations to gauge transformations of the NS-NS B-field, we are now ready to turn to our original question of constructing a gravity theory having the new fields $\fa g ^{ij}, \beta^{ij},\phi$ as dynamical degrees of freedom. In the original bosonic low energy effective action of string theory \eqref{stringaction_intro}, besides the Ricci-scalar and the dilaton we also have the $H$-flux term. Therefore we expect to get a corresponding $\Theta$-flux term. The original action was invariant under diffeomorphisms and gauge transformations of the $B$-field and consequently we are going to establish a differential geometric framework which is invariant under both, standard- and $\beta$-diffeomorphisms.

\paragraph{The appropriate Lie algebroid}
As we have seen in section \ref{ch-math}, to get a coordinate-independent differential geometry calculus (e.g. defining covariant differentiation, torsion and curvature), we have to define an underlying  Lie algebroid structure. In addition to diffeomorphism invariance the resulting theory should also be invariant under $\beta$-diffeomorphisms.

The first immediate but naive guess for the new framework would be the triple already discussed: \mbox{$(T^*M,[\cdot,\cdot]\ks,\beta^{\sharp})$}. But this is only a proper Lie algebroid for vanishing $\Theta$-flux. Indeed, evaluating for example the homomorphism property of the anchor in proposition \ref{homoeigenschaft} on basis one-forms $dx^i$ results in:
\eq{
[\beta^{\sharp}(dx^i),\beta^{\sharp}(dx^j)]_L - \beta^{\sharp}([dx^i,dx^j]\ks) = \Theta^{ijk}\partial_k \;.
}
But instead of giving up this guess, let us try to see the restrictions on the bracket by \emph{demanding} the homomorphism property. Define the following bracket (the terminology will become clear later in this section)
\eq{
\left[dx^i,dx^j\right]^H \ks &=\, (Q_k{}^{ij} + K_k{}^{ij})dx^k \\
&=:\, \cq_k{}^{ij} \;,
}
and extended for general one-forms $\xi,\eta$ by the Leibniz rule
\eq{
[dx^i,fdx^j]^H\ks = \beta^{\sharp}(dx^i)(f)\,dx^j + f\,[dx^i,dx^j]^H\ks \;.
}
From the homomorphism property, we can determine the field $K_k{}^{ij}$:
\eq{ \label{curQ}
0 \stackrel{!}{=}&\,\left[\beta^{\sharp}(dx^i),\beta^{\sharp}(dx^j)\right]_L- \beta^{\sharp}\left([dx^i,dx^j]^H \ks\right) \\
=&\, \left(Q_k{}^{ij} \beta^{kn} + \Theta^{ijn} - Q_k{}^{ij}\beta^{kn} - K_k{}^{ij}\beta^{kn}\right)\partial_n \;.
}
We now assume the invertibility of $\beta^{ij}$ and its relation to the NS-NS $B$-field, suggested by \cite{Seiberg:1999vs}
\eq{
\beta^{ij} = (1/B)^{ij} \;.
}
This will be enough to prove the Lie algebroid properties and the relation of $\Theta$- and $H$-fluxes.
From equation \eqref{curQ} we get $K_k{}^{ij} = \Theta^{ijn}(\beta^{-1})_{nk}$, and from the Leibniz property, we get the form of the bracket for general one-forms:
\eq{
[\xi,\eta]^H \ks = [\xi,\eta]\ks + (\beta^{-1})^{\sharp}\left(\iota_\eta \iota_\xi \Theta \right) \;.
}
The Jacobi-identity of this bracket also follows from the invertibility of the anchor and the homomorphism property:
\eq{
\beta^{\sharp}\left(\left[\xi,[\eta,\zeta]^H\ks\right]^H \ks +\mathrm{cycl}\right)
=&\, \bigl(\left[\beta^{\sharp}(\xi),[\beta^{\sharp}(\eta),\beta^{\sharp}(\zeta)]_L\right]_L + \mathrm{cycl} \bigr) \\
=&\,0 \;,
}
because after applying the anchor we deal with vector fields and the standard Lie bracket where the Jacobi identity is valid. Thus, under these conditions, the triple $(T^*M, [\cdot,\cdot]^H \ks, \beta^{\sharp})$ is a Lie algebroid. Using proposition \ref{Algebroidhomodiff}, we are able to relate the de Rham differential $d$ to the associated differential $d_\beta ^H$ on $TM$ (defined by \eqref{algebroiddiff}), and therefore the $H$ to the $\Theta$-flux:
\eq{
\wedge^3 \beta^{\sharp}(H) =&\, \wedge^3 \beta^{\sharp}(d\,B) \\
=&\, d_\beta^H\,\wedge^2\beta^{\sharp}(B) =\, - d_\beta ^H (\beta) \\
=&\,-\Theta \;,
}
which means in components:
\eq{ \label{hochziehen}
\Theta^{ijk} =\,\beta^{im}\beta^{jn}\beta^{kp}\,H_{mnp} \;.
}
Conversely, this relation is also sufficient for $(T^*M, [\cdot,\cdot]^H \ks,\beta^{\sharp})$ being a Lie algebroid, because we can write the Jacobiator for the bracket also as follows:
\eq{\label{jacobiH}
	\bigl[\xi &,[\eta,\zeta]^H \ks \bigr]^H \ks + \mathrm{cycl.} \\
	&\hspace{3pt}=\, d\bigl(\mathcal{R}(\xi,\eta,\zeta)\bigr) + \iota_{(\iota_\zeta\iota_\eta\mathcal{R})}d\xi
		+ \iota_{(\iota_\xi\iota_\zeta\mathcal{R})}d\eta
		+ \iota_{(\iota_\eta\iota_\xi\mathcal{R})}d\zeta \;,
}
\noindent where $\mathcal{R}^{abc} = \Theta^{abc} - \beta^{am}\,\beta^{bn}\,\beta^{ck}\,H_{mnk}$ and therefore vanishes if condition \eqref{hochziehen} holds.

\paragraph{$[\cdot,\cdot]^H\ks$ and $\beta$-diffeomorphisms}

Let us finally investigate the transformation behavior of the bracket $[\cdot,\cdot]^H \ks $ under infinitesimal $\beta$-diffeomorphisms. It is important to note that, in order to get a $\beta$-diffeomorphism invariant differential geometry calculus, we have to start with a bracket which maps $\beta$-tensors into $\beta$-tensors. This is a non-trivial requirement which is not valid for the standard Koszul-Schouten bracket where we have
\eq{
	\fa\delta_{\xi} \bigl([\eta,\zeta]\ks\bigr)_a
	= \bigl( \fa{\mathcal L}_{\xi}[\eta,\zeta]\ks \bigr)_a + \eta_m\,\zeta_n\,\Theta^{mnk}\bigl(\hspace{1pt}d\xi\hspace{1pt}\bigr)_{ak}\,.
}
But fortunately, the $H$-twisted Koszul-Schouten bracket has the required property as the following calculation shows
\eq{\label{HKbt}
	\fa\delta_{\xi} \bigl([\eta,\zeta]^H\ks \bigr)_a
	&= \fa\delta_{\xi} \bigl([\eta,\zeta]\ks\bigr)_a
		+ \fa\delta_{\xi}\bigl(\Theta^{mnk}(\beta^{-1})_{ka}\,\eta_m\,\xi_n\bigr) \\
	&= \bigl( \fa{\mathcal L}_{\xi}[\eta,\zeta]^H\ks \bigr)_a + \eta_m\hspace{1pt}\zeta_n\hspace{1pt} \Theta^{mnk}(d\xi)_{ak} \\
& \hspace{12pt}+\Theta^{mnk}\bigl((\fa \delta_\xi - \klie )(\beta^{-1})_{ka}\bigr)\eta_m\,\zeta_n \\
	&= \bigl( \fa{\mathcal L}_{\xi}[\eta,\zeta]^H \ks \bigr)_a \,,
}
where in the last line, we used the transformation behavior \eqref{varbeta2}, which can be transformed to its inverse $(\beta^{-1})_{ij}$ in the standard way.
To sum up, the Lie algebroid $(T^*M,[\cdot,\cdot]^H \ks , \beta^{\sharp})$ with $\beta$ being the inverse of the NS-NS $B$-field (or equivalently ${\cal R}=\, 0$) is a Lie algebroid and $[\cdot,\cdot]^H \ks$ maps $\beta$-tensors to $\beta$-tensors. This Lie algebroid will therefore be used in the next section to set up a differential geometry invariant under both, standard- and $\beta$-diffeomorphisms.

\subsubsection{Differential geometry on $(T^*M,[\cdot,\cdot]^H \ks,\beta^{\sharp})$}
\label{subsec-dglie}
Having identified a proper Lie algebroid setting in the last section, we are now ready to apply the general theory of section \ref{ch-math}. We will provide explicit formulas for the Christoffel connection coefficients, torsion and curvature. By construction, the bi-vector $\beta$ will enter in these expressions non-trivially to ensure diffeomorphism- and $\beta$-diffeomorphism invariance. To distinguish the new quantities from the ones of standard Riemannian geometry, we will denote them by an additional hat-symbol.

\paragraph{Covariant derivative}
As stated for general Lie algebroids in definition \ref{defconnection}, let us define a connection on \mbox{$(T^*M,[\cdot,\cdot]^H \ks,\beta^{\sharp})$} by its action on basis one-forms $dx^i$:
\eq{ \label{Christoffels}
	\fa\nabla_{dx^a}\,dx^b \equiv \fa\nabla^a\,dx^b = \fa\Gamma_c{}^{ab}\, dx^c \;.
}
Note, that by replacing the tangent bundle by the cotangent bundle as total space of our Lie algebroid, the connection coefficients of $\fa \nabla$ now have the opposite index structure as the original ones. Using the Leibniz rule, definition \eqref{Christoffels}  means for general one-forms $\eta$:
\eq{\label{cd1}
	\fa\nabla^a\eta_b = D^a\eta_b + \fa\Gamma_b{}^{am}\,\eta_m \,.
}
Compatibility of the connection with the insertion map $\iota$, that is $D^a(\iota_X\eta) = \iota_X(\fa\nabla^a\eta)+\iota_\eta(\fa\nabla^a X)$ results in the corresponding covariant derivative of vector fields $X$:
\eq{
	\fa\nabla^a X^b = D^a X^b - \fa\Gamma_m{}^{ab}\,X^m \,.
}
Together with the product rule, we obtain the following formula for applying the covariant derivative to an $(r,s)$-tensor:
\eq{ \label{covdergen}
	\fa\nabla^c\,T_{a_1\dots a_r}{}^{b_1\dots b_s}
	= D^c\, T_{a_1\dots a_r}{}^{b_1\dots b_s}
	&+\sum_{i=1}^r\fa\Gamma_{a_i}{}^{cm}\,
		T_{a_1\dots a_{i-1}ma_{i+1}\dots a_r}{}^{b_1\dots b_s}\\
	&-\sum_{i=1}^s\fa\Gamma_{m}{}^{cb_i}\,
		T_{a_1\dots a_r}{}^{b_1\dots b_{i-1}mb_{i+1}\dots b_s} \,.
}
In addition to the general tensor properties of the covariant derivative stated in section \ref{ch-math}, we want to have $\beta$-diffeomorphism invariance, i.e. the covariant derivative of a $\beta$-tensor should be again a $\beta$-tensor. Similar to standard differential geometry, the derivative of a one-form or vector component is not a proper tensor (for this reason one has to introduce the covariant derivative). Indeed, in our case e.g. for the derivative of a one-form $\eta_a$ we compute:
\eq{
  \fa \delta_\xi \left(D^a \eta_b\right) =&\, \klie \left(D^a \eta_b\right) - D^a\left(D^m\xi_b - \xi_k {\cal Q}_b{}^{km}\right) \eta_m \;,
}
and therefore to cancel the anomalous second term, we have to propose the following transformation behavior for the connection coefficients:
\eq{ \label{anomalousgamma}
\fa\delta_\xi\,\fa\Gamma_c{}^{ab} = \klie \fa \Gamma_c{}^{ab} + D^a(D^b\xi_c
		- \xi_m\,\mathcal{Q}_c{}^{mb}) \; .
}
We will see in the next section that this condition is compatible with the definition of torsion, and finally it can be shown that the Levi-Civita connection satisfies this condition.

\paragraph{Torsion}
Following definition \ref{defcurv} of section \ref{ch-math}, we are able to define the torsion tensor in $(T^*M,[\cdot,\cdot]^H \ks,\beta^{\sharp})$ by using the corresponding bracket:
\eq{\label{torsion}
  \fa T(\xi,\eta) = \fa\nabla_{\!\xi}\,\eta-\fa\nabla_{\!\eta}\, \xi-[\xi,\eta]^H \ks \;.
}
As was proven in section \ref{ch-math}, it is ${\cal C}^{\infty}(M)$-linear in its arguments and as shown in the last two subsections, it maps two $\beta$-tensors into a $\beta$ tensor. The explicit coordinate expression reads
\eq{ \label{torsioncoordinates}
	\fa T_c{}^{ab} = \iota_{\partial_c}\fa T(dx^a,dx^b)
	= \fa\Gamma_c{}^{ab} -\fa\Gamma_c{}^{ba} - \mathcal{Q}_c{}^{ab} \,.
}
Finally, let us check that the anomalous transformation behavior \eqref{anomalousgamma} is compatible with the definition of torsion. From \eqref{torsioncoordinates}, we have to show ($\fa \Delta_\xi :=  \delta_\xi - \klie$)
\eq{
2\,\fa \Delta_\xi \Gamma_a{}^{[\underline{bc}]} =\,\fa \Delta_\xi{\cal Q}_a{}^{bc} \;.
}
The right hand side can be computed directly using the anomalous transformation behavior of $\beta$, given in \eqref{varbeta2}. The result is
\eq{
	\fa\Delta_{\xi}\,\mathcal{Q}_c{}^{ab} &= \mathcal{Q}_m{}^{ab}\,D^m\xi_c
	 + 2\,\mathcal{Q}_c{}^{m[\underline{a}}\,D^{\underline{b}]}\xi_m - 2\,\xi_m\,D^{[\underline{a}}\mathcal{Q}_c{}^{m\underline{b}]} \;,
}
but this is exactly two times the anti-symmetrization of the second term in \eqref{anomalousgamma}.

\paragraph{Levi-Civita connection}
As seen in section \ref{ch-math}, for general Lie algebroids, it is possible to distinguish a unique connection having vanishing torsion and being compatible with the metric. In our case, we use the torsion defined in the last section, and the metric will be on the total space $T^*M$, i.e. compatibility with the connection means:
\eq{
	(\beta^\sharp\xi)\hspace{1pt}\left(\fa g(\eta,\zeta)\right) = \fa g\bigl(\fa\nabla_{\!\xi} \eta,\zeta\bigr)
		+ \fa g\bigl(\eta, \fa\nabla_{\!\xi}\zeta\bigr) \;.
}
Combining this with vanishing of the torsion results in the Koszul formula, as given in \eqref{Koszulformula}. In our case it is given in terms of basis one-forms $dx^i$ by:
\eq{ \label{koszulformuladxi}
2\hspace{1pt} \fa g \bigl(\fa\nabla_{dx^i} dx^j,dx^k\bigr) = & \;D^i\,\fa g^{ik}
+ D^j\, \fa g^{kj} - D^k \,\fa g^{ij} - \fa g\left( dx^i ,[dx^j,dx^k]^H \ks\right) \\
&+ \fa g \left( dx^j,[dx^k,dx^i]^H \ks\right) + \fa g \left(dx^k,[dx^i,dx^j]^H\ks \right)\;.
}
This enables us to compute the Christoffel connection coefficients explicitly. They are given by:
\eq{\label{Christoffelsym}
	\fa\Gamma_c{}^{ab} = \frac{1}{2}\,\fa g_{cm}\bigl(D^a\fa g^{bm}
	+D^b\fa g^{am}-D^m\fa g^{ab}\bigr)
	-\fa g_{cm}\,\fa g^{(\underline{a}n}\,\mathcal{Q}_n{}^{\underline{b})m} + \frac{1}{2}\,\mathcal{Q}_c{}^{ab}\,.
}
We will again denote the Christoffel connection and its coefficients by the same symbols, because there is no danger of confusion. Furthermore, it can be checked by a straightforward computation using the tensor transformation behavior of $\fa g^{ij}$ and the anomalous behavior \eqref{varbeta1}, that the Christoffel symbols have the right anomalous behavior \eqref{anomalousgamma} to ensure $\beta$-tensoriality of the covariant derivative.

\paragraph{Relation to standard geometry}
\label{subsubsec-realisation}
Having established essential parts of differential geometry on our Lie algebroid, we want to give important relations to standard geometry on the tangent bundle. Using the anchor map $\beta^{\sharp}: T^*M \rightarrow TM$, we can relate the metric $\fa g$ to a metric $G$ on the tangent bundle:
\eq{
\otimes^2 \beta^{\sharp}\left(G_{mn}\,dx^m\otimes dx^n\right) = \fa g^{ij}\, \partial_i \otimes \partial_j \;.
}
For the components of the two metrics we get:
\eq{ \label{transmetric}
\fa g^{ij} = \beta^{im}\beta^{jn}\, G_{mn} \;.
}
Using this redefinition in the Christoffel symbols for the metric $G$,
\eq{
\Gamma^i{}_{jk} =\,\frac{1}{2}G^{im}\left(\partial_j G_{km} + \partial_k G_{jm} - \partial_m G_{jk}\right) \;,
}
we are able to relate these to the Christoffel symbols \eqref{Christoffelsym} of the metric $\fa g$:
\eq{
\label{transgamma}
\Gamma^i{}_{jk}=-\fa \beta^{ip}\, \fa \beta_{jm}\, \fa \beta_{kn}\,
         \fa \Gamma_p{}^{mn} - \fa\beta_{nk} \, \partial_j \fa\beta^{in}\, .
}
This relation enables us to relate standard covariant derivatives of standard tensors to Lie algebroid-covariant derivatives of $\beta$-tensors. As an example, let $\hat T _i$ be a $\beta$-tensor. Then we can relate it to a standard tensor on the tangent bundle by:
\eq{\label{transtensor}
T^k =\, \beta^{kn} \, \fa T_n \;.
}
Now let us prove the following relation between the standard and $\beta$-covariant derivatives:
\eq{
\beta^{mi}\beta_{nj} \nabla _i T^j =\, \fa \nabla^m \fa T _n \;.
}
\begin{proof}
\eq{
  \beta^{mi}\beta_{nj}\left(\nabla_i T^j \right) =&\, \beta^{mi}\beta_{nj} \left( \partial_i T^j + \Gamma^j{}_{ik} T^k \right) \\
=&\,\beta^{mi}\beta_{nj}\left(\partial_i T^j -\beta^{jj'}\beta_{ii'}\beta_{kk'} T^k\, \fa \Gamma_{j'}{}^{i'k'}  - T^k \,\beta_{k'k} \partial_i \beta^{jk'} \right) \\
=&\, D^m\left(\beta_{nj} T^j \right) + \fa \Gamma_n{}^{mk'}\beta_{k'k} T^k \\
=&\, \fa \nabla^m \fa T_n \;,
}
where in the second step we employed the transformation \eqref{transgamma} of the Christoffel symbols.
\end{proof}
\noindent Similarly we can relate the two covariant derivatives of tensors and their anchors of arbitrary index structure. As this result will be used frequently in the following, let us formulate it separately as a proposition:
\begin{prop}\label{covdertrans}
Let $\fa T^{a_1 \dots a_r}{}_{b_1\dots b_s}$ be a $\beta$-tensor with anchor given by:
\eq{ \label{ntensor12}
T_{a_1\dots a_r}{}^{b_1\dots b_s} = \beta_{a_1 a_1'}\cdots\beta_{a_r a_r'}\beta^{b_1 b_1'}\cdots \beta^{b_s b_s'} \, \fa T^{a_1'\dots a_r'}{}_{b_1'\dots b_s'}\;, }
then the covariant derivatives $\fa \nabla$ and $\nabla$ are related by:
\eq{
\fa \nabla^m \fa T^{a_1\dots a_r}{}_{b_1\dots b_s} =\, \beta^{mn}\beta^{a_1 a_1'}\cdots \beta^{a_r a_r'}\beta_{b_1 b_1'}\cdots \beta_{b_s b_s'} \,\nabla_n T_{a_1'\dots a_r'}{}^{b_1'\dots b_s'} \;.
}
\end{prop}
\noindent The proof is similar to the above example by using the product rule for covariant derivatives and the general definition of the covariant derivative, given in \eqref{covdergen}.
\paragraph{Curvature}
Again, following the general construction of section \ref{ch-math}, especially definition \ref{defcurv}, we write down the Riemann curvature tensor for our Lie algebroid. It is given by:
\eq{\label{Riemtensor}
	\fa R(\xi,\eta)\zeta = \bigl[\fa\nabla_{\!\xi},\fa\nabla_{\!\eta}\bigr]\zeta
		- \fa\nabla_{\![\xi,\eta]^H \ks}\, \zeta \;,
}
which in components reads:
\eq{
  \label{Riemtensorcomp}
	\fa R_a{}^{bcd}\equiv \iota_{\partial_a}\bigl(\fa R(dx^c,dx^d)dx^b\bigr)
	= 2\bigl(D^{[\underline{c}}\fa\Gamma_a{}^{\underline{d}]b} + \fa\Gamma_a{}^{[\underline{c}m}\,\fa\Gamma_m{}^{\underline{d}]b}\bigr)
		- \fa\Gamma_a{}^{mb}\,\mathcal{Q}_m{}^{cd} \, .
}
Since  the covariant derivative and the bracket give $\beta$-tensors,
also $\fa R$ is a $\beta$-tensor. Note that the first two terms are similar to the standard Riemann tensor (with partial derivatives replaced by $D^a$), but there is an additional term involving the ${\cal Q}$-flux, i.e. the structure functions of the $[\cdot,\cdot]^H\ks$-bracket.

\vspace{0.3cm}

\noindent To derive the symmetries and Bianchi identities of the curvature tensor, let us note the following: By using the transformation \eqref{transgamma} in the expression \eqref{Riemtensorcomp}, a straightforward calculation results in:
\eq{
\label{curvaturetrafo}
     R^d{}_{cab}=-\fa\beta^{dq}\, \fa\beta_{cp}\,
    \fa\beta_{am}\, \fa\beta_{bn}\,
      \fa R_q{}^{pmn}     \, .
}
Together with proposition \ref{covdertrans} about the different covariant derivatives, it is easy to derive the symmetries and Bianchi identities of the new curvature tensor $\fa R_i{}^{jkl}$ from the corresponding identities of the standard Riemann tensor. Raising indices with the metric $\fa g^{ab}$, the new Riemann tensor inherits the symmetries:
\eq{\label{Rsymm}
	\fa R^{abcd} = -\fa R^{bacd} \;,\hspace{40pt}
	\fa R^{abcd} = -\fa R^{abdc} \;, \hspace{40pt}
	\fa R^{abcd} = \fa R^{cdab} \;,
}
as well as the Bianchi identities:
\eq{
\label{Rsymm2}
       \fa R^{abcd}+\fa R^{adbc}+\fa R^{acdb}&=0 \;,\\
       \fa\nabla^m \fa R^{abcd}+\fa\nabla^d \fa R^{abmc}+\fa\nabla^c \fa
  		R^{abdm}&=0\;.
}

\vspace{0.3cm}

The Ricci tensor is defined  by $\fa R^{ab} = \fa R_m{}^{amb}$, which is symmetric in its indices due to \eqref{Rsymm}. In terms of the connection, it can be written as
\eq{
	\fa R^{ab} = D^m\fa\Gamma_m{}^{ba} - D^b\fa\Gamma_m{}^{ma}
		+ \fa\Gamma_n{}^{ba}\,\fa\Gamma_m{}^{mn}
		- \fa\Gamma_n{}^{ma}\,\fa\Gamma_m{}^{nb} \, .
}
Finally the Ricci scalar $\fa R = \fa g_{ab}\fa R^{ab}$ can be expanded in terms of the metric and the derivative $D^a$ in the following way
\eq{
\label{Rexpand}
     \hat R=-\Bigl[ &\hspace{15pt}D^a D^b \hat g_{ab}
          - D^a\left( \hat g_{ab} \, \hat g^{mn}\,  D^b \hat g_{mn}  \right)\\
          &-{1\over 4} \hat g_{ab}\Bigl( D^a  \hat g_{mn}\, D^b  \hat g^{mn}
                            -2 D^a  \hat g_{mn}\, D^m  \hat g^{nb}
           -  \hat g_{mn}\,  \hat g_{pq}\, D^a  \hat g^{mn}\, D^b  \hat g^{pq}\Bigr)\\
      &+{1\over 4}  \hat g_{ab}\,  \hat g_{mn} \, \hat g^{pq}\,
		{\cal Q}_p{}^{ma} {\cal Q}_q{}^{nb}
      +{1\over 2}  \hat g_{ab}\, {\cal Q}_m{}^{nb}\, {\cal Q}_n{}^{ma}
     +  \hat g_{ab}\, {\cal Q}_m{}^{ma}\, {\cal Q}_n{}^{nb} \\
      &+2   D^a \bigl( \hat g_{ab} \, {\cal Q}_m{}^{mb} \bigr)
      -  \hat g_{ab}\, \hat g_{mn}\,  D^a \hat g^{pn}  \,   {\cal Q}_p{}^{bm}
     +   \hat g_{ab} \,\hat g^{mn}\,  D^a \hat g_{mn}   \,  {\cal
       Q}_p{}^{bp}\,\Bigr]\;.
}
\noindent We now have the basic geometric building blocks to formulate a gravity theory (or more generally a low energy effective action for the bosonic string) on the Lie algebroid  $(T^*M,[\cdot,\cdot]^H \ks,\beta^{\sharp})$. Quantities like the Ricci scalar are now invariant under both, diffeomorphisms and $\beta$-diffeo\-mor\-phisms by construction.

\subsubsection{Bi-invariant theory of gravity}

In this section, we  construct an Einstein-Hilbert action invariant under standard as well as $\beta$-diffeomorphisms, which we call \emph{bi-invariant} for short.
This action contains the metric $\hat{g}^{ab}$ and the bi-vector ${\beta}^{ab}$ as dynamical fields. In addition, it will be coupled to the dilaton $\phi$ and the $\Theta$-flux.

\paragraph{Invariant action}

As we have illustrated in the last section, it is possible to construct a Ricci scalar $\fa{R}$ which behaves as a scalar with respect to both types of diffeomorphisms.
Furthermore, by construction, the derivative of the dilaton $D^a \phi$ is a $\beta$-tensor and therefore the corresponding kinetic term $\hat{g}_{ab} D^a \phi D^b\phi$ behaves as a $\beta$-scalar. Also, the $R$-flux $\Theta^{abc}$ is a tensor with respect to $\beta$-diffeomorphisms as was shown in \eqref{Rflusstensor}, and it behaves as a standard tensor due to its definition \eqref{R-flux} in terms of the Schouten-Nijenhuis bracket of ${\beta}$ with itself.
Therefore, the following Lagrangian is a scalar with respect to both types of diffeomorphisms:
\eq{
   \hat{\cal L}=e^{-2\phi} \left(   \hat R -{1\over 12} \Theta^{abc}\, \Theta_{abc}
              +4\hspace{0.5pt} \hat g_{ab}\, D^a\phi D^b \phi\right)\;.
}
This Lagrangian has been constructed in a way to resemble the bosonic low-energy effective action \eqref{stringaction_intro}. Analogous to the geometric case, $\Theta$ can also be included as (con-)torsion of the connection.

To obtain a bi-invariant action, we have to find an appropriate measure $\mu$. More precisely, the variation of
\eq{
     \hat S={1\over 2\kappa^2} \int d^nx\, \mu(\fa g,  \beta)\, \hat{\cal L}\, ,
}
under standard and $\beta$-diffeomorphisms, has to vanish. As it turns out, the direct analogue to Riemannian geometry, namely the measure $\mu = \sqrt{-|\hat{g}|}$ with $|\fa g|=\det \fa g^{ab}$, does not lead to the desired result. This can be seen from
\eq{\label{varmu}
\delta_X  \bigl(\sqrt{-|\fa g|}\,\fa{\cal L}\hspace{1pt}\bigr) &= \partial_m\bigl(X^m \sqrt{-|\fa g|}\,\fa {\cal L} \hspace{1pt}\bigr) -2\sqrt{-|\fa g|}\,(\partial_m X^m) \,\fa{\cal L}\, ,\\[3pt]
\fa\delta_{\xi}\bigl(\sqrt{-|\fa g|}\,\fa{\cal L}\hspace{1pt}\bigr)&=
       \partial_m\bigl( \sqrt{-|\fa g|}\, \fa {\cal L}\,
       \xi_n\,\bigr)\beta^{nm}
      -  \sqrt{-|\fa g|}\, \fa {\cal L}\,  \xi_m (\partial_n
       \beta^{mn}) \, .
   }
Obviously, the right-hand sides in \eqref{varmu} are not total derivatives which would be required for the action to be invariant. However, taking as an additional factor the determinant of  $ \beta^{-1}$ into account, that means
\eq{\label{correctmeasure}
\mu =  \sqrt{-|\fa g|}\, \bigl| \beta^{-1}\bigr| \;,
}
we obtain the correct behavior under both types of diffeomorphisms. This can be seen by considering the variation of the determinant of the bi-vector:
\eq{\label{varbeta}
\delta_X  \bigl|  \beta^{-1} \bigr| &= X^m \partial_m \bigl| \beta^{-1}\bigr| + 2\hspace{1pt}\bigl| \beta^{-1} \bigr| \partial_m X^m \;, \\
\fa \delta_{ \xi} \bigl| \beta^{-1}\bigr| &= 2 \bigl| \beta^{-1}\bigr|\, \xi_m\, \partial_n \beta^{mn} + \xi_m \beta^{mk} \partial_k \bigl| \beta^{-1}\bigr| \;,
}
so that  the combination of \eqref{varmu} and \eqref{varbeta} results in a total derivative. We therefore propose the following bi-invariant Einstein-Hilbert action coupled to a dilaton $\phi$ and $R$-flux $\Theta^{abc}$:
\eq{
\label{theorem1}
      \hat S={1\over 2\kappa^2} \int d^nx\, \sqrt{-|\hat g|}\,\bigl|  \beta^{-1} \bigr|\, e^{-2\phi}
     \Bigl(   \hat R -{1\over 12} \Theta^{abc}\, \Theta_{abc}
              +4\hspace{1pt} \hat g_{ab}\, D^a\phi D^b \phi\Bigr)\, .
}
Due to the appearance of the \mbox{(quasi-)}symplectic structure $\beta^{ab}$, we will call the theory defined by the action \eqref{theorem1} \emph{symplectic gravity}.

\paragraph{Remarks}

Let us close this section with two remarks about the  measure \eqref{correctmeasure}.
\begin{itemize}

\item In general, the determinant of an anti-symmetric matrix vanishes in odd
  dimensions. Thus, our measure \eqref{correctmeasure} only makes sense for
  even dimensions, e.g. for symplectic manifolds. Denoting the components of $\beta ^{-1}$ by $\beta_{ab}$\footnote{This is only a convenient notation and does not mean that the indices are lowered by the metric.}, for the  latter case one has:
\eq{
  \det \beta_{ab} = \bigl( {\rm Pfaff}\, \beta_{ab} \bigr)^2 \;,
}
so that the determinant $| \beta^{-1} |$ is always non-negative.

\item In the Lie-algebroid construction of section \ref{subsec-geoHR}, we have effectively replaced the tangent bundle of a manifold by the co-tangent bundle. Performing the same procedure for
an integral, we would formally obtain:
\eq{ \label{newint}
\int \sqrt{-|G|}\,dx^1 \wedge \ldots \wedge dx^{n} \quad\to\quad\int
\sqrt{-|\fa g|}\,\partial_1 \wedge \ldots \wedge \partial_{n}\, .
}
Employing then the inverse of the anchor, we can relate the right-hand side to a standard integral by using $\partial_a = \beta_{ab}\, dx^b$ which results in the same measure as in \eqref{correctmeasure}
\eq{
\int \sqrt{-|\fa g|}\, \partial_1 \wedge \dots \wedge \partial_{n} = \int \sqrt{-|\fa g|} \hspace{1pt}\bigl|  \beta^{-1} \bigr|\, dx^1 \wedge \ldots \wedge dx^n \,.
}
However, let us
note again that this replacement is only possible in an even number of dimensions, otherwise the determinant of $\beta$ would vanish and the anchor would not be invertible.
\end{itemize}

\subsubsection{Equations of motion}
\label{subsec-eom}
The first immediate physical question related to the symplectic gravity action \eqref{theorem1} is about the equations of motion for its dynamical fields, the metric $\fa g ^{ab}$, the bi-vector $\beta^{ab}$ and the dilaton $\phi$. The variation with respect to the different fields is done in the standard way and we only want to mention the main steps in the calculation and the final result. We will use two relations known from standard geometry, which also hold in the Lie algebroid case. The first one concerns integration by parts of a divergence:
\eq{
  \int d^nx\, \sqrt{-|\hat g|}\hspace{2pt}  \bigl| \beta^{-1} \bigr|\, \fa \nabla^a \eta_a
  = -\int d^nx\,\, \partial_a\Bigl(\sqrt{-|\hat g|}\hspace{2pt}  \bigl| \beta^{-1} \bigr|\,  \beta^{am}\, \eta_m\Bigr) = 0\;,
}
where we assume vanishing of the fields and their derivatives at infinity (or having a manifold without boundary). The second one is the analogue of the Palatini-identity for the variation of the Ricci tensor, which can be transformed straightforwardly to our setup:
\eq{
  \delta_{\fa\Gamma} \fa R^{ab} = \fa \nabla^m \delta \fa\Gamma_m{}^{ab}
   - \fa\nabla^a \delta \fa\Gamma_m{}^{mb}
   + \fa\Gamma_n{}^{mb} \bigl( \delta \fa\Gamma_m{}^{na} - \delta\fa\Gamma_m{}^{an} \bigr) \;.
}
Now using the first identity, the variation of the action \eqref{theorem1} with respect to the dilaton $\phi$ can be done to give:
\eq{
  \label{eom1}
  0 =  \hat R -{1\over 12}\hspace{1pt} \Theta^{abc}\, \Theta_{abc}
   -4\hspace{0.5pt} \hat g_{ab}\hspace{1pt} \fa\nabla^a\phi \fa\nabla^b \phi
   +4\hspace{0.5pt} \hat g_{ab}\hspace{1pt} \fa\nabla^a \fa\nabla^b \phi \;.
}
Next, the variation with respect to the metric $\fa g^{ab}$ is done with the help of the Palatini identity. The result is:
\eq{
  \label{eom2}
  0 = \fa R^{ab} + 2 \hspace{0.5pt} \fa\nabla^a \fa\nabla^b \phi - \frac14 \hspace{0.5pt} \Theta^{amn}
  \Theta^b{}_{mn} - \frac12 \, \fa g^{ab} \bigl[ \, \mbox{$\phi$ eom} \, \bigr] \;.
}
Here the terms are ordered in a way that the last term collects all the contributions which vanish due to the first equation of motion \eqref{eom1}. Finally, varying with respect to the bi-vector $\beta^{ij}$ results in:
\eq{ \label{eom3}
  0 = \frac12 \hspace{1pt} \fa\nabla^m \Theta_{mab} - (\fa\nabla^m \phi) \hspace{1pt}\Theta_{mab}
  &+ 2 \hspace{1pt} \fa g_{ap} \beta_{bq}  \bigl[ \, \mbox{$\fa g$ eom} \, \bigr] ^{pq} \\
  &+ \beta_{ab}  \bigl[ \, \mbox{$\phi$ eom} \, \bigr]  \;,
}
where we can again drop all the terms which vanish due to the equations of motion derived before. Let us also write down the contraction of \eqref{eom2} with the metric:
\eq{
  \label{eom_03}
  0 = \fa R + 2 \hspace{0.5pt} \fa g_{ab} \fa\nabla^a \fa\nabla^b \phi - \frac14 \hspace{0.5pt} \Theta^{abc}
  \Theta_{abc} \;.
}
Using this equation to eliminate the Ricci-scalar term in \eqref{eom1}, we arrive at the following set of three independent equations of motion for the dynamical fields:
\eq{
  \label{eom_final}
  &0 =
   -\frac12\hspace{0.5pt} \hat g_{ab}\hspace{1pt} \fa\nabla^a \fa\nabla^b \phi
   +\hat g_{ab}\hspace{1pt} \fa\nabla^a\phi \fa\nabla^b \phi
    -\frac{1}{24}\hspace{1pt} \Theta^{abc}\, \Theta_{abc} \;, \\
  & 0 = \fa R^{ab} + 2 \hspace{0.5pt} \fa\nabla^a \fa\nabla^b \phi
  - \frac14 \hspace{0.5pt} \Theta^{amn}
  \Theta^b{}_{mn} \;, \\
  &0 = \frac12 \hspace{1pt} \fa\nabla^m \Theta_{mab} - (\fa\nabla^m \phi) \hspace{1pt}\Theta_{mab}
  \;.
}
We observe that they are formally the same as the standard equations for the bosonic string (by replacing $\fa g^{ab} \rightarrow G_{ab}$, $\fa \nabla^a \rightarrow \nabla_a$, $\Theta \rightarrow H$ and $\fa R^{ab} \rightarrow R_{ab}$. However note, that due to the appearance of $\beta$ for example in the derivative $D^a$ and the additional ${\cal Q}$-flux terms (as for example in the Christoffel symbols \eqref{Christoffelsym} and in the Riemann curvature tensor \eqref{Riemtensorcomp}), the dynamics described by this set of equations is very different.

\subsubsection{Relations to string theory}
We are now going to compare the low energy effective action of the bosonic string with the symplectic gravity action \eqref{theorem1}.  It will turn out that they are related by a field redefinition similar to changes of fields encountered in double field theory but followed by a Seiberg-Witten limit described in section \ref{ch-quant}. Even though open strings were considered there, we will formally take the same scalings.

\paragraph{Field redefinitions}
Either from double field theory \cite{Andriot:2012an} or from open string theory \cite{Seiberg:1999vs}, let us recall an important field redefinition\footnote{See also the sigma model considerations of section \ref{sec-Poissonsigma}, especially \eqref{basicfieldredef1}. To focus on geometrical aspects, we set $2\pi \alpha' = 1$. }, relating the standard metric to one with upper indices and the $B$-field to a bi-vector:
\eq{\label{LMUfield}
\tilde g ^{ij}  =&\hspace{12pt}\left(\frac{1}{G + B} G \frac{1}{G - B} \right)^{ij} \;,\\
\tilde \beta^{ij} =&\,-\left(\frac{1}{G + B}B\frac{1}{G - B} \right)^{ij} \;.
}
As was shown in \cite{Andriot:2012an}, starting with the standard bosonic low energy effective string action:
\eq{
\label{stringaction}
S={1\over 2\kappa^2}
\int
\hspace{-0.75pt}
d^nx \hspace{1pt}\sqrt{-|G|}\hspace{1pt} e^{-2\phi}\Bigl(R-{\textstyle{1\over
    12}} H_{abc} H^{abc}
+4 \hspace{0.5pt} G^{ab}\, \partial_a \phi \hspace{1pt}\partial_b \phi
 \Bigr) \;,
}
and changing fields by the full redefinition \eqref{LMUfield} does not lead to the symplectic gravity action \eqref{theorem1}. However, looking at the Seiberg-Witten scaling limit, introduced in section \ref{ch-quant}, equation \eqref{scaling}, we observe that the field redefinition changes as follows:
\eq{\label{MPIfield}
\fa g =&\, -B^{-1}\,G\, B^{-1} =\, -\beta \,G \,\beta \;, \\
\beta =&\, \hspace{14pt}B^{-1}\;,
}
which means in components:
\eq{
\beta^{ij} = (B^{-1})^{ij}\;,  \qquad \fa g^{ij} = \beta^{im}\beta^{jn}\,G_{mn}\;.
}
Recall, that we already encountered these expressions by relating the metric and bi-vector from the Lie algebroid $(T^*M,[\cdot,\cdot]^H \ks,\beta^{\sharp})$ to the tangent bundle, e.g. \eqref{transmetric}. Therefore, we are able to use the results derived there to perform the field redefinition \eqref{MPIfield}. First of all, the measure transforms as:
\eq{
      \sqrt{-|G|}=\sqrt{-|\fa g|}\: \bigl|\beta^{-1}\bigr| \;.
}
This is the measure encountered in \eqref{theorem1}. Furthermore we can use the results of section \ref{subsec-dglie}, especially \eqref{transgamma}, to transform the Riemann curvature tensor as given there:
\eq{
\label{curvaturetrafo}
     R^d{}_{cab}=-\beta^{dq}\, \beta_{cp}\,
    \beta_{am}\, \beta_{bn}\,
      \fa R_q{}^{pmn}     \, .
}
But this means for the Ricci tensor and Ricci scalar:
\eq{
  R_{ab}=  \beta_{am} \, \beta_{bn} \, \fa R^{mn}\, , \hspace{40pt}
  R=\fa R\, .
}
The transformation of the $H$-flux can be computed to give the important relation \eqref{hochziehen}, which was proven to be equivalent to $(T^*M,[\cdot,\cdot]^H \ks,\beta^{\sharp})$ being a Lie algebroid. Contracting two fluxes completely gives:
\eq{
H_{abc} H^{abc}=\, \Theta^{abc}\hspace{0.5pt} \Theta_{abc}\;.
}
Finally we transform the kinetic term for the dilaton. Since the field redefinition was chosen in such a way that the dilaton does not change, we can simply rewrite
\eq{
      \partial_a\phi = \, \beta_{am} D^m\phi \;.
}
To sum up, we see that the field redefinition \eqref{MPIfield} relates the standard string action to the action in the symplectic frame with variables $(\fa g, \beta, \phi)$, i.e.
\eq{
S\bigl(\,G(\fa g,\beta),\,B(\fa g,\beta),\,\phi\,\bigr)=\hat S\bigl(\fa g,\beta,\phi\bigr)\, ,
}
where we denote the symplectic gravity action by $\fa S$.

\paragraph{Higher order corrections}

The effective action \eqref{stringaction} for the massless string modes is known to receive
higher-order $\alpha'$-corrections. Due to the freedom of
performing field redefinitions these are not unique. However, choosing a specific set of field variables, all the terms appearing at next to leading order
 \cite{Metsaev:1987bc,Metsaev:1987zx,Hull:1987yi}
can be expressed in terms of  covariant derivatives of
the curvature tensor $R_{abcd}$, the
three-form $H_{abc}$ and the dilaton $ \phi$.
Since we have determined how each of these building blocks transforms
under $\beta$-diffeomorphisms, it is possible to transform every term in the action se\-pa\-ra\-tely. As a consequence, we have a well-motivated  guess for the
form of the higher-order corrections in the symplectic gravity frame.
For instance,   the next to leading order corrections
to the bosonic string effective action are expected to take the form:
\eq{
 &  \fa S^{(1)} =  \frac{1}{2\kappa^2}
    \,\frac{\alpha'}{4}
  \int d^{n}x\, \sqrt{-|\fa g|}\, \bigl|\beta^{-1}\bigr|\, e^{-2\phi}
   \Bigl( \fa R^{abcd}\, \fa R_{abcd}
     -{\textstyle {1\over 2}}  \fa R^{abcd}\,  \Theta_{abm}
     \Theta_{cd}{}^m\\
   &\hspace{142pt}
      + {\textstyle{1\over 24}}    \Theta_{abc}\,  \Theta^{a}{}_{mn}\,  \Theta^{bm}{}_p\,
         \Theta^{cnp}
    -{\textstyle {1\over 8}} (\Theta^2)_{ab} \, (\Theta^2)^{ab}     \Bigr) \,,
}
where we have abbreviated $(\Theta^2)_{ab}=\Theta_{amn}\, \Theta_b{}^{mn}$. At the moment, it is not clear if such a Lagrangian can also be understood from first principles, i.e. without referring to the standard bosonic string action.

\subsubsection{Extension to the superstring}
\label{subsec-super}
Having reformulated successfully the bosonic low energy effective action of string theory in a frame with dynamical variables $(\fa g, \beta, \phi)$, the question arises if this is extendable also to the type II superstring actions. As spinors are independent of gauge transformations of the $B$-field and are not affected by the Lie algebroid constructions so far, we do not expect a nontrivial transformation behavior of them under the field redefinitions \eqref{MPIfield}. In contrast to this, the spin connection uses the tangent bundle and should therefore be generalized to Lie algebroids. Let us discuss the different sectors in turn.

\paragraph{The R-R sector}
In the R-R sector we first consider the $p$-form fields $C_{a_1\dots a_p}$. They are tensorial in the standard sense and they remain invariant under gauge transformations of the $B$-field. As discussed in section \ref{subsubsec-realisation}, especially equation \eqref{ntensor12} we can convert the $p$-form fields by use of the anchor form the $(G,B,\phi)$-frame to the $(\fa g,\beta,\phi)$-frame
\eq{
\fa C^{a_1\dots a_p} =\, \beta^{a_1 b_1}\cdots\beta^{a_p b_b}\,C_{b_1\dots b_p} \;.
}
From the theory established in section \ref{subsubsec-realisation}, we know that the totally antisymmetric $p$-vectors $\fa C^{a_1 \dots a_p}$ are standard- and $\beta$-tensors. In addition we know that the $\beta$-covariant derivative $\fa \nabla $ of these multi-vectors again gives a bi-covariant tensor:
\eq{
  \label{gen_potentials}
    \hat F^{a_1\ldots a_{p+1}}= \fa\nabla^{[\underline{a_1}} \fa C^{\underline{a_2\ldots a_{p+1}}]}\;.
}
However, to interpret them as proper physical field strengths, we have to show that they are invariant under the corresponding gauge transformations of the $p$-form fields $\fa C^{a_1\dots a_p}$:
\eq{
  \delta_{\Lambda} \fa C^{a_1\ldots a_p} = \fa\nabla^{[\underline{a_1}} \Lambda^{\underline{a_2 \ldots a_p}]} \;.
}
Let us show the invariance explicitly:
\eq{
\delta_{\Lambda} \fa F^{a_1 \dots a_{p+1}} =&\, \fa \nabla^{[\underline{a_1}} \fa \nabla^{\underline{a_2}} \Lambda^{\underline{a_3\dots a_{p+1}}]} \\
=&\,\fa R_k{}^{[ \underline{a_3 a_1 a_2}}\Lambda^{\underline{a_4\dots a_{p+1}}]k} =\,0 \;,
}
where in the last line we first used vanishing torsion for the Levi-Civita connection and then the Bianchi identity \eqref{Rsymm2} for the Riemann curvature tensor in $(\fa g, \beta, \phi)$-frame. Therefore, identifying $C_{a}$ and $C_{a_1 a_2 a_3}$ with the one- and three-form gauge potentials of type IIA supergravity, we have found corresponding expressions in the symplectic frame.

In analogy to the standard formulation, we then introduce generalized field strengths of the form:
\eq{\label{IIBRR}
  \fa {\cal F}_2 =\fa F_2 \ , \hspace{40pt} \fa{\cal F}_4=\fa F_4- \Theta\wedge \hat C_1\, ,
}
and for the corresponding action we consider:
\eq{
\label{iiaactions}
\fa S^{\mbox{\scriptsize R-R}}_{\rm IIA}=\frac1{2\kappa_{10}^2}\int d^{10}x
\hspace{1pt}
\sqrt{- |\fa g|}\hspace{1pt}\bigl| \beta^{-1}\bigr| \,
\Bigl(-{\textstyle {1\over 2}} \hspace{0.5pt} |\fa {\cal F}_2|^2-{\textstyle{1\over 2}}\hspace{0.5pt}|\fa
{\cal F}_4|^2\Bigr)\, ,
}
where we employ:
\eq{
\label{formsquareabs}
|\fa{\mathcal F}_p|^2={1\over p!} \hspace{1pt} \fa {\mathcal F}_{a_1\dots a_p} \hspace{0.5pt}\fa {\mathcal F}^{a_1\dots a_p}\; .
}
As was explained in detail, these expressions are bi-invariant scalars and combined with the corresponding measure in \eqref{iiaactions}, we get a bi-invariant action for the R-R field strengths. It remains to reformulate the Chern-Simons part of the R-R sector of type IIA supergravity which is given by:
\eq{
S^{\rm CS}_{\rm IIA}&=\frac1{4\hspace{1pt}\kappa_{10}^2}\int H\wedge F_4\wedge
C_3\\
&=  \frac1{4\hspace{1pt}\kappa_{10}^2} \, \frac{1}{3! \;4! \;3!}
\int d^{10}x \; \epsilon^{a_1\ldots a_{10}}\,
H_{a_1 a_2 a_3}\hspace{1pt} F_{(4)a_4 a_5 a_6 a_7}\hspace{1pt} C_{(3)a_8 a_9 a_{10}} \;.
}
Here, we denote by $\epsilon^{a_1\dots a_{10}} =\,\pm 1$ the epsilon-tensor density. Noting the fact that $\epsilon^{a_1\ldots a_{10}}/\sqrt{-|G|}$ transforms as a standard tensor field and is invariant under $B$-field gauge transformations, one can show in a straightforward way that $\epsilon_{b_1\ldots b_{10}}/\sqrt{-|\fa g|}$ is a $\beta$-tensor. Thus, we are able to write down the Chern-Simons part of the type IIA action in the symplectic frame:
\eq{
  \label{csaction}
\fa S^{\rm CS}_{\rm IIA}
=  \frac1{4\hspace{1pt}\kappa_{10}^2} \, \frac{1}{3! \;4! \;3!}
\int d^{10}x \hspace{1pt}\bigl| \beta^{-1}\bigr| \,
   \epsilon_{b_1\ldots b_{10}}\,
    \Theta^{b_1 b_2 b_3}\, \fa F_{(4)}^{b_4 b_5 b_6 b_7}\,
    \fa C_{(3)}^{b_8 b_9 b_{10}} \;.
}
Let us conclude our exposition of the R-R part by the following remark. Having already seen the invariance of the field strengths $\fa F$ under gauge transformations of the multi-vectors $\fa C$, by using also the Bianchi identity of $\Theta$, rewritten in the form
\eq{
	\fa\nabla^{[\underline{a}}\Theta^{\underline{bcd}]} = 0 \;,
}
it is possible to show the invariance of the complete transformed type IIA action under the following set of gauge transformations:
\eq{
    &\delta_{\Lambda_{(0)}} \fa C^{ a} = \fa\nabla^{a} \Lambda_{(0)} \;,
    \hspace{90pt}
    \delta_{\Lambda_{(2)}} \fa C^{ a_1a_2 a_3} = \fa\nabla^{[\underline{a_1}} \Lambda_{(2)}^{\underline{a_2  a_3}]} \;, \\
    &\delta_{\Lambda_{(0)}} \fa C^{a_1 a_2 a_3} = - \Lambda_{(0)} \hspace{1pt} \Theta^{a_1 a_2 a_3} \;.
}

\paragraph{The NS-R and R-NS sectors}
As was mentioned in the introduction to this section, the spin connection uses the tangent bundle and therefore should be transformed carefully to the Lie algebroid setting. Due to its structure, involving Christoffel symbols and derivatives of vielbein matrices, it is not clear a priori, if the transformation works. But it turns out, due to the inhomogeneous transformation behavior of the Christoffel symbols \eqref{transgamma}, that we get a spin connection in the $(\fa g, \beta, \phi)$-frame having the same structure as the standard one with partial derivatives replaced by the $D^a$-derivative and standard Christoffel symbols replaced by the ones in the symplectic frame.

Let us  first establish our notation and  state that
\begin{equation}
  \nonumber
  \begin{array}{l@{\hspace{30pt}}l}
   \alpha, \beta, \gamma,\ldots      &\mbox{denote Lorentz-frame indices,} \\
   a,b,c,\ldots  &\mbox{denote space-time indices.}
   \end{array}
\end{equation}
The vielbein matrices $e_{\alpha}{}^a$ relating these two frames via $e_a = e_a{}^\alpha\,e_\alpha$ and $e^a = e_\alpha{}^a\,e^\alpha$ are  defined in the usual way by requiring that:
\eq{
  \label{vielbein_04}
  e_\alpha{}^a \, e_\beta{}^b \hspace{1pt} G_{ab} = \eta_{\alpha\beta} \;,
}
with $\eta_{\alpha\beta} = \mbox{diag}\, (-1,+1,\ldots, +1)$. In these conventions the spin connection can be expressed in terms of the Christoffel symbols $\Gamma^c{}_{ab}$ as follows
\eq{
  \label{spinchris}
  \omega_c{}^{\alpha}{}_{\beta}=e_a{}^{\alpha}{}\, e_{\beta}{}^b\, \Gamma^a{}_{cb}
  + e_a{}^{\alpha}{}\,  \partial_c   e_{\beta}{}^a\,.
}
Introducing spacetime-dependent gamma matrices by contracting with the vielbein $\gamma^a = \, \gamma^{\alpha}e_{\alpha}{}^a$ and using standard notation $\gamma^{a_1\dots a_p} = \gamma^{[\underline{a_1}}\gamma^{\underline{a_2}}\cdots \gamma^{\underline{a_p}]}$, we write the kinetic term for the dilatino $\lambda$ as:
\eq{
  \label{action_dilatini}
    {\cal L}^\lambda_{\rm IIA} =   \overline{\lambda}\hspace{1pt} \gamma^{a}\left(
            \partial_a - {\textstyle {i\over 4}} \hspace{1pt}
    \omega_{a\,\alpha\beta}\hspace{1pt} \gamma^{\alpha\beta}\right) \lambda \;.
}
Let us now transform this expression into the symplectic frame. By applying the anchor map to a basis in the Lorentz frame, we get:
\eq{
e^{\alpha} =\,e^{\alpha}{}_a \,dx^a \rightarrow e^{\alpha}{}_a\beta^{ab}\,\partial_b =:\eta^{\alpha \gamma} \fa e^m{}_\gamma \,\partial_m \;.
}
Note, that this is compatible with the expected definition of the vielbein matrices in the symplectic frame:
\eq{
  \fa e^{\alpha}{}_a \, \fa e^{\beta}{}_b \hspace{1pt} \fa g^{ab} = \eta^{\alpha\beta} \;.
}
We thus take the following definition for the new vielbein matrices in terms of the standard ones:\eq{\label{transvielbein}
     \fa e^{\alpha}{}_a=\eta^{\alpha\beta}\, e_{\beta}{}^b\, \beta_{ba}  \, .
}
To finally state the result for the dilatino kinetic term and the corresponding form of the new spin connection let us recall that spinors are not affected by transformation to a Lie algebroid. Thus, $\fa \lambda = \lambda$ and let us in addition define the new gamma matrices by:
\eq{
\fa \gamma_a =\, \gamma_\alpha \fa e^\alpha{}_a =\, \gamma^\beta e_\beta{}^b\beta_{ba} \;.
}
Due to its importance let us state the result in a separate proposition and give a detailed proof of the transformation.
\begin{prop}
 With the definitions given above, the kinetic term for the dilatino in the symplectic frame is given by:
\eq{
    \fa{\cal L}^\lambda_{\rm IIA} =   \overline{\fa\lambda}\hspace{1pt} \fa\gamma_{a}\left(
            D^a - {\textstyle {i\over 4}} \hspace{1pt}
    \fa\omega^a{}_{\beta \delta} \fa \gamma^{\beta\delta}\right) \fa\lambda \;,
}
where the transformed spin connection takes the form:
\eq{
    \fa\omega^a{}_{\alpha}{}^{\beta}=\fa e^b{}_{\alpha}\, \fa e^{\beta}{}_c\, \fa\Gamma_b{}^{ac}
   + \hat e^b{}_{\alpha}  D^a \fa e^{\beta}{}_b\, .
}
\end{prop}
\begin{proof}
Writing out the definition and employing the relation \eqref{transgamma} together with the transformation behavior \eqref{transvielbein} gives
\eq{
{\cal L}_{IIA} =&\,\bar{\lambda}\gamma^a\left[\partial_a - \tfrac{i}{4}\left(e_k{}^{\alpha'}e_{\beta}{}^b \Gamma^k{}_{ab} + e_k{}^{\alpha'}\partial_a e_\beta{}^k\right)\eta_{\alpha \alpha'}\gamma^{\alpha \beta}\right] \lambda \\
=&\,\bar{\lambda}\fa\gamma_{a'}\beta^{a'a}\Bigl[\partial_a - \tfrac{i}{4}(e_k{}^{\alpha'} e_\beta{}^b(-\beta^{kk'}\beta_{aa''}\beta_{bb'}\fa \Gamma_{k'}{}^{a''b'} - \beta_{b'b} \partial_a \beta^{kb'})  \\
&\,+ \beta_{kk'}\eta^{\alpha'\alpha''} \fa e^{k'}{}_{\alpha''}\,\partial_a(\beta^{k''k}\eta_{\beta\beta'} \fa e^{\beta'}{} _{k''}))\eta_{\alpha \alpha'}\gamma^{\alpha\beta}\Bigr]\lambda \\
=&\, \bar{\fa \lambda}\fa\gamma_{a'} \left[D^{a'} - \tfrac{i}{4}(\fa e^{k'}{}_\alpha \fa e^{\beta'}{}_b \fa \Gamma_{k'}{}^{a'b'} + \fa e^{k'}{}_\alpha \partial_a \fa e^{\beta'}{}_{k'})\eta^{\alpha \alpha'} \gamma_{\alpha'\beta'}\right]\fa \lambda\;,
}
where we also used the transformation behavior of the inverse vielbein in the symplectic frame $\fa e^a{}_\alpha =\, \eta_{\alpha\beta}\beta^{ab}e_b{}^\beta$, which gives again the right sign in the second term.
\end{proof}

The remaining part in this sector is the kinetic term for the gravitino $\Psi_a$, which also has a spacetime-vector index. It is given by the Rarita-Schwinger Lagrangian:
\eq{ \label{RaritaSchwinger}
      {\cal L}^{\Psi}_{\rm IIA} =   \overline{\Psi}\vphantom{\fa\Psi}_a\hspace{1pt}
      \gamma^{abc} \left(
            \nabla_b -{\textstyle {i\over 4}} \hspace{1pt}
       \omega_{b\, \alpha\beta}\, \gamma^{\alpha\beta}\right) \Psi_c\, .
}
In this form, the connection coefficients of the covariant derivative drop out due to the totally anti-symmetrized gamma matrices. Note that this is a special property of the Levi-Civita connection of Riemannian geometry and is different e.g. in theories with torsion. In the case of a general Lie algebroid, even the Levi-Civita connection can have an antisymmetric part and therefore we keep the general form \eqref{RaritaSchwinger}. To transform the above expression into the symplectic frame, we first note the transformation of the vector index of the gravitino:
\eq{
\fa \Psi^a =  \beta^{ab} \hspace{1pt}  \Psi_b \; .
}
Performing an analogous transformation as done for the dilatino kinetic term, we observe that this factor of $\beta^{ab}$ is crucial: It is needed to transform the covariant derivative in the right way to the symplectic frame. The result is then given by:
\eq{ \label{symplecticRaritaSchwinger}
      \fa{\cal L}^{\Psi}_{\rm IIA} =   \overline{\fa\Psi}\vphantom{\fa\Psi}^a\hspace{1pt}
      \fa\gamma_{abc} \left(
            \fa \nabla^b -{\textstyle {i\over 4}} \hspace{1pt}
      \fa\omega^b{}_{\alpha\beta}\, \gamma^{\alpha\beta}\right) \fa\Psi^c \;.
}
Note, that now the antisymmetric part of the symplectic Christoffel symbols enter in the new kinetic term: $\fa\Gamma_c{}^{[\underline{ab}]}={1\over 2}{\cal Q}_c{}^{ab}$.

This completes our exposition of the fermionic terms in type IIA supergravity, since all the terms can be transformed to the symplectic frame in the ways des\-cribed above. We therefore constructed a complete type IIA supergravity theory in the symplectic frame. A natural question to ask is, whether this theory can be obtained from first principles, i.e. by supersymmetrizing $\beta$-diffeomorphims and then writing down an invariant theory. We will comment about this question in section \ref{ch-future} on future directions.

\subsubsection{First solutions in the symplectic frame}
Having established a complete differential geometry framework on the previous Lie algebroid \mbox{$(T^*M,[\cdot,\cdot]^H \ks,\beta^{\sharp})$} together with the corresponding supergravity Lagrangians, it is natural to ask for solutions to the equations of motion given in section \ref{subsec-eom} and their relation to solutions in the standard geometric frame and to frames explored in double field theory \cite{Andriot:2012an}. Whereas the standard frame and the above Lie algebroid allow for differential geometric descriptions, the frame used in \cite{Andriot:2012an} is more appropriate to describe the prime examples of non-geometric fluxes. Finding the appropriate fluxes which are well described in our Lie algebroid setting will be one of the tasks of this section. On the other hand giving an appropriate Lie algebroid setting for the frame discovered in \cite{Andriot:2012an} will be sketched briefly in section \ref{ch-future}.

\paragraph{Calabi-Yau manifolds in the symplectic frame}
In this and the next section we are going to discuss first attempts to find solutions to the equations of motion given in section \ref{subsec-eom}. This can only be the start of the enterprise to find and classify solutions which could include the analogue of fundamental string- and brane solutions. Whereas the latter go beyond the scope of this work, we are able to give at least one class of solutions in the symplectic frame, which are directly related to the standard geometric frame. These are the well-known Calabi-Yau geometries, which are given by the following conditions:
\eq{
\arraycolsep2pt
  \begin{array}{lclcrcl}
      R_{ab}&=&0\;, &\hspace{40pt} &   d\omega&=&0 \;, \\[5pt]
      H_{abc}&=&0\;, &&\phi&=&{\rm const.} \;,
    \end{array}
}
where $\omega$ denotes the K\"ahler form $\omega={i\over 2} \hspace{1pt}G_{a\overline b}\, dz^a\wedge
d\overline{z}^{\overline b}$ and we are using complex coordinates.
Looking at the expression for the Christoffel symbols \eqref{Christoffelsym} we observe that the additional ${\cal Q}$-flux terms vanish for a constant bi-vector and the resulting formula has the same pattern as in standard differential geometry. In addition, the $\Theta$-flux also vanishes, i.e. $d_{\beta}^H \beta =\,d_\beta \beta =\, 0$. Thus we make the ansatz:
\eq{
    \beta=\left(\begin{matrix}
       0 & +1 & 0 & 0 & 0 & 0\\
     - 1 & 0 & 0 & 0& 0 & 0 \\
     0 & 0 & 0 & +1 & 0 & 0\\
     0 & 0 &  -1& 0 & 0 & 0\\
   0 & 0 & 0 & 0 & 0 & +1\\
     0 & 0 &  0& 0 & -1 & 0
\end{matrix}\right) \;.
}
Using the field redefinition \eqref{MPIfield}, we obtain a smooth solution to the
string equations of motion in the non-geometric frame
 characterized by:
\eq{
      \fa R_{ab}=0  \;,\hspace{40pt}
      \Theta_{abc}=0\;,\hspace{40pt}
       \phi={\rm const.}
}
Applying the anchor map to the K\"ahler form $\omega$ we observe that we get as components the corresponding metric in the symplectic frame, i.e.
\eq{
W = \frac{i}{2}\,\fa g^{a\bar{b}}\,\partial_{z_a}\wedge\partial_{\bar{z}_b}\;.
}
Since we assumed the anchor to be bijective, the non-degeneracy of $W$ is inherited from the original K\"ahler form $\omega$. Finally, we want to find the condition on $W$ corresponding to the closedness of the K\"ahler form. As detailed in section \ref{subsec-geoHR}, we identified $(T^*M,[\cdot,\cdot]^H \ks,\beta^{\sharp})$ to be a proper Lie algebroid, meaning that we have the following properties:
\begin{itemize}
\item The anchor $\beta : T^*M \rightarrow TM$ is an algebra homomorphism, i.e.
\eq{
\beta^{\sharp}([\xi,\eta]^H \ks) =\, \left[\beta^{\sharp}(\xi),\beta^{\sharp}(\eta)\right]_L
}
\item Using proposition \ref{d-E} of section \ref{ch-math} we can define the corresponding differential $d^H _\beta$ which is by construction nilpotent:
\eq{
(d^H _\beta)^2 =\, 0 \;.
}
\item There is the corresponding Poisson-cohomology theory for multi- vectorfields, as was described generally in \eqref{Poissonkohomologie}.
\end{itemize}
By observing that the anchor is a Lie algebroid homomorphism from  $(T^*M,[\cdot,\cdot]^H \ks,\beta^{\sharp})$ to the trivial Lie algebroid $(TM,[\cdot,\cdot]_L, \textrm{id})$ we can use proposition \ref{Algebroidhomodiff} of section \ref{ch-math} to conclude the following equivalence from the relation between $\omega$ and $W$:
\eq{
d\omega =\,0 \quad \Longleftrightarrow \quad d^H _\beta W =\,0 \;.
}
Thus we identified the latter property of $W$ to be the analogue of the necessary Calabi-Yau condition $d\omega =\, 0$, and we see that transforming standard Calabi-Yau solutions lead to manifolds bearing similar properties, where the old conditions get replaced by the corresponding objects in the new Lie algebroid frame. Note, that this construction is valid in a general Lie algebroid. Thus, as long as such a framework is identified, we can always transport the Calabi-Yau condition to it. Because of the similarity of the new conditions to the old ones, we finally want to give the following definition:
\begin{definition}
A \emph{co-Calabi-Yau manifold} is a complex manifold together with a non-degenerate $d_\beta^H$-closed two-vector field $W$ and  associated Ricci-flat ($\fa R_{ab} =\, 0$) Hermitean metric $\fa g$. The corresponding K\"ahler form is given by
\eq{
	W = \frac{i}{2}\,\fa g^{a\bar{b}}\,\partial_{z_a}\wedge\partial_{\bar{z}_b}\,.
}
\end{definition}

\subsubsection{An approximate solution with constant $\Theta$-flux}
We now want to take the example of the constant $H$-flux background described in section \ref{subsec-t3h} and construct in a similar way a solution to the equations of motion with constant $\Theta$-flux. Although this is only an approximate solution, it contains important strategies to think about finding solutions.

As a concrete example, let us choose a flat four-dimensional metric $\fa g^{ab} = \delta^{ab}$, together with a constant dilaton and a specific ansatz for the bi-vector. We assume the latter to be invertible and therefore we are bound to even-dimensional spacetime. Note that even for a flat metric the Ricci tensor in the symplectic frame does not vanish for general ${\cal Q}$-flux. Intuitively, the presence of a non-trivial $\beta$-field influences the geometry (e.g. producing non-vanishing curvature) even in the case of a flat metric. Let $\epsilon$ and $\theta$ be constant parameters, then for the bi-vector in a four dimensional flat space we choose the ansatz:
\eq{ \label{solution+}
    \beta=\left(\begin{matrix}
       0 & +\epsilon^{-1}\hspace{1pt} (1+x_4) & 0 & 0\\
     - \epsilon^{-1}\hspace{1pt} (1+x_4) & 0 & 0 & 0\\
     0 & 0 & 0 &+ \epsilon\hspace{1pt}\theta\\
     0 & 0 &  -\epsilon\hspace{1pt} \theta & 0\end{matrix}\right)
     \hspace{30pt}
      {\rm for}\ x_4>0 \;,
}
and
\eq{ \label{solution-}
   \beta=\left(\begin{matrix}
       0 &  +\epsilon^{-1}\hspace{1pt} (1-x_4) & 0 & 0\\
     - \epsilon^{-1}\hspace{1pt} (1-x_4) & 0 & 0 & 0\\
    0 & 0 & 0 & -\epsilon\hspace{1pt}\theta\\
    0 & 0 &  +\epsilon\hspace{1pt} \theta & 0\end{matrix}\right)
    \hspace{30pt}
    {\rm for}\ x_4<0 \;.
}

\noindent Because the ${\cal Q}$-flux contains inverse $\beta$-fields, we want to avoid zeros of the latter and thus we are forced to take the above two patches. With this ansatz, the $\Theta$-flux is constant whereas the ${\cal Q}$-flux is a continuous function of $x_4$. In components the result is:
\begin{gather}
 \Theta^{123} =\, \theta\;, \\
    {\cal Q}_1{}^{31}=-{\cal Q}_1{}^{13}={\cal Q}_2{}^{32}=-{\cal
       Q}_2{}^{23}=  {\theta\hspace{1pt}\epsilon\over  1+|x_4|}\, ,
\end{gather}
with all other components vanishing. The corresponding non-vanishing Christoffel-symbols are then given by:
\eq{
     \fa \Gamma_3{}^{11}
   =\fa \Gamma_3{}^{22}
   =-\fa \Gamma_1{}^{13}
   =-\fa \Gamma_2{}^{23}
   = {\theta\hspace{1pt}\epsilon\over  1+|x_4|}\, ,
}
and therefore the non-vanishing part of the Ricci tensor in the symplectic frame is given by:
\eq{
      \fa R^{11}= \fa R^{22}= \frac{3}{4}\, \fa R^{33}=-
     3 \, {(\theta\hspace{1pt}\epsilon)^2\over  (1+|x_4|)^2} \;.
}
In the following we assume weak $\Theta$-flux and thus only linear terms in the flux contribute to the equations of motion. In this limit, the first and last equation in the set \eqref{eom_final} are satisfied. For the second, we have to take the limit $\epsilon \rightarrow 0$ in which (independent of the coordinate $x_4$) the ${\cal Q}$-flux and therefore the Ricci tensor $\fa R^{ab}$ vanish with the $\Theta$-flux staying constant.

\noindent Let us close this section with two remarks:
\begin{itemize}
\item In the above limit, the ${\cal Q}$-flux is well defined, but the $Q$-flux component
\eq{
Q_4{}^{12} =\, \partial_4 \beta^{12}
}
gets infinite. This shows that the appropriate fluxes in the symplectic frame are the $\Theta$- and ${\cal Q}$-flux.
\item Using the field redefinition \eqref{MPIfield}, we can transform the above ansatz for a solution from the $(\fa g, \beta, \phi)$-frame into the original $(G,B,\phi)$-frame. The non-vanishing components of the metric and $B$- field are given by:
\eq{
B_{12} &=\, -\frac{\epsilon}{1+x_4}\;, \qquad
\hspace{40pt} B_{34} =\,-\frac{1}{\epsilon \hspace{1pt}\theta}  \;, \\
G_{11} = G_{22} &=  \frac{\epsilon^2}{(1+x_4)^2} \;, \hspace{27pt} 
G_{33}=G_{44} = \frac{1}{(\epsilon \hspace{1pt}\theta)^2} \;.
}
Thus, in the limit $\epsilon \rightarrow 0$, the original $H$-flux is well-defined, whereas the original metric is ill-defined: Even locally, it is not invertible and therefore not a proper metric. We interpret this as a hint that there exist well-defined solutions in the symplectic frame which cannot be described by the original Riemannian geometry setup.
\end{itemize}

\subsubsection{Summary}
In the last section we employed the differential geometry of Lie algebroids, a well known construction in the mathematics literature, to reformulate the standard bosonic low energy effective string action in a frame containing a metric $\fa g^{ab}$ on the cotangent-bundle, a (quasi-)symplectic bi-vector $\beta^{ab}$ and a dilaton $\phi$ as dynamical fields. By transforming the gauge-transformations of the original $B$-field to this frame, we were led to a new symmetry, called $\beta$-diffeomorphisms. Using a specific Lie algebroid, we were able to construct a covariant derivative, torsion and curvature being tensorial under both, standard and $\beta$-diffeomorphisms. As a consequence it was possible to write down a Lagrangian containing the new Ricci-scalar $\fa R$, dilaton and the field strength $\Theta$. The resulting action, which we called symplectic gravity due to the non-trivial dependence on the bi-vector $\beta$ is given by:
\eq{
      \hat S={1\over 2\kappa^2} \int d^nx\, \sqrt{-|\hat g|}\,\bigl|  \beta^{-1} \bigr|\, e^{-2\phi}
     \Bigl(   \hat R -{1\over 12} \Theta^{abc}\, \Theta_{abc}
              +4\hspace{1pt} \hat g_{ab}\, D^a\phi D^b \phi\Bigr)\, .
}
This action can be related to the original bosonic string action via a field redefinition involving the anchor map of the Lie algebroid. Due to its simple structure, a generalization to higher order corrections in $\alpha'$ and to type II supergravity was possible. Finally we constructed simple solutions to the resulting equations of motion in the bosonic case.


\section{Outlook and future directions}
\label{ch-future}

\subsection{Remarks on non-geometric frames and supersymmetry}
In this section we comment on two further directions which are interesting applications of the results presented in the last section. We will be very brief and refer the reader for the first part to the corresponding publications \cite{Blumenhagen:2013aia, Felixphd}. The aim of the second part only is to give ideas and therefore will not be as precise as the sections before.

\subsubsection{Relation to non-geometric frames}
In the last section, we constructed a differential geometry framework which allowed us to write down an action determining the dynamics of the field variables $(\hat g^{ij}, \beta^{ij}, \phi)$. These fields, and as a consequence the whole symplectic gravity action, are related to the standard bosonic low energy effective supergravity action for the standard metric $G_{ij}$ and $B$-field $B_{ij}$ via the following field redefinition:
\eq{
\hat g ^{ij} =\, -\beta^{im}G_{mn}\beta^{nj}\;, \qquad \beta^{ij} =\, \left(B^{-1}\right)^{ij} \;.
}
We want to argue, that this is a way to describe the standard $(G,B,\phi)$ background in a special \emph{frame}, and that in general the geometry of different frames can be described by Lie algebroids with different anchors. As this is an outlook, we only describe the idea. The detailed connection of geometric/non-geometric frames and Lie algebroids is given in the original paper \cite{Blumenhagen:2013aia} and will be contained in a forthcoming PhD-thesis \cite{Felixphd}.

\paragraph{Field redefinitions and $O(d,d)$-transformations}
The basic object of generalized geometry is the so-called generalized tangent bundle $E$. For a $d$-dimensional manifold $M$, it is given by the direct sum of the tangent- and cotangent bundle: $E=TM \oplus T^*M$. This direct sum is equipped with the following natural bilinear pairing: For sections $X,Y \in \Gamma(TM)$ and $\xi, \eta \in \Gamma(T^*M)$ we have:
\eq{\label{bilinearform}
\langle X + \xi, Y + \eta \rangle = \xi(Y) + \eta(X) \;.
}
In coordinates, the bilinear form is represented by the $2d \times 2d$-matrix $\eta$ which is explicitly given by:
\eq{
\begin{pmatrix}
0 & \mathbf{1} \\
\mathbf{1} & 0
\end{pmatrix} \:,
}
where $\mathbf{1}$ is the $d$-dimensional identity matrix. The transformations which leave the bilinear form \eqref{bilinearform} invariant are given by $2d \times 2d$-matrices ${\cal M}$, which have the property:
\eq{\label{oddcond}
{\cal M}^t \,\eta {\cal M} =\, \eta \;.
}
This is the definition of the group $O(d,d)$. Writing the matrices ${\cal M}$ in terms of $d\times d$- blocks ${\cal M} = \begin{pmatrix}
a & b \\
c& d
\end{pmatrix}$, the condition \eqref{oddcond} can be rewritten in the form:
\eq{
a^t c + c^ta &= \, 0\;, \\
b^t d + d^t b &=\, 0\;, \\
b^t c + d^t a &=\, \mathbf{1} \;.
}
Having the standard metric $G$ on $M$ and B-field $B$ as basic fields, it is common in generalized geometry to combine them into the generalized metric ${\cal H}$ on the bundle $E$ in the following way:
\eq{\label{genermetric}
{\cal H}(G,B) =\, \begin{pmatrix}
G - BG^{-1}B & BG^{-1} \\
-G^{-1}B & G^{-1}
\end{pmatrix}\;.
}
The group of $O(d,d)$-transformations acts by conjugation on the generalized metric:
\eq{
\widehat{\cal H}(G,B) =\, {\cal M}^t {\cal H}(G,B) {\cal M}\;.
}
To write the new metric $\widehat{\cal H}$ again into the standard form \eqref{genermetric}, we have to \emph{redefine} the metric and $B$-field:
\eq{\label{framedefinition}
{\cal H}(G,B)\; \xrightarrow{{\cal M}^t {\cal H} {\cal M}}\; \widehat{\cal H}(G,B) \;\xrightarrow{\widehat{G}(G,B),\widehat{B}(G,B)}\; {\cal H}(\widehat{G},\widehat{B})
}
This determines new field variables $(\widehat{G},\widehat{B},\phi)$, which describe the old background $(G,B,\phi)$ in a new \emph{frame}. Note that we did not change the dilaton up to now, because the generalized metric \eqref{genermetric} only contains the metric and $B$-field\footnote{To compare this to the action of T-duality, also the dilaton has to transform, as we observed in the Buscher rules. The transformation of the dilaton manifests itself in the choice of the appropriate integration measures if we write down actions in different frames.}. It turns out that for every such frame, it is possible to construct a Lie algebroid and its corresponding differential geometry which in turn enables us to write down the low energy effective string action in these frames in a manifest invariant way.

\paragraph{Field redefinitions and the anchor}
As stated in the last subsection, the choice of a frame can be geometrized by selecting an appropriate Lie algebroid. Let us sketch briefly how one can determine the anchor map in terms of the field redefinition \eqref{framedefinition}. The construction of the complete Lie algebroid corresponding to a chosen frame together with the differential geometry and action functional is detailed in \cite{Blumenhagen:2013aia}.

Comparing the transformed generalized metric $\widehat{\cal H}(G,B)$ to the standard one in terms of the new field variables ${\cal H}(\widehat{G},\widehat{B})$, it is possible to infer the corresponding field redefinition in terms of the matrix entries $a,b,c,d$ of the $O(d,d)$-matrix ${\cal M}$:
\eq{\label{frame12}
\widehat{G}(G,B) =&\,\gamma^{-1}(G,B)\,G\,\left(\gamma^{-1}\right)^t(G,B)\;, \\
\widehat{B}(G,B) =&\, \gamma^{-1}(G,B)\left[\gamma(G,B) \delta^t(G,B) - G \right]\left(\gamma^{-1}(G,B)\right)^t \;,
}
where we introduced the parameters $\gamma(G,B), \delta(G,B)$, which contain the matrix blocks $a,b,c,d$ in the following combination:
\eq{
\gamma(G,B) =\, d + (G-B)b \;, \qquad \delta(G,B) =\, c + (G-B)a \;.
}
Having a general Lie algebroid $(E, [\cdot,\cdot]_E,\rho)$ as described in the mathematical introduction of section \ref{ch-math}, it is possible to define a metric $\widehat{G} \in \Gamma (E^* \otimes_{\textrm{sym}} E^*)$ on it via \eqref{metric}. The anchor map $\rho$ enables us to relate it to the standard metric $G$ on the base manifold by the dual anchor:
\eq{
G =\, \left(\otimes^2 \rho^*\right) (\widehat{G})\;, \qquad \rho^* =\, \left(\rho^t\right)^{-1} : \Gamma(E^*) \rightarrow \Gamma(T^*M) \;.
}
The dual anchor is represented in coordinates by a $d\times d$-matrix, which we denote by\footnote{For distinction, we denote coordinate indices of sections in $E^*$ by Greek letters and sections in $T^*M$ by Latin letters, i.e. $s=s_\alpha e^\alpha \in \Gamma(E^*)$ and $\xi = \xi_k\,dx^k \in \Gamma(T^*M)$.} $(\rho^*)_m{}^\alpha$. In this notation the relation between the metric $\widehat{G}_{\alpha \beta}$ on the Lie algebroid and the standard metric $G_{mn}$ is given by:
\eq{\label{coordinateanchor}
G_{mn} =\, \left(\rho^*\right)_m{}^\alpha \left(\rho^*\right)_n{}^\beta \,\widehat{G}_{\alpha \beta} \;.
 }
If we compare this general relation to the redefined metric $\widehat{G}$ in \eqref{frame12}, we can read off the anchor corresponding to the field redefinition: $\rho = \left(\gamma^{-1}\right)^t$. Thus to every field redefinition, there exists a corresponding anchor and as a consequence it is possible to construct a Lie algebroid, which is proven in \cite{Blumenhagen:2013aia}.

\paragraph{Two examples}
To conclude this section, we want to mention two major examples of $O(d,d)$-induced field redefinitions and give the corresponding anchors. The first one is used in \cite{Andriot:2012an} for the double field theory description of non-geometric backgrounds. The field redefinition is given explicitly by:
\eq{\label{LMUfield12}
\widehat{G}(G,B) &=\, \hspace{10pt}\left(\mathbf{1} + BG^{-1}\right)\,G\,\left(\mathbf{1} -G^{-1}B \right)\;, \\
\widehat{B}(G,B) &=\, -\left(\mathbf{1} + BG^{-1} \right)\,B\,\left(\mathbf{1} -G^{-1}B\right)\;.
}
It turned out in \cite{Andriot:2012an} and becomes also clear from the preceding sections that it is convenient to describe the geometry in terms of a bi-vector which we denote in this special example also by a hat: $\hat{\beta} \in \Gamma(\wedge^2 TM)$. With the definition of the bi-vector given in \eqref{basicfieldredef2}, motivated by Poisson sigma models, we have:
\eq{
\hat{\beta} =\, -\widehat{G}^{-1}\widehat{B} \widehat{G}^{-1} \;.
}
We can now express the old variables $(G,B)$ in terms of the hatted ones conveniently in terms of the bi-vector:
\eq{
G(\widehat{G},\widehat{B}) &=\, \hspace{10pt}\left(\widehat{G}^{-1} + \hat{\beta}\right)^{-1} \widehat{G}^{-1} \left( \widehat{G}^{-1} - \hat{\beta} \right)^{-1}\;, \\
B(\widehat{G},\widehat{B}) &=\, \left(\widehat{G}^{-1} + \hat{\beta}\right)^{-1} \,\hat{\beta}\, \left( \widehat{G}^{-1} - \hat{\beta} \right)^{-1} \;.
}
In this form it is now easy to determine the anchor by comparing with the expression for $G$ with the corresponding one in \eqref{frame12}. It is shown in \cite{Blumenhagen:2013aia} that also the transformation of the $B$-field is consistent with the corresponding formula in \eqref{frame12}. The anchor is finally determined by the relation \eqref{coordinateanchor} to be:
\eq{
\rho^t =\, \widehat{G} + \hat{\beta} \;.
}
As a second example, let us see how the frame described by the Lie algebroid of the previous section, for which we explicitly formulated a symplectic gravity action can be obtained from the preceding example. In the relation \eqref{LMUfield12} we want to perform formally a Seiberg-Witten limit as described in section \ref{ch-quant} in the sense that we scale the metric and $B$-field as given in \eqref{scaling}. Note that we did not introduce factors of $\alpha'$ in this section, so the scaling in our case will be:
\eq{
G_{mn} \propto \epsilon\;, \qquad B_{mn} \propto \sqrt{\epsilon} \;,
}
and we only keep the leading order terms in quantities depending on $G$ and $B$ in the limit $\epsilon \rightarrow 0$. The results for the redefined fields $\widehat{G}$ and $\widehat{B}$ are:
\eq{
\widehat{G} =\, -BG^{-1}B\;, \quad \widehat{B} =\, BG^{-1}BG^{-1}B \; \Rightarrow \hat{\beta} =\, B^{-1}\;.
}
We thus recover the field redefinitions \eqref{MPIfield}, relating the standard $(G,B)$-frame to the $(\widehat{G},\hat{\beta})= (\fa g, \beta)$-frame of the previous section. Consequently the anchor is given by the bi-vector $\hat{\beta}$, because this dominates $\widehat{G}^{-1}$ in the limit $\epsilon \rightarrow 0$. We conclude that the Lie algebroid geometry described in the previous section can be seen as a Seiberg-Witten type limit of the preceding example. The corresponding field-dependent $O(d,d)$-transformation together with further details is computed in the original paper \cite{Blumenhagen:2013aia}.

\subsubsection{Speculations about supersymmetry}
As an application of the formalism of Lie algebroids, we showed in section \ref{subsec-super} how to transform type IIA supergravity into the frame with basic variables $(\hat g^{ij}, \beta^{ij},\phi)$. It turned out, that this is easily possible by \emph{transforming} the vielbein and spin connection by the anchor map. As an example, we repeat the expressions derived for the spin connection:
\eq{
    \fa\omega^a{}_{\alpha}{}^{\beta}=\fa e^b{}_{\alpha}\, \fa e^{\beta}{}_c\, \fa\Gamma_b{}^{ac}
   + \hat e^b{}_{\alpha}  D^a \fa e^{\beta}{}_b\, .
}
With the help of this result it was possible to write down the Rarita-Schwinger Lagrangian for the gravitino and the kinetic term for the dilatino in the new frame. For example, the latter was given by:
\eq{ \label{susy1}
    \fa{\cal L}^\lambda_{\rm IIA} =   \overline{\fa\lambda}\hspace{1pt} \fa\gamma_{a}\left(
            D^a - {\textstyle {i\over 4}} \hspace{1pt}
    \fa\omega^a{}_{\beta \delta} \gamma^{\beta\delta}\right) \fa\lambda \;.
}
For more details of the construction we refer the reader to section \ref{subsec-super}. All of these results were constructed out of the standard type IIA expressions by applying the anchor map (which was given by the bi-vector $\beta$) in an appropriate way.

Because these constructions worked out without problems (note, that this is not obvious as shown in section \ref{subsec-super}), one might ask if there is a possibility to construct the transformed type IIA supergravity theory independent of the transformation by the anchor map. One possible way (which is still to be analyzed and therefore we want to stress that it is not clear if this is possible at all) would be to begin with the supersymmetry algebra itself. As an example, in four dimensions, denoting by\footnote{We use indices $\alpha$ for the two dimensional $(\tfrac{1}{2},0)$-representation of $SL(2,\mathbb{C})$ and indices $\dot \beta$ for its complex conjugate representation $(0,\tfrac{1}{2})$. For further details on this notation we refer the reader to \cite{Wess:1992cp}.} $Q_\alpha , \bar{Q}_{\dot \beta}$ the supersymmetry generators and their Hermitean conjugates (we suppress further indices labeling the supercharges), their anti-commutator gives the generator of spacetime-translations $P_\mu$. With the Pauli-matrices denoted by\footnote{We refer the reader to the appendix for their definition.} $\sigma^\mu{}_{\alpha \dot{\beta}}$, the anti-commutator is:
\eq{
[Q_\alpha , \bar{Q}_{\dot \beta} ]_+ =  \, \sigma^\mu{} _{\alpha \dot \beta} \,P_\mu \;.
}
We observe, that the spacetime-indices of $P_\mu$ can be related by the anchor map to the ``symplectic'' frame by $\hat P ^\mu = \beta^{\mu \nu} P_\nu$, whereas spin-$\tfrac{1}{2}$- indices remain the same. Denoting in addition the Pauli-matrices (which also do not change) for convenience by $\hat \sigma_{\mu}{} ^{\alpha \dot \beta}$, we are able to write down the anti-commutator of two supercharges in the symplectic frame:
\eq{\label{symplsusy}
[\hat Q^\alpha, \bar{\hat Q}^{\dot \beta} ]_+ = \, \hat \sigma_\mu{}^{\alpha \dot \beta} \, \hat P^\mu \;.}
It would be interesting to study the consequences of this algebra in representation theory and in the construction of supersymmetric gauge theories and supergravity. If this is possible, there are immediate questions in physics and mathematics: What is the connection to symplectic and Poisson geometry (because the Poisson tensor $\beta$ enters non-trivially already into the algebra \eqref{symplsusy})? Can we recover the transformed type IIA theory, e.g. \eqref{susy1}? Is there a counterpart of the type IIB theory or eleven dimensional supergravity? We observe that the transformation with the help of the anchor was only possible in even dimensions, but the algebra \eqref{symplsusy} is possible a priori in arbitrary dimensions. These are all open questions to be studied in the future.

\subsection{Nambu-structures and Lie 3-algebroids}
The generalization of Poisson brackets to multilinear maps with more than two entries was first introduced by Nambu in \cite{Nambu:1973qe} to describe phase spaces of odd dimensions and mechanical systems with multiple Hamiltonians. A \emph{Nambu-Poisson n-bracket} is given by a bracket $\{\cdot,\dots, \cdot\}$ with $n$ entries, which is characterized by the following two properties. For functions $f_1,\dots f_{n-1}, g, h$ we have a generalized Leibniz rule:
\eq{
\{f_1,\dots,f_{n-1}, g\, h\} =\,\{f_1,\dots,f_{n-1},g \}h + \{f_1,\dots,f_{n-1},h\}g \;.
}
In addition, for $f_1,\dots , f_{n-1},\, g_1, \dots g_{n} $ we have the so-called fundamental identity, which is the n-ary generalization of the Jacobi identity:
\eq{
\{f_1,\dots,f_{n-1},\{g_1,\dots g_n\}\} =\, \sum_{i=1} ^n \, \{g_1,\dots g_{i-1},\{f_1,\dots,f_{n-1},g_i\},g_{i+1},\dots,g_n \} \;.
}
The quantization of phase spaces equipped with these structures is not completely understood up to now \cite{Dito:1996xr, Curtright:2002fd}, but might have applications to the dynamics of M-theory branes. As an example, Basu and Harvey \cite{Basu:2004ed} suggested a generalization of Nahm's equations \cite{Nahm:1979yw} by analyzing M2-branes ending on a stack of M5-branes. The fluctuations $X^\mu$ of the M5-brane world-volume coordinates in one of the directions of the M2-brane can be described by:
\eq{
\frac{d}{ds}\,X^\mu = \, \epsilon^{\mu}{} _{\nu \lambda \rho} \left[X^\nu,X^\lambda,X^\rho\right] \;,
}
where the right hand side is given by a three-bracket $[\cdot,\cdot,\cdot]$ (we will give more details on n-brackets in the next sections). Finally, in \cite{Blumenhagen:2010hj}, a similar type of 3-algebra was realized in the context of WZW models. Further applications of 3-algebras in string- and M-theory can be found in \cite{Bagger:2012jb}. An interesting connection of non-geometric fluxes and M(atrix) theory was established in \cite{Chatzistavrakidis:2012qj}.


\subsubsection{3- Koszul brackets}
In the introductory remarks we observed that a generalization of Poisson structures and their quantization to the case of more than two entries is an important mathematical structure in string theory with non-vanishing background fluxes as well as M-theory. We therefore want to describe in more detail the mathematical properties of such three- or n-ary algebras and give the first steps to generalize the construction of section \ref{sec-biinvsymplgrav} to the case of an underlying 3-algebroid.

As we described in section \ref{ch-math}, the Koszul bracket is intimately connected to Poisson structures and thus it is a natural question if there is a generalization to a Koszul bracket with more than two arguments. It is well known that this is possible (e.g. \cite{0615.58029, 2012arXiv1207.3590B} and the first part of \cite{0993.53025}) and we review the case of three entries in more detail. First, let us introduce a unifying concept which allows to express the standard Koszul bracket in a form which is easily generalized to the case of three or more arguments.

Let $M$ be a Poisson manifold. We will use the generalization of the standard Lie derivative in the direction of a multi-vectorfield $X \in \Gamma(\wedge^k TM)$. This is defined by its action on multi-vectorfields $Y \in \Gamma(\wedge^{\bullet}TM)$ and forms $\alpha \in \Gamma(\wedge^{\bullet}T^*M)$ by using the Schouten-Nijenhuis bracket $[\cdot,\cdot]_{SN}$ in the following way
\eq{ \label{multilie}
L_X Y =&\, [X,Y]_{SN}\;, \\
L_X \alpha =&\, \left(\iota_X \circ d - (-1)^k d \circ \iota_X \right)\alpha \;,
}
where we use the insertion $\iota_X$ of a multi-vectorfield which can be traced back to the insertion of a vector field by:
\eq{
  \iota_{fX} =\, f\iota_X \;, \qquad \iota_{X\wedge Y} =\, \iota_X \circ \iota_Y \;,
}
and is simply zero if the degree of the multi-vectorfield $X$ is greater than the argument of the insertion. Furthermore, let $F,G$ be differential operators on $\Gamma(\wedge^{\bullet}T^*M)$ of degree $|F|$ and $|G|$ respectively, i.e.
\eq{
F\,&: \quad \Gamma(\wedge^k T^*M) \rightarrow \Gamma(\wedge^{k + |F|}T^*M)\;, \\
G\,&: \quad \Gamma(\wedge^k T^*M) \rightarrow \Gamma(\wedge^{k + |G|}T^*M)\;.
}
The graded commutator of two such operators is then defined by:
\eq{
[F,G] := F\circ G -(-1)^{|F||G|} G \circ F \;.
}
Let us consider examples of operators on $\Gamma(\wedge^{\bullet}T^*M)$ which are important in the following. Clearly, the exterior differential $d$ is an operator of degree one and the insertion of a multi-vectorfield $X \in \Gamma(\wedge^kTM)$ has degree $-k$. From its definition \eqref{multilie} it is clear that $L_X$ is an operator of degree $(1-k)$. Another operator which is important for the following is the wedge-multiplication by a $k$-form $\alpha$ defined by
\eq{
\mu_\alpha\;: \quad \Gamma(\wedge^m T^*M) &\rightarrow \Gamma(\wedge^{m+k}T^*M) \\
\mu_\alpha(\xi) &:=\, \alpha \wedge \xi \;.
}
It is clearly of the same degree as the form $\alpha$. We are now ready to prove the following form of the Koszul bracket which was also given by Koszul \cite{0615.58029}:

\begin{prop}
Let $M$ be a Poisson manifold with bi-vector $\beta = \tfrac{1}{2}\beta^{ij}\,\partial_i \wedge \partial_j$. Then the Koszul-bracket for one-forms $\xi, \eta$ has the following form
\eq{
\left[\xi , \eta \right]\ks = \, \left[\, [L_\beta, \mu_\xi], \mu_\eta \right](1) \;.
}
\end{prop}

\begin{proof}
Using the definition of the graded commutator, by a straight forward computation we get
\eq{
\left[\,[L_\beta, \mu_\xi ],\mu_\eta \right](1) = L_\beta (\xi \wedge \eta) - \eta \wedge L_\beta \xi + \xi \wedge L_\beta \eta \;.
}
Now we use the following relation for two-forms $\omega$ and one-forms $\alpha$:
\eq{
\iota_\beta(\omega \wedge \alpha) =\, (\iota_\beta \omega) \wedge \alpha - \iota_{\beta^\sharp(\alpha)} \omega \;.
}
With the help of this relation we get to the result:
\eq{
\left[\,[L_\beta, \xi ],\eta \right](1) =\, d\left(\beta(\xi,\eta)\right) - \iota_{\beta^\sharp(\eta)}d\xi + \iota_{\beta^\sharp(\xi)}d\eta \;,
}
which is just a rewritten form of the Koszul bracket.
\end{proof}

In the following, we are going to generalize the form of the bracket given in the preceding proposition. Let $M$ now be a manifold equipped with a 3-vectorfield $\gamma \in \Gamma(\wedge^3 TM)$. The latter defines a map $\gamma^\sharp$ relating two-forms with vectorfields:
\eq{
\gamma^\sharp\, : \quad &\Gamma(\wedge^2 T^*M) \rightarrow TM \\
&\gamma^\sharp(\xi \wedge \eta) =\, \iota_{\eta \wedge \xi} \gamma \;.
}
In addition we have the degree $(-2)$ operator $L_\gamma$ which is the defining operator for the following 3-Koszul bracket of one-forms $\eta,\xi,\omega$:
\eq{ \label{3-Koszulbracket}
\left[\eta,\xi,\omega\right] := \left[\,[\,[L_\gamma, \mu_\eta],\mu_\xi],\mu_\omega \right](1) \;.
}
Evaluating the last expression in the same way as done in the previous proposition for the bi-vector case, we get the following expression for the 3-Koszul bracket:
\eq{ \label{3-Koszulbracket2}
\left[\eta, \xi, \omega \right] = \, \iota_{\gamma^\sharp(\eta \wedge \xi)} d\omega +\iota_{\gamma^\sharp(\xi \wedge \omega)} d\eta + \iota_{\gamma^\sharp(\omega \wedge \eta)}d\xi + d\left(\gamma(\eta, \xi, \omega)\right)\;.
}
This means in local coordinates $\eta = \eta_i\,dx^i, \xi = \xi_i\,dx^i, \omega = \omega_i \,dx^i$:
\eq{
[\eta, \xi, \omega] =&\, \bigl( \gamma^{mnk}\omega_n (\xi_m \partial_k \eta_i - \eta_m \partial_k \xi_i) + \gamma^{mnk}\eta_m \xi_n \partial_k \omega_i \\
&\, + \eta_m \xi_n \omega_k Q_i{}^{mnk} \bigr)\, dx^i \;,
}
where we introduced the field $Q_i{}^{mnk} = \partial_i \gamma^{mnk}$ which appears similarly to the $Q$-flux in case of the standard Koszul-bracket.

Having defined a generalization of the Koszul-bracket to the case of three entries (similar constructions can be done with an arbitrary number of entries, if the manifold has the corresponding multi-vectorfield to define the appropriate Lie derivative), the next step is to ask if there is a structure on the cotangent bundle of the manifold which is similar to that of the Lie algebroid $(T^*M, [\cdot,\cdot]\ks, \beta^\sharp)$ given in detail in section \ref{ch-math}. This is known to be true and we will review the corresponding notion in the next subsection and extend it by proposing expressions for torsion and curvature operators.

\subsubsection{Differential geometry of Filippov 3-algebroids}
In section \ref{ch-math} we used the notion of Lie algebroids to generalize the structures of the tangent bundle (like vectorfields and Lie bracket) to arbitrary vector bundles. It turns out that we can do a similar generalization to vector bundles having a three-bracket instead of a two-bracket (and similarly for the general case of an $n$-bracket). The construction is very close to that of a Lie algebroid and was done by Filippov \cite{0594.17002}: With the help of an anchor map, which is a vector bundle homomorphism relating three- and two brackets, the action of a general section on functions can be defined. More precisely\footnote{For a review, see also the first part of \cite{0993.53025}.}:

\begin{definition} \label{Filippov3algebroid}
Let $F\rightarrow M$ be a vector bundle over a manifold $M$ together with an n-linear, skew symmetric bracket $[\cdot,\dots,\cdot] : \Gamma(F)\times \dots \times \Gamma(F) \rightarrow \Gamma(F)$
and a bundle morphism $a: \Gamma(\wedge^{n-1} F) \rightarrow TM$ (called \emph{anchor}). The triple $(F,[\cdot,\dots,\cdot],a)$ is called a \emph{Filippov n-algebroid} if the following homomorphism property, Leibniz rule and \emph{fundamental identity} hold:
\begin{itemize}
\item For sections $u_1, \dots, u_{n-1},v_1,\dots v_{n-1} \in \Gamma(F)$:
\begin{gather}\label{3homo}
\left[a(u_1 \wedge \dots \wedge u_{n-1}),a(v_1,\dots,v_{n-1})\right]_L \nonumber \\
 = \sum_{i=1} ^{n-1} a \left(v_1 \wedge \dots \wedge [u_1,\dots,u_{n-1},v_i]\wedge \dots \wedge v_{n-1} \right) \;.
\end{gather}
\item For sections $v_1,\dots,v_{n-1}, u \in \Gamma(F)$ and functions $f \in {\cal C}^{\infty}(M)$:
\begin{gather} \label{3Leibniz}
[v_1,\dots,v_{n-1}, fu] \nonumber \\
= f\,[v_1,\dots,v_{n-1},u] + a(v_1 \wedge \dots \wedge v_{n-1})(f) \, u \;.
\end{gather}
\item For sections $u_1,\dots, u_{n-1}, v_1, \dots v_n \in \Gamma(F)$:
\begin{gather}\label{fundamentalidentity}
\left[u_1,\dots,u_{n-1},[v_1,\dots,v_n]\,\right] \nonumber \\
= \sum_{i=1}^n \left[v_1,\dots,v_{i-1},[u_1,\dots,u_{n-1},v_i],v_{i+1},\dots,v_n\right] \;.
\end{gather}
\end{itemize}
\end{definition}

We easily see, that in the case of $n=2$, the structure reduces to that of an ordinary Lie algebroid as introduced in section \ref{ch-math}. In the following we are mainly focused on the case of a Filippov 3-algebroid (which is often also called Lie 3-algebroid). As was reviewed in detail in section \ref{ch-math}, all the standard differential geometry constructions like Lie- and covariant derivatives, torsion and curvature can be generalized to the Lie algebroid case. In the following we will give the fist steps of the same program for Filippov 3-algebroids.

Let $(F,[\cdot,\cdot,\cdot],a)$ be a Filippov 3-algebroid over a manifold $M$. The Lie derivative with respect to a wedge product of sections $s_1 \wedge s_2 \in \Gamma(\wedge^2 F)$ is defined by its action on functions and sections by
\eq{\label{Lie-3-der}
{\cal L}_{s_1 \wedge s_2}(f) := \, a(s_1 \wedge s_2)(f)\;, \quad {\cal L}_{s_1 \wedge s_2} \,s_3 :=\, \left[s_1,s_2,s_3\right] \;.
}
By using a generalization of the Lie algebroid differential \eqref{algebroiddiff} of section \ref{ch-math} it is also possible to get the action of the Lie derivative on dual sections. We call the Lie derivative of \eqref{Lie-3-der} the Lie 3-derivative. It has the following property for sections $s, s_i, e_i \in \Gamma(F)$, which has the same pattern as for the standard Lie derivative with respect to vector fields:
\eq{
[{\cal L}_{s_1 \wedge s_2}, {\cal L}_{e_1 \wedge e_2}]s = \, {\cal L}_{[s_1,s_2,e_1]\wedge e_2} s + {\cal L}_{e_1 \wedge [s_1,s_2,e_2]}s \;.
}
As a next step let us introduce the notion of covariant derivative in the setting of Filippov 3-algebroids. It should be an operator acting into the direction of a wedge product of sections:
\eq{ \label{3-covder}
\nabla : \quad \Gamma(F\wedge F)\times \Gamma(F) \rightarrow \Gamma(F)\;,
}
and in analogy to the standard or Lie algebroid case it should have the following properties for sections $s_i, e_i$ and functions $f\in {\cal C}^\infty(M)$:
\begin{itemize}
\item $\nabla_{s_1 \wedge s_2 + e_1 \wedge e_2} = \nabla_{s_1 \wedge s_2} + \nabla_{e_1 \wedge e_2}$.
\item $\nabla_{s_1 \wedge s_2}(e_1 + e_2) = \nabla_{s_1 \wedge s_2} e_1 + \nabla_{s_1 \wedge s_2} e_2$.
\item $\nabla_{f\, s_1 \wedge s_2} =\, f\,\nabla_{s_1 \wedge s_2}  $.
\item $\nabla_{s_1 \wedge s_2} (f\,s_3) =\, a(s_1 \wedge s_2)(f)\, s_3 + f\,\nabla_{s_1 \wedge s_2} s_3 $.
\end{itemize}
Having this definition, it is also possible to define torsion and curvature operators. As mentioned in section \ref{ch-math}, one has to check the tensoriality of the operators, meaning ${\cal C}^\infty(M)$-linearity in every argument. In the case of Filippov 3-algebroids and the covariant derivative \eqref{3-covder} one has to find the appropriate combinations of derivatives to get the tensoriality properties by using the Leibniz rule and anchor property of definition \ref{Filippov3algebroid}. First, let us define torsion $T \in \Gamma(\wedge^3 F^* \otimes F)$. For sections $s_1, s_2, s_3 \in \Gamma(F)$ we set:
\eq{\label{3torsion}
T(s_1, s_2, s_3) :=\, \nabla_{s_1 \wedge s_2 } s_3 + \nabla_{s_2 \wedge s_3} s_1 + \nabla_{s_3 \wedge s_1} s_2 - [s_1,s_2,s_3] \;.
}
The last expression clearly is tensorial in each of its entries because of the properties of the covariant derivative and the Leibniz rule \eqref{3Leibniz}. Furthermore it is totally antisymmetric in its arguments and therefore it has analogous properties to the standard torsion.

For the curvature operator, we have to find the right combinations of covariant derivatives to cancel terms coming from the commutator of two such derivatives. As we will show, the following operator $R \in \Gamma(\otimes^4 F^* \otimes F^* \otimes F)$ has the right tensoriality properties for sections $s_i \in \Gamma(F)$:
\eq{\label{3curvature}
R(s_1,s_2,s_3,s_4) =\, &\left[\nabla_{s_1 \wedge s_2}, \nabla_{s_3 \wedge s_4} \right] - \nabla_{[s_1,s_2,s_3]\wedge s_4} - \nabla_{s_3 \wedge [s_1,s_2,s_4]} \\
-&\left[\nabla_{s_1 \wedge s_4},\nabla_{s_2 \wedge s_3}\right] + \nabla_{[s_1,s_4,s_2]\wedge s_3} + \nabla_{s_2 \wedge [s_1,s_4,s_3]} \\
+&\left[\nabla_{s_2 \wedge s_4},\nabla_{s_3 \wedge s_1} \right] - \nabla_{[s_2,s_4,s_3]\wedge s_1} - \nabla_{s_3 \wedge [s_2,s_4,s_1]} \;.
}
To be a proper curvature operator, the right hand side of this expression should be a section in the endomorphism bundle of $F$ or in other words it should map sections of $F$ into sections of $F$. For this it suffices to prove ${\cal C}^\infty(M)$-linearity, i.e. for a section $e \in \Gamma(F)$ we should have $R(s_1,s_2,s_3,s_4)(f\,e) = f\,R(s_1,s_2,s_3,s_4)(e)$. However, this is easy to show by noting for sections $s_i, e$ and functions $f$:
\eq{
\left[\nabla_{s_1 \wedge s_2}, \nabla_{s_3 \wedge s_4} \right](f\,e) =\, \left[a(s_1 \wedge s_2), a(s_3 \wedge s_4) \right](f)\,e + f\,\left[\nabla_{s_1 \wedge s_2}, \nabla_{s_3 \wedge s_4} \right] \;.
}
The first summand on the right can be rewritten by using the defining property \eqref{3Leibniz}. We get the following form :
\eq{
\left[a(s_1 \wedge s_2), a(s_3 \wedge s_4) \right](f)\,e = a\left( [s_1,s_2,s_3]\wedge s_4 + s_3 \wedge [s_1,s_2,s_4] \right)(f) \,e \;.
}
But this shows that every line in \eqref{3curvature} separately is an endomorphism and we proved the first property which one demands on a curvature operator.

In addition to that, a curvature operator should be ${\cal C}^\infty(M)$-linear in each of its entries, e.g. for the first entry one has to demand $ R(fs_1,s_2,s_3,s_4) = f\,R(s_1,s_2,s_3,s_4)$. The proof is a lengthy but straightforward calculation using the fourth property of the covariant derivative introduced above and the Leibniz rule of the 3-bracket \eqref{3Leibniz}.

To sum up, we see that the notion of Filippov 3-algebroid allows to generalize the most important constructions of differential geometry to the case of a vector bundle with 3-bracket. The next step would be the introduction of a metric on $F$, metric compatible connections and the notion of Levi-Civita connection which finally would make it possible to formulate the analogue of an Einstein-Hilbert action for Filippov 3-algebroids. Similar to earlier sections, it would be interesting to generalize also gauge transformations, e.g. of the NS-NS $B$-field and formulate corresponding string actions.

\subsubsection{Future directions and concluding remarks}
As sketched in the last subsection, the notion of Filippov 3-algebroid allows to set up the main constructions of differential geometry and is therefore well-suited to generalize gravity and (low energy) bosonic string theory to the case of manifolds with 3-brackets. In addition, the description of M-branes suggests that this structure could be important in a formulation of M-theory suitable for quantization questions.

However, to realize a concrete calculus, one has to choose a definite 3-algebroid or in other words one has to choose a three-bracket. The first guess would be the generalization of the Koszul bracket to three entries given in \eqref{3-Koszulbracket}. But it is possible to show that a manifold equipped with a three-vector $\gamma \in \Gamma(\wedge^3 TM)$ and the latter generalization of the Koszul bracket is hard to find because the defining properties \ref{Filippov3algebroid} impose very strong algebraic and differential conditions on $\gamma$. As an example \cite{1026.17006}, demanding $\gamma$ only to be a Nambu-Poisson three-tensor is not sufficient. Thus it may be possible that the choice of the Koszul 3-bracket (or even a twisted version of it) does not lead to non-trivial geometries.

Clearly there are at least two possible ways to avoid the latter problem. The first would be to choose another three-bracket than \eqref{3-Koszulbracket}. The notion of a Filippov 3-algebroid is completely general and thus the total space of the bundle on which the 3-bracket (and the differential calculus) is defined needs not to be the cotangent bundle or the tangent bundle of the manifold. The correct choice of the total space and thus the three bracket is still ongoing work. Secondly we remark that the choice of the fundamental identity \eqref{fundamentalidentity} is connected to the notion of three-associativity if one is interested in quantization questions \cite{Curtright:2002fd}. The correct definition of the generalization of associativity to three-structures is to our knowledge an open problem. Choosing another type of fundamental identity could allow for non-trivial geometries even by using the simple generalization \eqref{3-Koszulbracket} of the Koszul bracket.

Once having found the correct Filippov 3-algebroid and the corresponding differential calculus (which then directly follows as discussed in the previous subsection), the next question would be about the quantization of such a structure. Setting up star-products, analyzing the corresponding associativity requirements or even trying to find a covariant star-product calculus using the differential geo\-me\-try of the previous subsection are only the first ideas in this interesting direction.

\paragraph{Acknowledgments}
This article is based on the author's PhD thesis, written at the Max Planck Institute for Physics. The author wants to thank his supervisors Ralph Blumenhagen and Dieter L\"ust for their invaluable support and continuous encouragement. Furthermore he wants to thank Erik Plauschinn, Felix Rennecke and Christian Schmid for collaboration and David Andriot, Peter Patalong for discussion. Finally the author wants to thank his parents Anneliese and Friedrich Deser. 

\newpage

\appendix

\section{The Rogers dilogarithm}
In this appendix, we collect some of the definitions and important properties of the Rogers dilogarithm function. This beautiful object has lots of applications in mathematics and string theory, reaching from scattering amplitudes to invariants of 3-geometries.
For a more detailed analysis of the mathematical aspects of the Rogers dilogarithm function we would like to refer the reader to  \cite{0744.33011, 0894.11052}, whereas its generalization is described in detail in \cite{0902.57013, 1053.57010}.


\subsection{Definition and fundamental properties}

The Rogers dilogarithm function $L(x)$ for real arguments $x$ is defined in the following way:
\begin{equation}
\label{rogers}
L(x) :=\textrm{Li}_2(x) + \frac{1}{2}\, \log(x)\log(1-x) \;,  \hspace{40pt} 0<x<1 \;,
\end{equation}
where $\textrm{Li}_2(x)$ denotes the Euler dilogarithm function, given by:
\begin{equation}
\label{euler}
\textrm{Li}_2(x) := \sum_{n=1}^{\infty} \:\frac{x^n}{n^2} = -\int_0 ^x \frac{\log(1-y)}{y} \;,
\hspace{40pt} 0\leq x \leq 1 \;.
\end{equation}
With the help of \eqref{euler}, the integral representation of the Rogers dilogarithm can be deduced as:
\begin{equation}
\label{integraldarstellung}
L(x)= -\frac{1}{2} \int_0 ^x \left( \frac{\log(1-y)}{y} + \frac{\log(y)}{1-y}\right) dy \;.
\end{equation}
Furthermore, from these definitions one can derive two functional relations, which in turn uniquely characterize the Rogers dilogarithm function:
\eq{
\label{app_Lrel_eral}
&L(x) + L(1-x) = L(1) \;, \\
& L(x) - L(y) + L\bigl({\textstyle \frac{y}{x}}\bigr) - L \bigl( \textstyle{\frac{1-x^{-1}}{1-y^{-1}}} \bigr)
+ L\bigl( \textstyle{\frac{1-x}{1-y}} \bigr) = 0 \;.
}

Employing the integral representation \eqref{integraldarstellung}, one can analytically continue $L(x)$ to the domain $\mathbb{C} \setminus \{0,1\}$. However, the resulting function $L(z)$ is not single valued any more and one should use the universal cover of $\mathbb{C} \setminus \{0,1\}$ as the domain of definition.
For the complex Rogers dilogarithm  the relation
\eq{
\label{complfund}
  L(z)+L(1-z)=L(1)
}
still holds, but the five-term relation in \eqref{app_Lrel_eral} receives logarithmic corrections. The systematics of those corrections can be described by the following generalization which is due to Neumann \cite{0902.57013}:
\begin{equation}
\label{NeumannRogers}
R(z\,;p,q) := L(z) -\frac{\pi^2}{6} + \frac{\pi i }{2}\Bigl(p \log(z-1) + q \log
z \Bigr)\; .
\end{equation}
Here, $p,q$ are integer numbers and the constant is just a convenient normalization. In the main text, we will not need the complicated properties of this function. However, in the investigation of Virasoro-Shapiro amplitudes in the $\textrm{CFT}_H$-framework, as performed in \cite{Blumenhagen:2011ph}, its properties play an important role.


\section{Tachyon correlation functions}
In this appendix we review standard results on scattering theory of tachyon vertex operators which are used in the main text. For a detailed introduction into basic concepts of quantum field theory, conformal field theory and further information on tachyon vertex operators we refer the reader to the standard literature (e.g. \cite{Polchinski:1998rq, DiFrancesco:1997nk} and references therein). We first consider the case of the free theory and give the results which are important for detection of open string non-commutativity in section \ref{ch-quant}.  In a second section we focus on a special interacting case which was used in section \ref{sec-cft}.

\subsection{The free case}
Consider the theory of $n$ free bosonic fields $X^i(z,\bar z)$, interpreted as the embedding of a 2-dimensional world sheet into flat target space. The action is given by:
\eq{\label{freeactionA1}
{\cal S} = \, \frac{1}{2\pi \alpha'}\int_\Sigma \, d^2 z\, \delta_{ij}\, \partial X^i \bar \partial X^j \;.
}
If we specialize to the sphere $\Sigma = \mathbb{P}^1$, the propagator of the theory is given by:
\eq{\label{freepropA1}
\langle X^i(z_1,\bar z_1)X^j(z_2,\bar z_2) \rangle =\, -\frac{\alpha'}{2} \delta^{ij} \ln |z_1 - z_2|^2 \;.
}
The energy-momentum tensor can be constructed by using the \mbox{(anti-)holomorphic} currents $J^i(z) = i\partial X^i (z)$ and $\bar J^i (\bar z) = i \bar \partial X^i(\bar z)$. We denote the normal ordering of field operators by $: \dots :$. The holomorphic component of the energy momentum tensor is:
\eq{\label{freeEMtensor}
T(z) =\, \frac{1}{\alpha'} \delta_{ij}:\!J^i J^j\!:(z) \;,
}
and similar for the anti-holomorphic component. Finally, by using the contraction of a momentum vector $p_i$ with the coordinates $X^i$ which we denote by $p \cdot X$,  the definition a \emph{tachyon vertex operator} carrying momentum $p$ is as follows:
\eq{\label{tachyonvertexA1}
V_p(z,\bar z) \stackrel{\textrm{def}}{=}\, :\!e^{i\,p\cdot X(z,\bar z)}\!: \;.
}
To begin with, we want to show that tachyon vertex operators are conformal fields and determine their conformal dimension by calculating the operator product expansion with the energy momentum tensor \eqref{freeEMtensor}. Only the singular parts are important and we write $\simeq$ if an equation holds up to regular terms.
\eq{
T(z_1)V_p(z_2) =&\, \frac{1}{\alpha'} \sum_{n=0}^\infty \frac{1}{n!} :\! \partial X^i(z_1)\partial X^j(z_1) \!: \delta_{ij}:\left[i\,p\cdot X(z_2,\bar z_2)\right]^n\!: \\
\simeq &\, \frac{\alpha'}{4} \sum_{n=2}^\infty \frac{1}{(n-2)!} \frac{1}{(z_1 - z_2)^2} p^2 :\!\left[i\,p\cdot X(z_2,\bar z_2) \right]^{n-2}\!: \\
&\,+ \sum_{n=1}^\infty \frac{1}{(n-1)!}\frac{1}{z_1 - z_2} i\,p_i :\!\partial X^i(z_2)\left[i\,p\cdot X(z_2,\bar z_2)\right]^{n-1} \!: \\
\simeq &\, \frac{\alpha'}{4} p^2 \frac{V_p(z_2,\bar z_2)}{(z_1 - z_2)^2} + \frac{\partial_{z_2} V_p(z_2, \bar z_2)}{z_1 - z_2} \;.
}
In the second step, the first summand is the result of taking two Wick-contractions with $:\left[i\,p\cdot X(z_2,\bar z_2)\right]^n\!:$ and the second summand results by taking one contraction. Thus the vertex operator is primary with conformal dimension $h = \alpha '\,p^2/4$. Si\-mi\-lar results hold for the OPE with the anti-holomophic component of the energy momentum tensor.

Let us now compute the OPE of two vertex operators with momenta $p_1$ and $p_2$. To simplify the calculation, we first note useful identities for the calculation. First, let $A(z_1), B(z_2)$ be field operators depending on the world sheet coordinates. Then we have the following result for the contraction with exponentials:
\eq{\label{contractionexp}
\contraction{}{A(z_1)}{}{:\! e^{B(z_2)}\!:}
\contraction{A(z_1)\,:\!e^{B(z_2)}\!: \;=\;}{A(z_1)}{}{B(z_2)}
A(z_1)\,:\!e^{B(z_2)}\!: \;=\; A(z_1)B(z_2)\,:\!e^{B(z_2)}\!: \;,
}
which is proven by expanding the second exponential. Using this, we can compute the OPE between two normal ordered exponentials:
\eq{
:\! e^{A(z_1)}\!:\;:\!e^{B(z_2)}\!: \;=&\, \sum_{n=0}^\infty \frac{1}{n!} :\!A(z_1)^n\!:\,:e^{B(z_2)}\!: \\
\simeq &\,\sum_{n=0}^\infty \frac{1}{n!} \sum_{k =0} ^n
\begin{pmatrix}
n \\
k
\end{pmatrix}
\left[
\contraction{}{A(z)}{}{B(z)}
A(z_1)B(z_2)
\right]^k\,:\!A(z_2)^{n-k} e^{B(z_2)} \!: \\
\simeq &\, e^{
\contraction{}{A(z_1)}{}{B(z_2)}
A(z_1) B(z_2)
} \,:\!e^{A(z_2)} e^{B(z_2)}\!: \;.
}
This result is general and independent of the specific propagator of the theory. Thus, it can be applied also in the special limits which were encountered in section \ref{ch-quant}. Specializing to the propagator \eqref{freepropA1}, this leads to the following OPE of two vertex operators of the type \eqref{tachyonvertexA1} with momenta $p_1$ and $p_2$:
\eq{
V_{p_1} (z_1) V_{p_2}(z_2) \simeq |z_1 - z_2|^{-\alpha' \, p_1 \cdot p_2} \, V_{p_1 + p_2} \;.
}
With this expansion, the $n$-point tachyon correlator on the sphere can now be readily calculated by iterated application of the previous formulas:
\eq{\label{ntachyoncorA1}
\bigl\langle \prod_{i=1} ^n \,:\!e^{i\,p_i \cdot X(z_i,\bar z_i)}\!: \bigr\rangle =\, \prod_{1\leq i<j \leq n} |z_i - z_j|^{\alpha' p_i \cdot p_j} \, \delta (\sum_{i=1}^n p_i) \;.
}
The delta function ensures momentum conservation: The action \eqref{freeactionA1} is invariant under the shift symmetry $X \mapsto X +a$ and therefore correlation functions should also be invariant under this symmetry. However, the correlator \eqref{ntachyoncorA1} obtains a phase $ e^{i\,\left(\sum_{i=1}^n p_i\right)\cdot a}$, which has to vanish. But the latter is equivalent to the momentum conservation condition because the field $a$ can take arbitrary values.

\subsection{Interacting case}
Let us now treat one special case of an interacting theory, which we used in the main text. We consider an action ${\cal S} = {\cal S}_0 + {\cal S}_1$, where we treat ${\cal S}_1$ as a perturbation:
\eq{
{\cal S}_0 = \, \frac{1}{2\pi \alpha'}\int_\Sigma \, d^2 z\, \delta_{ij}\, \partial X^i \bar \partial X^j\;, \quad {\cal S}_1 =\, \frac{1}{2\pi \alpha'} \frac{H_{abc}}{3} \int_\Sigma \; d^2z\, X^a\partial X^b \bar \partial X^c \;.
}
As detailed in the main text (see for example \eqref{pertexp}), if we are interested in correlation functions up to linear order in the $H$-flux, we only have to consider the following two contributions:
\eq{\label{perturbationA1}
\langle {\cal O}_1 \cdots {\cal O}_n \rangle = \langle {\cal O}_1 \cdots {\cal O}_n \rangle_0 - \langle {\cal O}_1 \cdots {\cal O}_n {\cal S}_1 \rangle_0 \;,
}
where $\langle \cdots \rangle_0$ denotes the correlator with respect to the free theory. Let us now consider tachyon vertex operators in the $\textrm{CFT}_H$-context, introduced in section \ref{subsec-tachyonvertex}:
\eq{
{\cal V}_i \stackrel{\textrm{def}}{=} {\cal V}_{p_i}(z_i,\bar z_i)  \stackrel{\textrm{def}}{=}\, :\!e^{i\, p_i \cdot {\cal X}(z_i,\bar z_i)}\!: \;,
}
where the ${\cal X}$-fields were defined by \eqref{rightcoords}. We are interested in the correlation function of $N$ such operators. This  will be denoted by $\langle {\cal V}_1 \cdots {\cal V}_N \rangle^\mp$, where the superscript ``$\mp$'' is to indicate whether we are in the $H$-flux background or its complete T-dual $R$-flux background, as was explained in the main text. According to \eqref{perturbationA1}, we have to consider two contributions. Let us begin with the second contribution, which is already linear in the flux and therefore we can replace the fields ${\cal X}(z,\bar z)$ by their free analogues $X(z,\bar z)$ and as a consequence we only have to consider the free vertex operators $:\!e^{i\,p_i \cdot X(z_i,\bar z_i)}\!:$, which we want to denote by $V_i$ in the following. Using contraction techniques of the previous section, we get:
\eq{ \label{firstpartA1}
 &\bigl\langle {\cal V}_1 \ldots {\cal V}_N \:(-\mathcal S_1) \bigr\rangle^\mp_0
  = \bigl\langle {V}_1 \ldots {V}_N \bigr\rangle^\mp_0  \times \\
  &\hspace{30pt} \left( \sum_{1\leq i <  j <  k\leq N}
 (-i)\:p_{i, a} \,p_{j,b}\, p_{k,c} \, \bigl\langle X^a(z_i,\ov z_i) X^b(z_j,\ov z_j) X^c(z_k,\ov z_k)
  \bigr\rangle^\mp \right)\;,
}
where we used contractions of the type \eqref{contractionexp}, which give the same result as calculating the perturbed three-$X$ correlator linear in the $H$-flux, to get the second factor. Now, let us turn to the first part in \eqref{perturbationA1}, which is the free correlator of the vertex operators ${\cal V}_i$. We recall, that the fields ${\cal X}$ are given by
\eq{
{\cal X}^i(z,\bar z) =\, X^i(z,\bar z) + \frac{1}{2}H^i{}_{jk} X_L^j(z)X_R^k(\bar z) \;,
}
where the free field was decomposed into $X^i(z,\bar z) = X^i_L(z) + X^i_R(\bar z)$. Expanding now the exponentials ${\cal V}_i$ up to linear order in the $H$-flux and collecting only terms linear in the flux, we get:
\eq{
\hspace{-8pt}\langle {\cal V}_1 \cdots {\cal V}_N \rangle^\mp_0 =& \, \langle V_1 \cdots V_N\rangle^\mp_0 \\
 &\,+\frac{1}{2}ip_{1,i_1} H^{i_1}{}_{j_1 k_1} \langle:\!e^{i\,p_1 \cdot X(z_1,\bar z_1)}  X_L^{j_1}(z_1) X_R^{k_1}(\bar z_1)\!: \times \\
&\,\times :\! e^{i\,p_2 \cdot X(z_2,\bar z_2)}\!: \cdots :\! e^{i\,p_N \cdot X(z_N,\bar z_N)}\!:\rangle_0^\mp \, + \; \dots \\
&\, + \frac{1}{2} i\, p_{N,i_N} H^{i_N}{}_{j_N k_N} \langle :\! e^{i\,p_1 \cdot X(z_1,\bar z_1)}\!:\times \dots \\
&\,\times \,:\! e^{i\,p_{N-1} \cdot X(z_{N-1},\bar z_{N-1})}\!::\! e^{i\,p_N \cdot X(z_N,\bar z_N)} X_L^{j_N}(z_N)X_R^{k_N}(\bar z_N)\!:\rangle_0 ^\mp \;.
}
Combining these terms with the result \eqref{firstpartA1}, we observe, that the same terms appear as if one expands the perturbed correlator $\langle {\cal X}^a(z_1,\bar z_1) {\cal X}^b(z_2,\bar z_2){\cal X}^c(z_3,\bar z_3)\rangle^\mp$ up to linear order in the $H$-flux and takes the product with the appropriate momenta. As a consequence, we can write the $N$-point tachyon correlator as
\eq{\label{ntachyonresultA1}
  \bigl\langle {\cal V}_1 \ldots {\cal V}_N \bigr\rangle^\mp &=
  \bigl\langle {V}_1 \ldots {V}_N \bigr\rangle^\mp_0 \:\biggl[ 1 +
  \!\!\!\! \sum_{1\leq i <  j <  k\leq N}  \!\!\!\!
 (-i)\:p_{i, a} \,p_{j,b}\, p_{k,c} \times \\[-0.2cm]
 & \hspace{143pt}
 \bigl\langle {\cal X}^a(z_i,\ov z_i)\, {\cal X}^b(z_j,\ov z_j) \,{\cal X}^c(z_k,\ov z_k)
  \bigr\rangle^\mp \biggr] \\[0.1cm]
 &=
 \bigl\langle {V}_1 \ldots {V}_N \bigr\rangle^\mp_0 \:\biggl[ 1 -i\hspace{0.5pt}
   \theta^{abc}\!\!\!\! \sum_{1\leq i<j<k\leq N}  \!\!\!\!
    p_{i ,a} \,p_{j,b}\, p_{k,c}
   \Bigl[ {\cal L}\bigl({\textstyle \frac{z_{ij}}{ z_{ik}}}\bigr)
  -  {\cal L}\bigl({\textstyle \frac{\ov z_{ij}}{ \ov z_{ik}}}\bigr)
  \Bigr] \biggr] \;,
}
where in the last step we inserted the basic three-point correlator \eqref{3coordinatecor} from the main text.

The same calculation can be done for the case of the $R$-flux background with the only difference that there is the sum of the Rogers dilogarithm terms instead of the difference in \eqref{ntachyonresultA1}. We only kept linear terms in the $H$-flux (or equivalently in the parameter $\theta$) and the results suggest that they could be seen as the beginning of the Taylor series expansion of the exponential function. We therefore rewrite the last expression as follows, the $[\dots]_\theta$ indicating that we only mean up to linear order in the $\theta$-parameter:
\eq{\label{NtachyonresultA12}
\bigl\langle {\cal V}_1 \ldots {\cal V}_N \bigr\rangle^\mp
  &=
  \bigl\langle {V}_1 \ldots {V}_N \bigr\rangle^\mp_0 \: \exp \biggl[  \:
   -i \hspace{0.5pt} \theta^{abc} \!\!\!\! \sum_{1\leq i<j<k\leq N} \!\!\!\!  p_{i ,a} \,p_{j,b}\, p_{k,c}
   \Bigl[ {\cal L}\bigl({\textstyle \frac{z_{ij}}{ z_{ik}}}\bigr)
  \mp  {\cal L}\bigl({\textstyle \frac{\ov z_{ij}}{ \ov z_{ik}}}\bigr)
  \Bigr] \biggr]_{\theta} \;.
}
These results will be used in the main text to investigate the possibility of a three- or even $N$-product structure for functions on the target space in the case of the $R$-flux. Note that for this task, the exponential structure of the tachyon vertex operators is particular useful because one can evaluate easily star products on exponential functions. Although being not physical, tachyons are therefore ideal probes for the algebra of functions on the target space.

\section{Notation conventions}
\begin{longtable}{ll}
\label{notations}
$\beta^\sharp$ & anchor map \\
&$\beta^\sharp : T^*M \rightarrow TM\;, \quad \beta^\sharp(\xi)(\eta) = \beta(\xi,\eta)$ \\
 & \\
$\gamma^\sharp$ & 3-anchor map \\
& $\gamma^\sharp : \wedge^2 T^*M \rightarrow TM, \quad \gamma^\sharp(\xi \wedge \eta)(\zeta) = \gamma(\xi, \eta, \zeta )$ \\
& \\
$[\cdot,\cdot]_+$ & anti-commutator of operators \\
 & \\
$[\cdot,\cdot]$ & commutator of operators \\
& \\
$\langle \cdot,\cdot \rangle_\pm$ & bilinear forms on $TM \oplus T^*M$ \\
& $\langle X+\xi,Y+\eta\rangle_\pm = \xi(Y)\pm \eta(X)$ \\
& \\
$[\cdot,\cdot,\cdot]$ & 3-Koszul bracket, e.g. \eqref{3-Koszulbracket} \\
 & \\
$\iota_X ,\, \iota_{\xi}$ & insertion map, \hspace{3pt} $\iota_{\partial_i}\, dx^j = \delta^j _i , \, \iota_{dx^i} \partial_j = \delta^i _j$ \\
 & \\
$\textrm{Li}_2(x)$ & Euler dilogarithm function \\
 &\\
$d$ & exterior differential \\
 & \\
$d_{\beta}$ & exterior differential on vector fields \\
 & \\
$\mathfrak{J}(\cdot,\cdot,\cdot)$ & Jacobiator \\
 & \\
$[\cdot,\cdot]\ks$ & Koszul-Schouten bracket \\
 &$[\xi,\eta]\ks = \left(\xi_m D^m \eta_n - \eta_m D^m \xi_n - \xi_m \eta_k Q_n{}^{mk}\right)$ \\
 & \\
$[\cdot,\cdot]\ks ^H$ & $H$-twisted Koszul-Schouten bracket \\
& $[\xi,\eta]\ks ^H = [\xi,\eta]\ks +\iota_{\beta^\sharp(\eta)}\iota_{\beta^\sharp(\xi)} H$ \\
& \\
$\klie$ & Koszul-(quasi) Lie derivative \\
 & \\
$[\cdot,\cdot]_L$ & Lie bracket \\
 &$[X,Y]_L = \left(X^m \partial_m Y^n - Y^m\partial_m X^n\right)\partial_n $\\
 & \\
$[\cdot,\cdot]_L ^H $ & $H$-twisted Lie bracket \\
&$[X,Y]_L ^H = [X,Y]_L -\beta^\sharp(\iota_Y \iota_X H)$ \\
& \\
$L_X$ & Lie derivative \\
 & \\
 $g$ & metric on tangent bundle \\
 & \\
 $ \fa g$ & metric on cotangent bundle \\
 & \\
$: \dots : $ & normal ordering \\
& \\
$D^a$ & partial derivative on $T^* M $, $D^a f = \beta^{am}\partial_m$ \\
 & \\
$\sigma^\mu$ & Pauli matrices, \\
& $\sigma^0 = \begin{pmatrix}
-1 & 0 \\
0 & -1
\end{pmatrix}\;, \quad \sigma^1 = \begin{pmatrix}
\;0 & \;1\; \\
\;1 & \;0\;
\end{pmatrix}\;,$ \\
& $\sigma^2 = \begin{pmatrix}
\;0 & \;-i\;\, \\
\;i & \;0\;\,
\end{pmatrix}\;, \quad \sigma^3 = \begin{pmatrix}
\;1 & \;0\; \\
\;0 & -1\;
\end{pmatrix} $ \\
& \\

$\{\cdot,\cdot \}$ & Poisson bracket \\
& \\
$\gl \cdot,\cdot \gr $ & Roytenberg bracket, e.g.  \eqref{Roytenbergbracket}\\
 & \\
$L(x)$ & Rogers dilogarithm function \\
 & \\
$[\cdot, \cdot]_{SN}$ & Schouten-Nijenhuis bracket \\

\end{longtable}



\clearpage




\bibliography{references}
\bibliographystyle{utphys}

\end{document}